\documentclass[fleqn,reqno,11pt,fleqn]{amsart}
\usepackage[T1]{fontenc}
\usepackage[latin9]{inputenc}
\usepackage{geometry}
\usepackage{prettyref}
\usepackage{bbding}
\usepackage{mathrsfs}
\usepackage{enumitem}
\usepackage{amstext}
\usepackage{amsthm}
\usepackage{amssymb}

\makeatletter
\numberwithin{equation}{section}
\numberwithin{figure}{section}
\theoremstyle{plain}
\newtheorem{thm}{\protect\theoremname}[section]
\theoremstyle{plain}
\newtheorem{cor}[thm]{\protect\corollaryname}
\theoremstyle{plain}
\newtheorem{prop}[thm]{\protect\propositionname}
\theoremstyle{plain}
\newtheorem{fact}[thm]{\protect\factname}
\theoremstyle{plain}
\newtheorem{lem}[thm]{\protect\lemmaname}
\theoremstyle{definition}
\newtheorem{defn}[thm]{\protect\definitionname}
\theoremstyle{remark}
\newtheorem{claim}[thm]{\protect\claimname}
\theoremstyle{plain}
\newtheorem*{thm*}{\protect\theoremname}
\theoremstyle{plain}
\newtheorem*{lem*}{\protect\lemmaname}
\theoremstyle{plain}
\newtheorem*{prop*}{\protect\propositionname}

\usepackage{algorithmic}
\usepackage{amsmath}
\usepackage{sherstov}
\usepackage{small-caption}
\usepackage{xcolor}
\usepackage{setspace}
\usepackage{centernot}
\usepackage{bbm}
\usepackage{geometry}

\newcommand{\moo}{\{-1,+1\}}

\newtheoremstyle{myplain}      {10pt}{10pt}{\itshape}{}{\scshape}{.}{.5em}{}
\newtheoremstyle{mydefinition} {10pt}{10pt}{}{}{\scshape}{.}{.5em}{}
\newtheoremstyle{myremark} {10pt}{10pt}{}{}{\itshape}{.}{.5em}{}

\@namedef{thm}{\@thm{\let \thm@swap \@gobble \th@myplain }{thm}{Theorem}}
\@namedef{lem}{\@thm{\let \thm@swap \@gobble \th@myplain }{thm}{Lemma}}
\@namedef{cor}{\@thm{\let \thm@swap \@gobble \th@myplain }{thm}{Corollary}}
\@namedef{prop}{\@thm{\let \thm@swap \@gobble \th@myplain }{thm}{Proposition}}
\@namedef{claim}{\@thm{\let \thm@swap \@gobble \th@myplain }{thm}{Claim}}
\@namedef{fact}{\@thm{\let \thm@swap \@gobble \th@myplain }{thm}{Fact}}
\@namedef{defn}{\@thm{\let \thm@swap \@gobble \th@mydefinition }{thm}{Definition}}
\@namedef{rem}{\@thm{\let \thm@swap \@gobble \th@mydefinition }{thm}{Remark}}
\@namedef{example}{\@thm{\let \thm@swap \@gobble \th@mydefinition }{thm}{Example}}

\@namedef{thm*}{\@thm {\th@myplain }{}{Theorem}}
\@namedef{lem*}{\@thm {\th@myplain }{}{Lemma}}
\@namedef{cor*}{\@thm {\th@myplain }{}{Corollary}}
\@namedef{prop*}{\@thm {\th@myplain }{}{Proposition}}
\@namedef{claim*}{\@thm {\th@myplain }{}{Claim}}
\@namedef{fact*}{\@thm {\th@myplain }{}{Fact}}
\@namedef{defn*}{\@thm {\th@mydefinition }{}{Definition}}
\@namedef{rem*}{\@thm {\th@mydefinition }{}{Remark}}
\@namedef{example*}{\@thm {\th@mydefinition }{}{Example}}

\renewcommand{\mathcal}[1]{\mathscr{#1}}

\makeatletter
\def\@seccntformat#1{%
  \protect\textup{%
    \protect\@secnumfont
    \expandafter\protect\csname format#1\endcsname 
    \csname the#1\endcsname
    \protect\@secnumpunct
  }%
}

\makeatletter
\newcommand \SparseDotfill {\leavevmode \leaders \hb@xt@ .7em{\hss .\hss }\hfill \kern \z@}
\makeatother	

\makeatletter
\def\@tocline#1#2#3#4#5#6#7{\relax
  \ifnum #1>\c@tocdepth 
  \else
    \par \addpenalty\@secpenalty\addvspace{\ifnum #1=1 2mm \else #2\fi}%
    \begingroup \hyphenpenalty\@M
    \@ifempty{#4}{%
      \@tempdima\csname r@tocindent\number#1\endcsname\relax
    }{%
      \@tempdima#4\relax
    }%
    \parindent\z@ \leftskip#3\relax \advance\leftskip\@tempdima\relax
    \rightskip\@pnumwidth plus4em \parfillskip-\@pnumwidth
          \ifnum #1=1 \bfseries #5\else #5\fi 
   \leavevmode\hskip-\@tempdima
      \ifcase #1
       \or\or \hskip 1em \or \hskip 2em \else \hskip 3em \fi%
#6     \nobreak\relax
{\ifnum #1=1\hfill \else \SparseDotfill\fi}
 \hbox to\@pnumwidth{\@tocpagenum{
    \ifnum #1=1 \bfseries \fi #7}}\par
    \nobreak
    \endgroup
  \fi}
\makeatother

\usepackage{xcolor}

\usepackage{enumitem}
\setenumerate{label=\rm (\roman*),labelsep=5mm}

\usepackage{graphicx}

\providecommand{\noopsort}[1]{}

\DeclareMathOperator{\diag}{diag}

\renewcommand{\epsilon}{\varepsilon}

\makeatletter
\newcommand{\pushright}[1]{\ifmeasuring@#1\else\omit\hfill$\displaystyle#1$\fi\ignorespaces}
\newcommand{\pushleft}[1]{\ifmeasuring@#1\else\omit$\displaystyle#1$\hfill\fi\ignorespaces}
\makeatother

\usepackage{hyperref}
\hypersetup{
plainpages = false,
bookmarks = true,
bookmarksopen = true,
pdftex,
colorlinks=false,
citecolor = teal,
linkcolor = teal,
urlcolor  = teal,
}
\newrefformat{fact}{Fact~\ref{#1}}
\newrefformat{prop}{Proposition~\ref{#1}}
\newrefformat{lem}{Lemma~\ref{#1}}
\newrefformat{cor}{Corollary~\ref{#1}}
\newrefformat{thm}{Theorem~\ref{#1}}

\renewcommand{\phi}{\varphi}
\newcommand{\RANK}{\operatorname{RANK}}
\newcommand{\INTERSECT}{\operatorname{INTERSECT}}
\newcommand{\SUM}{\operatorname{SUM}}
\DeclareMathOperator{\ssign}{\widetilde{sgn}}

\DeclareMathOperator{\GL}{GL}
\DeclareMathOperator{\SL}{SL}

\DeclareMathOperator{\DET}{DET}
\DeclareMathOperator{\RANKDET}{RANKDET}

\DeclareMathOperator{\cost}{cost}

\DeclareMathOperator{\BLQ}{BLQ}
\DeclareMathOperator{\Rank}{rank}
\DeclareMathOperator{\Rankdet}{rankdet}
\DeclareMathOperator{\Det}{det}

\AtBeginDocument{
  
}

\makeatother

\providecommand{\claimname}{Claim}
\providecommand{\corollaryname}{Corollary}
\providecommand{\definitionname}{Definition}
\providecommand{\factname}{Fact}
\providecommand{\lemmaname}{Lemma}
\providecommand{\propositionname}{Proposition}
\providecommand{\theoremname}{Theorem}

\begin{document}
\title[The Communication Complexity of Approximating Matrix Rank]{The Communication Complexity of Approximating \\ Matrix Rank}
\author{Alexander A. Sherstov and Andrey A. Storozhenko}
\thanks{$^{*}$ Computer Science Department, UCLA, Los Angeles, CA~90095.
Supported by NSF grant CCF-2220232.\\
 {\large{}\Envelope ~}\texttt{\{sherstov,storozhenko\}@cs.ucla.edu }}
\begin{abstract}
We fully determine the communication complexity of approximating matrix
rank, over any finite field $\mathbb{F}$. We study the most general
version of this problem, where $0\leq r<R\leq n$ are given integers,
Alice and Bob's inputs are matrices $A,B\in\mathbb{F}^{n\times n}$,
respectively, and they need to distinguish between the cases $\rk(A+B)=r$
and $\rk(A+B)=R$. We show that this problem has randomized communication
complexity $\Omega(1+r^{2}\log|\mathbb{F}|)$. This is optimal in
a strong sense because $O(1+r^{2}\log|\mathbb{F}|)$ communication
is sufficient to determine, for arbitrary $A,B$, whether $\rk(A+B)\leq r$.
Prior to our work, lower bounds were known only for \emph{consecutive}
integers $r$ and $R$, with no implication for the approximation
of matrix rank. Our lower bound holds even for quantum protocols and
even for error probability $\frac{1}{2}-\frac{1}{4}|\mathbb{F}|^{-r/3}$,
which too is virtually optimal because the problem has a two-bit classical
protocol with error $\frac{1}{2}-\Theta(|\mathbb{F}|^{-r})$.

As an application, we obtain an $\Omega(\frac{1}{k}\cdot n^{2}\log|\mathbb{F}|)$
space lower bound for any streaming algorithm with $k$ passes that
approximates the rank of an input matrix $M\in\mathbb{F}^{n\times n}$
within a factor of $\sqrt{2}-\delta$, for any $\delta>0$. Our result
is an exponential improvement in $k$ over previous work.

We also settle the randomized and quantum communication complexity
of several other linear-algebraic problems, for all settings of parameters.
This includes the \emph{determinant problem} (given matrices $A$
and $B$, distinguish between the cases $\det(A+B)=a$ and $\det(A+B)=b$,
for fixed field elements $a\ne b)$ and the \emph{subspace sum} and
\emph{subspace intersection problem} (given subspaces $S$ and $T$
of known dimensions $m$ and $\ell$, respectively, approximate the
dimensions of $S+T$ and $S\cap T$).
\end{abstract}

\maketitle
\thispagestyle{empty}

\newpage\thispagestyle{empty}
\hypersetup{linkcolor=black} 
\setstretch{0.9}\tableofcontents{}

\setstretch{1}\newpage{}

\hypersetup{linkcolor=teal} 
\thispagestyle{empty}

\hyphenation{com-po-nent-wise}

\section{Introduction}

The exact and approximate computation of matrix rank is a fundamental
problem in theoretical computer science, studied for its intrinsic
importance as well as its connections to other algorithmic and complexity-theoretic
questions. In particular, a large body of research has focused on
the communication complexity of the matrix rank problem in Yao's two-party
model~\cite{yao79cc,yao93quantum}, with both classical and quantum
communication. In this problem, the two parties Alice and Bob receive
matrices $A,B\in\mathbb{F}^{n\times n}$, respectively, over a finite
field $\mathbb{F}$ and are tasked with determining the rank of $A+B$
using minimal communication. The first result in this line of research
was obtained three decades ago by Chu and Schnitger~\cite{chu-schnitger89matrix-integer},
who proved a lower bound of $\Omega(kn^{2})$ for the deterministic
communication complexity of computing the rank of $A+B$ when the
matrix entries are $k$-bit integers. Several years later, Chu and
Schnitger~\cite{chu-schnitger95matrix-finite} further showed that
this communication problem has deterministic complexity $\Omega(n^{2}\log p)$
when the matrix entries are in $\mathbb{F}_{p}$, the finite field
with $p$ elements. The first result on the \emph{randomized} communication
complexity of the matrix rank problem was obtained by Sun and Wang~\cite{sunwang12communication-linear},
who proved that determining whether $A+B$ is singular requires $\Omega(n^{2}\log p)$
bits of communication for matrices $A,B$ over the finite field $\mathbb{F}_{p}$
for prime $p$. In a follow-up paper, Li, Sun, Wang, and Woodruff~\cite{LiSWW14communication-linear}
showed that this $\Omega(n^{2}\log p)$ lower bound holds even for
a promise version of the matrix rank problem, where the matrix $A+B$
is guaranteed to have rank either $n-1$ or $n$. The lower bounds
of~\cite{sunwang12communication-linear,LiSWW14communication-linear}
further apply to quantum communication. 

Despite these exciting developments, no progress has been made on
lower bounds for \emph{approximating} matrix rank. Our main contribution
is the complete resolution of the approximate matrix rank problem.
In what follows, we state our results for matrix rank and several
other approximation problems, and present applications of our work
to streaming complexity.

\subsection{Matrix rank problem}

We study the problem of approximating matrix rank in its most general
form. Specifically, let $\mathbb{F}$ be any finite field. For integer
parameters $n,m,R,r$ such that $\min\{n,m\}\geq R>r\geq0,$ we consider
the promise communication problem defined on pairs of matrices $A,B\in\mathbb{F}^{n\times m}$
by
\[
\RANK_{r,R}^{\mathbb{F},n,m}(A,B)=\begin{cases}
-1 & \text{if }\rk(A+B)=r,\\
1 & \text{if }\rk(A+B)=R,\\
* & \text{otherwise,}
\end{cases}
\]
where the asterisk indicates that the communication protocol is allowed
to exhibit arbitrary behavior when $\rk(A+B)\notin\{r,R\}$. In words,
the problem amounts to distinguishing input pairs with $\rk(A+B)=r$
from those with $\rk(A+B)=R.$ The corresponding \emph{total} communication
problem is given by
\[
\RANK_{r}^{\mathbb{F},n,m}(A,B)=\begin{cases}
-1 & \text{if }\rk(A+B)\leq r,\\
1 & \text{otherwise.}
\end{cases}
\]
Clearly, the total problem $\RANK_{r}^{\mathbb{F},n,m}$ is more challenging
than the promise problem $\RANK_{r,R}^{\mathbb{F},n,m}$. Prior to
our work, the strongest known result was the $\Omega(n^{2}\log p)$
lower bound of \cite{LiSWW14communication-linear} on the bounded-error
quantum communication complexity of $\RANK_{n-1,n}^{\mathbb{F}_{p},n,n}$
for fields $\mathbb{F}_{p}$ of prime order. Unfortunately, this lower
bound has no implications for the approximation of matrix rank because
the ratio $(n-1)/n$ rapidly tends to $1$. We resolve this question
in full in the following theorem.
\begin{thm}[Lower bound for rank problem]
\label{thm:RANK-lowerbound-intro}There is an absolute constant $c>0$
such that for all finite fields $\mathbb{F}$ and all integers $n,m,R,r$
with $\min\{n,m\}\geq R>r\geq0,$
\[
Q_{\frac{1}{2}-\frac{1}{4|\mathbb{F}|^{r/3}}}^{*}(\RANK_{r,R}^{\mathbb{F},n,m})\geq c(1+r^{2}\log|\mathbb{F}|).
\]
In particular, 
\[
Q_{1/4}^{*}(\RANK_{r,R}^{\mathbb{F},n,m})\geq c(1+r^{2}\log|\mathbb{F}|).
\]
\end{thm}

In the statement above, $Q_{\epsilon}^{*}$ denotes $\epsilon$-error
quantum communication complexity with arbitrary prior entanglement,
which is the most powerful model of probabilistic computation. Clearly,
all our lower bounds apply to the randomized (classical) model as
well. Two other remarks are in order. Even in the special case of
$r=n-1$ and \textbf{$R=n$}, our result is a significant improvement
on previous work because our theorem is proved in the \emph{large-error
regime,} with the error probability exponentially close to $1/2$.
This should be contrasted with the communication lower bounds of~\cite{sunwang12communication-linear,LiSWW14communication-linear},
which were proved for error probability $1/3$. Moreover, \prettyref{thm:RANK-lowerbound-intro}
is the first result of its kind because it allows for an arbitrary
gap between $r$ and $R$. In particular, \prettyref{thm:RANK-lowerbound-intro}
shows for the first time that approximating the matrix rank to any
constant factor requires $\Omega(n^{2}\log|\mathbb{F}|)$ bits of
communication, even for protocols that succeed with exponentially
small probability (take $R=n$ and $r=cn$ for a small constant $c>0$). 

\prettyref{thm:RANK-lowerbound-intro} is optimal in a strong sense.
Specifically, we have the following matching upper bound, which we
prove by adapting Clarkson and Woodruff's streaming algorithm for
matrix rank~\cite{clarkson-woodruff09}. In the statement below,
$R_{\epsilon}$ denotes randomized (classical) communication complexity
with error $\epsilon$.
\begin{thm}[Upper bound for rank problem]
\label{thm:RANK-upperbound}There is an absolute constant $c>0$
such that for all finite fields $\mathbb{F}$ and all integers $n,m,r$
with $\min\{n,m\}>r\geq0,$
\begin{align*}
 & R_{1/3}(\RANK_{r}^{\mathbb{F},n,m})\leq c(1+r^{2}\log|\mathbb{F}|),\\
 & R_{\frac{1}{2}-\frac{1}{32|\mathbb{F}|^{r}}}(\mathrm{RANK}_{r}^{\mathbb{F},n,m})\leqslant2.
\end{align*}
\end{thm}

This result shows that the lower bound of \prettyref{thm:RANK-lowerbound-intro}
is tight not only for quantum protocols solving the partial problem
$\RANK_{r,R}^{\mathbb{F},n,m}$ but even for \emph{classical}, \emph{bounded-error}
protocols solving the \emph{total} problem $\RANK_{r}^{\mathbb{F},n,m}$.
Moreover, \prettyref{thm:RANK-upperbound} shows that the error regime
for which we prove our lower bound in~\prettyref{thm:RANK-lowerbound-intro}
is also optimal, in that the total rank problem has a classical protocol
with cost only $2$ bits and error probability $\frac{1}{2}-|\mathbb{F}|^{-\Theta(r)}$.

\prettyref{thm:RANK-lowerbound-intro} generalizes to multiparty communication,
as we discuss below in Section~\ref{subsec:Multiparty-intro}.

\subsection{Streaming complexity}

The streaming complexity of matrix rank has received extensive attention
in the literature~\cite{clarkson-woodruff09,sunwang12communication-linear,LiSWW14communication-linear,bury-schwiegelshohn15matchings,AKL17matching,AKSY20graph-streaming,CKPSSY21super-constant-pass}.
In this model, an algorithm with limited space is presented with a
matrix $M$ of order $n$ over a given field, in row-major order.
The objective is to compute or approximate the rank of $M$, using
either a single pass or multiple passes over $M$. Via standard reductions,
our \prettyref{thm:RANK-lowerbound-intro} implies an essentially
optimal lower bound on the streaming complexity of approximating matrix
rank. Unlike previous work, our result remains valid even for polynomially
many passes and even for correctness probability exponentially close
to $1/2.$ Stated in its most general form, our result is as follows.
\begin{thm}
\label{thm:RANK-streaming-intro} Let $n,r,R$ be nonnegative integers
with $n/2\leq r<R\leq n,$ and let $\mathbb{F}$ be a finite field.
Define $f\colon\mathbb{F}^{n\times n}\to\{-1,1,*\}$ by
\[
f(M)=\begin{cases}
-1 & \text{if }\rk M=r,\\
1 & \text{if }\rk M=R,\\
* & \text{otherwise.}
\end{cases}
\]
Let $\mathcal{A}$ be any randomized streaming algorithm for $f$
with error probability $\frac{1}{2}-\frac{1}{4}|\mathbb{F}|^{-(r-\lceil n/2\rceil)/3}$
that uses $k$ passes and space $s$. Then
\[
sk=\Omega\left(\left(r-\left\lceil \frac{n}{2}\right\rceil \right)^{2}\log|\mathbb{F}|\right).
\]
 
\end{thm}

\noindent By way of notation, recall that $f$ in the above statement
is a \emph{partial} function, and $\mathcal{A}$ is allowed to exhibit
arbitrary behavior on matrices $M$ where $f(M)=*$.
\begin{cor}
\label{cor:streaming}Let $\mathbb{F}$ be a finite field, and let
$\delta\in(1/2,1)$ be any constant. Let $\mathcal{A}$ be a $k$-pass
streaming algorithm that takes as input a matrix $M\in\mathbb{F}^{n\times n}$
$($for any $n\geq\frac{5}{\delta-0.5})$ such that either $\rk M=n$
or $\rk M=\lfloor\delta n\rfloor,$ and determines which is the case
with probability of correctness at least $\frac{1}{2}+|\mathbb{F}|^{-(\delta-0.5)n/5}$.
Then $\mathcal{A}$ uses $\Omega(\frac{1}{k}\cdot n^{2}\log|\mathbb{F}|)$
space.
\end{cor}

\begin{proof}
Take $R=n$ and $r=\lfloor\delta n\rfloor$ in~\prettyref{thm:RANK-streaming-intro}.
\end{proof}
The space lower bound in \prettyref{cor:streaming} is essentially
optimal since the rank of a matrix $M\in\mathbb{F}^{n\times n}$ can
be computed exactly by a trivial, single-pass algorithm with space
$O(n^{2}\log|\mathbb{F}|)$. Prior to our work, the strongest streaming
lower bound for approximating matrix rank was due to Chen et al.~\cite{CKPSSY21super-constant-pass}.
For any constants $\epsilon>0$ and $\delta>0,$ they proved that
no $o(\sqrt{\log n})$-pass algorithm with space $n^{2-\epsilon}$
can distinguish between the cases $\rk M=n$ and $\rk M\leq\delta n$
with probability $2/3$, where $M$ is an input matrix of order $n$
over a finite field of size $\omega(n)$. Our~\prettyref{cor:streaming}
shows that distinguishing between the cases $\rk M=n$ and $\rk M=\lfloor\delta n\rfloor$
requires $n^{2-\epsilon}\log|\mathbb{F}|$ space even with $k=\Theta(n^{\epsilon})$
passes, an exponential improvement on~\cite{CKPSSY21super-constant-pass}.
Moreover, \prettyref{cor:streaming} is valid for all finite fields
regardless of size, and holds even when the correctness probability
is exponentially close to $1/2$.

We now restate our streaming lower bound in more standard terminology.
Recall that an algorithm $\mathcal{A}$ with input $M\in\mathbb{F}^{n\times n}$
approximates, with probability $p$, the rank of $M$ within a factor
of $c\in[1,\infty)$\emph{ }if for every input matrix $M,$ the output
of $\mathcal{A}$ is in the range $[\frac{1}{c}\rk M,c\rk M]$ with
probability at least $p$. We have:
\begin{cor}
Let $\mathbb{F}$ be a finite field, and let $c\in[1,\sqrt{2})$ be
any constant. Let $\mathcal{A}$ be a $k$-pass streaming algorithm
with input $M\in\mathbb{F}^{n\times n}$ $($for any $n\geq\frac{40}{2-c^{2}})$
that approximates, with probability at least $\frac{1}{2}+|\mathbb{F}|^{-(2-c^{2})n/40},$
the rank of $M$ within a factor of $c$. Then $\mathcal{A}$ uses
$\Omega(\frac{1}{k}\cdot n^{2}\log|\mathbb{F}|)$ space.
\end{cor}

\begin{proof}
Define $\delta=\frac{1}{2}(\frac{1}{2}+\frac{1}{c^{2}}).$ Since $\delta<1/c^{2},$
algorithm $\mathcal{A}$ can be used to distinguish, with correctness
probability at least $\frac{1}{2}+|\mathbb{F}|^{-(2-c^{2})n/40}$,
matrices $M\in\mathbb{F}^{n\times n}$ of rank $\lfloor\delta n\rfloor$
from those of rank $n$ (simply check if $\mathcal{A}$'s output is
$<n/c$ or $\geq n/c$). The correctness probability of this distinguisher
exceeds $\frac{1}{2}+|\mathbb{F}|^{-(\lfloor\delta n\rfloor-\lceil n/2\rceil)/3}$
due to $n\geq40/(2-c^{2})$. Therefore, it uses $\Omega(\frac{1}{k}\cdot n^{2}\log|\mathbb{F}|)$
space by \prettyref{thm:RANK-streaming-intro}.
\end{proof}

\subsection{Determinant problem}

Recall that a square matrix over a field $\mathbb{F}$ has full rank
if and only if its determinant is nonzero. As a result, the problem
of computing the determinant has received considerable attention in
previous work on matrix rank, e.g.,~\cite{chu-schnitger95matrix-finite,sunwang12communication-linear,LiSWW14communication-linear}.
We are interested in the most general form of the determinant problem,
where Alice and Bob receive as input matrices $A,B\in\mathbb{F}^{n\times n}$,
respectively, and need to determine whether the determinant of $A+B$
equals $a$ or $b$. The problem parameters $a$ and $b$ are distinct
field elements that are fixed in advance. Formally, the determinant
problem is the partial\emph{ }communication problem on matrix pairs
$(A,B)$ given by
\[
\DET_{a,b}^{\mathbb{F},n}(A,B)=\begin{cases}
-1 & \text{if }\det(A+B)=a,\\
1 & \text{if }\det(A+B)=b,\\
* & \text{otherwise.}
\end{cases}
\]

Prior to our work, the strongest result on the determinant problem
was due to Sun and Wang~\cite{sunwang12communication-linear}, who
proved a tight lower bound of $\Omega(n^{2}\log|\mathbb{F}|)$ for
the randomized and quantum communication complexity of $\DET_{a,b}^{\mathbb{F},n}$
for nonzero $a,b$ over any finite field $\mathbb{F}$ of prime order.
They conjectured the same lower bound for the case of arbitrary $a,b.$
To see why the case of nonzero $a,b$ is rather special, observe that
the number of matrices with determinant $a$ is always the same as
the number of matrices with determinant $b$, with natural bijections
between these two sets; but this is no longer true if one of $a,b$
is zero. This asymmetry suggests that the determinant problem requires
a substantially different approach when one of $a,b$ is zero. In
this work, we develop sufficiently strong techniques to solve the
determinant problem in full, thereby settling Sun and Wang's conjecture
in the affirmative.
\begin{thm}
\label{thm:DET-intro}There is an absolute constant $c>0$ such that
for every finite field $\mathbb{F},$ every pair of distinct elements
$a,b\in\mathbb{F},$ and all integers $n\geq2,$
\[
Q_{\frac{1}{2}-\frac{1}{4|\mathbb{F}|^{(n-1)/3}}}^{*}(\DET_{a,b}^{\mathbb{F},n})\geq cn^{2}\log|\mathbb{F}|.
\]
\end{thm}

\noindent The communication lower bound of \prettyref{thm:DET-intro}
is best possible, up to the multiplicative constant~$c$. It matches
the trivial, deterministic protocol where Alice sends her input matrix
$A$ to Bob using $n^{2}\lceil\log|\mathbb{F}|\rceil$ bits, at which
point Bob computes $\det(A+B)$ and announces the output of the protocol.
Furthermore, the error regime in \prettyref{thm:DET-intro} is also
essentially optimal because, for example, the problem $\DET_{0,b}^{\mathbb{F},n}$
has a randomized protocol with only $2$ bits of communication and
error probability $\frac{1}{2}-\Theta(|\mathbb{F}|^{n-1})$, by taking
$r=n-1$ and $R=m=n$ in \prettyref{thm:RANK-upperbound}. Lastly,
we note that the requirement that $n\geq2$ in \prettyref{thm:DET-intro}
is also necessary because the determinant problem for $1\times1$
matrices reduces to the equality problem with domain $\mathbb{F}\times\mathbb{F}$
and therefore has randomized communication complexity $O(1)$.

We prove \prettyref{thm:DET-intro} for all $a,b$ from first principles,
without relying on the work of Sun and Wang~\cite{sunwang12communication-linear}.
In the case of nonzero $a,b$, we give a new proof that is quite short
and uses only basic Fourier analysis, unlike the rather technical
proof of \cite{sunwang12communication-linear}. To settle the complementary
case where one of $a,b$ is zero, we prove a stronger result of independent
interest. Here, we introduce a natural problem that we call $\RANKDET_{r,a}^{\mathbb{F},n}$,
which combines features of the matrix rank and determinant problems.
It is parameterized by a nonzero field element $a\in\mathbb{F}$ and
a nonnegative integer $r<n,$ and Alice and Bob's objective is to
distinguish input pairs $(A,B)$ with $\rk(A+B)=r$ from those with
$\det(A+B)=a.$ We prove the following.
\begin{thm}
\label{thm:RANKDET-intro}There is an absolute constant $c>0$ such
that for every finite field $\mathbb{F},$ every field element $a\in\mathbb{F}\setminus\{0\},$
and all integers $n>r\geq0,$ 
\[
Q_{\frac{1}{2}-\frac{1}{4|\mathbb{F}|^{r/3}}}^{*}(\RANKDET_{r,a}^{\mathbb{F},n})\geq c(1+r^{2}\log|\mathbb{F}|).
\]
\end{thm}

\noindent Taking $r=n-1$ in this result settles \prettyref{thm:DET-intro}
when one of $a,b$ is zero, as desired. \prettyref{thm:RANKDET-intro}
is optimal in a strong sense: even the \emph{total }problem $\RANK_{r}^{\mathbb{F},n,n}$,
which subsumes $\RANKDET_{r,a}^{\mathbb{F},n}$, has bounded-error
classical communication complexity $O(1+r^{2}\log|\mathbb{F}|)$ by
\prettyref{thm:RANK-upperbound}. \prettyref{thm:RANKDET-intro} for
the $\RANKDET_{r,a}^{\mathbb{F},n}$ problem significantly strengthens
our main result, \prettyref{thm:RANK-lowerbound-intro}, for the matrix
rank problem $\RANK_{r,n}^{\mathbb{F},n,n}$ (in the former problem,
Alice and Bob distinguish rank $r$ from determinant $a\ne0$; in
the latter problem, they distinguish rank $r$ from rank $n$).

Theorems~\ref{thm:DET-intro} and~\ref{thm:RANKDET-intro} generalize
to multiparty communication, as we discuss below in Section~\ref{subsec:Multiparty-intro}.

\subsection{Subspace sum and intersection problems}

There are two natural ways to recast the computation of matrix rank
as a communication problem. One approach, discussed in detail above,
is to assign matrices $A$ and $B$ to Alice and Bob, respectively,
and require them to compute the rank of $A+B.$ Alternatively, one
can require Alice and Bob to compute the rank of the matrix $\begin{bmatrix}A & B\end{bmatrix}$.
This alternative approach is best described in the language of linear
subspaces: letting $S$ and $T$ stand for the column space of $A$
and $B,$ respectively, the rank of $\begin{bmatrix}A & B\end{bmatrix}$
is precisely the dimension of the linear subspace $S+T$ generated
by $S$ and $T$. Here, we may assume that the dimensions of $S$
and $T$ are known in advance because this information can be communicated
at negligible cost.

In this way, one arrives at the \emph{subspace sum problem} over a
finite field $\mathbb{F}$, where Alice receives as input an $m$-dimensional
linear subspace $S\subseteq\mathbb{F}^{n}$ and Bob receives an $\ell$-dimensional
linear subspace $T\subseteq\mathbb{F}^{n}$. The integers $m$ and
$\ell$ are part of the problem specification and are fixed in advance.
In the promise version of the subspace sum problem, the objective
is to distinguish subspace pairs with $\dim(S+T)=d_{1}$ from those
with $\dim(S+T)=d_{2},$ for distinct integers $d_{1},d_{2}$ fixed
in advance. This corresponds to the partial function given by
\[
\SUM_{d_{1},d_{2}}^{\mathbb{F},n,m,\ell}(S,T)=\begin{cases}
-1 & \text{if }\dim(S+T)=d_{1},\\
1 & \text{if }\dim(S+T)=d_{2},\\
* & \text{otherwise. }
\end{cases}
\]
The corresponding total communication problem is that of determining
whether $S+T$ has dimension at most $d$, for an integer $d$ fixed
in advance:
\[
\SUM_{d}^{\mathbb{F},n,m,\ell}(S,T)=\begin{cases}
-1 & \text{if }\dim(S+T)\leq d,\\
1 & \text{otherwise. }
\end{cases}
\]
The total problem is more challenging than the promise problem in
that $\SUM_{d_{1},d_{2}}^{\mathbb{F},n,m,\ell}$ is a restriction
of $\SUM_{d_{1}}^{\mathbb{F},n,m,\ell}$, for any integers $d_{1}<d_{2}$.
As noted by many authors, from the standpoint of communication complexity,
computing the dimension of the subspace sum $S+T$ is equivalent to
computing the dimension of the subspace intersection $S\cap T.$ This
equivalence follows from the identity $\dim(S+T)=\dim(S)+\dim(T)-\dim(S\cap T)$.

Despite the syntactic similarity between the matrix sum $A+B$ and
the corresponding subspace sum $S+T$, the subspace sum problem appears
to be significantly more subtle and technical. Previous work has focused
on a special case that we call \emph{subspace disjointness }(determining
whether Alice and Bob's subspaces have trivial intersection, $\{0\}$)
and the dual problem that we call \emph{vector space span} (determining
if the sum of Alice and Bob's subspaces is the entire vector space).
These two problems were studied in \cite{lovasz-saks88logrank,chu-schnitger95matrix-finite},
with an optimal lower bound of $\Omega(n^{2}\log p)$ on their deterministic
communication complexity over a field with $p$ elements. Sun and
Wang~\cite{sunwang12communication-linear} showed that the $\Omega(n^{2}\log p)$
lower bound for subspace disjointness remains valid even for randomized
and quantum communication. In follow-up work, Li, Sun, Wang, and Woodruff~\cite{LiSWW14communication-linear}
proved an $\Omega(n^{2}\log p)$ quantum lower bound for a promise
version of subspace disjointness, where Alice and Bob's inputs are
$n/2$-dimensional subspaces that either have trivial intersection
or intersect in a one-dimensional subspace. The authors of~\cite{MNSW95asymmetric-communication}
considered an asymmetric problem where Alice receives an $n$-bit
vector, Bob receives a subspace, and their objective is to determine
whether Alice's vector is contained in Bob's subspace. They showed
that in any randomized one-way protocol that solves this problem,
either Alice sends $\Omega(n)$ bits, or Bob sends $\Omega(n^{2})$
bits.

In summary, all previous lower bounds for two-way communication complexity
have focused on subspace disjointness or vector space span. The general
problem, where Alice and Bob need to distinguish between the cases
$\dim(S+T)=d_{1}$ and $\dim(S+T)=d_{2}$, is substantially harder
and has remained unsolved. The difficulty is that previous results~\cite{sunwang12communication-linear,LiSWW14communication-linear}
are based on a reduction from the matrix rank problem to subspace
disjointness, and this straightforward strategy does not produce optimal
results for the subspace sum problem with arbitrary parameters. In
this paper, we approach the subspace sum problem from first principles
and solve it completely. Our solution settles both the promise version
of subspace sum and the corresponding total version. For clarity,
we first state our result in the regime of constant error.
\begin{thm}
\label{thm:MAIN-subspace-sum-constant-error}Let $\mathbb{F}$ be
a finite field with $q=|\mathbb{F}|$ elements, and let $n,m,\ell,d,D$
be nonnegative integers with $\max\{m,\ell\}\leq d<D\leq\min\{m+\ell,n\}.$
If $m=\ell=d,$ then
\begin{align*}
 & R_{1/3}(\SUM_{d}^{\mathbb{F},n,m,\ell})=O(1).
\end{align*}
If $m,\ell,d$ are not all equal, then
\begin{align*}
 & Q_{1/3}^{*}(\SUM_{d,D}^{\mathbb{F},n,m,\ell})=\Theta((d-m+1)(d-\ell+1)\log q),\\
 & R_{1/3}(\SUM_{d}^{\mathbb{F},n,m,\ell})=\Theta((d-m+1)(d-\ell+1)\log q).
\end{align*}
\end{thm}

Several remarks are in order. Recall that in $\mathbb{F}^{n}$, the
sum of an $m$-dimensional subspace and an $\ell$-dimensional subspace
has dimension between $\max\{m,\ell\}$ and $\min\{m+\ell,n\}$. This
justifies the above requirement that $d,D\in[\max\{m,\ell\},\min\{m+\ell,n\}]$.
\prettyref{thm:MAIN-subspace-sum-constant-error} shows that the promise
version of the subspace sum problem has the same communication complexity
as the total version, up to a constant factor. Moreover, the theorem
shows that this communication complexity is the same, up to a constant
factor, for quantum and classical communication protocols. Both the
lower and upper bounds in \prettyref{thm:MAIN-subspace-sum-constant-error}
require substantial effort. Lastly, the degenerate case $d=m=\ell$
of the subspace sum problem is easily seen to be equivalent to the
equality problem, which explains the $O(1)$ bound in the theorem
statement. 

In addition to the constant-error regime of \prettyref{thm:MAIN-subspace-sum-constant-error},
we are able to determine the communication complexity of subspace
sum for essentially all settings of the error parameter, as follows.
\begin{thm}
\label{thm:MAIN-subspace-sum-general}Let $\mathbb{F}$ be a finite
field with $q=|\mathbb{F}|$ elements, and let $n,m,\ell,d,D$ be
nonnegative integers with $\max\{m,\ell\}\leq d<D\leq\min\{m+\ell,n\}.$
If $m=\ell=d,$ then
\begin{align*}
 & R_{1/3}(\SUM_{d}^{\mathbb{F},n,m,\ell})=O(1).
\end{align*}
If $m,\ell,d$ are not all equal, then for all $\gamma\in[\frac{1}{3}q^{-(2d-m-\ell)/5},\frac{1}{3}],$
\begin{align*}
Q_{\frac{1-\gamma}{2}}^{*}(\SUM_{d,D}^{\mathbb{F},n,m,\ell}) & =\Theta((\log_{q}\lceil q^{d-m}\gamma\rceil+1)(\log_{q}\lceil q^{d-\ell}\gamma\rceil+1)\log q),\\
R_{\frac{1-\gamma}{2}}(\SUM_{d}^{\mathbb{F},n,m,\ell}) & =\Theta((\log_{q}\lceil q^{d-m}\gamma\rceil+1)(\log_{q}\lceil q^{d-\ell}\gamma\rceil+1)\log q)
\end{align*}
and moreover
\begin{equation}
R_{\frac{1}{2}-\frac{1}{16q^{2d-m-\ell+16}}}(\SUM_{d}^{\mathbb{F},n,m,\ell})\leq2.\label{eq:sum-4-bit}
\end{equation}
\end{thm}

\prettyref{thm:MAIN-subspace-sum-general} determines the communication
complexity of subspace sum for every error probability in $[\frac{1}{3},\frac{1}{2}-\Theta(|\mathbb{F}|^{-(2d-m-\ell)/5})]$.
This is essentially the complete range of interest because by \prettyref{eq:sum-4-bit},
the communication cost drops to $2$ bits when the error probability
is set to $\frac{1}{2}-|\mathbb{F}|^{-(2d-m-\ell)-\Theta(1)}$. Analogous
to the constant-error regime, \prettyref{thm:MAIN-subspace-sum-general}
shows that the communication complexity of subspace sum for any error
in $[\frac{1}{3},\frac{1}{2}-\Theta(|\mathbb{F}|^{-(2d-m-\ell)/5})]$
is the same, up to a constant factor, for both the partial and total
versions of the problem, and for both quantum and classical communication.
Theorems~\ref{thm:MAIN-subspace-sum-constant-error} and~\ref{thm:MAIN-subspace-sum-general}
reveal a rather subtle dependence of the communication complexity
on the problem parameters $d,m,\ell$, particularly as one additionally
varies the error parameter. This explains why we were not able to
obtain these theorems via a reduction from the matrix rank problem,
as was done in previous work~\cite{sunwang12communication-linear,LiSWW14communication-linear}
in the special case of subspace disjointness. 

In view of the aforementioned identity $\dim(S+T)=\dim(S)+\dim(T)-\dim(S\cap T)$,
our results for subspace sum can be equivalently stated in terms of
subspace intersection. Formally, the \emph{subspace intersection problem}
requires Alice and Bob to distinguish subspace pairs $(S,T)$ with
$\dim(S\cap T)=d_{1}$ from those with $\dim(S\cap T)=d_{2},$ where
$S$ is an $m$-dimensional subspace given as input to Alice, $T$
is an $\ell$-dimensional subspace given to Bob, and $d_{1},d_{2}$
are distinct integers fixed in advance. This corresponds to the partial
function
\[
\INTERSECT_{d_{1},d_{2}}^{\mathbb{F},n,m,\ell}(S,T)=\begin{cases}
-1 & \text{if }\dim(S\cap T)=d_{1},\\
1 & \text{if }\dim(S\cap T)=d_{2},\\
* & \text{otherwise. }
\end{cases}
\]
The total version of the subspace intersection problem is given by
\[
\INTERSECT_{d}^{\mathbb{F},n,m,\ell}(S,T)=\begin{cases}
-1 & \text{if }\dim(S\cap T)\geq d,\\
1 & \text{otherwise,}
\end{cases}
\]
where $d$ is a problem parameter fixed in advance. \prettyref{thm:MAIN-subspace-sum-general}
fully settles the complexity of the subspace intersection problem,
as follows.
\begin{thm}
\label{thm:MAIN-subspace-intersection-general}Let $\mathbb{F}$ be
a finite field with $q=|\mathbb{F}|$ elements, and let $n,m,\ell,r,R$
be nonnegative integers with $\max\{0,m+\ell-n\}\leq r<R\leq\min\{m,\ell\}.$
If $m=\ell=R,$ then
\[
R_{1/3}(\INTERSECT_{R}^{\mathbb{F},n,m,\ell})=O(1).
\]
If $m,\ell,R$ are not all equal, then for all $\gamma\in[\frac{1}{3}q^{-(m+\ell-2R)/5},\frac{1}{3}],$
\begin{align*}
Q_{\frac{1-\gamma}{2}}^{*}(\INTERSECT_{r,R}^{\mathbb{F},n,m,\ell}) & =\Theta((\log_{q}\lceil q^{m-R}\gamma\rceil+1)(\log_{q}\lceil q^{\ell-R}\gamma\rceil+1)\log q),\\
R_{\frac{1-\gamma}{2}}(\INTERSECT_{R}^{\mathbb{F},n,m,\ell}) & =\Theta((\log_{q}\lceil q^{m-R}\gamma\rceil+1)(\log_{q}\lceil q^{\ell-R}\gamma\rceil+1)\log q)
\end{align*}
and moreover
\[
R_{\frac{1}{2}-\frac{1}{16q^{m+\ell-2R+16}}}(\INTERSECT_{R}^{\mathbb{F},n,m,\ell})\leq2.
\]
\end{thm}

A moment's reflection (see~\prettyref{prop:possible-intersections-and-sums})
shows that in $\mathbb{F}^{n},$ the intersection of an $m$-dimensional
subspace and an $\ell$-dimensional subspace is a subspace of dimension
between $\max\{0,m+\ell-n\}$ and $\min\{m,\ell\}$, hence the requirement
that $r,R\in[\max\{0,m+\ell-n\},\min\{m,\ell\}].$ Remarks analogous
to those for subspace sum apply to \prettyref{thm:MAIN-subspace-intersection-general}
as well. Specifically, \prettyref{thm:MAIN-subspace-intersection-general}
determines the $\epsilon$-error communication complexity of subspace
intersection for all $\epsilon\in[\frac{1}{3},\frac{1}{2}-\Theta(|\mathbb{F}|^{-(m+\ell-2R)/5})]$,
which is essentially the complete range of interest because the communication
cost drops to $2$ bits when the error probability is set to $\frac{1}{2}-|\mathbb{F}|^{-(m+\ell-2R)-\Theta(1)}$.
Also, \prettyref{thm:MAIN-subspace-intersection-general} shows that
in this range of interest, the $\epsilon$-error communication complexity
of subspace intersection is the same (up to a constant factor) for
both the partial and total versions of the problem, and for both quantum
and classical communication.

\subsection{\label{subsec:Multiparty-intro}Multiparty lower bounds}

Via a blackbox reduction which we will now describe, our lower bounds
for the matrix rank and determinant problems scale to multiparty communication.
We adopt the standard multiparty model known as the\emph{ number-in-hand
blackboard model}, which features $t$ communicating players and a
(possibly partial) function $F\colon X_{1}\times X_{2}\times\cdots\times X_{t}\to\{-1,1,*\}$
with $t$ arguments. An input $(x_{1},x_{2},\dots,x_{t})$ is partitioned
among the $t$ players by assigning $x_{i}$ to the $i$-th player.
The players communicate by writing on a shared blackboard. They also
have access to an unbounded supply of shared random bits, which they
can use in deciding what to do at any given point in the protocol.
In the end, they must all agree on a bit ($-1$ or $1$) that represents
the output of the protocol. The\emph{ cost }of a communication protocol
is the maximum number of bits written on the blackboard in the worst-case
execution. The\emph{ $\epsilon$-error randomized communication complexity
$R_{\epsilon}(F)$} of a given function $F$ is the least cost of
a protocol that computes $F$ with probability of error at most $\epsilon$
on every input. As usual, the standard setting of the error parameter
is $\epsilon=1/3$, which can be replaced with any other constant
in $(0,1/2)$ at the expense of a constant-factor change in communication
complexity. This model subsumes Yao's two-party randomized model as
a special case, which justifies our continued use of the notation
$R_{\epsilon}(F)$. We note that there are alternative number-in-hand
models, where instead of a shared blackboard, the parties communicate
via private channels (the \emph{message-passing model}) or through
an intermediary (the \emph{coordinator model}). The blackboard model
is more powerful than these alternative models, and lower bounds in
it are more widely applicable.

Phillips, Verbin, and Zhang~\cite{NIH12symmetrization} developed
a symmetrization technique that transforms two-party communication
lower bounds for a class of problems into multiparty lower bounds.
Our communication problems have a large symmetry group and are particularly
well-suited for the methods of~\cite{NIH12symmetrization}. Using
their technique, we prove the following. 
\begin{prop}
\label{prop:2party-multiparty}Let $(X,+)$ be a finite Abelian group,
and let $f\colon X\to\{-1,1,*\}$ be a given function. For $t\geq2,$
let $F_{t}\colon X^{t}\to\{-1,1,*\}$ be the $t$-party communication
problem given by $F_{t}(x_{1},x_{2},\ldots,x_{t})=f(x_{1}+x_{2}+\cdots+x_{t}).$
Then for all $t\geq2,$
\[
R_{1/6}(F_{t})\geq\frac{1}{12}tR_{1/3}(F_{2}).
\]
\end{prop}

In other words, as one transitions from two parties to $t$ parties,
the communication complexity scales by a factor of $\Omega(t)$. This
proposition, proved in Appendix~\ref{sec:2party-to-multiparty},
simplifies and generalizes an earlier result due to Li, Sun, Wang,
and Woodruff~\cite[Theorem~7]{LiSWW14communication-linear}. The
matrix rank problem, determinant problem, and rank versus determinant
problem all admit multiparty generalizations that fit perfectly into
the framework of~\prettyref{prop:2party-multiparty}, with the Abelian
group in all cases being the group of matrices under addition. To
begin with, the \emph{$t$-party matrix rank problem} is given by
$\RANK_{r,R}^{\mathbb{F},n,m,t}(M_{1},M_{2},\ldots,M_{t})=\Rank_{r,R}^{\mathbb{F},n,m}(\sum M_{i}),$
where the matrix function $\Rank_{r,R}^{\mathbb{F},n,m}\colon\mathbb{F}^{n\times m}\to\{-1,1,*\}$
outputs $-1$ on matrices of rank $r,$ outputs $1$ on matrices of
rank $R$, and outputs $*$ in all other cases. \prettyref{thm:RANK-lowerbound-intro}
and \prettyref{prop:2party-multiparty} imply the following.
\begin{thm}
\label{thm:multiparty-rank}For all finite fields $\mathbb{F},$ all
integers $n,m,R,r$ with $\min\{n,m\}\geq R>r\geq0,$ and all $t\geq2,$
\[
R_{1/3}(\RANK_{r,R}^{\mathbb{F},n,m,t})=\Omega(t+tr^{2}\log|\mathbb{F}|).
\]
\end{thm}

\noindent Continuing, the \emph{$t$-party determinant problem} is
given by $\DET_{a,b}^{\mathbb{F},n,t}(M_{1},M_{2},\ldots,M_{t})=\Det_{a,b}^{\mathbb{F},n}(\sum M_{i}),$
where the matrix function $\Det_{a,b}^{\mathbb{F},n}\colon\mathbb{F}^{n\times n}\to\{-1,1,*\}$
outputs $-1$ on matrices with determinant $a,$ outputs $1$ on matrices
with determinant $b$, and outputs $*$ in all other cases. \prettyref{thm:DET-intro}
and \prettyref{prop:2party-multiparty} yield:
\begin{thm}
\label{thm:multiparty-det}For every finite field $\mathbb{F},$ every
pair of distinct elements $a,b\in\mathbb{F},$ and all integers $n\geq2$
and $t\geq2,$
\[
R_{1/3}(\DET_{a,b}^{\mathbb{F},n,t})=\Omega(tn^{2}\log|\mathbb{F}|).
\]
\end{thm}

\noindent Finally, the \emph{$t$-party rank versus determinant problem}
is given by $\RANKDET_{r,a}^{\mathbb{F},n,t}(M_{1},M_{2},\ldots,M_{t})=\Rankdet_{r,a}^{\mathbb{F},n}(\sum M_{i}),$
where the matrix function $\Rankdet_{r,a}^{\mathbb{F},n}\colon\mathbb{F}^{n\times n}\to\{-1,1,*\}$
outputs $-1$ on matrices of rank $r$, outputs $1$ on matrices with
determinant $a$, and outputs $*$ in all other cases. The following
multiparty result is immediate from \prettyref{thm:RANKDET-intro}
and \prettyref{prop:2party-multiparty}. 
\begin{thm}
\label{thm:multiparty-rankdet}For every finite field $\mathbb{F},$
every field element $a\in\mathbb{F}\setminus\{0\},$ and all integers
$n>r\geq0$ and $t\geq2,$ 
\[
R_{1/3}(\RANKDET_{r,a}^{\mathbb{F},n,t})=\Omega(t+tr^{2}\log|\mathbb{F}|).
\]
 
\end{thm}

Theorems~\ref{thm:multiparty-rank} and \ref{thm:multiparty-rankdet}
are tight for every $r\geq1$ in a very strong sense: we give a $t$-party
protocol with error $1/3$ and communication cost $O(t(r^{2}+1)\log|\mathbb{F}|)$
for checking whether the sum of the players' matrices has rank at
most $r$ (see \prettyref{cor:RANK-upperbound-high-accuracy} in Section~\ref{subsec:communication-upper-bounds}).
\prettyref{thm:multiparty-det} is tight because the stated lower
bound matches the trivial, deterministic protocol where each party
announces their input. Since the blackboard model is more powerful
than the message-passing and coordinator models, Theorems~\ref{thm:multiparty-rank}\textendash \ref{thm:multiparty-rankdet}
are valid in those alternative models as well.

\subsection{Bilinear query complexity}

Our communication lower bounds additionally imply new results in query
complexity. We adopt the\emph{ bilinear query model} due to Rashtchian,
Woodruff, and Zhu~\cite{RWZ20bilinear-query}, which subsumes a large
number of other query models and is particularly well-suited for linear-algebraic
problems. Formally, let $f\colon\mathbb{F}^{n\times m}\to\{-1,1,*\}$
be a (possibly partial) Boolean function on matrices over a field
$\mathbb{F}$. In the bilinear query model, the query algorithm accesses
the input $X\in\mathbb{F}^{n\times m}$ in an adaptive manner with
\emph{bilinear queries}. Each such query reveals the value $u\tr Xv\in\mathbb{F}$
for a pair of vectors $u\in\mathbb{F}^{n},$ $v\in\mathbb{F}^{m}$
of the algorithm's choosing. As usual, a randomized query algorithm
is a probability distribution on deterministic query algorithms. The
\emph{cost} of a query algorithm is the maximum number of queries
in the worst-case execution. The \emph{$\epsilon$-error bilinear
query complexity of $f,$ }which we denote by $\BLQ_{\epsilon}(f),$
is the minimum cost of a bilinear query algorithm that computes $f$
with probability of error at most $\epsilon$ on every input. As always,
the algorithm may exhibit arbitrary behavior on inputs $X$ with $f(X)=*.$

Recall the matrix functions $\Rank_{r,R}^{\mathbb{F},n,m},\Det_{a,b}^{\mathbb{F},n},\Rankdet_{r,a}^{\mathbb{F},n}$
that correspond to the matrix rank problem, determinant problem, and
rank versus determinant problem and were formally defined in Section~\ref{subsec:Multiparty-intro}.
Our next result settles their bilinear query complexity for all settings
of the parameters $n,m,r,R,a,b$.
\begin{thm}
\label{thm:BLQ-lower}Let $\mathbb{F}$ be a finite field. Then:
\begin{enumerate}
\item \label{enu:query-rank}for all integers $n,m,R,r$ with $\min\{n,m\}\geq R>r\geq0,$
\[
\BLQ_{\frac{1}{2}-\frac{1}{4|\mathbb{F}|^{r/3}}}(\Rank_{r,R}^{\mathbb{F},n,m})=\Omega(r^{2}+1);
\]
\item \label{enu:query-det}for every pair of distinct elements $a,b\in\mathbb{F}$
and all integers $n\geq1,$
\[
\BLQ_{\frac{1}{2}-\frac{1}{4|\mathbb{F}|^{(n-1)/3}}}(\Det_{a,b}^{\mathbb{F},n})=\Omega(n^{2});
\]
\item \label{enu:query-rankdet}for every field element $a\in\mathbb{F}\setminus\{0\}$
and all integers $n>r\geq0,$
\[
\BLQ_{\frac{1}{2}-\frac{1}{4|\mathbb{F}|^{r/3}}}(\Rankdet_{r,a}^{\mathbb{F},n})=\Omega(r^{2}+1).
\]
 
\end{enumerate}
\end{thm}

\begin{proof}
For a matrix function $f\colon\mathbb{F}^{n\times m}\to\{-1,1,*\}$,
consider the associated communication problem $F\colon\mathbb{F}^{n\times m}\times\mathbb{F}^{n\times m}\to\{-1,1,*\}$
given by $F(A,B)=f(A+B).$ As observed by the authors of~\cite{RWZ20bilinear-query},
a cost-$c$ randomized algorithm for $f$ in the bilinear query model
gives a randomized communication protocol for $F$ of cost $2\lceil\log|\mathbb{F}|\rceil c.$
Specifically, on input $A$ for Alice and $B$ for Bob, they simulate
the query algorithm on $A+B$. Computing a query $u\tr(A+B)v$ for
given vectors $u,v$ amounts to exchanging the field elements $u\tr Av$
and $u\tr Bv$. In summary, 
\[
R_{\epsilon}(F)\leq2\lceil\log|\mathbb{F}|\rceil\BLQ_{\epsilon}(f).
\]
Now the claimed query lower bounds in~\prettyref{enu:query-rank}\textendash \prettyref{enu:query-rankdet}are
immediate from our corresponding communication complexity results
(Theorems~\ref{thm:RANK-lowerbound-intro}, \ref{thm:DET-intro},
and~\ref{thm:RANKDET-intro}) as well as the trivial query lower
bound of $1$ for any nonconstant function.
\end{proof}
Every lower bound in \prettyref{thm:BLQ-lower} is tight, even for
computation with error probability $1/3$. To prove the tightness
of \prettyref{thm:BLQ-lower}\prettyref{enu:query-rank} and~\ref{thm:BLQ-lower}\prettyref{enu:query-rankdet},
we give a query algorithm with error probability $1/3$ and cost $O(r^{2}+1)$
for checking whether the input matrix has rank at most $r$ (see \prettyref{thm:BLQ-upper}
in Section~\prettyref{subsec:communication-upper-bounds}). Finally,
the lower bound in \prettyref{thm:BLQ-lower}\prettyref{enu:query-det}
matches the trivial, deterministic upper bound of $n^{2}$ queries.

The strongest result prior to our work was an $\Omega(n^{2})$ query
lower bound due to Rashtchian, Woodruff, and Zhu~\cite{RWZ20bilinear-query}
for distinguishing, with probability $2/3,$ matrices of rank $n-1$
from those of rank $n$. \prettyref{thm:BLQ-lower}\prettyref{enu:query-rank}
shows that the $\Omega(n^{2})$ query lower bound remains valid even
for distinguishing matrices of rank $cn$ (for any constant $c>0)$
from those of rank $n$, and even for correctness probability exponentially
close to $1/2.$ In particular, \prettyref{thm:BLQ-lower}\prettyref{enu:query-rank}
shows that $\Omega(n^{2})$ bilinear queries are needed to approximate
the rank of a matrix to any constant factor.

\subsection{Previous approaches}

A powerful tool for proving lower bounds on randomized and quantum
communication complexity is the \emph{approximate trace norm}~\cite{yao93quantum,kremer95thesis,razborov03quantum,linial07factorization-stoc,sherstov07quantum}.
In more detail, let $F\colon X\times Y\to\{-1,1\}$ be a given communication
problem, and let $M=[F(x,y)]_{x,y}$ be its characteristic matrix.
The \emph{$\delta$-approximate trace norm of $M$}, denoted $\|M\|_{\Sigma,\delta}$,
is the minimum trace norm of a real matrix $\widetilde{M}$ that approximates
$M$ entrywise within $\delta$. The approximate trace norm bound
states that
\begin{equation}
Q_{\epsilon}^{*}(F)\geq\frac{1}{2}\log\left(\frac{\|M\|_{\Sigma,2\epsilon}}{3\sqrt{|X||Y|}}\right)\label{eq:quantum-lower-intro-1}
\end{equation}
for all $\epsilon\geq0,$ making it possible to prove communication
lower bounds by analyzing the approximate trace norm of $M$. To bound
the approximate trace norm from below, it is useful to appeal to its
dual formulation as a maximization problem, whereby
\begin{equation}
\|M\|_{\Sigma,2\epsilon}\geq\frac{\langle M,\Phi\rangle-2\epsilon\|\Phi\|_{1}}{\|\Phi\|}\label{eq:approxtrace-lower-intro-total}
\end{equation}
for every nonzero real matrix $\Phi.$ As a result, proving a communication
lower bound reduces to constructing a matrix $\Phi$ whose spectral
norm and $\ell_{1}$ norm are small relative to the inner product
of $\Phi$ with the communication matrix $M$. The matrix $\Phi$
is often referred to as a \emph{dual matrix} or a \emph{witness. }The
lower bound~\prettyref{eq:quantum-lower-intro-1} remains valid for
partial functions $F\colon X\times Y\to\{-1,1,*\}$ and their characteristic
matrices $M$, in which case the dual characterization of the approximate
trace norm is given by
\begin{equation}
\|M\|_{\Sigma,2\epsilon}\geq\frac{1}{\|\Phi\|}\left(\sum_{\dom F}M_{x,y}\Phi_{x,y}-2\epsilon\|\Phi\|_{1}-\sum_{\overline{\dom F}}|\Phi_{x,y}|\right)\label{eq:approxtrace-lower-partial}
\end{equation}
for all $\Phi\ne0$. In this equation, $\dom F=\{(x,y):F(x,y)\ne*\}$
denotes the domain of the partial function $F$. Comparing this dual
characterization with the original one \prettyref{eq:approxtrace-lower-intro-total}
for total functions, we notice that the inner product is now restricted
to the domain of $F,$ and there is an additional penalty term for
any weight placed by $\Phi$ outside the domain of $F.$ For more
background on the use of duality in proving communication lower bounds,
we refer the reader to the surveys~\cite{dual-survey,lee-shraibman09cc-survey}.

\subsubsection*{Main idea in \emph{\cite{sunwang12communication-linear}, \cite{LiSWW14communication-linear}}}

Constructing a good witness $\Phi$ can be very challenging. Sun and
Wang~\cite{sunwang12communication-linear} studied the \emph{nonsingularity
problem} over fields $\mathbb{F}_{p}$ of prime order $p$\emph{,}
where Alice and Bob's inputs are matrices $A,B\in\mathbb{F}_{p}^{n\times n}$,
respectively, and they are required to output $1$ if $A+B$ is nonsingular
and $-1$ otherwise. Let $M$ be the characteristic matrix of this
communication problem. To analyze the approximate trace norm of $M,$
the authors of~\cite{sunwang12communication-linear} use the witness
$\Phi=[(-1)^{n}\widehat{g}(A+B)]_{A,B},$ where $\widehat{g}$ is
the Fourier transform of the function $g\colon\mathbb{F}_{p}^{n\times n}\to\{0,1\}$
given by $g(X)=1$ if and only if $X$ is nonsingular. The calculations
in~\cite{sunwang12communication-linear} reveal the following, where
$C\geq6$ is an absolute constant:
\begin{enumerate}
\item $\left\Vert \Phi\right\Vert =1$;
\item $\langle M,\Phi\rangle=2p^{n^{2}-n}\prod_{i=1}^{n}(p^{i}-1)$;
\item $\|\Phi\|_{1}\leq Cp^{n^{2}-n}\prod_{i=1}^{n}(p^{i}-1)$.
\end{enumerate}
Using this witness $\Phi$ in~\prettyref{eq:approxtrace-lower-intro-total}
with a sufficiently small error parameter $\epsilon$, Sun and Wang
obtain $\|M\|_{\Sigma,2\epsilon}=\Omega(p^{n^{2}}p^{n(n-1)/2})$,
which in view of~\prettyref{eq:quantum-lower-intro-1} gives an $\Omega(n^{2}\log p)$
lower bound on the bounded-error communication complexity of the nonsingularity
problem.

In follow-up work, Li, Sun, Wang, and Woodruff~\cite{LiSWW14communication-linear}
studied the partial communication problem $F=\RANK_{n-1,n}^{\mathbb{F}_{p},n,n}$.
Let $M'$ denote its characteristic matrix. The authors of~\cite{LiSWW14communication-linear}
used the same witness $\Phi$ as Sun and Wang~\cite{sunwang12communication-linear}
and proved the following:
\begin{enumerate}
\item $\|\Phi\|=1$;
\item $\sum_{\dom F}M_{A,B}'\Phi_{A,B}=p^{n^{2}-n}(1+\frac{p-p^{-n+1}}{p-1})\prod_{i=1}^{n}(p^{i}-1)$;
\item $\|\Phi\|_{1}=p^{n^{2}-n}\prod_{i=0}^{n-1}(1+p^{-i})\cdot\prod_{i=1}^{n}(p^{i}-1);$
\item $\sum_{\overline{\dom F}}|\Phi_{A,B}|=\|\Phi\|_{1}-\sum_{\dom F}M_{A,B}'\Phi_{A,B}$.
\end{enumerate}
Making these substitutions in~\prettyref{eq:approxtrace-lower-partial}
and setting $\epsilon$ to a sufficiently small constant, the authors
of~\cite{LiSWW14communication-linear} obtain $\|M'\|_{\Sigma,2\epsilon}=\Omega(p^{n^{2}}p^{n(n-1)/2})$,
which along with~\prettyref{eq:quantum-lower-intro-1} results in
an $\Omega(n^{2}\log p)$ lower bound on the quantum communication
complexity of $F=\RANK_{n-1,n}^{\mathbb{F}_{p},n,n}$. We note that
we have described the work of~\cite{sunwang12communication-linear,LiSWW14communication-linear}
in the framework that we adopt in our paper, which differs somewhat
from the original presentation in~\cite{sunwang12communication-linear,LiSWW14communication-linear}.
These differences do not affect any of the ideas or bounds in question.

Unfortunately, the above analyses rely heavily on $\epsilon$ being
set to a small constant. This is because $\|\Phi\|_{1}$ is too large
compared to the inner product $\langle M,\Phi\rangle$ and the correlation
$\sum_{\dom F}M_{A,B}'\Phi_{A,B}$, which makes setting $\epsilon$
close to $1/2$ impossible. Since the authors of~\cite{LiSWW14communication-linear}
determined $\|\Phi\|_{1}$ and $\sum_{\dom F}M_{A,B}'\Phi_{A,B}$
exactly, with equality, there is no room for improved analysis and
no possibility of setting $\epsilon$ close to $1/2$ with this choice
of witness $\Phi$. This rules out the use of $\Phi$ for proving
\prettyref{thm:RANK-lowerbound-intro} even in the special case of
$\RANK_{n-1,n}^{\mathbb{F},n,n}$. 

When it comes to the general problem $\RANK_{r,n}^{\mathbb{F},n,n}$,
the witness $\Phi$ produces no meaningful results at all for any
$r\leq n-3,$ regardless of the error parameter $\epsilon$. The issue
is that the $\ell_{1}$ norm of $\Phi$ is concentrated on matrix
pairs $(A,B)$ for which $A+B$ has rank $n$ or $n-1$, whereas the
domain of $\RANK_{r,n}^{\mathbb{F},n,n}$ consists of matrix pairs
whose sum has rank $n$ or $r$. Quantitatively, the domain of $\RANK_{r,n}^{\mathbb{F},n,n}$
supports less than half of the $\ell_{1}$ norm of $\Phi,$ which
causes the lower bound in~\prettyref{eq:approxtrace-lower-partial}
to be negative for every $\epsilon$. Our attempts at simple modifications
to $\Phi$ were not successful.

\subsection{Our approach}

Our techniques depart substantially from the previous work in~\cite{sunwang12communication-linear,LiSWW14communication-linear}.
Instead of attempting to guess a good witness $\Phi$ and analyzing
its metric and analytic properties, we determine how exactly those
properties depend on the choice of a witness. In this way, we are
able to construct essentially optimal witnesses for the matrix rank,
determinant, subspace sum, and subspace intersection problems. We
first discuss the matrix rank problem, over an arbitrary finite field
$\mathbb{F}$. In this overview, we focus on the canonical case $F=\RANK_{k,n}^{\mathbb{F},n,n}$,
where Alice and Bob receive square matrices $A,B\in\mathbb{F}^{n\times n}$,
respectively, and need to decide whether $\rk(A+B)=k$ or $\rk(A+B)=n$.
This special case captures the matrix rank problem in its full generality
via straightforward reductions.

\subsubsection*{Reducing the degrees of freedom}

We will call a witness $\Phi$ \emph{symmetric} if each entry $\Phi_{A,B}$
is fully determined by the rank of $A+B$. In searching for a good
witness for the matrix rank problem, we will only consider symmetric
witnesses $\Phi$. This restriction is without loss of generality:
since $F(A,B)$ depends only on the rank of $A+B$, it is not hard
to verify that any witness for $F$ can be ``symmetrized'' without
harming the corresponding value of the approximate trace norm bound,~\prettyref{eq:approxtrace-lower-partial}.
The resulting witness matrix $\Phi$ has only $n+1$ degrees of freedom,
corresponding to every possible value of the rank of $A+B$.

Let $i\in\{0,1,\ldots,n\}$ be given. Consider the matrix whose rows
and columns are indexed by elements of $\mathbb{F}^{n\times n},$
and whose $(A,B)$ entry is defined to be $1$ if $\rk(A+B)=i$ and
zero otherwise. Normalize this matrix to have $\ell_{1}$ norm $1,$
and call the resulting matrix $E_{i}$. Then any symmetric witness
matrix is a linear combination of $E_{0},E_{1},\ldots,E_{n}$. With
this in mind, for any real function $\phi\colon\{0,1,\ldots,n\}\to\Re,$
we define
\[
E_{\phi}=\phi(0)E_{0}+\phi(1)E_{1}+\cdots+\phi(n)E_{n}.
\]
Taking $\Phi=E_{\phi}$ in the approximate trace norm bound~\prettyref{eq:approxtrace-lower-partial}
and simplifying, we arrive at the following bound for the characteristic
matrix $M$ of $F$:
\begin{align}
\|M\|_{\Sigma,2\epsilon} & \geq\frac{1}{\|E_{\phi}\|}\left(\phi(n)-\phi(k)-2\epsilon\|\phi\|_{1}-\sum_{i\notin\{k,n\}}|\phi(i)|\right).\label{eq:approxtrace-lower-phi-intro}
\end{align}
Our challenge now is to understand how $\phi$ affects the spectral
norm of $E_{\phi}$. For matrices with a large symmetry group, it
is reasonable to expect algebraic structure in the spectrum. For example,
the so-called \emph{combinatorial matrices,} studied by Knuth~\cite{knuth01discreet}
and used for communication lower bounds by Razborov~\cite{razborov03quantum},
have all eigenvalues described in terms of Hahn polynomials. We will
similarly see that the spectrum of each $E_{\phi}$ has strong algebraic
structure and is described in terms of what we call \emph{hyperpolynomials}.

By analyzing the singular values of $E_{\phi}$, we prove that
\begin{equation}
\|E_{\phi}\|=q^{-n^{2}}\max_{s=0,1,\ldots,n}\left|\sum_{t=0}^{n}\phi(t)\Gamma_{n}(s,t)\right|,\label{eq:spectral-in-terms-of-phi}
\end{equation}
where $q$ is the order of the finite field $\mathbb{F},$ and $\Gamma_{n}$
is an auxiliary function. In more detail, we define
\[
\Gamma_{n}(s,t)=\Exp_{\substack{\rk A=s\\
\rk B=t
}
}\omega^{\langle A,B\rangle},
\]
where $\omega$ is a primitive root of unity of order equal to the
characteristic of $\mathbb{F},$ with the operation $x\mapsto\omega^{x}$
for field elements $x\in\mathbb{F}$ deferred to Section~\ref{sec:Fourier}.
An exact expression for $\Gamma_{n}(n,t)$ can be obtained from the
analysis of the Fourier spectrum of the nonsingularity function in
\cite{sunwang12communication-linear}. Understanding $\Gamma_{n}(s,t)$
for general $s,t,$ however, is rather nontrivial. To this end, we
derive the representation
\[
\Gamma_{n}(s,t)=\sum_{r=0}^{n}P_{n}(s,t,r)\Gamma_{n}(n,r),
\]
where $P_{n}(s,t,r)$ is the probability that the upper-left $s\times t$
quadrant of a uniformly random nonsingular matrix of order $n$ has
rank $r$. By explicitly calculating the probabilities $P_{n}(s,t,r)$
and combining them with the closed-form expression for $\Gamma_{n}(n,r)$,
we obtain the upper bound $|\Gamma_{n}(s,t)|\leq cq^{-st/2}$ for
an absolute constant $c$. In addition to this \emph{analytic} property,
we establish the following \emph{algebraic} result:\emph{ }for $n,s$
fixed, $\Gamma_{n}(s,t)$ as a function of $t\in\{0,1,\ldots,n\}$
is a polynomial in $q^{-t}$ of degree at most $s$. These two properties
play a central role in our analysis. In what follows, we will refer
to a polynomial in $q^{-t}$ as a \emph{hyperpolynomial in $t$}.

\subsubsection*{Univariate object for the rank problem}

Since \prettyref{eq:approxtrace-lower-phi-intro} is invariant under
multiplication of $\phi$ by a positive factor, we will normalize
$\phi$ such that $\phi(n)=1$. To achieve a large value on the right-hand
side of~\prettyref{eq:approxtrace-lower-phi-intro}, we will construct
a function $\phi$ that is negative at $k,$ has $\ell_{1}$ norm
concentrated on $\{k,n\}$, and results in $E_{\phi}$ having a small
spectral norm. In view of~\prettyref{eq:spectral-in-terms-of-phi},
the spectral norm requirement amounts to a bound on $\max_{s}\left|\sum_{t=0}^{n}\phi(t)\Gamma_{n}(s,t)\right|$.
Quantitatively speaking, to obtain an asymptotically optimal lower
bound for the matrix rank problem, we need $\phi$ to satisfy the
following constraints:
\begin{enumerate}
\item $\phi(n)=1;$
\item $\phi(k)<0;$
\item $\sum_{i\notin\{k,n\}}|\phi(i)|=q^{-\Omega(k)};$
\item $|\sum_{t=0}^{n}\phi(t)\Gamma_{n}(s,t)|=q^{-\Omega(k^{2})}$ for every
$s\in\{0,1,\ldots,n\}$.
\end{enumerate}
The last requirement states that $\phi$ needs to be almost orthogonal
to each $\Gamma_{n}(s,t)$, viewed as a function of $t$ with fixed
$s$. Recall from our earlier discussion that for $s$ and $n$ fixed,
$\Gamma_{n}(s,t)$ is a hyperpolynomial of low degree, namely, a polynomial
in $q^{-t}$ of degree at most $s$. To achieve orthogonality to hyperpolynomials
of low degree, we leverage the \emph{Cauchy binomial theorem}~\cite[eqn.~(1.87)]{enumerative-comb86stanley},
which implies that
\begin{equation}
\sum_{t=0}^{n}\gbinom[q]{n}{t}(-1)^{t}q^{\binom{t}{2}}g(q^{-t})=0\label{eq:cauchy-binom-intro}
\end{equation}
for every polynomial $g$ of degree less than $n$. In particular,
defining $\phi(t)=\gbinom nt(-1)^{t}q^{\binom{t}{2}}$ for $t=0,1,\ldots,n$
ensures that $\phi$ is exactly orthogonal to each hyperpolynomial
$\Gamma_{n}(s,t)$ for $s<n$. Unfortunately, this choice of $\phi$
does not satisfy our constraint on the distribution of the $\ell_{1}$
norm because most of it would be concentrated on the values $\phi(t)$
at points $t\approx n$. To overcome this difficulty, we apply a hyperpolynomial
of low degree to achieve the desired distribution of the $\ell_{1}$
norm. Specifically, we set
\[
\phi(t)=\gbinom nt(-1)^{t-n}q^{\binom{t}{2}-\binom{n}{2}}\zeta(q^{-t})
\]
for a carefully constructed  polynomial $\zeta$; the factor $(-1)^{-n}q^{-\binom{n}{2}}$
in this formula serves to normalize $\phi$ and ensure the proper
signs. As we increase the degree of $\zeta,$ we improve the distribution
of the $\ell_{1}$ norm of $\phi$ at the expense of a weaker orthogonality
guarantee, for now $\phi$ is orthogonal only to hyperpolynomials
of degree less than $n-\deg\zeta$. With an appropriate choice of
$\zeta,$ we are able to ensure all four desiderata (i)\textendash (iv)
for the univariate function $\phi$. The most technical part of the
analysis is the upper bound in~(iv). For $s$ small, our construction
guarantees (iv) as a consequence of the Cauchy binomial theorem, with
$\sum_{t=0}^{n}\phi(t)\Gamma_{n}(s,t)=0.$ For $s$ large, we use
the pointwise bounds for $\phi$ and $\Gamma_{n}$ and show that $\sum_{t=0}^{n}|\phi(t)|\,|\Gamma_{n}(s,t)|=q^{-\Omega(k^{2})}.$ 

By combining equations~\prettyref{eq:approxtrace-lower-phi-intro}
and~\prettyref{eq:spectral-in-terms-of-phi} with the properties
(i)\textendash (iv) of the univariate function $\phi$, we derive
the following bound on the approximate trace norm: $\|M\|_{\Sigma,2\epsilon}\geq(1-2\epsilon-q^{-\Omega(k)})q^{n^{2}}q^{\Omega(k^{2})}.$
Applying the approximate trace norm method~\prettyref{eq:quantum-lower-intro-1},
we obtain the sought lower bound of $\Omega(k^{2}\log q)$ on the
quantum communication complexity of $F$ for error $\epsilon=\frac{1}{2}-q^{-\Theta(k)}$.
To achieve the error probability as stated in \prettyref{thm:RANK-lowerbound-intro},
we derive bounds for $\phi$ with explicit constants, which we did
not discuss in this proof sketch.

\subsubsection*{The determinant problem}

To solve the determinant problem $\DET_{a,b}^{\mathbb{F},n}$ for
all field elements $a,b$, we combine our approach to the matrix rank
problem presented above with additional Fourier-theoretic ideas. Recall
that we tackle the determinant problem from first principles, without
relying on the partial solution for nonzero $a,b$ due to Sun and
Wang~\cite{sunwang12communication-linear}. With this in mind, we
will first discuss the case of nonzero $a,b.$ Consider the function
$g_{a,b}\colon\mathbb{F}^{n\times n}\to\{-1,1,0\}$ given by 
\[
g_{a,b}(X)=\begin{cases}
-1 & \text{if }\det X=a,\\
1 & \text{if }\det X=b,\\
0 & \text{otherwise.}
\end{cases}
\]
A simple argument reveals that the Fourier coefficients of $g_{a,b}$
corresponding to singular matrices are zero, whereas those corresponding
to nonsingular matrices $M$ depend only on $\det(M).$ By applying
Parseval's identity, we obtain a strong upper bound on the absolute
value of every Fourier coefficient of $g_{a,b}$:
\[
\|\widehat{g_{a,b}}\|_{\infty}\leq\frac{1}{\sqrt{\lvert\SL(\mathbb{F},n)\rvert}},
\]
where $\SL(\mathbb{F},n)$ denotes the special linear group of order-$n$
matrices over $\mathbb{F}.$ Consider now the matrix $\Phi_{a,b}$
whose rows and columns are indexed by elements of $\mathbb{F}^{n\times n}$
and whose entries are given by $\Phi_{a,b}(A,B)=g_{a,b}(A+B)$. The
spectral norm of $\Phi_{a,b}$ is governed by the Fourier coefficients
of $g_{a,b},$ with 
\[
\|\Phi_{a,b}\|=q^{n^{2}}\|\widehat{g_{a,b}}\|_{\infty}\leq\frac{q^{n^{2}}}{\sqrt{\lvert\SL(\mathbb{F},n)\rvert}}.
\]
Observe that $\Phi_{a,b}$ is precisely the characteristic matrix
of $\DET_{a,b}^{\mathbb{F},n}$ with the $*$ entries replaced with
zeroes. Using $\Phi_{a,b}$ as a witness in the approximate trace
norm method, we immediately obtain \prettyref{thm:DET-intro} for
nonzero $a,b$.

Consider now the complementary case when one of $a,b$ is zero, say,
$a\ne0$ and $b=0$. Here, we study the rank versus determinant problem
$\RANKDET_{k,a}^{\mathbb{F},n}$, which in this case is a subproblem
of the determinant problem. Its parameters are an integer $k\in\{0,1,\ldots,n-1\}$
and a nonzero field element $a\in\mathbb{F}$. Recall that in this
problem, Alice and Bob are given matrices $A,B\in\mathbb{F}^{n\times n}$,
respectively, and are called upon to distinguish between the cases
$\rk(A+B)=k$ and $\det(A+B)=a.$ To construct a witness for $\RANKDET_{k,a}^{\mathbb{F},n}$,
we combine our solutions to the matrix rank problem and the determinant
problem for nonzero field elements. In more detail, consider the witness
$\Phi$ for the problem $\RANK_{k,n}^{\mathbb{F},n,n}$ that we sketched
above. Recall that $\Phi_{A,B}$ depends only on the rank of $A+B,$
and moreover the $\ell_{1}$ norm of $\Phi$ is concentrated on matrix
pairs $(A,B)$ with $\rk(A+B)\in\{k,n\}$. To turn $\Phi$ into a
witness for $\RANKDET_{k,a}^{\mathbb{F},n}$, we form a \emph{linear
combination} of $\Phi$ with the matrices $\Phi_{a,b}$ for all $b\in\mathbb{F}\setminus\{0,a\}$,
constructed in the previous paragraph for the determinant problem
with nonzero field elements. The coefficients in this linear combination
are chosen so as to transfer the $\ell_{1}$ weight placed by $\Phi$
on matrix pairs with $\det(A+B)\notin\{0,a\}$ to the matrix pairs
with $\det(A+B)=a$, without affecting any other entries of $\Phi$.
The resulting dual witness has low spectral norm (being the sum of
matrices with low spectral norm) and has its $\ell_{1}$ norm concentrated
on matrix pairs $(A,B)$ for which $A+B$ has rank $k$ or determinant
$a,$ ensuring strong correlation with the partial function $\RANKDET_{k,a}^{\mathbb{F},n}.$
By applying the approximate trace norm method, we obtain the claimed
communication lower bounds for $\RANKDET_{k,a}^{\mathbb{F},n}.$

\subsubsection*{Subspace sum and intersection}

We now present the main ideas in our solution to the subspace sum
and subspace intersection problems. Since these problems are equivalent,
we will discuss the intersection problem alone. As before, we work
with an arbitrary finite field $\mathbb{F},$ whose order we denote
by $q.$ Also by way of notation, recall that $m$ and $\ell$ stand
for the dimensions of Alice's subspace $S$ and Bob's subspace $T,$
respectively. For simplicity, we will assume in this overview that
the dimension $n$ of the ambient vector space satisfies $n\geq m+\ell,$
which ensures that $\dim(S\cap T)$ takes on every possible value
in $\{0,1,2,\ldots,\min\{m,\ell\}\}$ as one varies the subspaces
$S,T.$ We will focus on the canonical case of the subspace intersection
problem where Alice and Bob need to distinguish subspace pairs with
$\dim(S\cap T)=0$ from those with $\dim(S\cap T)=R,$ for an integer
$R$ with $0<R\leq\min\{m,\ell\}.$ In what follows, we let $F=\INTERSECT_{0,R}^{\mathbb{F},n,m,\ell}$
stand for this communication problem of interest. The general case
of the subspace intersection problem, which we will not discuss in
this overview, reduces to this canonical case.

As before, the challenge is to construct a dual matrix $\Phi$ that
witnesses a strong lower bound on the approximate trace norm of the
characteristic matrix $M$ of $F$. Note that the rows of $\Phi$
are indexed by $m$-dimensional subspaces, and the columns are indexed
by $\ell$-dimensional subspaces. Analogous to the matrix rank problem,
we start with the methodological observation that the symmetry of
$F$ greatly reduces the number of degrees of freedom in $\Phi.$
Specifically, $F(S,T)$ by definition depends only on $\dim(S\cap T).$
A moment's thought now shows that any dual matrix $\Phi$ for the
subspace intersection problem can be ``symmetrized'' such that its
$(S,T)$ entry depends only on $\dim(S\cap T)$, and this symmetrization
can only improve the resulting lower bound on the approximate trace
norm in~\eqref{eq:approxtrace-lower-partial}. 

For $r=0,1,\ldots,\min\{m,\ell\},$ let $J_{r}^{n,m,\ell}$ stand
for the matrix whose rows are indexed by $m$-dimensional subspaces
of $\mathbb{F}^{n},$ whose columns are indexed by $\ell$-dimensional
subspaces of $\mathbb{F}^{n}$, and whose $(S,T)$ entry is $1$ if
$\dim(S\cap T)=r$ and zero otherwise. Put another way, $J_{r}^{n,m,\ell}$
is the characteristic matrix of subspace pairs whose intersection
has dimension $r$. For an arbitrary function $\psi\colon\{0,1,\ldots,\min\{m,\ell\}\}\to\Re,$
we define 
\[
J_{\psi}^{n,m,\ell}=\sum_{r=0}^{\min\{m,\ell\}}\psi(r)J_{r}^{n,m,\ell}.
\]
We refer to this family of matrices, whose $(S,T)$ entry depends
only on $\dim(S\cap T)$, as \emph{subspace matrices. }It will also
be helpful to have notation for normalized versions of these matrices,
as follows:
\[
\overline{J}_{r}^{n,m,\ell}=\frac{1}{\|J_{r}^{n,m,\ell}\|_{1}}\,J_{r}^{n,m,\ell},\qquad\qquad\overline{J}_{\psi}^{n,m,\ell}=\sum_{r=0}^{\min\{m,\ell\}}\frac{\psi(r)}{\|J_{r}^{n,m,\ell}\|_{1}}\,J_{r}^{n,m,\ell}.
\]
In this notation, we are looking to construct a dual witness of the
form $\Phi=\overline{J}_{\psi}^{n,m,\ell}$ for some function $\psi.$
This matrix has $\min\{m,\ell\}+1$ degrees of freedom, corresponding
to every possible value that $\dim(S\cap T)$ can take. Setting $\Phi=\overline{J}_{\psi}^{n,m,\ell}$
in the approximate trace norm bound~\prettyref{eq:approxtrace-lower-partial}
and simplifying, one obtains the following bound for the characteristic
matrix $M$ of $F$:
\begin{align}
\|M\|_{\Sigma,2\epsilon} & \geq\frac{1}{\|\overline{J}_{\psi}^{n,m,\ell}\|}\left(-\psi(0)+\psi(R)-2\epsilon\|\psi\|_{1}-\sum_{i\notin\{0,R\}}|\psi(i)|\right).\label{eq:approxtrace-lower-phi-intro-1}
\end{align}
At first glance, this equation looks similar to the corresponding
equation~\prettyref{eq:approxtrace-lower-phi-intro} for the matrix
rank problem. However, there is a major difference: the spectral norm
of $E_{\phi}$ is now replaced with the spectral norm of $\overline{J}_{\psi}^{n,m,\ell}$,
and there is no reason to expect that these quantities depend on their
corresponding univariate objects $\phi$ and $\psi$ in a similar
way. Indeed, our spectral analysis of $\overline{J}_{\psi}^{n,m,\ell}$
is quite different and significantly more technical than that of $E_{\phi}.$ 

\subsubsection*{Analyzing the spectrum of subspace matrices}

Symmetric subspace matrices $J_{\psi}^{n,m,m}$ are classical objects
whose eigenvectors and eigenvalues have been studied in numerous works,
e.g.,~\cite{delsarte76schemes,eisfeld99eigenspaces,BCIM17eigenvalues-graphs,cioaba-gupta21grassmann}.
However, these previous analyses do not seem to apply to the general,
asymmetric case of interest to us, namely, that of subspace matrices
$J_{\psi}^{n,m,\ell}$ for arbitrary $m,\ell$. One way to reduce
the analysis of the spectral norm of $J_{\psi}^{n,m,\ell}$ to the
symmetric case is to express the product $J_{\psi}^{n,m,\ell}(J_{\psi}^{n,m,\ell})\tr=J_{\psi}^{n,m,\ell}J_{\psi}^{n,\ell,m}$
as the sum of symmetric subspace matrices and then apply known results
for the symmetric case. Unfortunately, multiplying these subspace
matrices leads to expressions so unwieldy and complicated that this
is clearly not the method of choice.

Instead, our analysis is inspired by a result of Knuth~\cite{knuth01discreet}
on what he called \emph{combinatorial matrices}, which we briefly
mentioned above. Specifically, Knuth investigated the eigenvalues
of symmetric matrices of order $\binom{n}{t}$ whose rows and columns
are indexed by $t$-element subsets of $\{1,2,\ldots,n\}$ and whose
$(A,B)$ entry depends only on $|A\cap B|$. To determine the eigenvectors
of a combinatorial matrix, Knuth studied certain homogeneous linear
systems with variables indexed by subsets of a fixed cardinality $s$,
and the equations themselves corresponding to sets of cardinality
$s-1$. He showed that any solution to such a system for $s\in\{1,2,\ldots,t\}$
is an eigenvector for every combinatorial matrix of order $\binom{n}{t}$.
Knuth also proved that for any given $s,$ the space of solutions
has a basis supported on the variables indexed by what he called \emph{basic
sets}. These sets have a simple combinatorial description, which the
author of~\cite{knuth01discreet} used to prove that the eigenvectors
arising from the homogeneous systems for $s=1,2,\ldots,t$, together
with the all-ones vector, form an exhaustive description of the eigenvectors
of each combinatorial matrix. Once the eigenvectors are determined,
one readily calculates their associated eigenvalues and in particular
the spectral norm.

With some effort, we are able to adapt Knuth's ideas to the context
of subspaces. Along the way, we encounter several obstacles. To begin
with, counting problems that are straightforward for sets become challenging
for subspaces, and some intuitive combinatorial principles no longer
work. For example, the inclusion-exclusion formula $\dim(S+T)=\dim(S)+\dim(T)-\dim(S\cap T)$
has no analogue for three or more subspaces. Another obstacle is that
Knuth's notion of a basic set does not seem to have a meaningful analogue
for subspaces. For this reason, we reformulate Knuth's ideas in a
purely linear-algebraic way and sidestep much of the combinatorial
machinery in~\cite{knuth01discreet}. The final hurdle is extending
Knuth's analysis to the asymmetric case. Ultimately, we are able to
determine the singular values and spectral norm of every subspace
matrix $J_{\psi}^{n,m,\ell}$ and in particular its normalized version
$\overline{J}_{\psi}^{n,m,\ell}$. We prove that
\begin{equation}
\|\overline{J}_{\psi}^{n,m,\ell}\|=\max_{s=0,1,\ldots,\min\{m,\ell\}}\left|\sum_{r=0}^{\min\{m,\ell\}}\psi(r)\overline{\Lambda}_{r}^{\,n,m,\ell}(s)\right|^{1/2}\left|\sum_{r=0}^{\min\{m,\ell\}}\psi(r)\overline{\Lambda}_{r}^{\,n,\ell,m}(s)\right|^{1/2},\label{eq:J-psi-spectral-intro}
\end{equation}
where $\overline{\Lambda}_{r}^{\,n,m,\ell}$ and $\overline{\Lambda}_{r}^{\,n,\ell,m}$
are functions with algebraic and analytic properties analogous to
those of the $\Gamma_{n}$ function in our solution to the matrix
rank problem. Specifically, we have:
\begin{enumerate}
\item for $n,m,\ell,s$ fixed, $\overline{\Lambda}_{r}^{\,n,m,\ell}(s)$
as a function of $r\in\{0,1,\ldots,\min\{m,\ell\}\}$ is a polynomial
in $q^{r}$ of degree at most $s;$
\item $|\overline{\Lambda}_{r}^{\,n,m,\ell}(s)|\leq8\binom{n}{m}_{q}^{-1}q^{-s(m-r)/2}$
for $r=0,1,\ldots,\min\{m,\ell\}$.
\end{enumerate}
By swapping the roles of $m$ and $\ell$, one obtains analogous properties
for $\overline{\Lambda}_{r}^{\,n,\ell,m}(s)$.

This spectral result gives us fine-grained control over the spectrum
of $J_{\psi}^{n,m,\ell}$ via the univariate function $\psi$. Our
construction of $\psi$ is based on the Cauchy binomial theorem and
is conceptually similar to our univariate function $\phi$ in the
matrix rank problem. In particular, we use the algebraic property~(i)
to bound the product in~\prettyref{eq:J-psi-spectral-intro} for
small $s$, and the analytic property~(ii) to bound it for large
$s$. We further ensure that the $\ell_{1}$ norm of $\psi$ is highly
concentrated on $\{0,R\},$ with $\psi(0)<0$ and $\psi(R)>0.$ This
results in a strong lower bound in~\prettyref{eq:approxtrace-lower-phi-intro-1},
which in turn leads to an optimal lower bound on the communication
complexity of $F$ by virtue of the approximate trace norm method.

\section{Preliminaries}

\subsection{General notation}

We view Boolean functions as mappings $X\to\{-1,1\}$, where $X$
is a nonempty finite set and the range elements $-1,1$ correspond
to ``true'' and ``false,'' respectively. A \emph{partial} Boolean
function is a mapping $f\colon X\to\{-1,1,*\},$ whose \emph{domain
}is defined as $\dom f=\{x\in X:f(x)\ne*\}.$ Recall that for an arbitrary
function $f\colon X\to Y,$ the restriction of $f$ to a subset $X'\subseteq X$
is defined to be the mapping $f|_{X'}\colon X'\to Y$ given by $(f|_{X'})(x)=f(x).$ 

We adopt the shorthand $[n]=\{1,2,\ldots,n\}.$ We use the letters
$p$ and $q$ throughout this manuscript to refer to a prime number
and a prime power, respectively. As usual, $\mathbb{F}_{q}$ stands
for the Galois field GF$(q),$ the $q$-element field which is unique
up to isomorphism. For a given set $X,$ the \emph{Kronecker delta}
$\delta_{x,y}$ is defined for $x,y\in X$ by 
\[
\delta_{x,y}=\begin{cases}
1 & \text{if }x=y,\\
0 & \text{otherwise.}
\end{cases}
\]
For a function $f\colon X\to\mathbb{C},$ we use the familiar norms
$\|f\|_{1}=\sum_{x\in X}|f(x)|$ and $\|f\|_{\infty}=\max_{x\in X}|f(x)|.$
Similarly, for a real or complex matrix $M,$ one defines $\|M\|_{1}=\sum|M_{i,j}|$
and $\|M\|_{\infty}=\max|M_{i,j}|$. The norms $\|v\|_{1}$ and $\|v\|_{\infty}$
for a real or complex vector $v$ are defined analogously. The Euclidean
norm is given by $\|v\|_{2}=\sqrt{\sum|v_{i}|^{2}}.$ We denote the
base-$q$ logarithm of $x$ by $\log_{q}x.$ In the special case of
the binary logarithm, we write simply $\log x$ rather than $\log_{2}x$. 

\subsection{Linear-algebraic preliminaries}

Let $\mathbb{F}$ be a given field. We denote the set of $n\times m$
matrices over $\FF$ by $\FF^{n\times m}.$ We use the standard notation
$\rk A,$ $\ker A,$ and $A\tr$ for the rank, null space, and transpose
of the matrix $A.$ As usual, the determinant of $A\in\mathbb{F}^{n\times n}$
is denoted $\det A$. The trace of a matrix $A\in\mathbb{F}^{n\times n}$
is denoted $\trace A$ and defined as the sum of the diagonal elements
of $A.$ The commutativity of the trace operator is often helpful:
$\trace(AB)=\trace(BA)$ for square matrices $A,B.$ We let $\diag(a_{1},a_{2},\ldots,a_{n})$
denote the diagonal matrix of order $n$ with diagonal entries $a_{1},a_{2},\ldots,a_{n}.$
Recall that $I_{n}$ normally denotes the identity matrix of order
$n,$ whereas $I$ denotes the identity matrix whose order is to be
inferred from the context. We generalize the meaning of $I_{n}$ somewhat
by defining
\[
I_{n}=\diag(\underbrace{1,1,\ldots1}_{n},0,\ldots0),
\]
where the order of the matrix (and hence the number of zeroes on the
diagonal) will be clear from the context. We let $J$ and $\1$ denote
the all-ones matrix and all-ones vector, respectively, whose dimensions
will be clear from the context.
\begin{fact}
\label{fact:rank-of-product}For square matrices $A,B$ of order $n$
over a given field $\mathbb{F},$
\[
\rk AB\geq\rk A+\rk B-n.
\]
\end{fact}

\begin{proof}
Recall that the dimension of $\ker AB$ is at most the sum of the
dimensions of $\ker A$ and $\ker B$. By the rank-nullity theorem,
this is equivalent to the claimed inequality.
\end{proof}
For $\mathbb{F}$ a finite field or the field of real numbers, the
inner product operation on vectors and matrices is defined as usual
by $\langle x,y\rangle=\sum x_{i}y_{i}$ and $\langle A,B\rangle=\sum A_{i,j}B_{i,j}.$
For $\mathbb{F}=\mathbb{C},$ the modified definitions $\langle x,y\rangle=\sum x_{i}\overline{y_{i}}$
and $\langle A,B\rangle=\sum A_{i,j}\overline{B_{i,j}}$ are used
instead. For complex-valued functions $f,g\colon X\to\mathbb{C},$
we write $\langle f,g\rangle=\sum_{x\in X}f(x)\overline{g(x)}.$ Again
for $\mathbb{F}=\mathbb{C},$ the conjugate transpose of a matrix
$A=[A_{i,j}]_{i,j}$ is denoted by $A^{*}=[\overline{A_{j,i}}]_{i,j},$
and a matrix $A\in\mathbb{C}^{n\times n}$ is called \emph{unitary
}if $A^{*}A=AA^{*}=I.$ The following useful fact relates the inner
product and trace operators.
\begin{fact}
\label{fact:inner-product-vs-trace}Let $A,B,C,D$ be matrices of
order $n$ over $\mathbb{R}$ or a finite field. Then:
\end{fact}

\begin{enumerate}
\item \label{enu:inner-product-as-trace}$\langle A,B\rangle=\trace(AB\tr)=\trace(A\tr B),$
\item \label{enu:inner-product-transfer}$\langle A,C_{1}BC_{2}\rangle=\langle C_{1}\tr AC_{2}\tr,B\rangle.$
\end{enumerate}
\begin{proof}
Item \prettyref{enu:inner-product-as-trace} is immediate from the
definition of matrix multiplication, whereas~\prettyref{enu:inner-product-transfer}
follows from~\prettyref{enu:inner-product-as-trace} and the commutativity
of the trace operator: $\langle A,C_{1}BC_{2}\rangle=\trace(AC_{2}\tr B\tr C_{1}\tr)=\trace(C_{1}\tr AC_{2}\tr B\tr)=\langle C_{1}\tr AC_{2}\tr,B\rangle.$
\end{proof}
Over any field $\mathbb{F},$ we let $e_{1},e_{2},\ldots,e_{n}$ denote
as usual the vectors of the standard basis for $\mathbb{F}^{n}.$
For any subset $S\subseteq\mathbb{F}^{n},$ recall that its span over
$\mathbb{F}$ is denoted $\Span S.$ For a linear subspace $S,$ the
symbols $\dim S$ and $S^{\perp}$ refer as usual to the dimension
of $S$ and the orthogonal complement of $S,$ respectively. For a
linear transformation $M,$ we let $M(S)=\{Mx:x\in S\}$ denote the
image of $S$ under $M$. Recall that the \emph{sum} of linear subspaces
$S$ and $T$ is defined as $S+T=\{x+y:x\in S,y\in T\}$ and is the
smallest subspace that contains both $S$ and $T.$ In expressions
involving subspaces, we adopt the convention that the union $\cup$
and intersection $\cap$ operators have higher precedence than the
subspace sum operator $+$. For a vector space $V$ and an integer
$k,$ we adopt the notation $\Scal(V,k)$ for the set of all subspaces
of $V$ of dimension $k$. For arbitrary subspaces $S,T$ in a finite-dimensional
vector space, the following identity is well-known, and we use it
extensively in our proofs without further mention:
\begin{equation}
\dim(S+T)=\dim(S)+\dim(T)-\dim(S\cap T).\label{eq:sum-and-intersection-of-S-and-T}
\end{equation}
This equation is one of the few instances when subspaces behave in
ways analogous to sets. Such instances are rare. For example, unlike
sets, general subspaces $S,T,U$ need not satisfy $S\cap(T+U)=S\cap T+S\cap U.$
The equality requires additional hypotheses, as recorded below.
\begin{fact}
\label{fact:sum-intersection-subspaces}For any linear subspaces $S,S',T$
with $S'\subseteq S,$
\[
S\cap(S'+T)=S'+S\cap T.
\]
\end{fact}

\begin{proof}
It is clear that $S'+S\cap T$ is a subspace of both $S$ and $S'+T.$
It remains to prove the opposite inclusion, $S\cap(S'+T)\subseteq S'+S\cap T.$
For this, consider an arbitrary vector $u+v\in S$ with $u\in S'$
and $v\in T.$ Then $v\in S+u=S.$ As a result, $v\in S\cap T$ and
therefore $u+v\in S'+S\cap T$ as claimed.
\end{proof}
We continue with a fact that relates the dimension of $S\cap T$ to
that of $S^{\perp}\cap T^{\perp}$.
\begin{fact}
\label{fact:Aperp-intersect-Bperp}Let $S,T\subseteq\mathbb{F}^{n}$
be subspaces over a given field $\mathbb{F}$. Then
\begin{align}
(S+T)^{\perp} & =S^{\perp}\cap T^{\perp},\label{eq:aUNIONb-perp}\\
(S\cap T)^{\perp} & =S^{\perp}+T^{\perp},\label{eq:aINTERSECTb-perp}\\
\dim(S\cap T) & =\dim(S)+\dim(T)+\dim(S^{\perp}\cap T^{\perp})-n.\label{eq:AperpBperp-AplusB}
\end{align}
\end{fact}

\begin{proof}
To begin with,
\begin{align*}
S^{\perp}\cap T^{\perp} & =\{x:\langle x,y\rangle=0\text{ for all }y\in S\}\cap\{x:\langle x,y\rangle=0\text{ for all }y\in T\}\\
 & =\{x:\langle x,y\rangle=0\text{ for all }y\in S\cup T\}\\
 & =\{x:\langle x,y\rangle=0\text{ for all }y\in S+T\}\\
 & =(S+T)^{\perp},
\end{align*}
where the third step uses the linearity of inner product. This settles~\eqref{eq:aUNIONb-perp}.
Applying~\eqref{eq:aUNIONb-perp} to the orthogonal complements of
$S$ and $T$ results in $(S^{\perp}+T^{\perp})^{\perp}=S\cap T,$
which upon orthogonal complementation of both sides yields~\eqref{eq:aINTERSECTb-perp}.
Equation~\eqref{eq:AperpBperp-AplusB} is also a straightforward
consequence of~\eqref{eq:aUNIONb-perp}, as follows:
\begin{align*}
\dim(S^{\perp}\cap T^{\perp}) & =\dim((S+T)^{\perp})\\
 & =n-\dim(S+T)\\
 & =n-\dim(S)-\dim(T)+\dim(S\cap T).\qedhere
\end{align*}
\end{proof}
It is well-known that for a symmetric real matrix, any pair of eigenvectors
corresponding to distinct eigenvalues are orthogonal. For completeness,
we state this simple fact with a proof below.
\begin{fact}
\label{fact:diff-eigens-orthog}Let $M$ be a symmetric real matrix.
Let $u,v$ be eigenvectors of $M$ corresponding to different eigenvalues.
Then $\langle u,v\rangle=0.$
\end{fact}

\begin{proof}
Suppose that $Mu=\alpha u$ and $Mv=\beta v$, where $\alpha\ne\beta.$
Then $(\alpha-\beta)\langle u,v\rangle=\langle\alpha u,v\rangle-\langle u,\beta v\rangle=\langle Mu,v\rangle-\langle u,Mv\rangle=0,$
where the last step uses $M=M\tr.$ This forces $\langle u,v\rangle=0,$
as claimed.
\end{proof}

\subsection{Matrix norms}

Associated with every matrix $A\in\mathbb{C}^{n\times m}$ are $\min\{n,m\}$
nonnegative reals that are called the \emph{singular values of $A$},
denoted $\sigma_{1}(A)\geq\sigma_{2}(A)\geq\cdots\geq\sigma_{\min\{n,m\}}(A)$.
Every matrix $A\in\mathbb{C}^{n\times m}$ has a \emph{singular value
decomposition} $A=U\Sigma V^{*},$ where $U$ and $V$ are unitary
matrices of order $n$ and $m,$ respectively, and $\Sigma$ is a
rectangular diagonal matrix whose diagonal entries are $\sigma_{1}(A),\sigma_{2}(A),\ldots,\sigma_{\min\{n,m\}}(A).$
In the case of real matrices $A,$ the matrices $U$ and $V$ in the
singular value decomposition can be taken to be real. An alternative
characterization of the singular values is given by
\begin{fact}
\label{fact:singular-values-eigenvalues}Let $A\in\mathbb{C}^{n\times m}$
be given, with $n\leq m$. Then the singular values of $A$ are precisely
the square roots of the eigenvalues of $AA^{*},$ counting multiplicities. 
\end{fact}

The spectral norm, trace norm, and Frobenius norm of $A$ are defined
in terms of the singular values as follows:
\begin{align}
 & \|A\|=\sigma_{1}(A),\label{eq:spectral-def-sigma}\\
 & \|A\|_{\Sigma}=\sum\sigma_{i}(A),\label{eq:trace-def-sigma}\\
 & \|A\|_{\mathrm{F}}=\sqrt{\sum\sigma_{i}(A)^{2}}.\label{eq:Frobenius-def-sigma}
\end{align}
Equivalently,
\begin{align}
 & \|A\|=\max_{x:\|x\|_{2}=1}\|Ax\|_{2},\label{eq:spectral-def}\\
 & \|A\|_{\mathrm{F}}=\sqrt{\sum|A_{ij}|^{2}}.\label{eq:Frob-def}
\end{align}
These equations agree with~\prettyref{eq:spectral-def-sigma} and~\prettyref{eq:Frobenius-def-sigma}
because the Euclidean norm on vectors is invariant under unitary transformations.
\begin{fact}
\label{fact:dual-matrix-norms}For any matrices $A,B\in\mathbb{C}^{n\times m},$
\[
|\langle A,B\rangle|\leq\|A\|\;\|B\|_{\Sigma}.
\]
\end{fact}

\noindent \prettyref{fact:dual-matrix-norms} follows directly from~\prettyref{eq:spectral-def}
and the singular value decomposition of~$B$. We now recall a relationship
between  the trace norm and Frobenius norm; see, e.g.,~\cite[Prop.~2.4]{sherstov07quantum}.
\begin{fact}
\label{fact:bound-on-trace-of-product} For all matrices $A$ and
$B$ of compatible dimensions, 
\begin{align*}
\|AB\|_{\Sigma}\leq\|A\|_{\mathrm{F}}\;\|B\|_{\mathrm{F}}.
\end{align*}
\end{fact}

Recall that a\emph{ sign matrix} is a real matrix with entries in
$\{-1,1\}.$ A \emph{partial} sign matrix, then, is a matrix with
entries in $\{-1,1,*\}$. We define the \emph{domain }of a partial
sign matrix $F$ by $\dom F=\{(i,j):F_{ij}\neq*\}$. The $\epsilon$\emph{-approximate
trace norm of $F$}, denoted $\|F\|_{\Sigma,\epsilon}$, is the least
trace norm of a real matrix $\widetilde{F}$ that satisfies
\begin{align}
|F_{ij}-\widetilde{F}_{ij}|\leq\epsilon & \qquad\qquad\text{if }F_{ij}\in\{-1,1\},\label{eq:partialapprox1}\\
|\widetilde{F}_{ij}|\leq1+\epsilon & \qquad\qquad\text{if }F_{ij}=*.\label{eq:partialapprox2}
\end{align}
The following lower bound on the approximate trace norm is well known~\cite{lee-shraibman09cc-survey,sherstov07quantum,sherstov11quantum-sdpt}.
For reader's convenience, we include a proof.
\begin{prop}
For any partial sign matrix $F$ and $\epsilon\geq0,$\label{prop:approxtracelower}
\begin{align*}
\|F\|_{\Sigma,\epsilon} & \geq\sup_{\Phi\neq0}\frac{1}{\|\Phi\|}\left(\sum_{(i,j)\in\dom F}F_{ij}\Phi_{ij}-\epsilon\|\Phi\|_{1}-\sum_{(i,j)\notin\dom F}|\Phi_{ij}|\right).
\end{align*}
\end{prop}

\begin{proof}
Let $\widetilde{F}$ be a real matrix that approximates $F$ in the
sense of~\eqref{eq:partialapprox1} and~\eqref{eq:partialapprox2}.
Then for any $\Phi\ne0,$
\begin{align*}
\langle\widetilde{F},\Phi\rangle & =\sum_{\dom F}F_{ij}\Phi_{ij}+\sum_{\dom F}(\widetilde{F}_{ij}-F_{ij})\Phi_{ij}+\sum_{\overline{\dom F}}\widetilde{F}_{ij}\Phi_{ij}\\
 & \geq\sum_{\dom F}F_{ij}\Phi_{ij}-\sum_{\dom F}|\widetilde{F}_{ij}-F_{ij}|\,|\Phi_{ij}|-\sum_{\overline{\dom F}}|\widetilde{F}_{ij}|\,|\Phi_{ij}|\\
 & \geq\sum_{\dom F}F_{ij}\Phi_{ij}-\sum_{\dom F}\epsilon|\Phi_{ij}|-\sum_{\overline{\dom F}}(1+\epsilon)|\Phi_{ij}|\\
 & =\sum_{\dom F}F_{ij}\Phi_{ij}-\epsilon\|\Phi\|_{1}-\sum_{\overline{\dom F}}|\Phi_{ij}|.
\end{align*}
On the other hand, Fact~\ref{fact:dual-matrix-norms} shows that
$\langle\widetilde{F},\Phi\rangle\leq\|\widetilde{F}\|_{\Sigma}\left\Vert \Phi\right\Vert $.
Combining these two bounds for $\langle\widetilde{F},\Phi\rangle$
gives
\[
\|\widetilde{F}\|_{\Sigma}\geq\frac{1}{\|\Phi\|}\left(\sum_{\dom F}F_{ij}\Phi_{ij}-\epsilon\|\Phi\|_{1}-\sum_{\overline{\dom F}}|\Phi_{ij}|\right).
\]
Taking the supremum over $\Phi\ne0$ completes the proof.
\end{proof}

\subsection{\label{sec:Fourier}Fourier transform}

Consider a prime power $q=p^{k},$ with $p$ a prime and $k$ a positive
integer. Recall that the additive group of $\FF_{q}$ is isomorphic
to the Abelian group $\mathbb{Z}_{p}^{k}$. Fix any such isomorphism
$\psi$. Let $\omega=e^{2\pi i/p}$, a primitive $p$-th root of unity.
For $x\in\FF_{q},$ define $\omega^{x}=\omega^{x_{1}}\omega^{x_{2}}\cdots\omega^{x_{k}},$
where $(x_{1},x_{2},\ldots,x_{k})$ is the image of $x$ under $\psi.$
Then for all $x,y\in\mathbb{F}_{q},$
\begin{align}
\omega^{x+y} & =\omega^{x}\omega^{y},\label{eq:omega-aplusb}\\
\omega^{-x} & =\overline{\omega^{x}}.\label{eq:omega-inverses}
\end{align}
One further calculates $\sum_{x\in\mathbb{F}_{q}}\omega^{x}=\prod_{i=1}^{k}(1+\omega+\omega^{2}+\cdots+\omega^{p-1})=0$,
which in turn generalizes to 
\begin{equation}
\sum_{x\in\mathbb{F}_{q}}\omega^{ax}=0,\qquad\qquad a\in\mathbb{F}_{q}\setminus\{0\}\label{eq:omegasum}
\end{equation}
since $x\mapsto ax$ is a permutation on $\mathbb{F}_{q}.$ 

Let $n$ be a positive integer. For $A\in\FF_{q}^{n\times n}$, define
a corresponding character $\chi_{A}:\FF_{q}^{n\times n}\to\CC$ by
\begin{align*}
\chi_{A}(X) & =\omega^{\langle A,X\rangle}.
\end{align*}
It follows from~\prettyref{eq:omega-aplusb} that
\begin{equation}
\chi_{A}(X+Y)=\chi_{A}(X)\chi_{A}(Y),\label{eq:character-homomorphism}
\end{equation}
making $\chi_{A}$ a homomorphism of the additive group $\mathbb{F}_{q}^{n\times n}$
into the multiplicative group of $\mathbb{C}$. Using~\prettyref{eq:omega-aplusb}
and \prettyref{eq:omega-inverses}, one obtains $\langle\chi_{A},\chi_{B}\rangle=\sum_{X}\omega^{\langle A,X\rangle}\overline{\omega^{\langle B,X\rangle}}=\sum_{X}\omega^{\langle A,X\rangle-\langle B,X\rangle}=\sum_{X}\omega^{\langle A-B,X\rangle},$
which along with~\eqref{eq:omegasum} leads to
\begin{align}
\langle\chi_{A},\chi_{B}\rangle & =\begin{cases}
q^{n^{2}} & \text{{if} }A=B,\\
0 & \text{{otherwise}.}
\end{cases}\label{eq:innerchar}
\end{align}
Hence, the characters $\chi_{A}$ for $A\in\mathbb{F}_{q}^{n\times n}$
form an orthogonal basis for the complex vector space of functions
$\FF_{q}^{n\times n}\to\CC$. In particular, every function $f\colon\mathbb{F}_{q}^{n\times n}\to\mathbb{C}$
has a unique representation as a linear combination of the characters:
\begin{align}
f(X) & =\sum_{A\in\FF_{q}^{n\times n}}\widehat{f}(A)\chi_{A}(X).\label{eq:f-reverse-fourier-transform}
\end{align}
The numbers $\widehat{f}(A)$ are called the \emph{Fourier coefficients
of $f$}. They are given by 
\begin{equation}
\widehat{f}(A)=q^{-n^{2}}\langle f,\chi_{A}\rangle=\Exp_{X\in\mathbb{F}_{q}^{n\times n}}f(X)\omega^{-\langle A,X)}.\label{eq:Fourier-transform}
\end{equation}
where the first step is justified by~\eqref{eq:innerchar}, and the
second step uses~\prettyref{eq:omega-inverses}. An immediate consequence
of \prettyref{eq:innerchar} and~\eqref{eq:f-reverse-fourier-transform}
is that $\langle f,f\rangle=q^{n^{2}}\sum_{A}|\widehat{f}(A)|^{2}$.
This result is known as \emph{Parseval's identity,} and it is typically
written in the form
\begin{equation}
\Exp_{X\in\mathbb{F}_{q}^{n\times n}}[|f(X)|^{2}]=\sum_{A\in\mathbb{F}_{q}^{n\times n}}|\widehat{f}(A)|^{2}.\label{eq:parsevals}
\end{equation}

With $\widehat{f}$ viewed as a complex-valued function on $\mathbb{F}_{q}^{n\times n},$
the linear transformation that sends $f\mapsto\widehat{f}$ is called
the \emph{Fourier transform}. Its matrix representation is easy to
describe. Specifically, define
\[
H_{n}=q^{-n^{2}/2}[\omega^{\langle A,B\rangle}]_{A,B},
\]
where the row and column indices range over all matrices in $\mathbb{F}_{q}^{n\times n}.$
Analogous to~\eqref{eq:innerchar}, one shows that $H_{n}$ is unitary:
\begin{equation}
H_{n}H_{n}^{*}=H_{n}^{*}H_{n}=I.\label{eq:Hnq-unitary}
\end{equation}
Then the Fourier transform $f\mapsto\widehat{f}$, given by~\eqref{eq:Fourier-transform},
corresponds to the linear transformation $q^{-n^{2}/2}H_{n}^{*}.$
Analogously, the inverse transformation $\widehat{f}\mapsto f$ of~\eqref{eq:f-reverse-fourier-transform}
corresponds to $q^{n^{2}/2}H_{n}.$ 

The following well-known fact relates the singular values of a matrix
$[\phi(A+B)]_{A,B}$ to the Fourier spectrum of the outer function
$\phi$. We include a proof adapted from \cite{LiSWW14communication-linear}
and generalized to the case of $\FF_{q}$.
\begin{fact}[adapted from Li et al., Lemma 20]
\label{fact:spectraltofourier}  Let $\phi:\FF_{q}^{n\times n}\to\CC$
be given. Define 
\[
\Phi=[\phi(X+Y)]_{X,Y\in\mathbb{F}_{q}^{n\times n}}.
\]
Then
\[
\Phi=H_{n}DH_{n},
\]
where $D$ is the diagonal matrix given by $D_{A,A}=q^{n^{2}}\widehat{\phi}(A).$
In particular, the singular values of $\Phi$ are $q^{n^{2}}|\widehat{\phi}(A)|$
for $A\in\mathbb{F}_{q}^{n\times n}.$
\end{fact}

\begin{proof}
Using the homomorphic property~\eqref{eq:character-homomorphism}
of the characters,
\begin{align*}
\phi(X+Y) & =\sum_{A\in\FF_{q}^{n\times n}}\widehat{\phi}(A)\chi_{A}(X+Y)\\
 & =\sum_{A\in\FF_{q}^{n\times n}}\widehat{\phi}(A)\chi_{A}(X)\chi_{A}(Y).
\end{align*}
Restated in matrix form, this equation becomes $\Phi=[\chi_{A}(X)]_{X,A}\,\diag(\ldots,\widehat{\phi}(A),\ldots)\,[\chi_{A}(Y)]_{A,Y}=H_{n}DH_{n},$
as desired.
\end{proof}

\subsection{Gaussian binomial coefficients}

\emph{Gaussian binomial coefficients}, also known as \emph{$q$-binomial
coefficients}, are defined by
\begin{align}
\gbinom nm & =\frac{(q^{n}-1)(q^{n}-q)\cdots(q^{n}-q^{m-1})}{(q^{m}-1)(q^{m}-q)\cdots(q^{m}-q^{m-1})}\label{eq:gbinom-def1}\\
 & =\frac{(q^{n}-1)(q^{n-1}-1)\cdots(q^{n-m+1}-1)}{(q^{m}-1)(q^{m-1}-1)\cdots(q-1)}\label{eq:gbinom-def2}
\end{align}
for all nonnegative integers $n,m$ and real numbers $q>1$. Observe
that $\binom{n}{0}_{q}=1$ since the above product is empty for $m=0$.
Note further that $\binom{n}{m}_{q}=0$ whenever $m>n$. One recovers
standard binomial coefficients from this definition via
\[
\lim_{q\searrow1}\binom{n}{m}_{q}=\binom{n}{m}.
\]
As a matter of convenience, one generalizes Gaussian binomial coefficients
to arbitrary integers $n,m$ by defining 
\[
\binom{n}{m}_{q}=0\qquad\qquad\text{if }\min\{n,m\}<0.
\]
With this convention, one has the familiar identity
\begin{align}
\binom{n}{m}_{q} & =\binom{n}{n-m}_{q},\qquad\qquad n,m\in\mathbb{Z}.\label{eq:gbinom-symmetry}
\end{align}
Gaussian binomial coefficients play an important role in enumerative
combinatorics. In particular, we recall the following classical fact.
\begin{fact}
\label{fact:qbinprop}Fix a prime power $q$ and integers $n\geq m\geq0$.
Then the number of $m$-dimensional subspaces of $\mathbb{F}_{q}^{n}$
is exactly $\binom{n}{m}_{q}.$
\end{fact}

\begin{proof}
This result is clearly true for $m=0.$ For $m\geq1,$ there are $(q^{n}-1)(q^{n}-q)\cdots(q^{n}-q^{m-1})$
ordered bases $(v_{1},v_{2},\ldots,v_{m})$ of vectors in $\mathbb{F}_{q}^{n}$.
Each such basis defines an $m$-dimensional subspace. Conversely,
every $m$-dimensional subspace has exactly $(q^{m}-1)(q^{m}-q)\cdots(q^{m}-q^{m-1})$
ordered bases. Thus, the number of $m$-dimensional subspaces is~\eqref{eq:gbinom-def1},
as claimed.
\end{proof}
The following monotonicity property of $q$-binomial coefficients
is well-known. We provide a proof for convenience.
\begin{fact}
\label{fact:q-binomial-monotone} Let $n\geq m\geq0$ be given integers.
Then for all integers $\ell\in[m,n-m]$ and reals $q>1,$
\begin{equation}
\binom{n}{m}_{q}\leq\binom{n}{\ell}_{q}.\label{eq:n-m-ell-monotone}
\end{equation}
\end{fact}

\begin{proof}
The defining equation~\prettyref{eq:gbinom-def2} gives
\[
\binom{n}{\ell}_{q}=\binom{n}{m}_{q}\cdot\prod_{i=m+1}^{\ell}\frac{q^{n-i+1}-1}{q^{i}-1}.
\]
If $\ell\leq n/2,$ then every fraction in the above product is greater
than $1$. As a result, \eqref{eq:n-m-ell-monotone} holds in this
case. In the complementary case $\ell>n/2$, we have $n-\ell\in[m,n/2]$
and therefore 
\[
\binom{n}{m}_{q}\leq\binom{n}{n-\ell}_{q}
\]
by the first part of the proof. Since $\binom{n}{n-\ell}_{q}=\binom{n}{\ell}_{q},$
we again arrive at~\eqref{eq:n-m-ell-monotone}. 
\end{proof}
We will use the next proposition to accurately estimate Gaussian binomial
coefficients. 
\begin{prop}
For any set $I$ of positive integers, and any real number $x\geq2$,
\label{prop:prodprop}
\begin{align*}
\frac{1}{4}\leq\prod_{i\in I}\left(1-\frac{1}{x^{i}}\right) & \leq1.
\end{align*}
\end{prop}

\begin{proof}
The upper bound is trivial. For the lower bound, we may clearly assume
that $I=\{1,2,3\ldots\}.$ A simple inductive argument shows that
$(1-a_{1})\cdots(1-a_{n})\geq1-a_{1}-\cdots-a_{n}$ for any $a_{1},\ldots,a_{n}\in(0,1)$.
It follows that
\[
\prod_{i=2}^{\infty}\left(1-\frac{1}{x^{i}}\right)\geq1-\frac{1}{x^{2}}-\frac{1}{x^{3}}-\ldots=1-\frac{1}{x(x-1)}
\]
and therefore
\[
\prod_{i=1}^{\infty}\left(1-\frac{1}{x^{i}}\right)\geq\left(1-\frac{1}{x}\right)\left(1-\frac{1}{x(x-1)}\right)\geq\frac{1}{4},
\]
where the last step uses $x\geq2.$
\end{proof}
\begin{cor}
For any integers $n\geq m\geq0$ and any real number $q\geq2,$ \label{cor:qbinbound}
\begin{align*}
q^{m(n-m)} & \leq\gbinom nm\leq4q^{m(n-m)}.
\end{align*}
\end{cor}

\begin{proof}
The lower bound follows directly from the fact that $(q^{n}-q^{i})/(q^{m}-q^{i})\geq q^{n}/q^{m}$
for $n\geq m$. The upper bound can be verified as follows:
\[
\gbinom nm=\frac{{(q^{n}-1)(q^{n}-q)\ldots(q^{n}-q^{m-1})}}{(q^{m}-1)(q^{m}-q)\ldots(q^{m}-q^{m-1})}\leq\frac{{q^{nm}}}{q^{m^{2}}\prod_{i=1}^{m}(1-q^{-i})}\leq4q^{m(n-m)},
\]
where the last step applies~\prettyref{prop:prodprop}.
\end{proof}
We now recall a classical result known as the \emph{Cauchy binomial
theorem}, see, e.g.,~\cite[eqn.~(1.87)]{enumerative-comb86stanley}.
\begin{fact}
For any integer $n\geq1$ and real number $q>1,$ the following identity
holds in $\mathbb{R}[t]$\emph{:}
\begin{align}
(1+t)(1+qt)\ldots(1+q^{n-1}t)= & \sum_{i=0}^{n}q^{\binom{i}{2}}\gbinom[q]nit^{i}.\label{eq:cauchybinom}
\end{align}
\end{fact}

\begin{cor}
For any integer $n\geq1$ and real number $q>1,$ and any real polynomial
$g$ of degree less than $n,$\label{cor:qbinompoly}
\begin{align}
\sum_{i=0}^{n}(-1)^{i}q^{\binom{i}{2}}\gbinom[q]{n}{i}g(q^{-i}) & =0.\label{eq:orthog-binom}
\end{align}
\end{cor}

\begin{proof}
For $d=0,1,\ldots,n-1$, take $t=-1/q^{d}$ in \eqref{eq:cauchybinom}
to obtain
\begin{align}
\sum_{i=0}^{n}(-1)^{i}q^{\binom{i}{2}}\gbinom ni(q^{-i})^{d} & =0.\label{eq:binom-monom}
\end{align}
This establishes~\eqref{eq:orthog-binom} when $g$ is a \emph{monomial}
of degree less than $n$. The general case follows by linearity: multiply~\eqref{eq:binom-monom}
by the degree-$d$ coefficient in $g$ and sum over $d.$

\end{proof}

\subsection{Counting and generating matrices of given rank}

For a field $\mathbb{F},$ we let $\Mcal_{r}^{\mathbb{F},n,m}$ denote
the set of matrices in $\FF^{n\times m}$ of rank $r$. Since we mostly
use $\mathbb{F}=\mathbb{F}_{q}$ in this work, we will usually omit
the reference to the field and write simply $\Mcal_{r}^{n,m}.$ As
a matter of convenience, we adopt the convention that for any $n\geq0$
there is exactly one ``matrix'' of size $0\times n$ and exactly
one ``matrix'' of size $n\times0,$ both of rank $0$. The role of
these empty matrices is to ensure that
\[
|\Mcal_{0}^{0,n}|=|\Mcal_{0}^{n,0}|=1,\qquad\qquad n\geq0,
\]
which simplifies the statement of several lemmas in this paper. Analogously,
we define
\begin{equation}
\Mcal_{r}^{n,m}=\varnothing\qquad\text{if }\min\{n,m,r\}<0.\label{eq:Mcal-nmr-empty}
\end{equation}
For nonsingular matrices of order $n\geq1$, we adopt the shorthand
$\Mcal_{n}=\Mcal_{n}^{n,n}$.
\begin{prop}
\label{prop:nummatr} Let $n,m,r$ be nonnegative integers with $r\leq\min\{n,m\}$.
Then
\begin{align}
|\Mcal_{r}^{n,m}| & =\gbinom nr(q^{m}-1)(q^{m}-q)\ldots(q^{m}-q^{r-1}).\label{eq:Mcal-nmr}
\end{align}
\end{prop}

\begin{proof}
If $r=0$, then the right-hand side of~\prettyref{eq:Mcal-nmr} evaluates
to~$1$. This is consistent with our convention that $|\Mcal_{0}^{n,m}|=1$
for all $n,m\geq0.$ 

We now consider the complementary case $r\geq1,$ which forces $n$
and $m$ to be positive. Fix an arbitrary $r$-dimensional subspace
$S\subseteq\mathbb{F}_{q}^{n}$ and consider the subset $\Mcal_{S}\subseteq\Mcal_{r}^{n,m}$
of matrices whose column space is $S$. Fix an $n\times r$ matrix
$A$ with column space $S.$ Since the columns of $A$ are linearly
independent, every matrix in $\Mcal_{S}$ has a unique representation
of the form $AB$ for some $B\in\Mcal_{r}^{r,m}.$ Conversely, any
product $AB$ with $B\in\Mcal_{r}^{r,m}$ is a matrix in $\Mcal_{S}.$
Therefore,
\begin{equation}
|\Mcal_{S}|=|\Mcal_{r}^{r,m}|.\label{eq:Mcal-S-Mcal-rm-n}
\end{equation}
Recall that $\Mcal_{r}^{n,m}$ is the disjoint union of $\Mcal_{S}$
over $r$-dimensional subspaces $S\subseteq\mathbb{F}_{q}^{n},$ and
there are precisely $\binom{n}{r}_{q}$ such subspaces (\prettyref{fact:qbinprop}).
With this in mind,~\prettyref{eq:Mcal-S-Mcal-rm-n} leads to
\begin{equation}
|\Mcal_{r}^{n,m}|=\binom{n}{r}_{q}|\Mcal_{r}^{r,m}|.\label{eq:Mcal-nm-rn}
\end{equation}
Finally, the number of $r\times m$ matrices of rank $r$ is precisely
the number of bases $(v_{1},v_{2},\ldots,v_{r})$ of row vectors in
$\mathbb{F}_{q}^{m},$ whence $|\Mcal_{r}^{r,m}|=(q^{m}-1)(q^{m}-q)\cdots(q^{m}-q^{r-1}).$
Making this substitution in~\prettyref{eq:Mcal-nm-rn} completes
the proof.
\end{proof}
Using \prettyref{prop:prodprop} and \prettyref{cor:qbinbound} to
estimate the right-hand side of~\prettyref{eq:Mcal-nmr}, we obtain:
\begin{cor}
\label{cor:nummatrbound} Let $m,n,r$ be nonnegative integers with
$r\leq\min\{n,m\}$. Then 
\begin{align*}
 & \frac{1}{4}q^{r(n+m-r)}\leq|\Mcal_{r}^{n,m}|\leq4q^{r(n+m-r)}.
\end{align*}
\end{cor}

The following fact is well-known; cf.~\cite{LiSWW14communication-linear}. 
\begin{prop}
\label{prop:random-matrices}Let $n\geq1$ be a given integer. Let
$X,Y$ be random matrices distributed independently and uniformly
on $\Mcal_{n}.$ Then: 
\begin{enumerate}
\item \label{enu:XA-AX}for any fixed $A\in\Mcal_{n},$ the matrices $XA$
and $AX$ are distributed uniformly on $\Mcal_{n};$
\item \label{enu:XAY}for any $r\in\{0,1,\ldots,n\}$ and fixed $A\in\Mcal_{r}^{n,n},$
the matrix $XAY$ is distributed uniformly on $\Mcal_{r}^{n,n}.$
\end{enumerate}
\end{prop}

\begin{proof}
\prettyref{enu:XA-AX}~For any $B\in\Mcal_{n},$ we have $\Prob[XA=B]=\Prob[X=BA^{-1}]=1/|\Mcal_{n}|.$
Therefore, $XA$ is distributed uniformly on $\Mcal_{n}.$ The argument
for $AX$ is analogous.

\prettyref{enu:XAY}~Fix $B\in\Mcal_{r}^{n,n}$ arbitrarily. Then
$B$ can be obtained from $A$ by a series of elementary row and column
operations, so that $B=M_{1}AM_{2}$ for nonsingular $M_{1},M_{2}.$
As a result, 
\[
\Prob[XAY=B]=\Prob[M_{1}^{-1}XAYM_{2}^{-1}=A]=\Prob[XAYM_{2}^{-1}=A]=\Prob[XAY=A],
\]
where the last two steps are valid by part~\prettyref{enu:XA-AX}.
To summarize, $XAY$ takes on every value in $\Mcal_{r}^{n,n}$ with
the same probability. Since $XAY\in\Mcal_{r}^{n,n}$, the proof is
complete.
\end{proof}

\subsection{Random projections}

Given a collection of subspaces $S_{1},S_{2},\ldots,S_{m}$ in a vector
space, we use random projections to reduce the dimension of the ambient
space while preserving algebraic relationships among the $S_{i}.$
This is done by choosing a uniformly random matrix $X$ and replacing
$S_{1},S_{2},\ldots,S_{m}$ with the subspaces $X(S_{1}),X(S_{2}),\ldots,X(S_{m}),$
respectively. The following lemma provides quantitative details.
\begin{lem}
\label{lem:random-proj-subspaces}Let $n$ and $d$ be positive integers,
$\mathbb{F}$ a finite field with $q=|\mathbb{F}|$ elements, and
$S\subseteq\mathbb{F}^{n}$ a subspace. Then for every integer $t\leq\min\{\dim(S),d\},$
\begin{equation}
\Prob_{X\in\mathbb{F}^{d\times n}}[\dim(X(S))\leq t]\leqslant4q^{-(\dim(S)-t)(d-t)}.\label{eq:subspace-doesnt-shrink}
\end{equation}
In particular, for every integer $T\leq\min\{\dim(S),d\},$
\begin{equation}
\Exp_{X\in\mathbb{F}^{d\times n}}\,q^{T-\min\{T,\dim(X(S))\}}\leq1+8q^{-(\dim(S)-T+1)(d-T+1)+1}.\label{eq:subspace-shrink-expectation}
\end{equation}
 
\end{lem}

\begin{proof}
Equations~\prettyref{eq:subspace-doesnt-shrink} and~\prettyref{eq:subspace-shrink-expectation}
hold trivially for negative $t$ and $T,$ respectively. As a result,
we may assume that $t\ge0$ and $T\geq0.$ Abbreviate $k=\dim(S)$.
Fix a basis $v_{1},v_{2},\dots,v_{k}$ for $S$ and extend it to a
basis $v_{1},v_{2},\ldots,v_{n}$ for $\mathbb{F}^{n}.$ Let $A\in\mathbb{F}^{n\times n}$
be the unique matrix such that $Av_{i}=e_{i}$ for each $i=1,2,\dots,n$.
In particular, $A(S)=\Span\{e_{1},e_{2},\ldots,e_{k}\}.$ Now, let
$X\in\mathbb{F}^{d\times n}$ be uniformly random. Then the rows of
$XA$ are independent random variables, each a uniformly random linear
combination of the rows of $A.$ Since $A$ is nonsingular of order
$n,$ it follows that the rows of $XA$ are independent random vectors
in $\mathbb{F}^{n}.$ Put another way, $XA\in\mathbb{F}^{d\times n}$
has the same distribution as $X$. As a result,
\begin{align}
\Prob[\dim(X(S))\leq t] & =\Prob[\dim(XA(S))\leq t]\nonumber \\
 & =\Prob[\dim(X(A(S)))\leq t]\nonumber \\
 & =\Prob[\dim(\Span\{Xe_{1},Xe_{2},\ldots,Xe_{k}\})\leq t]\nonumber \\
 & =\Prob[\exists B\in\Scal(\mathbb{F}^{d},t)\text{ such that }Xe_{1},Xe_{2},\ldots,Xe_{k}\in B]\nonumber \\
 & \leq\sum_{B\in\Scal(\mathbb{F}^{d},t)}\Prob[Xe_{1},Xe_{2},\ldots,Xe_{k}\in B],\label{eq:proj-bound-intermed}
\end{align}
where the third step uses $A(S)=\Span\{e_{1},e_{2},\ldots,e_{k}\}$,
and the last step applies a union bound. Now
\[
\Prob[\dim(X(S))\leq t]\leq\sum_{\Scal(\mathbb{F}^{d},t)}\left(\frac{q^{t}}{q^{d}}\right)^{k}=\binom{d}{t}_{q}q^{-k(d-t)}\leq4q^{t(d-t)}q^{-k(d-t)}=4q^{-(k-t)(d-t)},
\]
where the first step is justified by~\prettyref{eq:proj-bound-intermed}
and the fact that $Xe_{1},Xe_{2},\ldots,Xe_{k}$ are independent and
uniformly random vectors in $\mathbb{F}^{d}$; the second step applies~\prettyref{fact:qbinprop};
and the third step uses~\prettyref{cor:qbinbound}. This settles~\prettyref{eq:subspace-doesnt-shrink}.
Now~\prettyref{eq:subspace-shrink-expectation} can be verified as
follows:
\begin{align*}
\Exp\,q^{T-\min\{T,\dim(X(S))\}} & \leq1+\sum_{t=0}^{T-1}q^{T-t}\Prob[\dim(X(S))=t]\\
 & \leq1+\sum_{t=0}^{T-1}q^{T-t}\cdot4q^{-(k-t)(d-t)}\\
 & =1+\sum_{t=1}^{T}q^{t}\cdot4q^{-(k-T+t)(d-T+t)}\\
 & =1+\sum_{t=1}^{T}q^{t}\cdot4q^{-(k-T+1)(d-T+1)-(t-1)(d+k+t-2T+1)}\\
 & \leq1+\sum_{t=1}^{\infty}q^{t}\cdot4q^{-(k-T+1)(d-T+1)-(t^{2}-1)}\\
 & \leq1+4q^{-(k-T+1)(d-T+1)+1}\cdot\frac{q}{q-1}\\
 & \leq1+8q^{-(k-T+1)(d-T+1)+1},
\end{align*}
where the third step is a change of variable, the next-to-last step
bounds the series by a geometric series, and the last step is valid
due to $q\geq2.$
\end{proof}
The previous lemma yields an analogous results for matrices:
\begin{lem}
\label{lem:random-proj-matrices}Let $n,m,d$ be positive integers,
$\mathbb{F}$ a finite field with $q=|\mathbb{F}|$ elements, and
$M\in\mathbb{F}^{n\times m}$ a given matrix. Then for every integer
$t\leq\min\{\rk M,d\}$$:$
\begin{enumerate}
\item \label{enu:left-matrix-proj}$\Prob[\rk(XM)\leq t]\leq4q^{-(\rk(M)-t)(d-t)}$
for a uniformly random matrix $X\in\mathbb{F}^{d\times n};$
\item \label{enu:right-matrix-proj}$\Prob[\rk(MY)\leq t]\leq4q^{-(\rk(M)-t)(d-t)}$
for a uniformly random matrix $Y\in\mathbb{F}^{m\times d}.$
\end{enumerate}
\end{lem}

\begin{proof}
Let $S$ be the column span of $M.$ Then $\rk(XM)=\dim(X(S)),$ and
\prettyref{enu:left-matrix-proj} follows from~\prettyref{lem:random-proj-subspaces}.
For~\prettyref{enu:right-matrix-proj}, rewrite the probability of
interest as $\Prob[\rk(Y\tr M\tr)\leq t]$ and apply~\prettyref{enu:left-matrix-proj}.
\end{proof}

\subsection{Communication complexity}

An excellent reference on communication complexity is the monograph
by Kushilevitz and Nisan~\cite{ccbook}. In this overview, we will
limit ourselves to key definitions and notation. The \emph{public-coin
randomized model},~due to Yao~\cite{yao79cc}, features two players
Alice and Bob and a (possibly partial) Boolean function $F\colon X\times Y\to\{-1,1,*\}$
for finite sets $X$ and $Y.$ Alice is given as input an element
$x\in X,$ Bob is given $y\in Y$, and their objective is to evaluate
$F(x,y)$. To this end, Alice and Bob communicate by sending bits
according to a protocol agreed upon in advance. Moreover, they have
an unlimited supply of shared random bits which they can use when
deciding what message to send at any given point in the protocol.
Eventually, they must agree on a bit ($-1$ or $1$) that represents
the output of the protocol. An \emph{$\epsilon$-error protocol} for
$F$ is one which, on every input $(x,y)\in\dom F,$ produces the
correct answer $F(x,y)$ with probability at least $1-\epsilon.$
The protocol's behavior on inputs outside $\dom F$ can be arbitrary.
The \emph{cost} of a protocol is the total bit length of the messages
exchanged by Alice and Bob in the worst-case execution of the protocol.
The \emph{$\epsilon$-error randomized communication complexity of
$F,$} denoted $R_{\epsilon}(F),$ is the least cost of an $\epsilon$-error
randomized protocol for $F$. The standard setting of the error parameter
is $\epsilon=1/3,$ which can be replaced by any other constant in
$(0,1/2)$ with only a constant-factor change in communication cost.

A far-reaching generalization of the randomized model is Yao's \emph{quantum
model}~\cite{yao93quantum}, where Alice and Bob exchange \emph{quantum}
messages. As before, their objective is to evaluate a fixed function
$F\colon X\times Y\to\{-1,1,*\}$ on any given input pair $(x,y),$
where Alice receives as input $x$ and Bob receives $y.$ We allow
arbitrary \emph{prior entanglement} at the start of the communication,
which is the quantum analogue of shared randomness. A measurement
at the end of the protocol produces a one-bit answer, which is interpreted
as the protocol output. An\emph{ $\epsilon$-error protocol for $F$}
is required to output, on every input $(x,y)\in\dom F,$ the correct
value $F(x,y)$ with probability at least $1-\epsilon.$ As before,
the protocol can exhibit arbitrary behavior on inputs outside $\dom F$.
The \emph{cost} of a quantum protocol is the total number of quantum
bits exchanged in the worst-case execution. The \emph{$\epsilon$-error
quantum communication complexity of $F$}, denoted $Q_{\epsilon}^{*}(F),$
is the least cost of an $\epsilon$-error quantum protocol for $F.$
The asterisk in $Q_{\epsilon}^{*}(F)$ indicates that the parties
can share arbitrary prior entanglement. As before, the standard setting
of the error parameter is $\epsilon=1/3.$ For a detailed formal description
of the quantum model, we refer the reader to~\cite{dewolf-thesis,razborov03quantum,sherstov07quantum}.
For any protocol $\Pi,$ quantum or otherwise, we write $\cost(\Pi)$
for the communication cost of $\Pi.$

The following theorem, due to Linial and Shraibman~\cite[Lem.~10]{linial07factorization-stoc},
states that the matrix of the acceptance probabilities of a quantum
protocol has an efficient factorization with respect to the Frobenius
norm. Closely analogous statements were established earlier by Yao~\cite{yao93quantum},
Kremer~\cite{kremer95thesis}, and Razborov~\cite{razborov03quantum}.
\begin{thm}
\label{thm:protocol2matrix} Let $X,Y$ be finite sets. Let $P$ be
a quantum protocol $($with or without prior entanglement$)$ with
cost $C$ qubits and input sets $X$ and $Y.$ Then 
\[
\Big[\Prob[P(x,y)=1]\Big]_{x\in X,y\in Y}=AB
\]
for some real matrices $A,B$ with $\|A\|_{\mathrm{F}}\leq2^{C}\sqrt{|X|}$
and $\|B\|_{\mathrm{F}}\leq2^{C}\sqrt{|Y|}.$
\end{thm}

\noindent Theorem~\ref{thm:protocol2matrix} provides a transition
from quantum protocols to matrix factorization, which is by now a
standard technique that has been used by various authors in various
contexts. Among other things, Theorem~\ref{thm:protocol2matrix}
gives the following \emph{approximate trace norm method} for quantum
lower bounds; see, e.g.,~\cite[Thm.~5.5]{razborov03quantum}. For
the reader's convenience, we state and prove this result in the generality
that we require.
\begin{thm}[Approximate trace norm method]
\label{thm:approx-trace-norm-method}Let $F\colon X\times Y\to\{-1,1,*\}$
be a $($possibly partial$)$ communication problem. Then
\[
4^{Q_{\epsilon}^{*}(F)}\geq\frac{\|M\|_{\Sigma,2\epsilon}}{3\sqrt{|X|\,|Y|}},
\]
where $M=[F(x,y)]_{x\in X,y\in Y}$ is the characteristic matrix of
$F.$
\end{thm}

\begin{proof}
Let $P$ be a quantum protocol with prior entanglement that computes
$F$ with error $\epsilon$ and cost $C.$ Put 
\[
\Pi=\Big[\Prob[P(x,y)=1]\Big]_{x\in X,\,y\in Y}.
\]
Then the matrix $\widetilde{M}=2\Pi-J$ satisfies $|\widetilde{M}_{x,y}|\leq1$
for all $(x,y)\in X\times Y$ and $|M_{x,y}-\widetilde{M}_{x,y}|\leq2\epsilon$
for all $(x,y)\in\dom M$. In particular,
\begin{equation}
\|M\|_{\Sigma,2\epsilon}\leq\|\widetilde{M}\|_{\Sigma}.\label{eq:F-tilde-F}
\end{equation}
On the other hand, Theorem~\ref{thm:protocol2matrix} guarantees
the existence of matrices $A$ and $B$ with $AB=\Pi$ and $\|A\|_{\mathrm{F}}\,\|B\|_{\mathrm{F}}\leq4^{C}\sqrt{|X|\,|Y|}.$
Therefore,
\begin{align}
\|\widetilde{M}\|_{\Sigma} & =\|2AB-J\|_{\Sigma}\nonumber \\
 & \leq2\|AB\|_{\Sigma}+\|J\|_{\Sigma}\nonumber \\
 & \leq2\|A\|_{\mathrm{F}}\|B\|_{\mathrm{F}}+\|J\|_{\Sigma}\nonumber \\
 & \leq2\cdot4^{C}\sqrt{|X|\,|Y|}+\|J\|_{\Sigma}\nonumber \\
 & =2\cdot4^{C}\sqrt{|X|\,|Y|}+\sqrt{|X|\,|Y|},\label{eq:tilde-F-trace-norm}
\end{align}
where the third step uses \prettyref{fact:bound-on-trace-of-product}.
Equations~\prettyref{eq:F-tilde-F} and~\prettyref{eq:tilde-F-trace-norm}
give $\|M\|_{\Sigma,2\epsilon}\leq(2\cdot4^{C}+1)\sqrt{|X|\,|Y|},$
which implies the claimed lower bound on $4^{C}.$
\end{proof}
A \emph{distinguisher} for a communication problem $F\colon X\times Y\to\{-1,1,*\}$
is a communication protocol $\Pi$ for which the expected output on
every input in $F^{-1}(-1)$ is less than the expected output on every
input in $F^{-1}(1).$ We will use the following proposition to convert
any distinguisher for $F$ into a communication protocol that computes
$F$.
\begin{prop}
\label{prop:shift-probab}Let $F\colon X\times Y\to\{-1,1,*\}$ be
a $($possibly partial$)$ communication problem. Suppose that $\Pi$
is a cost-$c$ randomized protocol with output $\pm1$ such that
\begin{align}
\Exp[\Pi(x,y)] & \leq\alpha &  & \text{for all }(x,y)\in F^{-1}(-1),\label{eq:minus-inputs}\\
\Exp[\Pi(x,y)] & \geq\beta &  & \text{for all }(x,y)\in F^{-1}(1),\label{eq:plus-inputs}
\end{align}
where $\alpha,\beta$ are reals with $-1\leq\alpha\leq\beta\leq1.$
Then
\[
R_{\frac{1}{2}-\frac{1}{8}(\beta-\alpha)}(F)\leq c.
\]
\begin{proof}
For a real number $t,$ define $\ssign t$ to be $1$ if $t\geq0$
and $-1$ if $t<0.$ Set $p=|\alpha+\beta|/(2+|\alpha+\beta|)$ and
consider the following randomized protocol $\Pi'$ with input $(x,y)\in X\times Y$:
with probability $p,$ Alice and Bob output $-\ssign(\alpha+\beta)$
without any communication; with the complementary probability $1-p$,
they execute the original protocol $\Pi$ on $(x,y)$ and output its
answer. Clearly, $\Pi'$ has the same cost as $\Pi.$ On every $(x,y)\in F^{-1}(-1),$
\[
\Exp[\Pi'(x,y)]\leq-p\ssign(\alpha+\beta)+(1-p)\alpha=\frac{-(\alpha+\beta)+2\alpha}{2+|\alpha+\beta|}=\frac{\alpha-\beta}{2+|\alpha+\beta|}\leq-\frac{\beta-\alpha}{4},
\]
where the first step uses~\prettyref{eq:minus-inputs}, and the last
step uses $-1\leq\alpha\leq\beta\leq1$. Analogously, on every $(x,y)\in F^{-1}(1),$
\[
\Exp[\Pi'(x,y)]\geq-p\ssign(\alpha+\beta)+(1-p)\beta=\frac{-(\alpha+\beta)+2\beta}{2+|\alpha+\beta|}=\frac{\beta-\alpha}{2+|\alpha+\beta|}\geq\frac{\beta-\alpha}{4},
\]
where the first step uses~\prettyref{eq:plus-inputs}. We have shown
that $\Exp[\Pi'(x,y)F(x,y)]\geq(\beta-\alpha)/4$ on the domain of
$F,$ which is another way of saying that $\Pi'$ computes $F$ with
error at most $\frac{1}{2}-\frac{1}{8}(\beta-\alpha).$
\end{proof}
\end{prop}

\subsection{Communication problems defined}

Let $\mathbb{F}$ be a given field. For nonnegative integers $n,m,r$
with $r\leq\min\{n,m\}$, the \emph{rank problem} is the communication
problem in which Alice and Bob are given matrices $A,B\in\mathbb{F}^{n\times m},$
respectively, and their objective is to determine whether $\rk(A+B)\leq r.$
Formally, this problem corresponds to the Boolean function $\RANK_{r}^{\mathbb{F},n,m}\colon\mathbb{F}^{n\times m}\times\mathbb{F}^{n\times m}\to\{-1,1\}$
given by 
\[
\RANK_{r}^{\mathbb{F},n,m}(A+B)=-1\qquad\Leftrightarrow\qquad\rk(A+B)\leq r.
\]
We also study the corresponding partial problem $\RANK_{r,R}^{\mathbb{F},n,m}$
for nonnegative integers $n,m,r,R$ with $r<R\leq\min\{n,m\}$, defined
on $\mathbb{F}^{n\times m}\times\mathbb{F}^{n\times m}$ by
\[
\RANK_{r,R}^{\mathbb{F},n,m}(A,B)=\begin{cases}
-1 & \text{if }\rk(A+B)=r,\\
1 & \text{if }\rk(A+B)=R,\\
* & \text{otherwise.}
\end{cases}
\]

For a positive integer $n$ and a pair of distinct field elements
$a,b\in\mathbb{F}$, the \emph{determinant problem} $\DET_{a,b}^{\mathbb{F},n}\colon\mathbb{F}^{n\times n}\times\mathbb{F}^{n\times n}\to\{-1,1,*\}$
is given by
\[
\DET_{a,b}^{\mathbb{F},n}(A,B)=\begin{cases}
-1 & \text{if }\det(A+B)=a,\\
1 & \text{if }\det(A+B)=b,\\
* & \text{otherwise.}
\end{cases}
\]
The \emph{rank versus determinant problem} is a hybrid inspired by
the previous two problems. Specifically, for a number $r\in\{0,1,\ldots,n-1\}$
and a nonzero field element $a\in\mathbb{F}\setminus\{0\},$ we define
$\RANKDET_{r,a}^{\mathbb{F},n}\colon\mathbb{F}^{n\times n}\times\mathbb{F}^{n\times n}\to\{-1,1,*\}$
by
\[
\RANKDET_{r,a}^{\mathbb{F},n}(A,B)=\begin{cases}
-1 & \text{if }\rk(A+B)=r,\\
1 & \text{if }\det(A+B)=a,\\
* & \text{otherwise.}
\end{cases}
\]
Note that $\RANKDET_{r,a}^{\mathbb{F},n}$ is a \emph{subproblem}
of both $\RANK_{r,n}^{\mathbb{F},n,n}$ and $\DET_{0,a}^{\mathbb{F},n},$
in the sense that the domain of $\RANKDET_{r,a}^{\mathbb{F},n}$ is
a subset of the domain of each of these other two problems and it
agrees on its domain with those problems.

Consider now the setting where Alice is given an $m$-dimensional
subspace $S\subseteq\mathbb{F}^{n}$ and Bob is given an $\ell$-dimensional
subspace $T\subseteq\mathbb{F}^{n}$, for some nonnegative integers
$n,m,\ell$ with $\max\{m,\ell\}\leq n.$ In the \emph{subspace intersection
problem} with parameter $d,$ Alice and Bob need to determine whether
$S\cap T$ has dimension at least $d$. In the \emph{subspace sum
problem}, they need to determine whether $S+T$ has dimension at most
$d$. Formally, these problems correspond to the Boolean functions
$\INTERSECT_{d}^{\mathbb{F},n,m,\ell}$ and $\SUM_{d}^{\mathbb{F},n,m,\ell}$
that are defined on $\Scal(\mathbb{F}^{n},m)\times\Scal(\mathbb{F}^{n},\ell)$
by
\begin{align*}
\INTERSECT_{d}^{\mathbb{F},n,m,\ell}(S,T) & =-1\qquad\Leftrightarrow\qquad\dim(S\cap T)\geq d,\\
\SUM_{d}^{\mathbb{F},n,m,\ell}(S,T) & =-1\qquad\Leftrightarrow\qquad\dim(S+T)\leq d.
\end{align*}
Their partial counterparts $\INTERSECT_{d_{1},d_{2}}^{\mathbb{F},n,m,\ell}$
and $\SUM_{d_{1},d_{2}}^{\mathbb{F},n,m,\ell}$, for any pair of distinct
integers $d_{1},d_{2}$, are defined on $\Scal(\mathbb{F}^{n},m)\times\Scal(\mathbb{F}^{n},\ell)$
by
\begin{align*}
\INTERSECT_{d_{1},d_{2}}^{\mathbb{F},n,m,\ell}(S,T) & =\begin{cases}
-1 & \text{if }\dim(S\cap T)=d_{1},\\
1 & \text{if }\dim(S\cap T)=d_{2},\\
* & \text{otherwise, }
\end{cases}\\
\rule{0mm}{11mm}\SUM_{d_{1},d_{2}}^{\mathbb{F},n,m,\ell}(S,T) & =\begin{cases}
-1 & \text{if }\dim(S+T)=d_{1},\\
1 & \text{if }\dim(S+T)=d_{2},\\
* & \text{otherwise. }
\end{cases}
\end{align*}
These partial functions are well-defined for any $d_{1},d_{2}$ with
$d_{1}\ne d_{2}.$ Their communication complexity, however, is zero
unless both $d_{1}$ and $d_{2}$ have meaningful values for the problem
in question. Specifically, one must have $d_{1},d_{2}\in[\max\{m,\ell\},\min\{m+\ell,n\}]$
for the subspace sum problem and $d_{1},d_{2}\in[\max\{0,m+\ell-n\},\min\{m,\ell\}]$
for the subspace intersection problem. We record this simple fact
as a proposition below.
\begin{prop}
\label{prop:possible-intersections-and-sums}Let $\mathbb{F}$ be
a field. Let $n,m,\ell$ be nonnegative integers with $\max\{m,\ell\}\leq n.$
Then:
\begin{enumerate}
\item \label{enu:set-of-sums}there exist $S\in\mathcal{S}(\mathbb{F}^{n},m)$
and $T\in\Scal(\mathbb{F}^{n},\ell)$ with $\dim(S+T)=d$ if and only
if $d$ is an integer with $\max\{m,\ell\}\leq d\leq\min\{m+\ell,n\};$
\item \label{enu:set-of-intersections}there exist $S\in\mathcal{S}(\mathbb{F}^{n},m)$
and $T\in\Scal(\mathbb{F}^{n},\ell)$ with $\dim(S\cap T)=d$ if and
only if $d$ is an integer with $\max\{0,m+\ell-n\}\leq d\leq\min\{m,\ell\}.$
\end{enumerate}
\end{prop}

\begin{proof}
\prettyref{enu:set-of-sums} For any subspaces $S,T\subseteq\mathbb{F}^{n},$
we have the trivial bounds $\max\{\dim(S),\dim(T)\}\leq\dim(S+T)\leq\min\{\dim(S)+\dim(T),n\}$.
This proves the ``only if'' part of~\prettyref{enu:set-of-sums}.
For the converse, let $d$ be any integer with $\max\{m,\ell\}\leq d\leq\min\{m+\ell,n\}.$
Then the sets $A=\{1,2,\ldots,m\}$ and $B=\{d-\ell+1,\ldots,d-1,d\}$
satisfy $A,B\subseteq\{1,2,\ldots,n\}$ (because $\ell\leq d\leq n$)
and $A\cup B=\{1,2,\ldots,d\}$ (because $m\leq d\leq m+\ell$). As
a result, $\Span\{e_{1},e_{2},\ldots,e_{m}\}$ and $\Span\{e_{d-\ell+1},\ldots,e_{d-1},e_{d}\}$
are a pair of subspaces in $\mathbb{F}^{n}$ of dimension $m$ and
$\ell,$ respectively, whose sum has dimension $d.$ 

\prettyref{enu:set-of-intersections} Recall that $\dim(S\cap T)=\dim(S)+\dim(T)-\dim(S+T)$
for any subspaces $S,T$. As a result,
\begin{align*}
\{\dim(S\cap T) & :S\in\Scal(\mathbb{F}^{n},m),T\in\Scal(\mathbb{F}^{n},\ell)\}\\
 & =\{m+\ell-\dim(S+T):S\in\Scal(\mathbb{F}^{n},m),T\in\Scal(\mathbb{F}^{n},\ell)\}\\
 & =\{m+\ell-d:d\in\mathbb{Z}\text{ with }\max\{m,\ell\}\leq d\leq\min\{m+\ell,n\}\}\\
 & =\{\max\{0,m+\ell-n\},\ldots,\min\{m,\ell\}-1,\min\{m,\ell\}\},
\end{align*}
where the second step uses~\prettyref{enu:set-of-sums}.
\end{proof}
Let $F\colon X\times Y\to\{-1,1,*\}$ and $F'\colon X'\times Y'\to\{-1,1,*\}$
be (possibly partial) communication problems. A \emph{communication-free
reduction from $F$ to $F'$} is a pair of mappings $\alpha\colon X\to X'$
and $\beta\colon Y\to Y'$ such that $F(x,y)=F'(\alpha(x),\beta(y))$
for all $(x,y)\in\dom F.$ We indicate the existence of a communication-free
reduction from $F$ to $F'$ by writing $F'\succeq F$. In this case,
it is clear that the communication complexity of $F'$ in any given
model is bounded from below by the communication complexity of $F$
in the same model.

\begin{prop}
\label{prop:INTERSECT-reductions}Let $n,m,\ell,r,R$ be integers
with $0\leq r<R\leq\min\{m,\ell\}$ and $\max\{m,\ell\}\leq n.$ Then
\[
\INTERSECT_{r,R}^{\mathbb{F},n,m,\ell}\succeq\INTERSECT_{0,R-r}^{\mathbb{F},n-r,m-r,\ell-r}.
\]
\end{prop}

\begin{proof}
Consider the injective linear map $\phi\colon\mathbb{F}^{n-r}\to\mathbb{F}^{n}$
that takes any vector and extends it with $r$ zero components to
obtain a vector in $\mathbb{F}^{n}.$ Given arbitrary subspaces $S,T\subseteq\mathbb{F}^{n-r}$
of dimension $m-r$ and $\ell-r,$ respectively, define $S'=\Span(\phi(S)\cup\{e_{n-r+1},\ldots,e_{n-1},e_{n}\})$
and $T'=\Span(\phi(T)\cup\{e_{n-r+1},\ldots,e_{n-1},e_{n}\})$. Then
clearly 
\begin{align*}
\dim(S'\cap T') & =\dim(S')+\dim(T')-\dim(S'+T')\\
 & =\dim(S)+r+\dim(T)+r-\dim(S+T)-r\\
 & =\dim(S)+\dim(T)-\dim(S+T)+r\\
 & =\dim(S\cap T)+r,
\end{align*}
whence the reduction $\INTERSECT_{0,R-r}^{\mathbb{F},n-r,m-r,\ell-r}(S,T)=\INTERSECT_{r,R}^{\mathbb{F},n,m,\ell}(S',T').$ 
\end{proof}

\section{\label{sec:rank-problem}The matrix rank problem}

In this section, we prove a tight lower bound on the randomized and
quantum communication complexity of the rank problem. As discussed
in the introduction, we obtain this lower bound by constructing a
dual matrix $\Phi$ with certain properties, namely, low spectral
norm, low $\ell_{1}$ norm, and high correlation with the characteristic
matrix of the rank problem. We start in Section~\prettyref{subsec:the-pn-function}
by analyzing the probabilities $P_{n}$ that arise in the recurrence
relation for the $\Gamma_{n}$ function. The latter plays an important
role in our proof and is studied in Section~\prettyref{subsec:the-gamma-function}.
Section~\prettyref{subsec:univar-dual-obj} constructs a univariate
dual object $\phi$ defined on $\{0,1,\ldots,n\}$ and studies its
analytic and metric properties. We build on $\phi$ to construct a
dual matrix $E_{\phi}$ in Section~\prettyref{subsec:from-univar-dual-obj-to-dual-matr},
and discuss how the properties of $\phi$ give rise to analogous properties
of $E_{\phi}$. Sections~\ref{subsec:approx-trace-norm}~and~\ref{subsec:communication-lower-bounds}
establish lower bounds for the approximate trace norm of the characteristic
matrix and the communication complexity of the rank problem, with
$\Phi=E_{\phi}$ used as the dual witness. We prove a matching communication
upper bound in Section~\prettyref{subsec:communication-upper-bounds}.
Section~\prettyref{subsec:streaming} concludes our study of the
rank problem with an application to streaming complexity.

Throughout this section, the underlying field is $\mathbb{F}_{q}$
for an arbitrary prime power $q.$ The root of unity $\omega$ and
the notation $\omega^{x}$ for $x\in\mathbb{F}_{q}$ are as defined
in \prettyref{sec:Fourier}.

\subsection{\label{subsec:the-pn-function}The $P_{n}$ function}

The $P_{n}$ function, defined next, conveys useful information about
random nonsingular matrices of order $n$ over a given field.
\begin{defn}
\label{def:P_n}Let $n\geq1$ be a given integer. For nonnegative
integers $s,t,r\in\{0,1,\ldots,n\},$ define $P_{n}(s,t,r)$ to be
the probability that the upper-left $s\times t$ quadrant of a uniformly
random nonsingular matrix in $\FF_{q}^{n\times n}$ has rank $r$:
\begin{equation}
P_{n}(s,t,r)=\Prob_{X\in\Mcal_{n}}[\rk(I_{s}XI_{t})=r].\label{eq:def-Pn}
\end{equation}
\end{defn}

To derive a closed-form expression for $P_{n}$, we essentially need
to count the number of ways to complete a given $s\times t$ matrix
of rank $r$ to a nonsingular matrix of order $n$. We break this
counting task into two steps, where the first step is to count the
number of completions of an $s\times t$ matrix of rank $r$ to an
$s\times n$ matrix of rank $s$.
\begin{lem}
\label{lem:extend-quadrant}Let $s,t,r,m$ be nonnegative integers
with $r\leq\min\{s,t\}.$ Let $A\in\Mcal_{r}^{s,t}$ be given. Then
the number of matrices $B\in\mathbb{F}_{q}^{s\times m}$ for which
$\rk\begin{bmatrix}A & B\end{bmatrix}=s$ is 
\[
q^{rm}\,|\Mcal_{s-r}^{s-r,m}|.
\]
\end{lem}

\begin{proof}
If $r=0,$ then $\rk\begin{bmatrix}A & B\end{bmatrix}=\rk B$. As
a result, $\rk\begin{bmatrix}A & B\end{bmatrix}=s$ if and only if
$B\in\Mcal_{s}^{s,m}.$ Therefore, the lemma holds in this case. In
what follows, we consider $r\geq1,$ which forces $s$ and $t$ to
be positive integers.

Define the matrices $A'$ and $A''$ to be the top $r$ rows of $A$
and the bottom $s-r$ rows of $A,$ respectively. We first consider
the possibility when $A''$ is zero or empty. Here, the column span
of $A'$ is necessarily all of $\mathbb{F}_{q}^{r}.$ Given an $s\times m$
matrix $B,$ partition it into $B'$ and $B''$ conformably with the
partition of $A$. Then
\[
\rk\begin{bmatrix}A & B\end{bmatrix}=\rk\begin{bmatrix}A' & B'\\
0 & B''
\end{bmatrix}=\rk\begin{bmatrix}A' & 0\\
0 & B''
\end{bmatrix}=\rk(A')+\rk(B'')=r+\rk(B'').
\]
Thus, $\begin{bmatrix}A & B\end{bmatrix}$ has rank $s$ if and only
if $\rk(B'')=s-r.$ This means that there are $|\Mcal_{s-r}^{s-r,m}|$
ways to choose $B''$, and independently $q^{rm}$ ways to choose
$B'$, such that $\rk\begin{bmatrix}A & B\end{bmatrix}=s$.

It remains to examine the case of a general matrix $A$ of rank $r\geq1.$
Let $V$ be an invertible matrix such that the bottom $s-r$ rows
of $VA$ are zero. Let $\Mcal$ be the set of $s\times m$ matrices
$M$ for which $\rk\begin{bmatrix}VA & M\end{bmatrix}=s$. Then $\rk\begin{bmatrix}A & B\end{bmatrix}=s$
if and only if $VB\in\Mcal.$ In particular, the number of matrices
$B$ for which $\rk\begin{bmatrix}A & B\end{bmatrix}=s$ is $|\Mcal|$.
Since $|\Mcal|=q^{rm}\,|\Mcal_{s-r}^{s-r,m}|$ by the previous paragraph,
we are done.
\end{proof}
We now derive an exact expression for $P_{n}$ and establish its relevant
algebraic and analytic properties.
\begin{lem}
Let $n\geq1$ be a given integer. Then for all $s,t,r\in\{0,1,\ldots,n\}$\emph{:}\label{lem:probquadrant}
\end{lem}

\begin{enumerate}
\item $P_{n}(s,t,r)=0$ if $r>\min\{s,t\}$ or $r<s+t-n$;\label{enu:probsupp}
\item \label{enu:probexpr}$P_{n}(s,t,r)=q^{r(n-t)}|\Mcal_{r}^{s,t}|\,|\Mcal_{s-r}^{s-r,n-t}|/((q^{n}-1)(q^{n}-q)\cdots(q^{n}-q^{s-1}));$
\item for any fixed values of $n,s,r$, the quantity $P_{n}(s,t,r)$ as
a function of $t\in\{0,1,\ldots,n\}$ is a polynomial in $q^{-t}$
of degree at most $s;$ \label{enu:probpoly}
\item $P_{n}(s,t,r)\leq16q^{-(s-r)(t-r)}.$ \label{enu:probbound}
\end{enumerate}
\begin{proof}
\ref{enu:probsupp} Since the quadrant of interest is an $s\times t$
matrix, the first inequality is trivial. For the second inequality,
observe that the matrix $I_{s}XI_{t}$ in the defining equation~\prettyref{eq:def-Pn}
satisfies $\rk(I_{s}XI_{t})\geq\rk I_{s}+\rk(XI_{t})-n=s+t-n$ by~\prettyref{fact:rank-of-product}. 

\ref{enu:probexpr} If $r>\min\{s,t\},$ then the left-hand side and
right-hand side of~\prettyref{enu:probexpr} both vanish due to~\prettyref{enu:probsupp}
and the definition of $\Mcal_{r}^{s,t}$. We now treat the case $r\leq\min\{s,t\}.$
Letting $\Mcal$ stand for the set of nonsingular matrices of order
$n$ whose upper-left $s\times t$ quadrant has rank $r$, we have
\begin{equation}
P_{n}(s,t,r)=\frac{|\Mcal|}{|\Mcal_{n}|}.\label{eq:P-ratio}
\end{equation}
A matrix in $\Mcal$ can be chosen by the following three-step process:
choose a matrix in $\Mcal_{r}^{s,t}$ for the upper-left quadrant;
extend the quadrant to a matrix in $\Mcal_{s}^{s,n}$, which by \prettyref{lem:extend-quadrant}
can be done in $q^{r(n-t)}\,|\Mcal_{s-r}^{s-r,n-t}|$ ways; and finally
add $n-s$ rows to obtain an invertible matrix, which can be done
in $(q^{n}-q^{s})(q^{n}-q^{s+1})\cdots(q^{n}-q^{n-1})$ ways. Altogether,
we obtain
\[
|\Mcal|=|\Mcal_{r}^{s,t}|\cdot q^{r(n-t)}\,|\Mcal_{s-r}^{s-r,n-t}|\cdot(q^{n}-q^{s})(q^{n}-q^{s+1})\cdots(q^{n}-q^{n-1}),
\]
whereas \prettyref{prop:nummatr} gives 
\[
|\Mcal_{n}|=(q^{n}-1)(q^{n}-q)\cdots(q^{n}-q^{n-1}).
\]
Making these substitutions in \prettyref{eq:P-ratio} completes the
proof.

\prettyref{enu:probpoly} We claim that for all $s,t,r\in\{0,1,\ldots,n\},$
\begin{multline}
P_{n}(s,t,r)=q^{r(n-t)}\gbinom sr(q^{t}-1)(q^{t}-q)\cdots(q^{t}-q^{r-1})\\
\times\frac{(q^{n-t}-1)(q^{n-t}-q)\cdots(q^{n-t}-q^{s-r-1})}{(q^{n}-1)(q^{n}-q)\cdots(q^{n}-q^{s-1})}.\qquad\label{eq:P-str-explicit-formula}
\end{multline}
Indeed, in the case when $r>\min\{s,t\}$ or $r<s+t-n,$ the right-hand
side vanishes and therefore the equality holds due to~\prettyref{enu:probsupp}.
In the complementary case, \prettyref{prop:nummatr} gives closed-form
expressions for $|\Mcal_{r}^{s,t}|$ and $|\Mcal_{s-r}^{s-r,n-t}|$
which, when substituted in~\prettyref{enu:probexpr}, result in~\prettyref{eq:P-str-explicit-formula}.
This settles~\prettyref{eq:P-str-explicit-formula} for all $s,t,r\in\{0,1,\ldots,n\}.$

Rewrite~\prettyref{eq:P-str-explicit-formula} to obtain
\begin{multline}
P_{n}(s,t,r)=q^{rn}\gbinom sr(1-q^{-t})(1-q^{-t+1})\cdots(1-q^{-t+r-1})\\
\times\frac{(q^{n-t}-1)(q^{n-t}-q)\cdots(q^{n-t}-q^{s-r-1})}{(q^{n}-1)(q^{n}-q)\cdots(q^{n}-q^{s-1})}.\qquad\label{eq:P-str-explicit-formula-alternate}
\end{multline}
Now, fix $n,s,r$ arbitrarily. If $r\leq s,$ then \prettyref{eq:P-str-explicit-formula-alternate}
makes it clear that $P_{n}(s,t,r)$ is a polynomial in $q^{-t}$ of
degree at most $r+(s-r)=s.$ If $r>s,$ then $P_{n}(s,t,r)$ is identically
zero and thus trivially a polynomial in $q^{-t}$ of degree at most
$s.$

\prettyref{enu:probbound} For $r>s,$ we have $P_{n}(s,t,r)=0$ by~\prettyref{enu:probsupp}
and therefore the claimed upper bound holds trivially. In the complementary
case, simplify~\prettyref{eq:P-str-explicit-formula} to obtain
\begin{align*}
P_{n}(s,t,r) & \le q^{r(n-t)}\binom{s}{r}_{q}q^{tr}\cdot\frac{q^{(n-t)(s-r)}}{(q^{n}-1)(q^{n}-q)\cdots(q^{n}-q^{s-1})}\\
 & \leq q^{r(n-t)}\binom{s}{r}_{q}q^{tr}\cdot4q^{(n-t)(s-r)}q^{-ns}\\
 & \leq q^{r(n-t)}\cdot4q^{r(s-r)}q^{tr}\cdot4q^{(n-t)(s-r)}q^{-ns}\\
 & =16q^{-(s-r)(t-r)},
\end{align*}
where the second and third steps apply~\prettyref{prop:prodprop}
and~\prettyref{cor:qbinbound}, respectively.
\end{proof}

\subsection{\label{subsec:the-gamma-function}The $\Gamma_{n}$ function}

A basic building block in our construction is the characteristic function
of matrices in $\mathbb{F}_{q}^{n\times n}$ of a given rank. Its
Fourier spectrum is best understood in terms of what we call the $\Gamma_{n}$
function.
\begin{defn}
\label{def:Gamma}Let $n\geq1$ be a given integer. For $s,t\in\{0,1,\ldots,n\},$
define
\begin{align*}
\Gamma_{n}(s,t) & =\Exp_{\substack{\rk A=s\\
\rk B=t
}
}\omega^{\langle A,B\rangle},
\end{align*}
where the expectation is taken with respect to the uniform distribution
on $\Mcal_{s}^{n,n}\times\Mcal_{t}^{n,n}$.
\end{defn}

Sun and Wang~\cite{sunwang12communication-linear} studied the Fourier
spectrum of the nonsingularity function on $\mathbb{F}_{q}^{n\times n},$
defined to be $1$ on nonsingular matrices and $0$ otherwise. In
our notation, they established the following result.
\begin{lem}
For any integers $n\geq1$ and $r\in\{0,1,\ldots,n\},$ \label{lem:sunwang}
\begin{align*}
\Gamma_{n}(n,r) & =\frac{(-1)^{r}q^{\binom{r}{2}}}{(q^{n}-1)(q^{n}-q)\cdots(q^{n}-q^{r-1})}.
\end{align*}
\end{lem}

\noindent The proof of Sun and Wang~\cite{sunwang12communication-linear}
is stated for fields $\mathbb{F}_{p}$ with prime $p$, but their
analysis readily extends to fields of cardinality a prime power. In
Appendix~\ref{sec:nonsingularity}, we prove \prettyref{lem:sunwang}
from scratch in our desired generality, using a simpler proof than
that of~\cite{sunwang12communication-linear}.

Our next lemma collects crucial properties of $\Gamma_{n}(s,t)$ for
general values of $s,t$.
\begin{lem}
Let $n\geq1$ be a given integer. Then for all $s,t\in\{0,1,\ldots,n\}$$:$
\label{lem:gammaprop}
\end{lem}

\begin{enumerate}
\item $|\Gamma_{n}(s,t)|\leq1;$\label{enu:gammatriv}
\item $\Gamma_{n}(s,t)=\Gamma_{n}(t,s);$\label{enu:gammasym}
\item $\Gamma_{n}(s,t)=\sum_{r=0}^{n}P_{n}(s,t,r)\Gamma_{n}(n,r);$\label{enu:gammasum}
\item for $n,s$ fixed, $\Gamma_{n}(s,t)$ as a function of $t\in\{0,1,\ldots,n\}$
is a polynomial in $q^{-t}$ of degree at most $s$;\label{enu:gammapoly}
\item $|\Gamma_{n}(s,t)|\leq128q^{-st/2}.$\label{enu:gammabound}
\end{enumerate}
\begin{proof}
\ref{enu:gammatriv} Using $|\omega|=1$ and the triangle inequality,
\[
\left|\Gamma_{n}(s,t)\right|=\left|\Exp_{A,B}\omega^{\langle A,B\rangle}\right|\leq\Exp_{A,B}\left|\omega^{\langle A,B\rangle}\right|=1.
\]

\ref{enu:gammasym} The symmetry of $\Gamma_{n}$ follows from the
independence of $A$ and $B$ in the defining equation for $\Gamma_{n}$,
and the symmetry of the inner product over $\mathbb{F}_{q}$.

\ref{enu:gammasum} We have:
\begin{align*}
\Gamma_{n}(s,t) & =\Exp_{\substack{A\in\Mcal_{s}^{n,n}\\
B\in\Mcal_{t}^{n,n}
}
}\;\omega^{\langle A,B\rangle}\\
 & =\Exp_{X,Y,Z_{1},Z_{2},W\in\Mcal_{n}}\omega^{\langle XI_{s}Y,Z_{1}Z_{2}I_{t}W\rangle}\\
 & =\Exp_{X,Y,Z_{1},Z_{2},W\in\Mcal_{n}}\omega^{\langle XI_{s}YW\tr I_{t}Z_{2}\tr,Z_{1}\rangle}\\
 & =\Exp_{X,U,Z_{1},Z_{2}\in\Mcal_{n}}\omega^{\langle X(I_{s}UI_{t})Z_{2}\tr,Z_{1}\rangle}\\
 & =\sum_{r=0}^{n}\Prob_{U\in\Mcal_{n}}[\rk(I_{s}UI_{t})=r]\Exp_{X,U,Z_{1},Z_{2}\in\Mcal_{n}}\left[\omega^{\langle X(I_{s}UI_{t})Z_{2}\tr,Z_{1}\rangle}\ \mid\ \rk(I_{s}UI_{t})=r\right]\\
 & =\sum_{r=0}^{n}\Prob_{U\in\Mcal_{n}}[\rk(I_{s}UI_{t})=r]\Exp_{\substack{B\in\Mcal_{r}^{n,n}\\
Z_{1}\in\Mcal_{n}
}
}\omega^{\langle B,Z_{1}\rangle}\\
 & =\sum_{r=0}^{n}P_{n}(s,t,r)\Gamma_{n}(n,r),
\end{align*}
where the first step restates the definition of $\Gamma_{n}$, the
second step uses \prettyref{prop:random-matrices}, the third step
applies \prettyref{fact:inner-product-vs-trace}\prettyref{enu:inner-product-transfer},
the fourth and sixth steps again use \prettyref{prop:random-matrices},
and the last step is immediate from the definitions of $P_{n}$ and
$\Gamma_{n}$.

\ref{enu:gammapoly} Immediate from \prettyref{enu:gammasum} and
\prettyref{lem:probquadrant}\ref{enu:probpoly}.

\ref{enu:gammabound} We have:
\begin{align*}
|\Gamma_{n}(s,t)| & =\left|\sum_{r=0}^{n}P_{n}(s,t,r)\Gamma_{n}(n,r)\right|\\
 & \leq\sum_{r=0}^{n}P_{n}(s,t,r)|\Gamma_{n}(n,r)|\\
 & =\sum_{r=\max\{0,s+t-n\}}^{n}P_{n}(s,t,r)\left|\Gamma_{n}(n,r)\right|\\
 & \leq\sum_{r=\max\{0,s+t-n\}}^{n}16q^{-(s-r)(t-r)}\cdot\frac{q^{\binom{r}{2}}}{(q^{n}-1)(q^{n}-q)\cdots(q^{n}-q^{r-1})}\\
 & \leq\sum_{r=\max\{0,s+t-n\}}^{n}64q^{-(s-r)(t-r)+\binom{r}{2}-nr}\\
 & \leq128q^{-st/2},
\end{align*}
where the first step appeals to~\prettyref{enu:gammasum}, the third
step is valid by \prettyref{lem:probquadrant}\prettyref{enu:probsupp},
the fourth step uses \prettyref{lem:probquadrant}\ref{enu:probbound}
and \prettyref{lem:sunwang}, the fifth step applies \prettyref{prop:prodprop},
and the last step which completes the proof is justified by the following
claim.
\end{proof}
\begin{claim}
For any integers $n\geq1$ and $s,t\in\{0,1,\ldots,n\}$,
\begin{align}
\sum_{r=\max\{0,s+t-n\}}^{\infty}q^{-(s-r)(t-r)+\binom{r}{2}-nr} & \leq2q^{-st/2}.\label{eq:claim-q-exponent-sum}
\end{align}
\end{claim}

\begin{proof}
The exponent of $q$ on the left-hand side of~\prettyref{eq:claim-q-exponent-sum}
is given by the function
\begin{align}
A(r) & =-(s-r)(t-r)+\binom{r}{2}-nr\label{eq:Ar}\\
 & =-st-\frac{1}{2}\left(r+n-s-t+\frac{1}{2}\right)^{2}+\frac{1}{2}\left(n-s-t+\frac{1}{2}\right)^{2}.\label{eq:Ar-squares}
\end{align}
The first equality shows that $A(r)$ is always an integer, whereas
the second shows that $A(r)$ is a strictly decreasing function in
the variable $r\in[\max\{0,s+t-n\},\infty).$ These two facts lead
to
\begin{align}
A(\max\{0,s+t-n\}+i) & \leq A(\max\{0,s+t-n\})-i, &  & i=0,1,2,\ldots.\label{eq:A-decreases}
\end{align}

We will now prove that
\begin{equation}
A(\max\{0,s+t-n\})\leq-\frac{st}{2}.\label{eq:A-at-max}
\end{equation}
There are two cases to consider. If $s+t\leq n$, then $A(\max\{0,s+t-n\})=A(0)=-st$
and therefore~\prettyref{eq:A-at-max} holds. The complementary case
$s+t\geq n+1$ is more challenging. Here, we have 
\[
A(\max\{0,s+t-n\})=A(s+t-n)\leq-st+\frac{1}{2}\left(n-s-t+\frac{1}{2}\right)^{2},
\]
where the second step uses~\prettyref{eq:Ar-squares}. Thus, the
proof of~\prettyref{eq:A-at-max} will be complete once we show that
\begin{equation}
\left(n-s-t+\frac{1}{2}\right)^{2}-st\leq0.\label{eq:am-gm-reverse}
\end{equation}
To prove~\prettyref{eq:am-gm-reverse}, suppose that of all pairs
$(s,t)\in\{0,1,\ldots,n\}^{2}$ with $s+t\geq n+1$, the left-hand
side of~\prettyref{eq:am-gm-reverse} is maximized at a pair $(s^{*},t^{*}).$
By symmetry, we may assume that $s^{*}\leq t^{*}.$ If we had $t^{*}\leq n-1,$
then it would follow that $s^{*}\geq2$ (due to the requirement that
$s^{*}+t^{*}\geq n+1$); as a result, the left-hand side of~\prettyref{eq:am-gm-reverse}
would be larger for the pair $(s,t)=(s^{*}-1,t^{*}+1)$ than it is
for the pair $(s,t)=(s^{*},t^{*}),$ an impossibility. Therefore,
$t^{*}=n.$ In addition, we have $s^{*}\geq1$ (due to the requirement
that $s^{*}+t^{*}\geq n+1$). Evaluating the right-hand side of~\prettyref{eq:am-gm-reverse}
at this pair $(s^{*},t^{*})=(s^{*},n)$, we obtain $(s^{*}-\frac{1}{2})^{2}-s^{*}n,$
which is clearly negative due to $s^{*}\in\{1,2,\ldots,n\}.$ This
completes the proof of~\prettyref{eq:am-gm-reverse} and therefore
that of~\prettyref{eq:A-at-max}.

Now,
\begin{align*}
\sum_{r=\max\{0,s+t-n\}}^{\infty}q^{-(s-r)(t-r)+\binom{r}{2}-nr} & =\sum_{r=\max\{0,s+t-n\}}^{\infty}q^{A(r)}\\
 & =\sum_{i=0}^{\infty}q^{A(\max\{0,s+t-n\}+i)}\\
 & \leq q^{A(\max\{0,s+t-n\})}\sum_{i=0}^{\infty}q^{-i}\\
 & \leq q^{-st/2}\cdot\frac{q}{q-1},
\end{align*}
where the first step uses the definition of $A(r),$ the third step
applies~\prettyref{eq:A-decreases}, and the final step appeals to~\prettyref{eq:A-at-max}
and a geometric series. Since $q\geq2,$ this completes the proof
of~\prettyref{eq:claim-q-exponent-sum}.
\end{proof}

\subsection{\label{subsec:univar-dual-obj}Univariate dual object}

Our construction of the univariate dual object is based on the Cauchy
binomial theorem along with a certain ``correcting'' polynomial $\zeta$.
The next lemma presents $\zeta$ as parametrized by two numbers $\ell$
and $m$ and gives its basic properties.

\begin{lem}
\label{lem:zeta}Let $n,k,\ell,m$ be nonnegative integers such that
$\ell+m\leq k<n$. Define a univariate polynomial $\zeta$ by 
\begin{align}
\zeta(t) & =\prod_{i=0}^{\ell-1}\frac{t-q^{-i}}{q^{-n}-q^{-i}}\cdot\prod_{i=k-m}^{k-1}\frac{t-q^{-i}}{q^{-n}-q^{-i}}\cdot\prod_{i=k+1}^{n-1}\frac{t-q^{-i}}{q^{-n}-q^{-i}}.\label{eq:zeta-def}
\end{align}
 Then:
\end{lem}

\begin{enumerate}
\item $\zeta(q^{-n})=1;$\label{enu:zeta-at-n}
\item $\sign\zeta(q^{-k})=(-1)^{n-k-1};$\label{enu:zeta-at-k}
\item $\zeta(q^{-r})=0$ for $r\in\{0,1,\ldots,n\}\setminus(\{\ell,\ell+1,\ldots,k-m-1\}\cup\{k,n\});$\label{enu:zeta-vanish}
\item \label{enu:zeta-deg}$\deg\zeta=n+\ell+m-k-1;$
\item $\left|\zeta(q^{-r})\right|\leq4q^{-r(n-k+m-1)+\binom{n}{2}-k-\binom{k-m}{2}}$
for $r\in\{\ell,\ell+1,\ldots,k-m-1\}$.\label{enu:zeta-small-weight}
\end{enumerate}
\begin{proof}
Items~\prettyref{enu:zeta-at-n}, \prettyref{enu:zeta-vanish}, and~\prettyref{enu:zeta-deg}
are immediate from the defining equation for $\zeta.$ Item~\prettyref{enu:zeta-at-k}
holds because for $t=q^{-k},$ the first and second products in~\prettyref{eq:zeta-def}
contain only positive factors, whereas the third product contains
exactly $n-k-1$ factors all of which are negative. For~\prettyref{enu:zeta-small-weight},
\begin{align*}
|\zeta(q^{-r})| & =\left|\prod_{i=0}^{\ell-1}\frac{q^{-r}-q^{-i}}{q^{-n}-q^{-i}}\cdot\prod_{i=k-m}^{k-1}\frac{q^{-r}-q^{-i}}{q^{-n}-q^{-i}}\cdot\prod_{i=k+1}^{n-1}\frac{q^{-r}-q^{-i}}{q^{-n}-q^{-i}}\right|\\
 & =\prod_{i=0}^{\ell-1}\frac{1-q^{i-r}}{1-q^{-(n-i)}}\cdot\prod_{i=k-m}^{k-1}\frac{q^{i-r}-1}{1-q^{-(n-i)}}\cdot\prod_{i=k+1}^{n-1}\frac{q^{i-r}-1}{1-q^{-(n-i)}}\\
 & \leq\prod_{i=0}^{\ell-1}\frac{1}{1-q^{-(n-i)}}\cdot\prod_{i=k-m}^{k-1}\frac{q^{i-r}}{1-q^{-(n-i)}}\cdot\prod_{i=k+1}^{n-1}\frac{q^{i-r}}{1-q^{-(n-i)}}.
\end{align*}
The product of the numerators in the last expression is $q^{-r(n-k+m-1)+\binom{n}{2}-k-\binom{k-m}{2}},$
whereas the product of the denominators is at least $1/4$ by~\prettyref{prop:prodprop}.
\end{proof}
With $\zeta$ in hand, we are now in a position to construct the promised
univariate dual object $\phi$. The properties of $\phi$ established
in the lemma below will give rise to analogous properties in the dual
matrix $E_{\phi}$.
\begin{lem}
\label{lem:phi}Let $n,k,\ell,m$ be nonnegative integers such that
$\ell+m\leq k<n$. Then there is a function $\phi\colon\{0,1,\ldots,n\}\to\Re$
such that:
\end{lem}

\begin{enumerate}
\item $\phi(n)=1;$\label{enu:phi-at-n}
\item $\phi(k)<0;$\label{enu:phi-at-k}
\item $\phi(r)=0$ for $r\in\{0,1,\ldots,n\}\setminus(\{\ell,\ell+1,\ldots,k-m-1\}\cup\{k,n\});$\label{enu:phi-vanish}
\item \label{enu:phi-orthog}$\sum_{r=0}^{n}\phi(r)\xi(q^{-r})=0$ for every
polynomial $\xi$ of degree at most $k-\ell-m;$
\item $\sum_{r\in\{0,\ldots,n\}\setminus\{k,n\}}|\phi(r)|\leq32q^{-m-1}$.\label{enu:phi-small-tail}
\end{enumerate}
\begin{proof}
Define
\[
\phi(r)=\gbinom nr(-1)^{r-n}q^{\binom{r}{2}-\binom{n}{2}}\zeta(q^{-r}),
\]
where $\zeta$ is the univariate polynomial from~\prettyref{lem:zeta}.
Then items~\prettyref{enu:phi-at-n}\textendash \prettyref{enu:phi-vanish}
are immediate from the corresponding items~\prettyref{enu:zeta-at-n}\textendash \prettyref{enu:zeta-vanish}
of \prettyref{lem:zeta}.

For~\prettyref{enu:phi-orthog}, fix a univariate polynomial $\xi$
of degree at most $k-\ell-m$. In view of~\prettyref{lem:zeta}\prettyref{enu:zeta-deg},
the product of $\zeta$ and $\xi$ has degree less than $n$. As a
result, the Cauchy binomial theorem (\prettyref{cor:qbinompoly})
implies that
\[
\sum_{r=0}^{n}\phi(r)\xi(q^{-r})=(-1)^{-n}q^{-\binom{n}{2}}\sum_{r=0}^{n}\gbinom nr(-1)^{r}q^{\binom{r}{2}}\zeta(q^{-r})\xi(q^{-r})=0.
\]

For~\prettyref{enu:phi-small-tail}, fix any $r\in\{\ell,\ell+1,\ldots,k-m-1\}$.
Then 
\begin{align}
|\phi(r)| & =\gbinom nrq^{\binom{r}{2}-\binom{n}{2}}|\zeta(q^{-r})|\nonumber \\
 & \leq4q^{r(n-r)}\cdot q^{\binom{r}{2}-\binom{n}{2}}\cdot4q^{-r(n-k+m-1)+\binom{n}{2}-k-\binom{k-m}{2}}\nonumber \\
 & =16q^{-\binom{k-m-r+1}{2}-m},\label{eq:phi-at-problematic-points}
\end{align}
where in the second step we bound the $q$-binomial coefficient via~\prettyref{cor:qbinbound}
and $|\zeta(q^{-r})|$ via~\prettyref{lem:zeta}\prettyref{enu:zeta-small-weight}.
Now
\[
\sum_{r\in\{0,\ldots n\}\setminus\{k,n\}}|\phi(r)|=\sum_{r=\ell}^{k-m-1}|\phi(r)|\leq\sum_{r=\ell}^{k-m-1}16q^{-\binom{k-m-r+1}{2}-m}\leq\sum_{i=2}^{\infty}16q^{-\binom{i}{2}-m}\leq\frac{16q^{-m-1}}{1-\frac{1}{q}},
\]
where the first step is valid by~\prettyref{enu:phi-vanish}, the
second step uses~\prettyref{eq:phi-at-problematic-points}, and the
fourth step uses a geometric series along with $\binom{i}{2}\geq i-1$
for $i\geq2$. Since $q\geq2,$ this completes the proof of~\prettyref{enu:phi-small-tail}.
\end{proof}

\subsection{\label{subsec:from-univar-dual-obj-to-dual-matr}From univariate
dual objects to dual matrices}

En route to the main result of this section, we now show how to convert
a univariate dual object $\phi$, such as the one constructed in~\prettyref{lem:phi},
into a dual matrix $E_{\phi}$.
\begin{defn}
\label{def:E-phi}Let $n\geq1$ be a given integer. For $r=0,1,\ldots,n,$
define $E_{r}$ to be the matrix with rows and columns indexed by
matrices in $\mathbb{F}_{q}^{n\times n},$ and entries given by
\[
(E_{r})_{A,B}=\begin{cases}
q^{-n^{2}}|\Mcal_{r}^{n,n}|^{-1} & \text{if \ensuremath{\rk(A+B)=r,}}\\
0 & \text{otherwise.}
\end{cases}
\]
For a function $\phi\colon\{0,1,\ldots,n\}\to\Re,$ define
\[
E_{\phi}=\sum_{r=0}^{n}\phi(r)E_{r}.
\]
\end{defn}

As one would expect, the metric and analytic properties of $E_{\phi}$
are closely related to those of $\phi.$
\begin{lem}[Metric properties of $E_{\phi}$]
\label{lem:E_phi-metric} Let $n\geq1$ be an integer and $\phi\colon\{0,1,\ldots,n\}\to\Re$
a given function. Then
\begin{align}
 & \sum_{A,B:\rk(A+B)=r}(E_{\phi})_{A,B}=\phi(r), &  & r=0,1,\dots,n,\label{eq:E-phi-sum-rank-r}\\
 & \sum_{A,B:\rk(A+B)=r}|(E_{\phi})_{A,B}|=|\phi(r)|, &  & r=0,1,\dots,n.\label{eq:E-phi-sum-rank-r-abs-values}
\end{align}
In particular, 
\[
\|E_{\phi}\|_{1}=\|\phi\|_{1}.
\]
\end{lem}

\begin{proof}
Recall that for any fixed matrix $A\in\mathbb{F}_{q}^{n\times n},$
the mapping $B\mapsto A+B$ is a permutation on $\mathbb{F}_{q}^{n\times n}.$
As a result, for any fixed matrix $A,$ there are exactly $|\Mcal_{r}^{n,n}|$
matrices $B$ such that $\rk(A+B)=r.$ Altogether, there are $q^{n^{2}}|\Mcal_{r}^{n,n}|$
matrix pairs $(A,B)$ with $\rk(A+B)=r.$ With this in mind, Definition~\prettyref{def:E-phi}
implies the following for each $r$:
\begin{align}
\sum_{\rk(A+B)=r}(E_{r})_{A,B} & =\sum_{\rk(A+B)=r}|(E_{r})_{A,B}|=1.\label{eq:E-normalization}
\end{align}
Now for each $r,$
\[
\sum_{\rk(A+B)=r}(E_{\phi})_{A,B}=\sum_{\rk(A+B)=r}\;\sum_{i=0}^{n}\phi(i)(E_{i})_{A,B}=\sum_{\rk(A+B)=r}\;\phi(r)(E_{r})_{A,B}=\phi(r),
\]
where the second step uses $(E_{i})_{A,B}=0$ for $i\ne r,$ and the
final step applies \prettyref{eq:E-normalization}. Analogously, 
\[
\sum_{\rk(A+B)=r}|(E_{\phi})_{A,B}|=\sum_{\rk(A+B)=r}\;\left|\sum_{i=0}^{n}\phi(i)(E_{i})_{A,B}\right|=\sum_{\rk(A+B)=r}\,|\phi(r)|\,|(E_{r})_{A,B}|=|\phi(r)|.
\]
This establishes~\prettyref{eq:E-phi-sum-rank-r} and~\prettyref{eq:E-phi-sum-rank-r-abs-values}.
Summing~\prettyref{eq:E-phi-sum-rank-r-abs-values} over $r$ gives
$\|E_{\phi}\|_{1}=\|\phi\|_{1}.$ 
\end{proof}
To discuss the spectrum of $E_{\phi}$, we first describe the Fourier
spectrum of the characteristic function of matrices of a given rank.
This is where the significance of the $\Gamma_{n}$ function becomes
evident.
\begin{lem}
\label{lem:f-t-spectrum}Let $n\geq1$ be a given integer. For $r\in\{0,1,\ldots,n\},$
define $f_{r}\colon\mathbb{\mathbb{F}}_{q}^{n\times n}\to\{0,1\}$
by $f_{r}(X)=1$ if and only if $\rk X=r.$ Then for all $M\in\mathbb{F}_{q}^{n\times n},$
\begin{align*}
\widehat{f_{r}}(M) & =\frac{|\Mcal_{r}^{n,n}|}{q^{n^{2}}}\cdot\Gamma_{n}(\rk M,r).
\end{align*}
\end{lem}

\begin{proof}
We have
\begin{align*}
\widehat{f_{r}}(M) & =\Exp_{X\in\mathbb{F}_{q}^{n\times n}}\omega^{-\langle M,X\rangle}f_{r}(X)\\
 & =q^{-n^{2}}\sum_{X\in\Mcal_{r}^{n,n}}\omega^{-\langle M,X\rangle}\\
 & =q^{-n^{2}}|\Mcal_{r}^{n,n}|\Exp_{X\in\Mcal_{r}^{n,n}}\omega^{-\langle M,X\rangle}\\
 & =q^{-n^{2}}|\Mcal_{r}^{n,n}|\Exp_{\substack{X\in\Mcal_{r}^{n,n}\\
U,V\in\Mcal_{n}
}
}\omega^{-\langle M,UXV\rangle}\\
 & =q^{-n^{2}}|\Mcal_{r}^{n,n}|\Exp_{\substack{X\in\Mcal_{r}^{n,n}\\
U,V\in\Mcal_{n}
}
}\omega^{\langle-U\tr MV\tr,X\rangle}\\
 & =q^{-n^{2}}|\Mcal_{r}^{n,n}|\Exp_{\substack{X\in\Mcal_{r}^{n,n}\\
Y\in\Mcal_{\rk M}^{n,n}
}
}\omega^{\langle Y,X\rangle}\\
 & =q^{-n^{2}}|\Mcal_{r}^{n,n}|\,\Gamma_{n}(\rk M,r),
\end{align*}
where the second step uses the definition of $f_{r}$, the fourth
step is valid by~\prettyref{prop:random-matrices}, the fifth step
invokes \prettyref{fact:inner-product-vs-trace}\prettyref{enu:inner-product-transfer},
the sixth step uses \prettyref{prop:random-matrices} once more, and
the last step applies the definition of $\Gamma_{n}.$ 
\end{proof}
We are now ready to describe the spectrum of $E_{\phi}$ in terms
of $\phi$ and the $\Gamma_{n}$ function.
\begin{lem}[Singular values of $E_{\phi}$]
\label{lem:E_phi-spectral} Let $n\geq1$ be an integer and $\phi\colon\{0,1,\ldots,n\}\to\Re$
a given function. Then the singular values of $E_{\phi}$ are 
\begin{align*}
q^{-n^{2}}\left|\sum_{t=0}^{n}\phi(t)\Gamma_{n}(s,t)\right|, &  & s=0,1,\dots,n,
\end{align*}
with corresponding multiplicities $|\Mcal_{s}^{n,n}|$ for $s=0,1,\ldots,n.$
\begin{proof}
For $t=0,1,\ldots,n,$ define $f_{t}$ as in \prettyref{lem:f-t-spectrum}.
In this notation, 
\[
E_{\phi}=\sum_{t=0}^{n}\phi(t)E_{t}=\sum_{t=0}^{n}\phi(t)\left[\frac{1}{q^{n^{2}}|\Mcal_{t}^{n,n}|}\cdot f_{t}(A+B)\right]_{A,B}=[f(A+B)]_{A,B},
\]
 where 
\[
f=\sum_{t=0}^{n}\frac{\phi(t)}{q^{n^{2}}|\Mcal_{t}^{n,n}|}\cdot f_{t}.
\]
By \prettyref{fact:spectraltofourier}, the singular values of $E_{\phi}$
are $q^{n^{2}}|\widehat{f}(M)|$ for $M\in\mathbb{F}_{q}^{n\times n}.$
Calculating,
\begin{align*}
q^{n^{2}}|\widehat{f}(M)| & =q^{n^{2}}\left|\sum_{t=0}^{n}\frac{\phi(t)}{q^{n^{2}}|\Mcal_{t}^{n,n}|}\cdot\widehat{f_{t}}(M)\right|\\
 & =q^{-n^{2}}\left|\sum_{t=0}^{n}\phi(t)\Gamma_{n}(\rk M,t)\right|,
\end{align*}
where the first step uses the linearity of the Fourier transform,
and the second step applies \prettyref{lem:f-t-spectrum}. Grouping
these singular values according to $\rk M$ shows that the spectrum
of $E_{\phi}$ is as claimed.
\end{proof}
\end{lem}

\subsection{\label{subsec:approx-trace-norm}Approximate trace norm of the rank
problem}

Using the machinery developed in previous sections, we now prove a
lower bound on the approximate trace norm of the characteristic matrix
of the rank problem. Combined with the approximate trace norm method,
this will allow us to obtain our communication lower bounds for the
rank problem.
\begin{thm}
\label{thm:approx-norm-rank-problem}Let $n>k\geq0$ be given integers.
Let $F$ be the matrix with rows and columns indexed by elements of
$\mathbb{F}_{q}^{n\times n},$ and entries given by
\[
F_{A,B}=\begin{cases}
1 & \text{if }\rk(A+B)=n,\\
-1 & \text{if }\rk(A+B)=k,\\
* & \text{otherwise.}
\end{cases}
\]
Then for all reals $\delta\geq0$ and all nonnegative integers $\ell,m$
with $\ell+m\leq k,$
\begin{align}
\|F\|_{\Sigma,\delta} & \geq\frac{1}{150}\left(1-\delta-\frac{64}{q^{m+1}}\right)q^{\ell(k-\ell-m+1)/2}\,q^{n^{2}},\label{eq:F-m-l-trace-norm-rank}\\
\|F\|_{\Sigma,\delta} & \geq\frac{1-\delta}{150}\cdot q^{k/2}\,q^{n^{2}}.\label{eq:F-m-l-trace-norm-rank-alt}
\end{align}
\end{thm}

\begin{proof}
Let $\phi\colon\{0,1,\ldots,n\}\to\Re$ be the function constructed
in \prettyref{lem:phi}. Then
\begin{align}
\sum_{\dom F}F_{A,B} & (E_{\phi})_{A,B}-\delta\|E_{\phi}\|_{1}-\sum_{\overline{\dom F}}|(E_{\phi})_{A,B}|\nonumber \\
 & =\sum_{\rk(A+B)=n}(E_{\phi})_{A,B}-\sum_{\rk(A+B)=k}(E_{\phi})_{A,B}-\delta\|E_{\phi}\|_{1}-\sum_{\rk(A+B)\notin\{n,k\}}|(E_{\phi})_{A,B}|\nonumber \\
 & =\phi(n)-\phi(k)-\delta\|\phi\|_{1}-\sum_{r\notin\{n,k\}}|\phi(r)|\nonumber \\
 & =|\phi(n)|+|\phi(k)|-\delta\|\phi\|_{1}-\sum_{r\notin\{n,k\}}|\phi(r)|\nonumber \\
 & =(1-\delta)\|\phi\|_{1}-2\sum_{r\notin\{n,k\}}|\phi(r)|\nonumber \\
 & \geq\left(1-\delta-2\sum_{r\notin\{n,k\}}|\phi(r)|\right)\|\phi\|_{1},\label{eq:Ephi-F-correlation}
\end{align}
where the second step uses \prettyref{lem:E_phi-metric}, the third
step is valid by \prettyref{lem:phi}\prettyref{enu:phi-at-n}\textendash \prettyref{enu:phi-at-k},
and the last step is justified by \prettyref{lem:phi}\prettyref{enu:phi-at-n}.

We now analyze the spectral norm of $E_{\phi}.$ Recall from \prettyref{lem:gammaprop}\prettyref{enu:gammapoly}
that for any fixed values of $n$ and $s,$ the quantity $\Gamma_{n}(s,t)$
as a function of $t\in\{0,1,\ldots,n\}$ is a polynomial in $q^{-t}$
of degree at most $s.$ In this light, \prettyref{lem:phi}\prettyref{enu:phi-orthog}
implies that 
\begin{equation}
\max_{s\in\{0,1,\ldots,k-\ell-m\}}\left|\sum_{t=0}^{n}\phi(t)\Gamma_{n}(s,t)\right|=0.\label{eq:sing-value-s-small}
\end{equation}
Continuing,
\begin{align}
\max_{s\in\{k-\ell-m+1,\ldots,n-1,n\}} & \left|\sum_{t=0}^{n}\phi(t)\Gamma_{n}(s,t)\right|\nonumber \\
 & =\max_{s\in\{k-\ell-m+1,\ldots,n-1,n\}}\left|\sum_{t=\ell}^{n}\phi(t)\Gamma_{n}(s,t)\right|\nonumber \\
 & \leq\max_{s\in\{k-\ell-m+1,\ldots,n-1,n\}}\left\{ \|\phi\|_{1}\max_{t\in\{\ell,\ell+1,\ldots,n\}}|\Gamma_{n}(s,t)|\right\} \nonumber \\
 & \leq\max_{s\in\{k-\ell-m+1,\ldots,n-1,n\}}\left\{ \|\phi\|_{1}\max_{t\in\{\ell,\ell+1,\ldots,n\}}128q^{-st/2}\right\} \nonumber \\
 & =128\|\phi\|_{1}q^{-\ell(k-\ell-m+1)/2},\label{eq:sing-value-s-large}
\end{align}
where the first step uses \prettyref{lem:phi}\prettyref{enu:phi-vanish},
and the third step applies the bound of \prettyref{lem:gammaprop}\prettyref{enu:gammabound}.
By~\prettyref{eq:sing-value-s-small}, \prettyref{eq:sing-value-s-large},
and \prettyref{lem:E_phi-spectral}, 
\begin{equation}
\|E_{\phi}\|\leq128\|\phi\|_{1}\,q^{-\ell(k-\ell-m+1)/2}\,q^{-n^{2}}.\label{eq:E-phi-spectral-norm-bound}
\end{equation}

\prettyref{prop:approxtracelower} with $\Phi=E_{\phi}$ implies,
in view of~\prettyref{eq:Ephi-F-correlation} and~\prettyref{eq:E-phi-spectral-norm-bound},
that
\begin{equation}
\|F\|_{\Sigma,\delta}\geq\frac{1}{128}\cdot\left(1-\delta-2\sum_{r\notin\{n,k\}}|\phi(r)|\right)q^{\ell(k-\ell-m+1)/2}\,q^{n^{2}}.\label{eq:F-trace-norm-intermediate}
\end{equation}
Since $\sum_{r\notin\{n,k\}}|\phi(r)|\leq32q^{-m-1}$ by~\prettyref{lem:phi}\prettyref{enu:phi-small-tail},
this settles~\prettyref{eq:F-m-l-trace-norm-rank}. The alternative
lower bound~\prettyref{eq:F-m-l-trace-norm-rank-alt} follows from~\prettyref{eq:F-trace-norm-intermediate}
by taking $\ell=k$ and $m=0$ and noting that $\sum_{r\notin\{n,k\}}|\phi(r)|=0$
in this case (by \prettyref{lem:phi}\prettyref{enu:phi-vanish}).
\end{proof}

\subsection{\label{subsec:communication-lower-bounds}Communication lower bounds}

We will now use our newly obtained lower bound on the approximate
trace norm to prove the main result of this section, a tight lower
bound on the communication complexity of the rank problem. We will
first examine the canonical case of distinguishing rank-$k$ matrices
in $\mathbb{F}^{n\times n}$ from full-rank matrices.
\begin{thm}
\label{thm:rank-lower-bound-CANONICAL}There is an absolute constant
$c>0$ such that for all finite fields $\mathbb{F}$ and all integers
$n>k\geq0,$
\begin{equation}
Q_{\frac{1}{2}-\frac{1}{4|\mathbb{F}|^{k/3}}}^{*}(\RANK_{k,n}^{\mathbb{F},n,n})\geq c(1+k^{2}\log|\mathbb{F}|).\label{eq:master-lower-bound-rank}
\end{equation}
\end{thm}

\begin{proof}
Abbreviate $q=|\mathbb{F}|$ and $\epsilon=\frac{1}{2}-\frac{1}{4q^{k/3}}$.
Since $\RANK_{k,n}^{\mathbb{F},n,n}$ is a nonconstant function, we
have the trivial lower bound
\begin{equation}
Q_{\epsilon}^{*}(\RANK_{k,n}^{\mathbb{F},n,n})\geq1.\label{eq:RANK-trivial}
\end{equation}
Let $F$ be the characteristic matrix of this communication problem.
We first examine the case $k\leq50.$ Here, taking $\delta=2\epsilon$
in equation~\prettyref{eq:F-m-l-trace-norm-rank-alt} of~\prettyref{thm:approx-norm-rank-problem}
shows that $\|F\|_{\Sigma,2\epsilon}\geq q^{k/6}q^{n^{2}}/300\geq q^{k^{2}/300}q^{n^{2}}/300,$
where the last step uses $k\leq50.$ It follows from \prettyref{thm:approx-trace-norm-method}
that
\[
Q_{\epsilon}^{*}(\RANK_{k,n}^{\mathbb{F},n,n})\geq\frac{1}{2}\log\frac{q^{k^{2}/300}}{3\cdot300}\geq\frac{1}{600}\,k^{2}\log q-5.
\]
Taking a weighted arithmetic average of this lower bound and~\prettyref{eq:RANK-trivial}
settles~\prettyref{eq:master-lower-bound-rank}.

Consider now the complementary case $k>50$. Taking $\delta=2\epsilon,$
$\ell=\lceil k/3\rceil$, and $m=\lfloor k/2\rfloor$ in equation~\prettyref{eq:F-m-l-trace-norm-rank}
of~\prettyref{thm:approx-norm-rank-problem} gives
\begin{align*}
\|F\|_{\Sigma,2\epsilon} & \geq\frac{1}{150}\left(\frac{1}{2q^{k/3}}-\frac{64}{q^{\lfloor k/2\rfloor+1}}\right)q^{\lceil k/3\rceil(k-\lceil k/3\rceil-\lfloor k/2\rfloor+1)/2}q^{n^{2}}\\
 & \geq\frac{1}{300}\left(1-\frac{128}{q^{k/6}}\right)q^{-k/3}q^{\lceil k/3\rceil k/12}q^{n^{2}}\\
 & \geq\frac{1}{600}\,q^{-k/3}q^{\lceil k/3\rceil k/12}q^{n^{2}}\\
 & \geq\frac{1}{600}\,q^{k^{2}/48}q^{n^{2}},
\end{align*}
where the last two steps use $k>50.$ As a result, \prettyref{thm:approx-trace-norm-method}
guarantees that
\[
Q_{\epsilon}^{*}(\RANK_{k,n}^{\mathbb{F},n,n})\geq\frac{1}{2}\log\frac{q^{k^{2}/48}}{3\cdot600}\geq\frac{1}{96}\,k^{2}\log q-6.
\]
Taking a weighted arithmetic average of this lower bound and~\prettyref{eq:RANK-trivial}
settles~\prettyref{eq:master-lower-bound-rank}.
\end{proof}
We now establish our main lower bound for the rank problem in its
full generality.
\begin{thm*}[restatement of \prettyref{thm:RANK-lowerbound-intro}]
\label{thm:rank-lower-bound-GENERAL}There is an absolute constant
$c>0$ such that for all finite fields $\mathbb{F}$ and all integers
$n,m,R,r$ with $\min\{n,m\}\geq R>r\geq0,$
\begin{equation}
Q_{\frac{1}{2}-\frac{1}{4|\mathbb{F}|^{r/3}}}^{*}(\RANK_{r,R}^{\mathbb{F},n,m})\geq c(1+r^{2}\log|\mathbb{F}|).\label{eq:master-lower-bound-rank-general}
\end{equation}
In particular, 
\begin{equation}
Q_{1/4}^{*}(\RANK_{r,R}^{\mathbb{F},n,m})\geq c(1+r^{2}\log|\mathbb{F}|).\label{eq:master-lower-bound-rank-const-error}
\end{equation}
\end{thm*}
\begin{proof}
There is a communication-free reduction from $\RANK_{r,R}^{\mathbb{F},R,R}$
to $\mathbb{\RANK}_{r,R}^{\mathbb{F},n,m}$, where Alice and Bob pad
their input matrices $A,B\in\mathbb{F}^{R\times R}$ with zeroes to
obtain matrices $A',B'\in\mathbb{F}^{n\times m}$ with $\rk(A+B)=\rk(A'+B').$
Therefore, $Q_{\epsilon}^{*}(\RANK_{r,R}^{\mathbb{F},n,m})\geq Q_{\epsilon}^{*}(\RANK_{r,R}^{\mathbb{F},R,R})$
for all $\epsilon.$ Now~\prettyref{thm:rank-lower-bound-CANONICAL}
implies~\prettyref{eq:master-lower-bound-rank-general}, which in
turn implies \prettyref{eq:master-lower-bound-rank-const-error}.
\end{proof}

\subsection{\label{subsec:communication-upper-bounds}Communication upper bounds}

To finalize our study of the rank problem, we will prove a matching
upper bound on its communication complexity. We emphasize that our
upper bound is achieved by a randomized (classical) protocol, whereas
our lower bound is valid even for quantum communication.
\begin{thm}
\label{thm:RANK-upperbound-high-accuracy}Let $n,m,R$ be nonnegative
integers with $\min\{n,m\}\geq R\geq0$. Let $\mathbb{F}$ be a finite
field with $q=|\mathbb{F}|$ elements. Then for all $t\geq2$ and
$\epsilon\in(0,1),$ there is a $t$-party randomized communication
protocol which: 
\begin{itemize}
\item takes as input matrices $A_{1},A_{2},\ldots,A_{t}\in\mathbb{F}^{n\times m}$
for players $1,2,\ldots,t,$ respectively;
\item computes $\min\{\rk(\sum A_{i}),R\}$ with probability of error at
most $\epsilon;$ and 
\item has communication cost $O(t(R+\lceil\log_{q}(1/\epsilon)\rceil)^{2}\log q)$. 
\end{itemize}
\end{thm}

\begin{proof}
We may assume that $n,m\geq1$ since the theorem is trivial otherwise.
The communication protocol is based on random projections and is inspired
by Clarkson and Woodruff's streaming algorithm~\cite{clarkson-woodruff09}
for matrix rank. Set $\Delta=\lceil\log_{q}(8/\epsilon)\rceil$ and
$A=\sum A_{i}$. The players use their shared randomness to pick a
pair of independent and uniformly random matrices $X\in\mathbb{F}^{(R+\Delta)\times n}$
and $Y\in\mathbb{F}^{m\times(R+\Delta)}$. Then each player $i$ sends
the matrix $XA_{i}Y\in\mathbb{F}^{(R+\Delta)\times(R+\Delta)}$, and
they all output $\min\{\rk(XAY),R\}$. The communication cost is $O(t(R+\Delta)^{2}\log q)$
as claimed, due to $XAY=\sum XA_{i}Y$. It is also clear that this
protocol always outputs a \emph{lower} bound on the correct value
$\min\{\rk A,R\},$ due to $\rk(XAY)\leq\rk A$ for all $X,Y$. It
remains to show that 
\begin{equation}
\Prob[\rk(XAY)\geq\min\{\rk A,R\}]\geq1-\epsilon.\label{eq:lower-bound-tight-rank-protocol}
\end{equation}

Conditioned on $X,$ we have $\rk(XAY)\geq\min\{\rk(XA),R\}$ with
probability at least $1-4q^{-\Delta-1}\geq1-\epsilon/2$ (apply \prettyref{lem:random-proj-matrices}\prettyref{enu:right-matrix-proj}
with $M=XA$ and $t=\min\{\rk(XA),R\}-1$). Similarly, $\rk(XA)\geq\min\{\rk A,R\}$
with probability at least $1-\epsilon/2$ (apply \prettyref{lem:random-proj-matrices}\prettyref{enu:left-matrix-proj}
with $M=A$ and $t=\min\{\rk A,R\}-1$). The union bound now gives~\prettyref{eq:lower-bound-tight-rank-protocol}.
\end{proof}
In the corollary below, $\RANK_{r}^{\mathbb{F},n,m,t}\colon(\mathbb{F}^{n\times m})^{t}\to\{-1,1\}$
denotes the total version of the matrix rank problem for $t$ parties,
given by $\RANK_{r}^{\mathbb{F},n,m,t}(A_{1},A_{2},\ldots,A_{t})=-1$
if and only if $\rk(\sum A_{i})\leq r.$
\begin{cor}
\label{cor:RANK-upperbound-high-accuracy}Let $n,m,r$ be integers
with $\min\{n,m\}>r\geq0.$ Let $\mathbb{F}$ be a finite field with
$q=|\mathbb{F}|$ elements. Then for all $\epsilon\in(0,1/2),$
\begin{equation}
R_{\epsilon}(\RANK_{r}^{\mathbb{F},n,m})=\begin{cases}
O(\log(1/\epsilon)) & \text{if }r=0,\\
O((r+\lceil\log_{q}(1/\epsilon)\rceil)^{2}\log q) & \text{otherwise.}
\end{cases}\label{eq:RANK-upper-total}
\end{equation}
More generally, for all $t\geq2,$
\begin{equation}
R_{\epsilon}(\RANK_{r}^{\mathbb{F},n,m,t})=O(t(r+\lceil\log_{q}(1/\epsilon)\rceil)^{2}\log q).\label{eq:RANK-upper-total-multiparty}
\end{equation}
\end{cor}

\begin{proof}
We have $\RANK_{r}^{\mathbb{F},n,m,t}(A_{1},A_{2},\ldots,A_{t})=-1$
if and only if $\min\{\rk(\sum A_{i}),r+1\}\leq r.$ To compute this
minimum with error $\epsilon$, one can use the $t$-party protocol
of \prettyref{thm:RANK-upperbound-high-accuracy} with $R=r+1,$ with
communication cost $O(t(r+\lceil\log_{q}(1/\epsilon)\rceil)^{2}\log q)$.
This settles the multiparty bound~\prettyref{eq:RANK-upper-total-multiparty},
which in turn implies the two-party bound~\prettyref{eq:RANK-upper-total}
for $r\geq1.$

Lastly, $\RANK_{0}^{\mathbb{F},n,m}(A,B)=-1$ if and only if $A=-B.$
Thus, $\RANK_{0}^{\mathbb{F},n,m}$ is equivalent to the equality
problem with domain $\mathbb{F}^{n\times m}\times\mathbb{F}^{n\times m}$.
It is well known~\cite{ccbook} that the $\epsilon$-error randomized
communication complexity of equality is $O(\log(1/\epsilon)).$ Therefore,
\prettyref{eq:RANK-upper-total} holds also for $r=0.$
\end{proof}
Our communication upper bound readily generalizes to the bilinear
query model, as follows.
\begin{thm}
\label{thm:BLQ-upper}Let $n,m,r$ be integers with $\min\{n,m\}>r\geq0.$
Let $\mathbb{F}$ be a finite field with $q=|\mathbb{F}|$ elements.
Then for all $\epsilon\in(0,1/2),$ there is a query algorithm in
the bilinear query model with cost $O((r+\lceil\log_{q}(1/\epsilon)\rceil)^{2})$
that takes as input a matrix $A\in\mathbb{F}^{n\times m}$ and determines
whether $\rk A\leq r$ with probability of error at most $\epsilon$.
\end{thm}

\begin{proof}
On input $A\in\mathbb{F}^{n\times m},$ choose $X\in\mathbb{F}^{(R+\Delta)\times n}$
and $Y\in\mathbb{F}^{m\times(R+\Delta)}$ uniformly at random, where
$\Delta=\lceil\log_{q}(8/\epsilon)\rceil$ and $R=r+1$. Then trivially
$\rk(XAY)\leq\rk A.$ In the opposite direction, we have the bound
\prettyref{eq:lower-bound-tight-rank-protocol}, established in the
last paragraph of the proof of \prettyref{thm:RANK-upperbound-high-accuracy}.
Therefore, we can determine whether $\rk A\leq r$ with probability
of error $\epsilon$ by checking whether $\rk(XAY)\leq r.$ Since
the entries of $XAY$ are bilinear forms in $A,$ the entire matrix
$XAY$ can be recovered using $(R+\Delta)^{2}$ queries. 
\end{proof}
We now prove an alternative communication upper bound, showing that
even a two-bit protocol can solve the rank problem with nontrivial
advantage. For simplicity, we will only consider the two-party model;
a similar statement can be proved for bilinear query complexity. 
\begin{thm}
\label{thm:RANK-upperbound-2bit}Let $n,m,r$ be integers with $\min\{n,m\}>r\geq0.$
Let $\mathbb{F}$ be a finite field with $q=|\mathbb{F}|$ elements.
Then 
\begin{equation}
R_{\frac{1}{2}-\frac{1}{32q^{r}}}(\mathrm{RANK}_{r}^{\mathbb{F},n,m})\leqslant2.\label{eq:rank-2bit-total}
\end{equation}
\end{thm}

\begin{proof}
Consider the following auxiliary protocol $\Pi$. On input $A,B\in\mathbb{F}^{n\times m},$
Alice and Bob use their shared randomness to pick a pair of independent
and uniformly random vectors $x\in\mathbb{F}^{n}$ and $y\in\mathbb{F}^{m},$
as well as a uniformly random function $H\colon\mathbb{F}\to\{-1,1\}.$
They exchange $H(x\tr Ay)$ and $H(-x\tr By)$ using $2$ bits of
communication and output $-H(x\tr Ay)H(-x\tr By).$

We now analyze the expected output of $\Pi(A,B)$ on a given matrix
pair $A,B$. To begin with,
\begin{equation}
\Exp[\Pi(A,B)\mid x,y]=\begin{cases}
-1 & \text{if }x\tr(A+B)y=0,\\
0 & \text{otherwise.}
\end{cases}\label{eq:Pi-cond-x-y}
\end{equation}
Indeed, if $x\tr(A+B)y=0$ then $x\tr Ay=-x\tr By$ and therefore
$\Pi$ outputs $-1.$ If, on the other hand, $x\tr(A+B)y\ne0$ then
$x\tr Ay\ne-x\tr By$, which means that $H(x\tr Ay)$ and $H(-x\tr By)$
are independent and their product has expected value $0$. Equation~\prettyref{eq:Pi-cond-x-y}
implies that $\Exp\Pi(A,B)=-\Prob[x\tr(A+B)y=0]$, which can be expanded
as
\[
\Exp\Pi(A,B)=-\Prob[x\tr(A+B)=0]-\Prob[x\tr(A+B)\ne0]\Prob[x\tr(A+B)y=0\mid x\tr(A+B)\ne0].
\]
The event $x\tr(A+B)=0$ is equivalent to $x$ being in the orthogonal
complement of the column span of $A+B,$ which happens with probability
$q^{n-\rk(A+B)}/q^{n}=q^{-\rk(A+B)}.$ Conditioned on $x\tr(A+B)\ne0,$
the field element $x\tr(A+B)y$ is uniformly random and in particular
is $0$ with probability $1/q.$ As a result,
\[
\Exp\Pi(A,B)=-\frac{1}{q^{\rk(A+B)}}-\left(1-\frac{1}{q^{\rk(A+B)}}\right)\cdot\frac{1}{q}=-\frac{1}{q}-\frac{q-1}{q^{\rk(A+B)+1}}.
\]
Therefore, the expected value of $\Pi(A,B)$ is at most $-1/q-(q-1)/q^{r+1}$
when $\rk(A+B)\leq r$ and is at least $-1/q-(q-1)/q^{r+2}$ when
$\rk(A+B)>r$. \prettyref{prop:shift-probab} now shows that $\RANK_{r}^{\mathbb{F},n,m}$
has a communication protocol with the same cost as $\Pi$ and error
at most $\frac{1}{2}-\frac{1}{8}(q-1)^{2}/q^{r+2}.$ This settles~\prettyref{eq:rank-2bit-total}
since $q\geq2.$
\end{proof}
\prettyref{cor:RANK-upperbound-high-accuracy} (with $\epsilon=1/3$)
and \prettyref{thm:RANK-upperbound-2bit} settle \prettyref{thm:RANK-upperbound}
from the introduction. 

\subsection{\label{subsec:streaming}Streaming complexity}

Fix a finite field $\mathbb{F}$ and a (possibly partial) function
$f\colon\mathbb{F}^{n\times n}\to\{-1,1,*\}$. A streaming algorithm
for $f$ receives as input a matrix $M\in\mathbb{F}^{n\times n}$
in row-major order. We say that $\mathcal{A}$ \emph{computes $f$
with error} $\epsilon$ if for every input in the domain of $f$,
the output of $\mathcal{A}$ agrees with $f$ with probability at
least $1-\epsilon$. We will now use a well-known reduction~\cite{LiSWW14communication-linear}
to transform our communication lower bound for the matrix rank problem
into a lower bound on its streaming complexity.
\begin{thm*}[restatement of Theorem~\ref{thm:RANK-streaming-intro}]
 Let $n,r,R$ be nonnegative integers with $n/2\leq r<R\leq n,$
and let $\mathbb{F}$ be a finite field. Define $f\colon\mathbb{F}^{n\times n}\to\{-1,1,*\}$
by
\[
f(M)=\begin{cases}
-1 & \text{if }\rk M=r,\\
1 & \text{if }\rk M=R,\\
* & \text{otherwise.}
\end{cases}
\]
Let $\mathcal{A}$ be any randomized streaming algorithm for $f$
with error probability $\frac{1}{2}-\frac{1}{4}|\mathbb{F}|^{-(r-\lceil n/2\rceil)/3}$
that uses $k$ passes and space $s$. Then
\begin{equation}
sk=\Omega\left(\left(r-\left\lceil \frac{n}{2}\right\rceil \right)^{2}\log|\mathbb{F}|\right).\label{eq:sp-tradeoff}
\end{equation}
\end{thm*}
\begin{proof}
Abbreviate $m=\lfloor n/2\rfloor$ and $F=\RANK_{r-\lceil n/2\rceil,R-\lceil n/2\rceil}^{\mathbb{F},m,m}$.
We will use a reduction from communication to streaming due to Li,
Sun, Wang, and Woodruff~\cite[Thm.~29]{LiSWW14communication-linear}.
Specifically, let $A,B\in\mathbb{F}^{m\times m}$ be Alice and Bob's
inputs, respectively, for $F$. Define
\[
M=\begin{bmatrix}A & -I_{m} & 0\\
B & I_{m} & 0\\
0 & 0 & I_{n-2m}
\end{bmatrix},
\]
where $I_{m}$ and $I_{n-2m}$ stand for the identity matrices of
order $m$ and $n-2m,$ respectively (in particular, $I_{n-2m}$ is
empty for even $n$). We have
\[
\rk M=\rk\begin{bmatrix}A+B & 0 & 0\\
B & I_{m} & 0\\
0 & 0 & I_{n-2m}
\end{bmatrix}=\rk(A+B)+n-m=\rk(A+B)+\left\lceil \frac{n}{2}\right\rceil .
\]
As a result, for all matrix pairs $(A,B)$ with $\rk(A+B)\in\{r-\lceil n/2\rceil,R-\lceil n/2\rceil\},$
we have $F(A,B)=f(M)$. This makes it possible for Alice and Bob to
compute $F$ by simulating $\mathcal{A}$ on $M$. Alice starts the
simulation by running $\mathcal{A}$ on the first $m$ rows of $M$,
which depend only on her input $A$. She then sends Bob the contents
of $\mathcal{A}$'s memory, and Bob runs $\mathcal{A}$ on the remaining
$n-m$ rows of $M$. This completes the first pass. Next, Bob sends
Alice the contents of $\mathcal{A}$'s memory, and they continue as
before until they simulate all $k$ passes. At the end of the $k$-th
pass, Bob announces the output of $\mathcal{A}$ as the protocol output.
The error probability of the described protocol is the same as that
of $\mathcal{A}$, and the communication cost is $s(2k-1)+1$ bits.
Therefore,
\[
R_{\frac{1}{2}-\frac{1}{4}|\mathbb{F}|^{-(r-\lceil n/2\rceil)/3}}(F)\leq s(2k-1)+1.
\]
Since the left-hand side is at least $\Omega((r-\lceil n/2\rceil)^{2}\log|\mathbb{F}|+1)$
by \prettyref{thm:RANK-lowerbound-intro}, the claimed trade-off~\prettyref{eq:sp-tradeoff}
follows.
\end{proof}

\section{The determinant problem}

In this section, we establish our lower bound on the communication
complexity of the determinant problem. We begin in Section~\prettyref{subsec:auxiliary-results}
with technical results on characteristic functions of matrices with
a given determinant value. In Section~\prettyref{subsec:determinant-nonzero},
we give our own proof of the lower bound for distinguishing two nonzero
values of the determinant, which is simpler and more elementary than
the proof in \cite{sunwang12communication-linear}. In Section~\prettyref{subsec:determinant-arbitrary},
we prove an optimal lower bound for the general case of distinguishing
two arbitrary values of the determinant, solving an open problem from~\cite{sunwang12communication-linear}.
Throughout this section, we use a generic finite field $\mathbb{F}$
with $q$ elements, where $q$ is an arbitrary prime power. The root
of unity $\omega$ and the notation $\omega^{x}$ for $x\in\mathbb{F}$
are as defined in \prettyref{sec:Fourier}.

\subsection{\label{subsec:auxiliary-results}Auxiliary results}

Fix a finite field $\mathbb{F}$ and a positive integer $n$. Recall
that the determinant function on $\mathbb{F}^{n\times n}$ is multiplicative,
with $\det(AB)=\det(A)\det(B).$ As a result, the set of matrices
in $\mathbb{F}^{n\times n}$ with nonzero determinants form a group
under matrix multiplication, called the \emph{general linear group}
and denoted by $\GL(\mathbb{F},n).$ Analogously, the matrices in
$\mathbb{F}^{n\times n}$ with determinant $1$ also form a group,
called the \emph{special linear group} and denoted by $\SL(\mathbb{F},n).$
The multiplicativity of the determinant further implies that $\SL(\mathbb{F},n)$
is a normal subgroup of $\GL(\mathbb{F},n)$, with quotient isomorphic
to the multiplicative group of the field: $\GL(\mathbb{F},n)/\SL(\mathbb{F},n)\cong\mathbb{F}^{\times}.$
For any given field element $u\ne0$, the set of matrices with determinant
$u$ form a coset of $\SL(\mathbb{F},n)$ in $\GL(\mathbb{F},n).$
In particular,
\begin{align}
|\{X & \in\mathbb{F}^{n\times n}:\det X=u\}|=\lvert\SL(\mathbb{F},n)\rvert=\frac{\lvert\Mcal_{n}\rvert}{|\mathbb{F}|-1}, &  & u\in\mathbb{F}\setminus\{0\}.\label{eq:coset-size}
\end{align}
Recall that for each $Y\in\mathbb{F}^{n\times n},$ the mapping $X\mapsto X+Y$
is a permutation on $\mathbb{F}^{n\times n}.$ As a result, the previous
equation implies that
\begin{align}
|\{(X,Y) & \in\mathbb{F}^{n\times n}\times\mathbb{F}^{n\times n}:\det(X+Y)=u\}|=|\mathbb{F}|^{n^{2}}\lvert\SL(\mathbb{F},n)\rvert, & u\in\mathbb{F}\setminus\{0\}.\label{eq:pairs-with-given-det}
\end{align}

To understand the spectral norm of the characteristic matrix of the
determinant problem, we now introduce a relevant function on $\mathbb{F}^{n\times n}$
and discuss its Fourier coefficients.
\begin{lem}
\label{lem:DETab-fourier}Let $n$ be a positive integer, $\mathbb{F}$
a finite field. For a pair of distinct elements $u,v\in\mathbb{F}\setminus\{0\},$
define $g_{u,v}\colon\mathbb{F}^{n\times n}\to\{-1,1,0\}$ by
\[
g_{u,v}(X)=\begin{cases}
-1 & \text{if }\det X=u,\\
1 & \text{if }\det X=v,\\
0 & \text{otherwise.}
\end{cases}
\]
Then:
\begin{enumerate}
\item \label{enu:Fourier-coeff-singular}$\widehat{g_{u,v}}(A)=0$ for every
singular matrix $A;$
\item \label{enu:Fourier-coeff-equality}$\widehat{g_{u,v}}(A)=\widehat{g_{u,v}}(B)$
whenever $\det A=\det B;$
\item \label{enu:Fourier-coeff-bound}$\|\widehat{g_{u,v}}\|_{\infty}\leq1/\sqrt{\lvert\SL(\mathbb{F},n)\rvert}$.
\end{enumerate}
\end{lem}

\begin{proof}
\prettyref{enu:Fourier-coeff-singular} In view of \prettyref{eq:coset-size},
we have
\begin{align*}
\widehat{g_{u,v}}(A) & =\Exp_{X\in\mathbb{F}^{n\times n}}g_{u,v}(X)\omega^{-\langle A,X\rangle}\\
 & =|\mathbb{F}|^{-n^{2}}\cdot\frac{|\Mcal_{n}|}{|\mathbb{F}|-1}\left(\Exp_{X:\det X=v}\omega^{-\langle A,X\rangle}-\Exp_{X:\det X=u}\omega^{-\langle A,X\rangle}\right).
\end{align*}
It remains to show that the expectations in the last expression are
equal. Since $A$ is singular, there exist nonsingular matrices $P$
and $Q$ such that $A=PI_{s}Q$ for $s=\rk A<n$. Consider the order-$n$
diagonal matrix $Z=\diag(1,1,\ldots,1,u^{-1}v)$. Using $I_{s}=I_{s}Z,$
we obtain $A=PI_{s}ZQ=PI_{s}QQ^{-1}ZQ=AQ^{-1}ZQ.$ As a result,
\begin{align*}
\Exp_{X:\det X=u}\omega^{-\langle A,X\rangle} & =\Exp_{X:\det X=u}\omega^{-\langle AQ^{-1}ZQ,X\rangle}\\
 & =\Exp_{X:\det X=u}\omega^{-\langle A,X(Q^{-1}ZQ)\tr\rangle}\\
 & =\Exp_{Y:\det Y=v}\omega^{-\langle A,Y\rangle},
\end{align*}
where the second step uses \prettyref{fact:inner-product-vs-trace}\prettyref{enu:inner-product-transfer},
and the last step is valid because the mapping $X\mapsto X(Q^{-1}ZQ)\tr$
is a bijection from the set of matrices with determinant $u$ onto
the set of matrices with determinant $u\cdot\det((Q^{-1}ZQ)\tr)=v.$

\prettyref{enu:Fourier-coeff-equality} For singular $A$ and $B$,
the claim is immediate from~\prettyref{enu:Fourier-coeff-singular}.
In the complementary case, 
\begin{align*}
\widehat{g_{u,v}}(A) & =\Exp_{X\in\mathbb{F}^{n\times n}}g_{u,v}(X)\omega^{-\langle A,X\rangle}\\
 & =\Exp_{X\in\mathbb{F}^{n\times n}}g_{u,v}((BA^{-1})\tr X)\omega^{-\langle A,(BA^{-1})\tr X\rangle}\\
 & =\Exp_{X\in\mathbb{F}^{n\times n}}g_{u,v}(X)\omega^{-\langle A,(BA^{-1})\tr X\rangle}\\
 & =\Exp_{X\in\mathbb{F}^{n\times n}}g_{u,v}(X)\omega^{-\langle B,X\rangle}\\
 & =\widehat{g_{u,v}}(B),
\end{align*}
where the second step is valid because $(BA^{-1})\tr$ is invertible
and hence $X\mapsto(BA^{-1})\tr X$ is a permutation on $\mathbb{F}^{n\times n};$
the third step is justified by $\det((BA^{-1})\tr X)=\det(B)\det(X)/\det(A)=\det X$;
and the fourth step is an application of \prettyref{fact:inner-product-vs-trace}\prettyref{enu:inner-product-transfer}.

\prettyref{enu:Fourier-coeff-bound} Let $M$ be a matrix with $|\widehat{g_{u,v}}(M)|=\|\widehat{g_{u,v}}\|_{\infty}.$
By~\prettyref{enu:Fourier-coeff-singular}, we know that $\det M\ne0.$
Now
\begin{align*}
1 & \geq\Exp_{X\in\mathbb{F}^{n\times n}}[|g_{u,v}(X)|^{2}]\\
 & =\sum_{A\in\mathbb{\mathbb{F}}^{n\times n}}|\widehat{g_{u,v}}(A)|^{2}\\
 & \geq\sum_{A:\det A=\det M}|\widehat{g_{u,v}}(A)|^{2}\\
 & =|\{A:\det A=\det M\}|\,|\widehat{g_{u,v}}(M)|^{2}\\
 & =\lvert\SL(\mathbb{F},n)\rvert\,\|\widehat{g_{u,v}}\|_{\infty}^{2},
\end{align*}
where the second step applies Parseval's inequality~\prettyref{eq:parsevals},
the fourth step is justified by~\prettyref{enu:Fourier-coeff-equality},
and the fifth step uses~$\det M\ne0$ along with~\prettyref{eq:coset-size}.
\end{proof}

\subsection{\label{subsec:determinant-nonzero}Determinant problem for nonzero
field elements}

As an application of the previous lemma, we now prove that the characteristic
matrix of the determinant problem $\DET_{a,b}^{\mathbb{F},n}$ for
any two nonzero field elements $a,b$ has small spectral norm.
\begin{lem}
\label{lem:DET-problem-witness}Let $\mathbb{F}$ be a finite field
with $q=|\mathbb{F}|$ elements. For each $u\in\mathbb{F}\setminus\{0\},$
define $G_{u}$ to be the matrix with rows and columns indexed by
elements of $\mathbb{F}^{n\times n},$ and entries given by 
\begin{equation}
(G_{u})_{X,Y}=\begin{cases}
q^{-n^{2}}\lvert\SL(\mathbb{F},n)\rvert^{-1} & \text{if }\det(X+Y)=u,\\
0 & \text{otherwise.}
\end{cases}\label{eq:G_u-defined}
\end{equation}
Then
\begin{align}
 & \|G_{u}\|_{1}=1, &  & u\in\mathbb{F}\setminus\{0\},\label{eq:Ga-ell1}\\
 & \|G_{v}-G_{u}\|\leq\lvert\SL(\mathbb{F},n)\rvert^{-3/2}\leq8q^{-3(n^{2}-1)/2}, &  & u,v\in\mathbb{F}\setminus\{0\}.\label{eq:Gab-spectral}
\end{align}
\end{lem}

\begin{proof}
Equation~\prettyref{eq:Ga-ell1} follows from~\prettyref{eq:pairs-with-given-det}.
For \prettyref{eq:Gab-spectral}, there are two cases to consider.
If $u=v,$ then $G_{v}-G_{u}=0$ and thus $\|G_{v}-G_{u}\|=0.$ If
$u\ne v,$ write $G_{v}-G_{u}=[q^{-n^{2}}\lvert\SL(\mathbb{F},n)\lvert^{-1}g_{u,v}(X+Y)]_{X,Y}$
with $g_{u,v}$ as defined in \prettyref{lem:DETab-fourier}. Then
\begin{equation}
\|G_{v}-G_{u}\|=\frac{\|\widehat{g_{u,v}}\|_{\infty}}{\lvert\SL(\mathbb{F},n)\rvert}\leq\frac{1}{\lvert\SL(\mathbb{F},n)\rvert^{3/2}},\label{eq:Gb-Ga-spectral-intermed}
\end{equation}
where the first step applies \prettyref{fact:spectraltofourier},
and the second step uses \prettyref{lem:DETab-fourier}\prettyref{enu:Fourier-coeff-bound}.
It remains to simplify the bound of~\prettyref{eq:Gb-Ga-spectral-intermed}:
\[
\frac{1}{\lvert\SL(\mathbb{F},n)\rvert^{3/2}}=\left(\frac{|\Mcal_{n}|}{q-1}\right)^{-3/2}=\left(q^{n-1}\prod_{i=0}^{n-2}(q^{n}-q^{i})\right)^{-3/2}\leq8q^{-3(n^{2}-1)/2},
\]
where the first step uses \prettyref{eq:coset-size}, the second step
applies \prettyref{prop:nummatr}, and the last step is justified
by \prettyref{prop:prodprop}.
\end{proof}
\prettyref{lem:DET-problem-witness} was originally obtained by Sun
and Wang~\cite{sunwang12communication-linear} using a different
and rather technical proof. By contrast, the proof presented above
is short and uses only basic Fourier analysis. With this newly obtained
bound on the spectral norm of the characteristic matrix of $\DET_{a,b}^{\mathbb{F},n}$
for nonzero $a,b$, we can use the approximate trace norm method to
obtain a tight communication lower bound for this special case of
the determinant problem.
\begin{thm}
\label{thm:DET-cc-ab-nonzero}Let $\mathbb{F}$ be a finite field,
and $n$ a positive integer. Then for every pair of distinct elements
$a,b\in\mathbb{F}\setminus\{0\}$ and every $\gamma\in(0,1),$ 
\begin{equation}
Q_{(1-\gamma)/2}^{*}(\DET_{a,b}^{\mathbb{F},n})\geq\frac{1}{4}(n^{2}-3)\log|\mathbb{F}|-\frac{1}{2}\log\frac{12}{\gamma}.\label{eq:Q-det-ab}
\end{equation}
\end{thm}

\begin{proof}
Let $F$ be the characteristic matrix of $\DET_{a,b}^{\mathbb{F},n}$.
For $u\in\mathbb{F}\setminus\{0\},$ define $G_{u}$ as in \prettyref{lem:DET-problem-witness}.
Since $G_{a}$ and $G_{b}$ are supported on disjoint sets of entries,
\prettyref{eq:Ga-ell1} leads to
\begin{equation}
\|G_{b}-G_{a}\|_{1}=\|G_{b}\|_{1}+\|G_{a}\|_{1}=2.\label{eq:Gba-ell1}
\end{equation}
Taking $\Phi=G_{b}-G_{a}$ in~\prettyref{prop:approxtracelower},
we obtain
\begin{align}
\hspace{-1cm}\|F\|_{\Sigma,1-\gamma} & \geq\frac{1}{\|G_{b}-G_{a}\|}\left(\sum_{\dom F}F_{A,B}(G_{b}-G_{a})_{A,B}-(1-\gamma)\|G_{b}-G_{a}\|_{1}-\sum_{\overline{\dom F}}|(G_{b}-G_{a})_{A,B}|\right)\nonumber \\
 & =\frac{1}{\|G_{b}-G_{a}\|}\left(\sum_{\dom F}|(G_{b}-G_{a})_{A,B}|-(1-\gamma)\|G_{b}-G_{a}\|_{1}\right)\nonumber \\
 & =\frac{\gamma\|G_{b}-G_{a}\|_{1}}{\|G_{b}-G_{a}\|}\nonumber \\
 & \geq\frac{1}{4}\,\gamma|\mathbb{F}|^{3(n^{2}-1)/2},\label{eq:F-trace-norm-gamma}
\end{align}
where the second and third steps are valid because $G_{b}-G_{a}$
by definition coincides in sign with $F$ on $\dom F$ and vanishes
on $\overline{\dom F}$; and the last step uses~\prettyref{eq:Gab-spectral}
and~\prettyref{eq:Gba-ell1}. Now~\prettyref{eq:Q-det-ab} follows
from~\prettyref{eq:F-trace-norm-gamma} in view of \prettyref{thm:approx-trace-norm-method}.
\end{proof}
We remind the reader that \prettyref{thm:DET-cc-ab-nonzero} was obtained
with different techniques by Sun and Wang~\cite{sunwang12communication-linear},
who settled the determinant problem $\DET_{a,b}^{\mathbb{F},n}$ for
nonzero $a,b$ and left open the complementary case when one of $a,b$
is zero.

\subsection{\label{subsec:determinant-arbitrary}Determinant problem for arbitrary
field elements}

Recall that the \emph{rank versus determinant problem,} $\RANKDET_{k,a}^{\mathbb{F},n}$,
is a hybrid problem that naturally generalizes the matrix rank problem
$\RANK_{k,n}^{\mathbb{F},n,n}$ and the determinant problem $\DET_{0,a}^{\mathbb{F},n}$.
Specifically, the rank versus determinant problem requires Alice and
Bob to distinguish matrix pairs with $\rk(A+B)=k$ from those with
$\det(A+B)=a,$ where $a$ is a nonzero field element, $k$ is an
integer with $k<n,$ and $A,B$ are Alice and Bob's respective inputs.
We will now construct a dual matrix for $\RANKDET_{k,a}^{\mathbb{F},n}$
and thereby obtain a lower bound on its approximate trace norm. As
a dual matrix, we will use a linear combination of the dual matrices
from our analyses of the rank and determinant problems.
\begin{thm}
\label{thm:RANK-vs-DET-trace-norm}Let $n>k\geq1$ be given integers.
Let $\mathbb{F}$ be a finite field with $q=|\mathbb{F}|$ elements,
and let $a\in\mathbb{F}\setminus\{0\}.$ Let $F$ be the characteristic
matrix of $\RANKDET_{k,a}^{\mathbb{F},n}.$ Then for all reals $\delta\geq0$
and all nonnegative integers $\ell,m$ with $\ell+m\leq k,$
\begin{align}
\|F\|_{\Sigma,\delta} & \geq\frac{1}{150}\left(1-\delta-\frac{64}{q^{m+1}}\right)q^{\ell(k-\ell-m+1)/2}\,q^{n^{2}},\label{eq:G-m-l-trace-norm}\\
\|F\|_{\Sigma,\delta} & \geq\frac{1-\delta}{150}\cdot q^{k/2}\,q^{n^{2}}.\label{eq:G-m-l-trace-norm-alt}
\end{align}
\end{thm}

\begin{proof}
This proof combines our ideas in Theorems~\ref{thm:approx-norm-rank-problem}
and \ref{thm:DET-cc-ab-nonzero}, and our dual matrix here will be
a linear combination of the dual matrices used in those theorems. 

Fix nonnegative integers $\ell,m$ with $\ell+m\leq k,$ and let $\phi\colon\{0,1,\ldots,n\}\to\Re$
be the corresponding function constructed in \prettyref{lem:phi}.
This univariate function gives rise to a matrix $E_{\phi},$ described
in Definition~\prettyref{def:E-phi}. To restate equation~\prettyref{eq:E-phi-spectral-norm-bound}
from our proof of \prettyref{thm:approx-norm-rank-problem},
\begin{equation}
\|E_{\phi}\|\leq128\|\phi\|_{1}\,q^{-\ell(k-\ell-m+1)/2}\,q^{-n^{2}}.\label{eq:E-phi-spectral}
\end{equation}
For $u\in\mathbb{F}\setminus\{0\},$ define $G_{u}$ as in~\prettyref{lem:DET-problem-witness}.
As our dual matrix, we will use
\begin{equation}
\Phi=E_{\phi}+\sum_{b\in\mathbb{F}\setminus\{0,a\}}\frac{\phi(n)}{q-1}\,(G_{a}-G_{b}).\label{eq:Phi-defined}
\end{equation}
\begin{claim}
\label{claim:Phi-simplified}For every matrix pair $(A,B),$
\[
\Phi_{A,B}=\begin{cases}
(E_{\phi})_{A,B} & \text{if }\det(A+B)=0,\\
\phi(n)q^{-n^{2}}\lvert\SL(\mathbb{F},n)\rvert^{-1} & \text{if }\det(A+B)=a,\\
0 & \text{otherwise.}
\end{cases}
\]
\end{claim}

\begin{proof}
If $\det(A+B)=0,$ then by definition $(G_{u})_{A,B}=0$ for every
nonzero field element $u.$ As a result, \prettyref{eq:Phi-defined}
gives $\Phi_{A,B}=(E_{\phi})_{A,B}$ in this case. 

In what follows, we treat the complementary case when $\det(A+B)\ne0.$
For all such matrix pairs,
\[
(E_{\phi})_{A,B}=\sum_{i=0}^{n}\phi(i)(E_{i})_{A,B}=\phi(n)(E_{n})_{A,B}=\frac{\phi(n)}{q^{n^{2}}\lvert\Mcal_{n}\lvert}=\frac{\phi(n)}{q^{n^{2}}(q-1)\lvert\SL(\mathbb{F},n)\rvert},
\]
where the first three steps are immediate from Definition~\prettyref{def:E-phi},
and the last step uses~\prettyref{eq:coset-size}. In particular,
\begin{equation}
\Phi_{A,B}=\frac{\phi(n)}{q^{n^{2}}(q-1)\lvert\SL(\mathbb{F},n)\rvert}+\sum_{b\in\mathbb{F}\setminus\{0,a\}}\frac{\phi(n)}{q-1}\,((G_{a})_{A,B}-(G_{b})_{A,B}).\label{eq:PhiAB}
\end{equation}
If $\det(A+B)=a,$ then by definition $(G_{a})_{A,B}=q^{-n^{2}}\lvert\SL(\mathbb{F},n)\rvert^{-1}$
and $(G_{b})_{A,B}=0$ for all $b\in\mathbb{F}\setminus\{0,a\},$
so that~\prettyref{eq:PhiAB} gives
\[
\Phi_{A,B}=\frac{\phi(n)}{q^{n^{2}}(q-1)\lvert\SL(\mathbb{F},n)\rvert}+\sum_{b\in\mathbb{F}\setminus\{0,a\}}\frac{\phi(n)}{(q-1)q^{n^{2}}\lvert\SL(\mathbb{F},n)\rvert}=\frac{\phi(n)}{q^{n^{2}}\lvert\SL(\mathbb{F},n)\rvert}.
\]
If, on the other hand, $\det(A+B)=c$ for some $c\in\mathbb{F}\setminus\{0,a\},$
then by definition $(G_{a})_{A,B}=0$ and likewise $(G_{b})_{A,B}=0$
for every $b\ne c,$ so that \prettyref{eq:PhiAB} simplifies to
\[
\Phi_{A,B}=\frac{\phi(n)}{q^{n^{2}}(q-1)\lvert\SL(\mathbb{F},n)\rvert}-\frac{\phi(n)}{(q-1)}(G_{c})_{A,B}=0.\qedhere
\]
\end{proof}
We proceed to establish key analytic and metric properties of $\Phi.$
To begin with,
\begin{align}
\|\Phi\| & \leq\|E_{\phi}\|+\sum_{b\in\mathbb{F}\setminus\{0,a\}}\frac{|\phi(n)|}{q-1}\,\|G_{a}-G_{b}\|\nonumber \\
 & \leq\|E_{\phi}\|+\sum_{b\in\mathbb{F}\setminus\{0,a\}}\frac{\|\phi\|_{1}}{q-1}\,\|G_{a}-G_{b}\|\nonumber \\
 & \leq128\|\phi\|_{1}\,q^{-\ell(k-\ell-m+1)/2}\,q^{-n^{2}}+\sum_{b\in\mathbb{F}\setminus\{0,a\}}\frac{\|\phi\|_{1}}{q-1}\cdot8q^{-3(n^{2}-1)/2}\nonumber \\
 & \leq(128q^{-\ell(k-\ell-m+1)/2}+8q^{-(n^{2}-3)/2})q^{-n^{2}}\|\phi\|_{1},\label{eq:Phi-spectral-intermed}
\end{align}
where the first step uses the triangle inequality, and the third step
is a substitution from~\prettyref{eq:E-phi-spectral} and equation~\prettyref{eq:Gab-spectral}
of~\prettyref{lem:DET-problem-witness}. To simplify this bound,
recall from the theorem hypothesis that $n>k\geq1$ and $\ell,m\geq0.$
Therefore, $\ell(k-\ell-m+1)\leq\ell(k-\ell+1)\leq(k+1)^{2}/4\leq n^{2}/4\leq n^{2}-3.$
This results in $q^{-(n^{2}-3)/2}\leq q^{-\ell(k-\ell-m+1)/2}$, and
thus \prettyref{eq:Phi-spectral-intermed} simplifies to
\begin{equation}
\|\Phi\|\leq136q^{-\ell(k-\ell-m+1)/2}q^{-n^{2}}\|\phi\|_{1}.\label{eq:Phi-spectral}
\end{equation}

Next, we examine $\|\Phi\|_{1}.$ We have
\[
\sum_{\rk(A+B)=n}|\Phi_{A,B}|=\sum_{\det(A+B)=a}|\Phi_{A,B}|=\sum_{\det(A+B)=a}\frac{|\phi(n)|}{q^{n^{2}}\lvert\SL(\mathbb{F},n)\rvert}=|\phi(n)|,
\]
where the first and second steps are immediate from Claim~\prettyref{claim:Phi-simplified},
and the last step applies~\prettyref{eq:pairs-with-given-det}. Also,
\[
\sum_{\rk(A+B)<n}|\Phi_{A,B}|=\sum_{\rk(A+B)<n}|(E_{\phi})_{A,B}|=\|E_{\phi}\|-\sum_{\rk(A+B)=n}|(E_{\phi})_{A,B}|=\|\phi\|_{1}-|\phi(n)|,
\]
where the first step uses Claim~\prettyref{claim:Phi-simplified},
and the last step invokes~\prettyref{lem:E_phi-metric}. These two
equations yield
\begin{equation}
\|\Phi\|_{1}=\|\phi\|_{1}.\label{eq:Phi-ell1}
\end{equation}

Continuing,
\begin{align}
\sum_{\dom F}F_{A,B}\Phi_{A,B} & =\sum_{\det(A+B)=a}\Phi_{A,B}-\sum_{\rk(A+B)=k}\Phi_{A,B}\nonumber \\
 & =\sum_{\det(A+B)=a}\frac{\phi(n)}{q^{n^{2}}\lvert\SL(\mathbb{F},n)\rvert}-\sum_{\rk(A+B)=k}(E_{\phi})_{A,B}\nonumber \\
 & =\phi(n)-\phi(k)\nonumber \\
 & =|\phi(n)|+|\phi(k)|\nonumber \\
 & =\|\phi\|_{1}-\sum_{r\notin\{k,n\}}|\phi(r)|,\label{eq:sum-over-domG}
\end{align}
where the second step uses Claim~\prettyref{claim:Phi-simplified},
the third step invokes \prettyref{lem:E_phi-metric} and~\prettyref{eq:pairs-with-given-det},
and the fourth step is valid due to~\prettyref{lem:phi}\prettyref{enu:phi-at-n},~\prettyref{enu:phi-at-k}.
Finally,
\begin{align}
\sum_{\overline{\dom F}}|\Phi_{A,B}| & =\sum_{\rk(A+B)\notin\{n,k\}}|\Phi_{A,B}|+\sum_{\substack{\rk(A+B)=n\\
\det(A+B)\ne a
}
}|\Phi_{A,B}|\nonumber \\
 & =\sum_{\rk(A+B)\notin\{n,k\}}|\Phi_{A,B}|\nonumber \\
 & =\sum_{\rk(A+B)\notin\{n,k\}}|(E_{\phi})_{A,B}|\nonumber \\
 & =\sum_{r\notin\{n,k\}}|\phi(r)|,\label{eq:sum-over-complement-domG}
\end{align}
where the second and third steps use Claim~\prettyref{claim:Phi-simplified},
and the last step uses~\prettyref{lem:E_phi-metric}. Now
\begin{align}
\sum_{\dom F}F_{A,B} & \Phi_{A,B}-\delta\|\Phi\|_{1}-\sum_{\overline{\dom F}}|\Phi_{A,B}|\nonumber \\
 & =\|\phi\|_{1}-\delta\|\phi\|_{1}-2\sum_{r\notin\{n,k\}}|\phi(r)|\nonumber \\
 & \geq\left(1-\delta-2\sum_{r\notin\{n,k\}}|\phi(r)|\right)\|\phi\|_{1},\label{eq:Phi-G-correlation}
\end{align}
where the first step uses \prettyref{eq:Phi-ell1}\textendash \prettyref{eq:sum-over-complement-domG},
and the last step is legitimate by \prettyref{lem:phi}\prettyref{enu:phi-at-n}.

\prettyref{prop:approxtracelower} implies, in view of~\prettyref{eq:Phi-spectral}
and~\prettyref{eq:Phi-G-correlation}, that
\begin{equation}
\|F\|_{\Sigma,\delta}\geq\frac{1}{136}\,\left(1-\delta-2\sum_{r\notin\{n,k\}}|\phi(r)|\right)q^{\ell(k-\ell-m+1)/2}\,q^{n^{2}}.\label{eq:G-trace-norm}
\end{equation}
Since $\sum_{r\notin\{n,k\}}|\phi(r)|\leq32q^{-m-1}$ by~\prettyref{lem:phi}\prettyref{enu:phi-small-tail},
this proves~\prettyref{eq:G-m-l-trace-norm}. The alternative lower
bound~\prettyref{eq:G-m-l-trace-norm-alt} follows by taking $\ell=k$
and $m=0$ in~\prettyref{eq:G-trace-norm} and noting that $\sum_{r\notin\{n,k\}}|\phi(r)|=0$
in this case (by \prettyref{lem:phi}\prettyref{enu:phi-vanish}).
\end{proof}
By virtue of the approximate trace norm method, \prettyref{thm:RANK-vs-DET-trace-norm}
yields the following tight lower bound on the communication complexity
of the rank versus determinant problem.
\begin{thm*}[restatement of \prettyref{thm:RANKDET-intro}]
There is an absolute constant $c>0$ such that for every finite field
$\mathbb{F},$ every field element $a\in\mathbb{F}\setminus\{0\},$
and all integers $n>k\geq0,$ 
\begin{equation}
Q_{\frac{1}{2}-\frac{1}{4|\mathbb{F}|^{k/3}}}^{*}(\RANKDET_{k,a}^{\mathbb{F},n})\geq c(1+k^{2}\log|\mathbb{F}|).\label{eq:master-lower-bound-rankdet}
\end{equation}
\end{thm*}
\begin{proof}
For $k=0,$ the claimed lower bound follows from the fact that $\RANKDET_{0,a}^{\mathbb{F},n}$
is nonconstant and hence has communication complexity at least $1$
bit. For $k\geq1,$ our lower bounds on the approximate trace norm
of $\RANKDET_{k,a}^{\mathbb{F},n}$ are identical to those for $\RANK_{k,n}^{\mathbb{F},n,n}$
(Theorems~\ref{thm:RANK-vs-DET-trace-norm} and~\prettyref{thm:approx-norm-rank-problem},
respectively). Accordingly, the proof here is identical to that of
\prettyref{thm:rank-lower-bound-CANONICAL}, with equations~\prettyref{eq:G-m-l-trace-norm}
and \prettyref{eq:G-m-l-trace-norm-alt} of \prettyref{thm:RANK-vs-DET-trace-norm}
used in place of the corresponding equations \prettyref{eq:F-m-l-trace-norm-rank}
and \prettyref{eq:F-m-l-trace-norm-rank-alt} of \prettyref{thm:approx-norm-rank-problem}.
\end{proof}
As a consequence, we obtain an optimal communication lower bound for
the unrestricted determinant problem.
\begin{thm*}[restatement of \prettyref{thm:DET-intro}]
There is an absolute constant $c>0$ such that for every finite field
$\mathbb{F},$ every pair of distinct elements $a,b\in\mathbb{F},$
and all integers $n\geq2,$
\begin{equation}
Q_{\frac{1}{2}-\frac{1}{4|\mathbb{F}|^{(n-1)/3}}}^{*}(\DET_{a,b}^{\mathbb{F},n})\geq cn^{2}\log|\mathbb{F}|.\label{eq:master-lower-bound-det}
\end{equation}
\end{thm*}
\begin{proof}
If $ab=0,$ then $\DET_{a,b}^{\mathbb{F},n}$ contains as a subproblem
either $\RANKDET_{n-1,b}^{\mathbb{F},n}$ (when $a=0$) or $\neg\RANKDET_{n-1,a}^{\mathbb{F},n}$
(when $b=0$), and therefore \prettyref{eq:master-lower-bound-det}
follows from~\prettyref{thm:RANKDET-intro}. If $a$ and $b$ are
both nonzero, \prettyref{thm:DET-cc-ab-nonzero} gives
\[
Q_{\frac{1}{2}-\frac{1}{4|\mathbb{F}|^{(n-1)/3}}}^{*}(\DET_{a,b}^{\mathbb{F},n})\geq c'n^{2}\log|\mathbb{F}|-\frac{1}{2}\log24
\]
for a small enough constant $c'>0.$ Taking a weighted average of
this lower bound with the trivial lower bound of $1$ bit settles~\eqref{eq:master-lower-bound-det}.
\end{proof}

\section{\label{sec:Subspace-sum-and-intersection}The subspace sum and intersection
problems}

As discussed in the introduction, our analysis of the subspace sum
and subspace intersection problems has similarities with the rank
problem but also diverges from it in important ways. Instead of additively
composed matrices whose rows and columns are indexed by elements of
$\mathbb{F}_{q}^{n\times n}$, we now have matrices with rows and
columns indexed by subspaces, and each entry $(A,B)$ depends solely
on the dimension of $A\cap B$. While the construction of the univariate
dual object is similar to that for the rank problem, its relation
to the singular values of the dual matrix is significantly more intricate,
and computing the spectral norm of the dual matrix is now a challenge.
Our study of the spectral norm is based on ideas due to Knuth~\cite{knuth01discreet}.
We start in Section~\ref{subsec:Equivalence-subset-sum-intersection}
by formalizing the equivalence of the subspace sum and subspace intersection
problems, which allows us to focus on the latter problem from then
on. As a first step toward solving the subspace intersection problem,
we collect necessary technical results about subspace combinatorics
in Section~\ref{subsec:Counting-subspaces-satisfying}. In Section~\ref{subsec:Subspace-matrices},
we give a formal definition of subspace matrices, state several auxiliary
results, and compare our analysis of their spectrum to that of Knuth.
In Section~\prettyref{subsec:Eigenvalues-and-eigenvectors}, we fully
determine the spectrum of subspace matrices. In Sections~\ref{subsec:Normalized-subspace-matrices}\textendash \ref{subsec:Communication-lower-bounds-subspace},
we use this spectral study along with our techniques developed in
Section~\ref{sec:rank-problem} to prove optimal lower bounds on
the communication complexity of the subspace intersection problem.
Sections~\ref{subsec:Communication-upper-bounds-subspace-small-error}
and~\ref{subsec:Communication-upper-bounds-large-error} conclude
with matching communication upper bounds. As in previous sections,
we let $q$ denote an arbitrary prime power and adopt $\mathbb{F}_{q}$
throughout as the underlying field. 

\subsection{\label{subsec:Equivalence-subset-sum-intersection}Equivalence of
the subspace sum and intersection problems}

The equivalence of the subspace sum and subspace intersection problems
from the standpoint of communication complexity is a straightforward
consequence of the identity~\prettyref{eq:sum-and-intersection-of-S-and-T},
valid for any linear subspaces $S$ and $T$ in a finite-dimensional
vector space. We formalize this equivalence below.
\begin{prop}
\label{prop:SUM-INTERSECT-reductions}Let $n,m,\ell$ be nonnegative
integers with $\max\{m,\ell\}\leq n.$ Then for all integers $d,D$
with $d\ne D,$
\begin{align}
 & \SUM_{d,D}^{\mathbb{F},n,m,\ell}=\INTERSECT_{m+\ell-d,m+\ell-D}^{\mathbb{F},n,m,\ell},\label{eq:sum-partial-intersect-partial}\\
 & \SUM_{d}^{\mathbb{F},n,m,\ell}=\INTERSECT_{m+\ell-d}^{\mathbb{F},n,m,\ell}.\label{eq:intersect-total-sum-total}
\end{align}
\end{prop}

\begin{proof}
Let $S,T\subseteq\mathbb{F}^{n}$ be arbitrary subspaces of dimension
$m$ and $\ell,$ respectively. Since $\dim(S+T)=m+\ell-\dim(S\cap T),$
we have
\begin{align*}
\SUM_{d,D}^{\mathbb{F},n,m,\ell}(S,T) & =\INTERSECT_{m+\ell-d,m+\ell-D}^{\mathbb{F},n,m,\ell}(S,T),
\end{align*}
settling~\prettyref{eq:sum-partial-intersect-partial}. Analogously,
for any subspaces $S,T\subseteq\mathbb{F}^{n}$ of dimension $m$
and $\ell,$ respectively, we have $\dim(S+T)\leq d$ if and only
if $\dim(S\cap T)\geq m+\ell-d,$ which implies~\prettyref{eq:intersect-total-sum-total}.
\end{proof}
We will now prove our main result on the subspace sum problem (stated
in the introduction as Theorems~\ref{thm:MAIN-subspace-sum-constant-error}
and~\ref{thm:MAIN-subspace-sum-general}) assuming our corresponding
result on subspace intersection (\prettyref{thm:MAIN-subspace-intersection-general}).
In the rest of this work, we will focus on proving \prettyref{thm:MAIN-subspace-intersection-general}.
\begin{proof}[Proof of Theorems~\emph{\ref{thm:MAIN-subspace-sum-constant-error}}
and~\emph{\ref{thm:MAIN-subspace-sum-general}} assuming Theorem~\emph{\ref{thm:MAIN-subspace-intersection-general}}.]
Recall that \prettyref{thm:MAIN-subspace-sum-constant-error} is
a special case of \prettyref{thm:MAIN-subspace-sum-general}, corresponding
to $\gamma=1/3$. Therefore, it suffices to prove \prettyref{thm:MAIN-subspace-sum-general}.
Define $r=m+\ell-D$ and $R=m+\ell-d.$ Then the hypotheses $\max\{m,\ell\}\leq d<D\leq\min\{m+\ell,n\}$
and $\gamma\in[\frac{1}{3}q^{-(2d-m-\ell)/5},{\textstyle \frac{1}{3}}]$
of \prettyref{thm:MAIN-subspace-sum-general} can be equivalently
stated as
\begin{align}
 & \max\{0,m+\ell-n\}\leq r<R\leq\min\{m,\ell\},\label{eq:intersect-hyp}\\
 & \gamma\in[{\textstyle \frac{1}{3}}q^{-(m+\ell-2R)/5},{\textstyle \frac{1}{3}}].\label{eq:intersect-gamma-hyp}
\end{align}
Recall from \prettyref{prop:SUM-INTERSECT-reductions} that $\SUM_{d,D}^{\mathbb{F},n,m,\ell}$
is the same function as $\INTERSECT_{R,r}^{\mathbb{F},n,m,\ell}$,
which in turn is the negation of $\INTERSECT_{r,R}^{\mathbb{F},n,m,\ell}$.
Now the bounds for $\SUM_{d,D}^{\mathbb{F},n,m,\ell}$ claimed in
\prettyref{thm:MAIN-subspace-sum-general} follow from the bounds
for $\INTERSECT_{r,R}^{\mathbb{F},n,m,\ell}$ in \prettyref{thm:MAIN-subspace-intersection-general},
upon substituting $R=m+\ell-d.$ This appeal to \prettyref{thm:MAIN-subspace-intersection-general}
is legitimate due to~\prettyref{eq:intersect-hyp} and~\prettyref{eq:intersect-gamma-hyp}. 

Analogously, $\SUM_{d}^{\mathbb{F},n,m,\ell}$ is the same function
as $\INTERSECT_{R}^{\mathbb{F},n,m,\ell}$ (\prettyref{prop:SUM-INTERSECT-reductions}),
and therefore the bounds claimed for $\SUM_{d}^{\mathbb{F},n,m,\ell}$
in \prettyref{thm:MAIN-subspace-sum-general} follow from the bounds
for $\INTERSECT_{R}^{\mathbb{F},n,m,\ell}$ in \prettyref{thm:MAIN-subspace-intersection-general},
upon substituting $R=m+\ell-d.$ 
\end{proof}

\subsection{\label{subsec:Counting-subspaces-satisfying}Counting subspaces satisfying
combinatorial constraints}

When it comes to counting, one could hope that the transition from
subsets to subspaces would be straightforward and amount to replacing
binomial coefficients with their Gaussian counterparts. Unfortunately,
this is not the case. Many basic results for sets have no analogues
in the subspace setting. For example, the well-known inclusion-exclusion
formula \prettyref{eq:sum-and-intersection-of-S-and-T} is valid for
two subspaces but does not generalize to any larger number. As a consequence,
it is in general a subtle task to count the subspaces of a given dimension
that satisfy basic combinatorial constraints relative to other given
subspaces. We start by counting, for given subspaces $A$ and $C$,
all $d$-dimensional subspaces that contain $C$ and avoid $A\setminus C$.
\begin{lem}[Counting subspaces externally]
\label{lem:external-subspaces} Let $A$ and $C$ be linear subspaces
of an $n$-dimensional vector space $V$ over $\mathbb{F}_{q}$. Let
$d\geq0$ be an integer. Then the number of dimension-$d$ linear
subspaces $X$ such that $C\subseteq X\subseteq V$ and $A\cap X=A\cap C$
is
\begin{equation}
q^{(\dim(A)-\dim(A\cap C))(d-\dim(C))}\binom{n-\dim(A+C)}{d-\dim(C)}_{q}.\label{eq:external-subspaces-bound}
\end{equation}
\end{lem}

\begin{proof}
The lemma is trivially true for $d\notin[\dim(C),n]$ since the Gaussian
binomial coefficient in~\eqref{eq:external-subspaces-bound} is zero
in that case. In what follows, we consider the complementary case
$d\in[\dim(C),n].$

Let $\mathscr{X}$ be the set of subspaces $X$ in the statement of
the lemma. Fix a basis $v_{1},v_{2},\ldots,v_{\dim(C)}$ for $C.$
Let $\mathscr{B}$ be the set of all $d$-tuples $(v_{1},\ldots,v_{\dim(C)},u_{1},\ldots,u_{d-\dim(C)})$
of vectors in $V$ such that for all $i,$
\begin{equation}
u_{i}\notin A+C+\Span\{u_{1},u_{2},\ldots,u_{i-1}\}.\label{eq:inductive-prop-basis}
\end{equation}
Then each element of $\mathscr{B}$ is an ordered basis, with
\begin{equation}
|\mathscr{B}|=\prod_{i=1}^{d-\dim(C)}(q^{n}-q^{\dim(A+C)+i-1}).\label{eq:B-bases}
\end{equation}
\begin{claim}
\label{claim:X-bases}Every subspace $X\in\mathscr{X}$ has precisely
\begin{equation}
\prod_{i=1}^{d-\dim(C)}(q^{d}-q^{\dim(C)+i-1})\label{eq:X-bases}
\end{equation}
ordered bases in $\mathscr{B}$.
\end{claim}

\begin{proof}
Let us say that a sequence of vectors $(u_{1},u_{2},\ldots,u_{k})$
in $X$ is \emph{good} if~\eqref{eq:inductive-prop-basis} holds
for all $i=1,2,\ldots,k.$ We will prove that for every good sequence
of $k$ vectors in $X$, where $k<d-\dim(C),$ there are exactly $q^{d}-q^{\dim(C)+k}$
vectors $u_{k+1}\in X$ such that the sequence $(u_{1},u_{2},\ldots,u_{k+1})$
is good. Indeed, letting $S,S',T$ in Fact~\ref{fact:sum-intersection-subspaces}
be the subspaces $X,C+\Span\{u_{1},u_{2},\ldots,u_{k}\},A,$ respectively,
we obtain
\begin{align*}
X\cap(A+C+\Span\{u_{1},u_{2},\ldots,u_{k}\}) & =(X\cap A)+C+\Span\{u_{1},u_{2},\ldots,u_{k}\}\\
 & =(C\cap A)+C+\Span\{u_{1},u_{2},\ldots,u_{k}\}\\
 & =C+\Span\{u_{1},u_{2},\ldots,u_{k}\}.
\end{align*}
Therefore, the only vectors $u_{k+1}\in X$ for which the sequence
$(u_{1},u_{2},\ldots,u_{k+1})$ is not good are the elements of $C+\Span\{u_{1},u_{2},\ldots,u_{k}\},$
which is a subspace of dimension $\dim(C)+k$ because it is spanned
by the linearly independent vectors $v_{1},v_{2},\ldots,v_{\dim(C)},u_{1},u_{2},\ldots,u_{k}.$
In conclusion, out of the $q^{d}$ vectors of $X$, there are precisely
$q^{\dim(C)+k}$ vectors $u_{k+1}$ for which the sequence $u_{1},u_{2},\ldots,u_{k+1}$
is not good.

It now follows immediately that the number of good sequences $(u_{1},u_{2},\ldots,u_{d-\dim(C)})$
of vectors in $X$ is~\eqref{eq:X-bases} as claimed, with $q^{d}-q^{\dim(C)}$
ways to choose $u_{1},$ then $q^{d}-q^{\dim(C)+1}$ ways to choose
$u_{2}$ given $u_{1},$ then $q^{d}-q^{\dim(C)+2}$ ways to choose
$u_{3}$ given $u_{1},u_{2},$ and so on.
\end{proof}
\begin{claim}
\label{claim:B-some-subspace}Every element of $\mathscr{B}$ is an
ordered basis for some subspace in $\mathscr{X}$. 
\end{claim}

\begin{proof}
Fix a tuple $(u_{1},u_{2},\ldots,u_{d-\dim(C)})$ with~\eqref{eq:inductive-prop-basis}
for all $i,$ and let 
\[
X=\Span\{v_{1},\ldots,v_{\dim(C)},u_{1},\ldots,u_{d-\dim(C)}\}.
\]
Then clearly $X$ is a $d$-dimensional subspace with $C\subseteq X\subseteq V.$
This in particular means that $A\cap X$ contains $A\cap C.$ It remains
to prove the opposite inclusion, $A\cap X\subseteq A\cap C.$ For
this, fix arbitrary scalars $\alpha_{i},\beta_{j}$ such that 
\[
\sum\alpha_{i}v_{i}+\sum\beta_{j}u_{j}\in A.
\]
If some $\beta_{j}$ were nonzero, we could take $j^{*}=\max\{j:\beta_{j}\ne0\}$
and obtain $u_{j^{*}}\in\beta_{j^{*}}^{-1}(A-\sum\alpha_{i}v_{i}-\sum_{j<j^{*}}\beta_{j}u_{j}),$
contradicting~\eqref{eq:inductive-prop-basis}. This means that $\beta_{j}=0$
for all $j,$ with the consequence that the vector $\sum\alpha_{i}v_{i}+\sum\beta_{j}u_{j}=\sum\alpha_{i}v_{i}$
belongs to $C.$ This settles the containment $A\cap X\subseteq A\cap C$
and completes the proof.
\end{proof}
Claims~\ref{claim:X-bases} and~\ref{claim:B-some-subspace} imply
that $|\mathscr{X}|$ is the quotient of \eqref{eq:B-bases} by~\eqref{eq:X-bases},
namely,
\begin{align*}
|\mathscr{X}| & =\prod_{i=1}^{d-\dim(C)}\frac{q^{n}-q^{\dim(A+C)+i-1}}{q^{d}-q^{\dim(C)+i-1}}\\
 & =\left(\frac{q^{\dim(A+C)}}{q^{\dim(C)}}\right)^{d-\dim(C)}\binom{n-\dim(A+C)}{d-\dim(C)}_{q}\\
 & =q^{(\dim(A)-\dim(A\cap C))(d-\dim(C))}\binom{n-\dim(A+C)}{d-\dim(C)}_{q}.
\end{align*}
This completes the proof of the lemma.
\end{proof}
\begin{cor}[Counting subspaces internally]
\label{cor:internal-subspaces} Let $S'\subseteq S$ be linear subspaces
in a vector space over $\mathbb{F}_{q}$. Let $d\geq0$ be an integer.
Then the number of dimension-$d$ linear subspaces $T$ with $S'\subseteq T\subseteq S$
is
\begin{equation}
\binom{\dim(S)-\dim(S')}{d-\dim(S')}_{q}.\label{eq:counting-internally}
\end{equation}
\end{cor}

\begin{proof}
Set $V=S,$ $C=S',$ and $A=\{0\}$ in the statement of \prettyref{lem:external-subspaces}
\end{proof}
We now generalize \prettyref{lem:external-subspaces} by allowing
$A\cap X$ to be any subspace of $A$ of a given dimension $t$.
\begin{lem}
\label{lem:external-subspaces-gen}Let $A,B$ be linear subspaces
of an $n$-dimensional vector space $V$ over $\mathbb{F}_{q}$. Define
$r=\dim(A\cap B).$ Let $d$ and $t$ be nonnegative integers. Then
the number of dimension-$d$ linear subspaces $X$ such that $B\subseteq X\subseteq V$
and $\dim(A\cap X)=t$ is 
\begin{equation}
q^{(\dim(A)-t)(d-t-\dim(B)+r)}\binom{n-\dim(A)-\dim(B)+r}{d-t-\dim(B)+r}_{q}\binom{\dim(A)-r}{t-r}_{q}.\label{eq:external-subspaces-bound-1}
\end{equation}
\end{lem}

\begin{proof}
The lemma is trivially true for $t\notin[r,\dim(A)]$ since the last
Gaussian binomial coefficient in~\eqref{eq:external-subspaces-bound-1}
is zero in that case. In what follows, we consider the complementary
case $t\in[r,\dim(A)]$.

Let $\mathscr{X}$ be the set of all dimension-$d$ subspaces $X$
with $B\subseteq X\subseteq V$ and $\dim(A\cap X)=t.$ Let $\mathscr{A}$
be the set of all dimension-$t$ subspaces $A'$ with $A\cap B\subseteq A'\subseteq A.$
By \prettyref{cor:internal-subspaces},
\begin{equation}
|\mathscr{A}|=\binom{\dim(A)-r}{t-r}_{q}.\label{eq:number-Aprimes}
\end{equation}
For any $X\in\mathscr{X},$ the subspace $A\cap X$ is by definition
a dimension-$t$ subspace of $A$ that contains $A\cap B$. This makes
it possible to define a function $f\colon\mathscr{X}\to\mathscr{A}$
by $f(X)=A\cap X.$
\begin{claim}
\label{claim:no-stray-vectors}For every $A'\in\mathscr{A},$
\begin{equation}
|f^{-1}(A')|=q^{(\dim(A)-t)(d-t-\dim(B)+r)}\binom{n-\dim(A)-\dim(B)+r}{d-t-\dim(B)+r}_{q}.\label{eq:extend-by-Aprime}
\end{equation}
\end{claim}

\begin{proof}
Define $C=A'+B.$ Then $A\cap C=A'+A\cap B$ by Fact~\ref{fact:sum-intersection-subspaces},
which in view of $A\cap B\subseteq A'$ further yields 
\begin{equation}
A\cap C=A'.\label{eq:A-C-Aprime}
\end{equation}
 Now
\begin{align}
|f^{-1}(A')| & =|\{X:\text{\ensuremath{X} is a subspace of dimension }d\text{ with }B\subseteq X\subseteq V\text{ and }A\cap X=A'\}\nonumber \\
 & =|\{X:\text{\ensuremath{X} is a subspace of dimension }d\text{ with }C\subseteq X\subseteq V\text{ and }A\cap X=A'\}\nonumber \\
 & =|\{X:\text{\ensuremath{X} is a subspace of dimension }d\text{ with }C\subseteq X\subseteq V\text{ and }A\cap X=A\cap C\}\nonumber \\
 & =q^{(\dim(A)-\dim(A\cap C))(d-\dim(C))}\binom{n-\dim(A+C)}{d-\dim(C)}_{q},\label{eq:intermediate-part}
\end{align}
where the first step is immediate from the definitions of $\mathscr{X}$
and $f$; the second step holds because the condition $B\subseteq X$
is logically equivalent to $A'+B\subseteq X$ due to $A'\subseteq X$;
the third step applies~\eqref{eq:A-C-Aprime}; and the final step
uses Lemma~\ref{lem:external-subspaces}. 

It remains to calculate the dimensions of the relevant subspaces in~\eqref{eq:intermediate-part}.
We have $\dim(C)=\dim(A'+B)=\dim(A')+\dim(B)-\dim(A'\cap B),$ which
along with $A'\cap B=A\cap B$ and $\dim(A')=t$ yields
\begin{equation}
\dim(C)=t+\dim(B)-r.\label{eq:plugin1}
\end{equation}
It is immediate from~\eqref{eq:A-C-Aprime} that
\begin{equation}
\dim(A\cap C)=t.\label{eq:plugin2}
\end{equation}
Finally, we have $\dim(A+C)=\dim(A+A'+B)=\dim(A+B)$ and therefore
\begin{equation}
\dim(A+C)=\dim(A)+\dim(B)-r.\label{eq:plugin3}
\end{equation}
Substituting~\eqref{eq:plugin1}\textendash \eqref{eq:plugin3} into~\eqref{eq:intermediate-part},
we arrive at the sought equality~\eqref{eq:extend-by-Aprime}.
\end{proof}
Claim~\ref{claim:no-stray-vectors} implies that $|\mathscr{X}|$
is the product of the right-hand side of~\eqref{eq:number-Aprimes}
and the right-hand side of~\eqref{eq:extend-by-Aprime}, as was to
be shown.
\end{proof}
\begin{cor}
\label{cor:internal-subspaces-gen}Let $S'\subseteq S$ be linear
subspaces in a vector space over $\mathbb{F}_{q}$. Let $d$ and $t$
be nonnegative integers. Then the number of dimension-$d$ linear
subspaces $T\subseteq S$ with $\dim(S'\cap T)=t$ is
\begin{equation}
q^{(\dim(S')-t)(d-t)}\binom{\dim(S)-\dim(S')}{d-t}_{q}\binom{\dim(S')}{t}_{q}.\label{eq:counting-internally-gen}
\end{equation}
\end{cor}

\begin{proof}
Invoke \prettyref{lem:external-subspaces-gen} with $V=S,$ $A=S',$
and $B=\{0\}.$
\end{proof}

\subsection{\label{subsec:Subspace-matrices}Subspace matrices}

In \cite{knuth01discreet}, Knuth defined \emph{combinatorial matrices
of type $(n,t)$} as matrices whose rows and columns are indexed by
$t$-element subsets of a fixed $n$-element set, and whose $(A,B)$
entry depends only on $|A\cap B|$. We begin with analogous definitions
in the setting of linear subspaces. Let $\mathbb{F}$ be a given finite
field. For each $d=0,1,2,\ldots,n,$ fix an ordering on the set of
dimension-$d$ subspaces of $\mathbb{F}^{n}.$
\begin{defn}
\label{def:Jr}Let $n,m,\ell$ be nonnegative integers with $\max\{m,\ell\}\leq n.$
For any $r\geq0,$ define $J_{r}^{\mathbb{F},n,m,\ell}$ to be the
matrix whose rows are indexed by dimension-$m$ subspaces of $\mathbb{F}^{n},$
columns indexed by dimension-$\ell$ subspaces of $\mathbb{F}^{n},$
and entries given by
\[
(J_{r}^{\mathbb{F},n,m,\ell})_{A,B}=\begin{cases}
1 & \text{if }\dim(A\cap B)=r,\\
0 & \text{otherwise},
\end{cases}
\]
where the row index $A$ and column index $B$ use the ordering on
the subspaces of $\mathbb{F}^{n}$ fixed at the beginning. 
\end{defn}

Thus, the $(A,B)$ entry of $J_{r}^{\mathbb{F},n,m,\ell}$ depends
only on the dimension of $A\cap B$ rather than the subspaces $A,B$
themselves. By passing to the linear span of all such matrices for
fixed $\mathbb{F},n,m,\ell$, we obtain a matrix family that we call
\emph{subspace matrices}.
\begin{defn}[Subspace matrices]
\label{def:J-phi}For a function $\phi\colon\mathbb{Z}\to\Re,$ we
define
\begin{equation}
J_{\phi}^{\mathbb{F},n,m,\ell}=\sum_{r=0}^{\min\{m,\ell\}}\phi(r)J_{r}^{\mathbb{F},n,m,\ell}.\label{eq:J-phi-def}
\end{equation}
\end{defn}

\noindent Recall that throughout this manuscript, the underlying field
is $\mathbb{F}=\mathbb{F}_{q}$ for an arbitrary prime power $q.$
To avoid notational clutter, we will write simply $J_{r}^{n,m,\ell}$
and $J_{\phi}^{n,m,\ell}$ to mean $J_{r}^{\mathbb{F},n,m,\ell}$
and $J_{\phi}^{\mathbb{F},n,m,\ell}$, respectively. 

To determine the eigenvalues of combinatorial matrices with rows and
columns indexed by $t$-element subsets of $\{1,2,\ldots,n\}$, Knuth
investigates a certain homogeneous system of linear equations with
variables indexed by $s$-element subsets and the equations themselves
corresponding to $(s-1)$-element subsets. He refers to the solutions
to such systems as \emph{$(n,s)$-kernel systems}. It turns out that
the linear space of kernel systems has a basis supported on variables
labeled by a certain type of sets, which Knuth calls\emph{ basic sets}
and which he fully describes in a combinatorial way. For any $s\in\{1,2,\ldots,t\}$
and any $(n,s)$-kernel system $(x_{u})$, he shows that the corresponding
vector $(z_{w})$, indexed by $t$-element subsets $w$ and given
by $z_{w}=\sum_{u\subseteq w}x_{u}$, is an eigenvector for any combinatorial
matrix of type $(n,t)$. These vectors $(z_{w})$ for various values
of $s$, together with the all-ones vector, make up a complete set
of eigenvectors, and Knuth's analysis also reveals the associated
eigenvalues.

Even setting aside the more subtle combinatorial nature of subspaces
described in Section~\prettyref{subsec:Counting-subspaces-satisfying},
it is not clear how to generalize Knuth's notion of basic sets to
linear subspaces. For this reason, we do not appeal to combinatorial
machinery and rely instead on linear-algebraic arguments. As another
point of departure, our problem requires understanding the singular
values of a general subspace matrix $J_{\phi}^{n,m,\ell},$ whereas
Knuth studied combinatorial matrices that are symmetric (analogous
to the symmetric subspace matrices $J_{\phi}^{n,m,m}$ in our setting).
We note that the eigenvalues of symmetric subspace matrices $J_{\phi}^{n,m,m}$
were also determined by Delsarte~\cite{delsarte76schemes} and~Eisfeld~\cite{eisfeld99eigenspaces},
and their properties were studied in \cite{BCIM17eigenvalues-graphs,cioaba-gupta21grassmann}.
However, these previous analyses do not seem to apply to the general
case of interest to us, namely, that of subspace matrices $J_{\phi}^{n,m,\ell}$
for arbitrary $m,\ell$.

We start by studying the subspace matrices $J_{k}^{n,m,k}$, which
play a particularly important role in our analysis. The following
lemma investigates their rank.
\begin{lem}
\label{lem:J-rank}Let $n,m,k$ be nonnegative integers with $m\geq k\geq0$
and $n\geq m+k$. Then 
\begin{equation}
\rk J_{k}^{n,m,k}=\binom{n}{k}_{q}.\label{eq:rk-bound-J-mkk}
\end{equation}
\end{lem}

\begin{proof}
In the degenerate case $n=0,$ the matrix $J_{k}^{n,m,k}=J_{0}^{0,0,0}=\begin{bmatrix}1\end{bmatrix}$
clearly has rank $\binom{0}{0}_{q}=1.$ In what follows, we treat
the case $n\geq1.$ Here, we will exhibit reals $z_{0},z_{1},\ldots,z_{k}$
such that for all $k$-dimensional subspaces $A,B\subseteq\mathbb{F}_{q}^{n}$,
\begin{equation}
\sum_{X\subseteq\mathbb{F}_{q}^{n}:\;\dim X=m}z_{\dim(A\cap X)}(J_{k}^{n,m,k})_{X,B}=\delta_{\dim(A\cap B),k}.\label{eq:Jz}
\end{equation}
Put differently, this means that every vector of the standard basis
$e_{1},e_{2},\ldots$ can be obtained as a linear combination of the
rows of $J_{k}^{n,m,k},$ immediately implying~\eqref{eq:rk-bound-J-mkk}.
Rewriting~\eqref{eq:Jz},
\begin{equation}
\sum_{i=0}^{k}z_{i}\sum_{\substack{X\subseteq\mathbb{F}_{q}^{n}:\;\dim X=m,\\
\dim(A\cap X)=i
}
}(J_{k}^{n,m,k})_{X,B}=\delta_{\dim(A\cap B),k}\qquad\qquad\forall A,B.\label{eq:Jz-1}
\end{equation}
The inner summation equals the number of $m$-dimensional subspaces
$X$ with $B\subseteq X\subseteq\mathbb{F}_{q}^{n}$ and $\dim(A\cap X)=i.$
Applying Lemma~\ref{lem:external-subspaces-gen}, we find that~\eqref{eq:Jz-1}
is equivalent to
\begin{equation}
\sum_{i=0}^{k}z_{i}q^{(k-i)(m-i-k+r)}\binom{n-2k+r}{m-i-k+r}_{q}\binom{k-r}{i-r}_{q}=\delta_{r,k},\qquad\qquad r=0,1,\ldots,k,\label{eq:Jz-2}
\end{equation}
where $r$ corresponds to $\dim(A\cap B)$ in~\eqref{eq:Jz-1}. Write~\eqref{eq:Jz-2}
in matrix form as 
\begin{equation}
Mz=\begin{bmatrix}0 & 0 & \cdots & 0 & 1\end{bmatrix}\tr,\label{eq:Mz}
\end{equation}
where $M=[M_{r,i}]$ is the real matrix of order $k+1$ given by
\[
M_{r,i}=q^{(k-i)(m-i-k+r)}\binom{n-2k+r}{m-i-k+r}_{q}\binom{k-r}{i-r}_{q}
\]
for $r,i\in\{0,1,\ldots,k\}.$ All entries of $M$ below the diagonal
are zero because $\binom{k-r}{i-r}_{q}=0$ for $r>i.$ The diagonal
entries, on the other hand, are
\[
M_{r,r}=q^{(k-r)(m-k)}\binom{n-2k+r}{m-k}_{q},
\]
which is nonzero because $n-2k+r\geq m-k$ by the hypothesis that
$n\geq m+k$. This makes $M$ an upper triangular matrix with nonzero
entries on the diagonal. Then $M$ is invertible, and a solution $z$
to~\eqref{eq:Mz} is guaranteed to exist. 
\end{proof}
We will recover the $k$-th eigenspace of $J_{\phi}^{n,m,m}$ as the
image of $\ker J_{k-1}^{n,k-1,k}$ under the linear map $J_{k}^{n,m,k}$.
The first step is to understand how the map $J_{i}^{n,m,k}$ acts
on $\ker J_{k-1}^{n,k-1,k}$ for different values of $i$.
\begin{lem}
\label{lem:increase-overlap}Let $n\geq m\geq k$ be positive integers.
Then for all $i=0,1,\ldots,k-1$ and $x\in\ker J_{k-1}^{n,k-1,k},$
\begin{equation}
J_{i}^{n,m,k}x=-q^{k-i-1}\cdot\frac{q^{i+1}-1}{q^{k-i}-1}J_{i+1}^{n,m,k}x.\label{eq:Jii-plus-1}
\end{equation}
\end{lem}

\begin{proof}
Consider the matrix product $M=J_{i}^{n,m,k-1}J_{k-1}^{n,k-1,k}$.
Let us compute the generic entry $M_{A,B},$ where $A,B\subseteq\mathbb{F}_{q}^{n}$
are subspaces of dimension $m$ and $k$, respectively. By definition,
$M_{A,B}$ is the number of $(k-1)$-dimensional subspaces $X\subseteq B$
such that $\dim(A\cap X)=i$. Invoking Corollary~\ref{cor:internal-subspaces-gen}
with $S=B$ and $S'=A\cap B$, we obtain
\[
(J_{i}^{n,m,k-1}J_{k-1}^{n,k-1,k})_{A,B}=q^{(r-i)(k-1-i)}\binom{k-r}{k-1-i}_{q}\binom{r}{i}_{q},
\]
where $r=\dim(A\cap B).$ Rewriting this equation in matrix form,
\[
J_{i}^{n,m,k-1}J_{k-1}^{n,k-1,k}=\sum_{r=0}^{k}q^{(r-i)(k-1-i)}\binom{k-r}{k-1-i}_{q}\binom{r}{i}_{q}J_{r}^{n,m,k}.
\]
In this equation, the product of the $q$-binomial coefficients vanishes
whenever $r>i+1$ or $r<i.$ Therefore, the above summation contains
only two nonzero terms, namely,
\[
J_{i}^{n,m,k-1}J_{k-1}^{n,k-1,k}=\sum_{r\in\{i,i+1\}}q^{(r-i)(k-1-i)}\binom{k-r}{k-1-i}_{q}\binom{r}{i}_{q}J_{r}^{n,m,k}.
\]
Simplifying, 
\[
J_{i}^{n,m,k-1}J_{k-1}^{n,k-1,k}=\frac{q^{k-i}-1}{q-1}J_{i}^{n,m,k}+q^{k-1-i}\cdot\frac{q^{i+1}-1}{q-1}J_{i+1}^{n,m,k}.
\]
Applying this matrix equation to a vector $x\in\ker J_{k-1}^{n,k-1,k}$
gives
\[
0=\frac{q^{k-i}-1}{q-1}J_{i}^{n,m,k}x+q^{k-1-i}\cdot\frac{q^{i+1}-1}{q-1}J_{i+1}^{n,m,k}x,
\]
which directly implies~\eqref{eq:Jii-plus-1}.
\end{proof}
\begin{cor}
\label{cor:increase-overlap}Let $n\geq m\geq k$ be positive integers.
Then for all $r=0,1,\ldots,k$ and $x\in\ker J_{k-1}^{n,k-1,k},$
\begin{equation}
J_{r}^{n,m,k}x=(-1)^{k-r}q^{\binom{k-r}{2}}\binom{k}{r}_{q}J_{k}^{n,m,k}x.\label{eq:base-case-r-to-m}
\end{equation}
\end{cor}

\begin{proof}
The proof is by induction on $k-r$ for fixed integers $n,m,k.$ For
the base case $r=k,$ the equality in~\eqref{eq:base-case-r-to-m}
is trivial. For the inductive step with $k-r>0,$ we have
\begin{align*}
J_{r}^{n,m,k}x & =-q^{k-r-1}\cdot\frac{q^{r+1}-1}{q^{k-r}-1}J_{r+1}^{n,m,k}x\\
 & =-q^{k-r-1}\cdot\frac{q^{r+1}-1}{q^{k-r}-1}\cdot(-1)^{k-r-1}q^{\binom{k-r-1}{2}}\binom{k}{r+1}_{q}J_{k}^{n,m,k}x\\
 & =(-1)^{k-r}q^{\binom{k-r}{2}}\binom{k}{r}_{q}J_{k}^{n,m,k}x,
\end{align*}
where the first step uses Lemma~\ref{lem:increase-overlap}, and
the second step applies the inductive hypothesis.
\end{proof}
Let $A,B\subseteq\mathbb{F}_{q}^{n}$ be arbitrary subspaces of dimension
$m$ and $\ell,$ respectively. Recall from \prettyref{fact:Aperp-intersect-Bperp}
that for fixed $n,m,\ell$, the dimension of $A\cap B$ is uniquely
determined by the dimension of $A^{\perp}\cap B^{\perp}$. This makes
the subspace matrix $J_{\phi}^{n,m,\ell}$ identical, up to a permutation
of the rows and columns, to the subspace matrix $J_{\phi'}^{n,n-m,n-\ell}$
for an appropriate function $\phi'$. We record this fact as our next
lemma. Its role in our work will be to simplify the calculation of
the singular values of $J_{\phi}^{n,m.\ell}$ and the eigenvalues
of $J_{\phi}^{n,m,m}$ by reducing the general case to the case $m+\ell\leq n$
and $m\leq n/2$, respectively.
\begin{lem}
\label{lem:J-orthog-complement}Let $n,m,\ell$ be nonnegative integers
with $\max\{m,\ell\}\leq n.$ Let $\phi\colon\mathbb{Z}\to\Re$ be
given. Then:
\end{lem}

\begin{enumerate}
\item \label{enu:m-ell-reorder}$J_{\phi}^{n,m,\ell}=PJ_{\phi'}^{n,n-m,n-\ell}Q$,
where $P,Q$ are permutation matrices and $\phi'\colon\mathbb{Z}\to\Re$
is defined by $\phi'(t)=\phi(t+m+\ell-n);$
\item \label{enu:m-m-reorder}$J_{\phi}^{n,m,m}=PJ_{\phi''}^{n,n-m,n-m}P^{-1},$
where $P$ is a permutation matrix and $\phi''\colon\mathbb{Z}\to\Re$
is defined by $\phi''(t)=\phi(t+2m-n).$
\end{enumerate}
\begin{proof}
Recall that for any $d\in\{0,1,\ldots,n\},$ the map $S\mapsto S^{\perp}$
is a bijection between the subspaces of $\mathbb{F}_{q}^{n}$ of dimension
$d$ and those of dimension $n-d.$ For subspaces $A,B\subseteq\mathbb{F}_{q}^{n}$
of dimension $m$ and $\ell,$ respectively, we have
\begin{align*}
(J_{\phi}^{n,m,\ell})_{A,B} & =\phi(\dim(A\cap B))\\
 & =\phi(\dim(A^{\perp}\cap B^{\perp})+m+\ell-n)\\
 & =\phi'(\dim(A^{\perp}\cap B^{\perp}))\\
 & =(J_{\phi'}^{n,n-m,n-\ell})_{A^{\perp},B^{\perp}},
\end{align*}
where the second step uses Fact~\ref{fact:Aperp-intersect-Bperp}.
Rewriting this conclusion in matrix form,
\[
J_{\phi}^{n,m,\ell}=[(J_{\phi'}^{n,n-m,n-\ell})_{A^{\perp},B^{\perp}}]_{A,B},
\]
where $A,B$ range over all subspaces of dimension $m$ and $\ell,$
respectively. The matrix on the right-hand side is clearly $J_{\phi'}^{n,n-m,n-\ell},$
up to a reordering of the rows and columns. This settles~\ref{enu:m-ell-reorder}.

An argument analogous to the above yields 
\[
J_{\phi}^{n,m,m}=[(J_{\phi''}^{n,n-m,n-m})_{A^{\perp},B^{\perp}}]_{A,B},
\]
where $A,B$ range over all subspaces of dimension $m.$ The matrix
on the right-hand side is the result of permuting the rows and columns
of $J_{\phi''}^{n,n-m,n-m}$ according to the same permutation, which
is another way of phrasing~\ref{enu:m-m-reorder}.
\end{proof}

\subsection{Eigenvalues and eigenvectors of subspace matrices\label{subsec:Eigenvalues-and-eigenvectors}}

Our description of the spectrum of each $J_{\phi}^{n,m,\ell}$ is
in terms of a function which we now define.
\begin{defn}
\label{def:Lambda}For nonnegative integers $n,m,\ell,r,k$ with $\max\{m,\ell\}\leq n$
and $k\leq\min\{m,\ell\},$ define
\[
\Lambda_{r}^{n,m,\ell}(k)=\sum_{i=0}^{k}(-1)^{i}\binom{k}{i}_{q}q^{\binom{i}{2}+(m-r)(\ell-r-i)}\binom{n-m-i}{\ell-r-i}_{q}\binom{m-k+i}{r-k+i}_{q}.
\]
More generally, for any $\phi\colon\mathbb{Z}\to\Re,$ define
\[
\Lambda_{\phi}^{n,m,\ell}(k)=\sum_{r=0}^{\min\{m,\ell\}}\phi(r)\Lambda_{r}^{n,m,\ell}(k).
\]
\end{defn}

As part of our analysis of the eigenvalues of $J_{\phi}^{n,m,m},$
we will determine its eigenspaces and show that they are pairwise
orthogonal. The orthogonality will follow from the pairwise distinctness
of the corresponding eigenvalues, with the following lemma playing
a crucial role.
\begin{lem}
\label{lem:lambdas-distinct}Let $n,m$ be nonnegative integers with
$m\leq n/2$. Then the numbers $\Lambda_{0}^{n,m,m}(k)$ for $k=0,1,\ldots,m$
are pairwise distinct.
\end{lem}

\begin{proof}
For $r=0$, the $q$-binomial coefficient $\binom{m-k+i}{r-k+i}_{q}$
in Definition~\ref{def:Lambda} vanishes unless $i=k$. As a result,
\begin{align*}
\Lambda_{0}^{n,m,m}(k) & =(-1)^{k}q^{\binom{k}{2}+m(m-k)}\binom{n-m-k}{m-k}_{q}\\
 & =(-1)^{k}q^{\binom{k}{2}+m(m-k)}\binom{n-m-k}{n-2m}_{q}.
\end{align*}
For $k\in\{0,1,\ldots,m\},$ the $q$-binomial coefficient in the
last expression is clearly positive and a nonincreasing function of
$k$, whereas the exponent of $q$ is a strictly decreasing function
of $k.$ It follows that the numbers $|\Lambda_{0}^{n,m,m}(k)|$ for
$k=0,1,\ldots,m$ form a strictly decreasing sequence. 
\end{proof}
As in \cite{knuth01discreet}, we treat the all-ones eigenvector separately.
\begin{prop}
\label{prop:J-normalization}Let $n,m,\ell,r$ be nonnegative integers
with $\max\{m,\ell\}\leq n$. Then 
\begin{align}
J_{r}^{n,m,\ell}\,\1 & =q^{(m-r)(\ell-r)}\binom{n-m}{\ell-r}_{q}\binom{m}{r}_{q}\1\label{eq:1-eigenvector}\\
 & =\Lambda_{r}^{n,m,\ell}(0)\,\1,\label{eq:eq:1-eigenvector-Lambda0}\\
\|J_{r}^{n,m,\ell}\|_{1} & =q^{(m-r)(\ell-r)}\binom{n-m}{\ell-r}_{q}\binom{m}{r}_{q}\binom{n}{m}_{q}.\label{eq:Jnml-r-ell1-norm}
\end{align}
More generally, for $\phi\colon\mathbb{Z}\to\Re,$
\begin{align}
J_{\phi}^{n,m,\ell}\,\1 & =\Lambda_{\phi}^{n,m,\ell}(0)\,\1,\label{eq:1-eigenvector-phi}\\
\|J_{\phi}^{n,m,\ell}\|_{1} & =\sum_{r=0}^{\min\{m,\ell\}}|\phi(r)|\,q^{(m-r)(\ell-r)}\binom{n-m}{\ell-r}_{q}\binom{m}{r}_{q}\binom{n}{m}_{q}.\label{eq:Jnml-phi-ell1-norm}
\end{align}
\end{prop}

\begin{proof}
Let $A\subseteq\mathbb{F}_{q}^{n}$ be a subspace of dimension $m.$
By definition, $(J_{r}^{n,m,\ell}\,\1)_{A}$ is the number of $\ell$-dimensional
subspaces $X\subseteq\mathbb{F}_{q}^{n}$ with $\dim(A\cap X)=r.$
Taking $S'=A$ and $S=\mathbb{F}_{q}^{n}$ in Corollary~\ref{cor:internal-subspaces-gen},
we obtain
\[
(J_{r}^{n,m,\ell}\,\1)_{A}=q^{(m-r)(\ell-r)}\binom{n-m}{\ell-r}_{q}\binom{m}{r}_{q}.
\]
This settles~\eqref{eq:1-eigenvector}, which in turn implies~\eqref{eq:eq:1-eigenvector-Lambda0}
by Definition~\ref{def:Lambda}. Since there are exactly $\binom{n}{m}_{q}$
subspaces $A\subseteq\mathbb{F}_{q}^{n}$ of dimension $m,$ equation~\eqref{eq:Jnml-r-ell1-norm}
is immediate from~\eqref{eq:1-eigenvector}. Equation~\eqref{eq:1-eigenvector-phi}
follows by linearity from~\eqref{eq:J-phi-def} and~\eqref{eq:eq:1-eigenvector-Lambda0}.
Analogously, \eqref{eq:Jnml-phi-ell1-norm} follows from~\eqref{eq:J-phi-def}
and~\eqref{eq:Jnml-r-ell1-norm} since the matrices $J_{r}^{n,m,\ell}$
for $r\geq0$ have disjoint support.
\end{proof}
The following lemma is the cornerstone of our analysis of the spectrum
of subspace matrices $J_{\phi}^{n,m,\ell}$. It generalizes Knuth's
work from sets to subspaces ($m=\ell$) and further to the asymmetric
case of interest to us ($m\ne\ell$).
\begin{lem}
\label{lem:eigenvectors}Let $n,m,\ell$ be positive integers with
$n\geq m+\ell.$ Let $k\in\{1,2,\ldots,\min\{m,\ell\}\}$ and $x\in\ker J_{k-1}^{n,k-1,k}$
be given. Then for all integers $t\geq0,$ 
\begin{equation}
J_{t}^{n,m,\ell}J_{k}^{n,\ell,k}x=\Lambda_{t}^{n,m,\ell}(k)J_{k}^{n,m,k}x.\label{eq:eigenvector-t}
\end{equation}
More generally, for all $\phi\colon\mathbb{Z}\to\Re,$
\begin{equation}
J_{\phi}^{n,m,\ell}J_{k}^{n,\ell,k}x=\Lambda_{\phi}^{n,m,\ell}(k)J_{k}^{n,m,k}x.\label{eq:eigenvector-phi}
\end{equation}
\end{lem}

\begin{proof}
Fix an arbitrary integer $t\geq0$ and define $M=J_{t}^{n,m,\ell}J_{k}^{n,\ell,k}.$
Let us compute the generic entry $M_{A,B}$, where $A,B$ are subspaces
of $\mathbb{F}_{q}^{n}$ of dimension $m$ and $k$, respectively.
By definition, $M_{A,B}$ is the number of $\ell$-dimensional subspaces
$X$ such that $\dim(A\cap X)=t$ and $B\subseteq X\subseteq\mathbb{F}_{q}^{n}.$
Lemma~\ref{lem:external-subspaces-gen} implies that 
\[
(J_{t}^{n,m,\ell}J_{k}^{n,\ell,k})_{A,B}=q^{(m-t)(\ell-t-k+r)}\binom{n-m-k+r}{\ell-t-k+r}_{q}\binom{m-r}{t-r}_{q},
\]
where $r=\dim(A\cap B).$ Rewriting this equation in matrix form,
we obtain
\[
J_{t}^{n,m,\ell}J_{k}^{n,\ell,k}=\sum_{r=0}^{k}q^{(m-t)(\ell-t-k+r)}\binom{n-m-k+r}{\ell-t-k+r}_{q}\binom{m-r}{t-r}_{q}J_{r}^{n,m,k}.
\]
Applying this matrix identity to a vector $x\in\ker J_{k-1}^{n,k-1,k}$,
we find
\begin{align*}
J_{t}^{n,m,\ell}J_{k}^{n,\ell,k}x & =\sum_{r=0}^{k}q^{(m-t)(\ell-t-k+r)}\binom{n-m-k+r}{\ell-t-k+r}_{q}\binom{m-r}{t-r}_{q}J_{r}^{n,m,k}x\\
 & =\sum_{r=0}^{k}q^{(m-t)(\ell-t-k+r)}\binom{n-m-k+r}{\ell-t-k+r}_{q}\binom{m-r}{t-r}_{q}\cdot(-1)^{k-r}q^{\binom{k-r}{2}}\binom{k}{r}_{q}J_{k}^{n,m,k}x\\
 & =\sum_{r=0}^{k}q^{(m-t)(\ell-t-k+r)}\binom{n-m-k+r}{\ell-t-k+r}_{q}\binom{m-r}{t-r}_{q}\cdot(-1)^{k-r}q^{\binom{k-r}{2}}\binom{k}{k-r}_{q}J_{k}^{n,m,k}x\\
 & =\sum_{i=0}^{k}q^{(m-t)(\ell-t-i)}\binom{n-m-i}{\ell-t-i}_{q}\binom{m-k+i}{t-k+i}_{q}\cdot(-1)^{i}q^{\binom{i}{2}}\binom{k}{i}_{q}J_{k}^{n,m,k}x\\
 & =\Lambda_{t}^{n,m,\ell}(k)\,J_{k}^{n,m,k}x,
\end{align*}
where the second step uses Corollary~\ref{cor:increase-overlap},
the fourth step is a change of variable, and the last step is immediate
by Definition~\ref{def:Lambda}. This settles~\eqref{eq:eigenvector-t}.
Now for any $\phi\colon\mathbb{Z}\to\Re$,
\[
J_{\phi}^{n,m,\ell}J_{k}^{n,\ell,k}x=\sum_{t=0}^{\min\{m,\ell\}}\phi(t)J_{t}^{n,m,\ell}J_{k}^{n,\ell,k}x=\sum_{t=0}^{\min\{m,\ell\}}\phi(t)\Lambda_{t}^{n,m,\ell}(k)J_{k}^{n,m,k}x=\Lambda_{\phi}^{n,m,\ell}(k)J_{k}^{n,m,k}x,
\]
where the first step uses~\eqref{eq:J-phi-def}, the second step
applies~\eqref{eq:eigenvector-t}, and the last step is valid by
Definition~\ref{def:Lambda}.
\end{proof}
We are now in a position to describe the eigenvalues of every symmetric
subspace matrix.
\begin{thm}[Eigenvalues of $J_{\phi}^{n,m,m}$]
\label{thm:J-eigenvalues-sparse} Let $n\geq m\geq0$ be given integers.
\begin{enumerate}
\item \label{enu:eigen-m-small}If $m\leq n/2,$ then the eigenvalues of
$J_{\phi}^{n,m,m}$ for a given function $\phi\colon\mathbb{Z}\to\Re$
are $\Lambda_{\phi}^{n,m,m}(k)$ for $k=0,1,\ldots,m,$ with corresponding
multiplicities $\binom{n}{k}_{q}-\binom{n}{k-1}_{q}$ for $k=0,1,\ldots,m.$
\item \label{enu:eigen-m-large}If $m\geq n/2,$ then the eigenvalues of
$J_{\phi}^{n,m,m}$ for a given function $\phi\colon\mathbb{Z}\to\Re$
are $\Lambda_{\psi}^{n,n-m,n-m}(k)$ for $k=0,1,\ldots,n-m,$ with
corresponding multiplicities $\binom{n}{k}_{q}-\binom{n}{k-1}_{q}$
for $k=0,1,\ldots,n-m,$ with $\psi\colon\mathbb{Z}\to\Re$ given
by $\psi(t)=\phi(t+2m-n).$
\end{enumerate}
\end{thm}

\begin{proof}
We first show that~\ref{enu:eigen-m-small} implies~\ref{enu:eigen-m-large}.
Recall from Lemma~\ref{lem:J-orthog-complement} that $J_{\phi}^{n,m,m}$
is permutation-similar to $J_{\psi}^{n,n-m,n-m}$ with $\psi\colon\mathbb{Z}\to\Re$
given by $\psi(t)=\phi(t+2m-n).$ The eigenvalues of $J_{\psi}^{n,n-m,n-m}$
are, by part~\ref{enu:eigen-m-small} of this theorem, $\Lambda_{\psi}^{n,n-m,n-m}(k)$
for $k=0,1,\ldots,n-m,$ with corresponding multiplicities $\binom{n}{k}_{q}-\binom{n}{k-1}_{q}$
for $k=0,1,\ldots,n-m$. It follows that these are also the eigenvalues
of $J_{\phi}^{n,m,m}$ because a similarity transformation preserves
the eigenvalues and their multiplicities. This settles~\ref{enu:eigen-m-large}.

It remains to prove~\ref{enu:eigen-m-small}, where by hypothesis
\begin{equation}
m\leq\frac{n}{2}.\label{eq:m-n-half}
\end{equation}
Define subspaces $S_{0},S_{1},\ldots,S_{m}$ of the $\binom{n}{m}_{q}$-dimensional
real vector space by 
\begin{align*}
S_{0} & =\Span\{\1\},\\
S_{k} & =\{J_{k}^{n,m,k}x:x\in\ker J_{k-1}^{n,k-1,k}\}, &  & k=1,2,\ldots,m.
\end{align*}
\begin{claim}
\label{claim:dimensions}Let $k\in\{0,1,\ldots,m\}$. Then $\dim S_{k}=\binom{n}{k}_{q}-\binom{n}{k-1}_{q}.$
\end{claim}

\begin{proof}
We need only consider $k\geq1,$ the claim being trivial otherwise.
Observe from~\eqref{eq:m-n-half} and Lemma~\ref{lem:J-rank} that
$J_{k}^{n,m,k}$ has rank $\binom{n}{k}_{q}.$ Put another way, its
columns are linearly independent. Since $S_{k}$ is the image of $\ker J_{k-1}^{n,k-1,k}$
under $J_{k}^{n,m,k},$ it follows that
\begin{equation}
\dim S_{k}=\dim\ker J_{k-1}^{n,k-1,k}.\label{eq:dim-S_k-claim}
\end{equation}
Another appeal to~\eqref{eq:m-n-half} and Lemma~\ref{lem:J-rank}
reveals that the columns of $J_{k-1}^{n,k,k-1}$ are linearly independent.
This makes $J_{k-1}^{n,k-1,k}=(J_{k-1}^{n,k,k-1})^{\tr}$ a matrix
of order $\binom{n}{k-1}_{q}\times\binom{n}{k}_{q}$ with linearly
independent rows, whence $\dim\ker J_{k-1}^{n,k-1,k}=\binom{n}{k}_{q}-\binom{n}{k-1}_{q}$.
In view of~\eqref{eq:dim-S_k-claim}, the proof is complete.
\end{proof}
\begin{claim}
\label{claim:eigens}Let $k\in\{0,1,\ldots,m\}$. Then every vector
of $S_{k}$ is an eigenvector of $J_{\phi}^{n,m,m}$ with eigenvalue
$\Lambda_{\phi}^{n,m,m}(k).$
\end{claim}

\begin{proof}
For $k=0$, the claim is immediate from~\eqref{eq:1-eigenvector-phi}
of Proposition~\ref{prop:J-normalization}. Consider now the complementary
case $k\in\{1,2,\ldots,m\}.$ Here, $n$ and $m$ are positive integers.
Invoking Lemma~\ref{lem:eigenvectors} with $\ell=m$ and \eqref{eq:m-n-half}
yields $J_{\phi}^{n,m,m}v=\Lambda_{\phi}^{n,m,m}(k)v$ for all $v\in S_{k},$
as desired. 
\end{proof}
\begin{claim}
\label{claim:ortho}For any $k,k'\in\{0,1,\ldots,m\}$ with $k\ne k',$
the subspaces $S_{k}$ and $S_{k'}$ are orthogonal.
\end{claim}

\begin{proof}
Taking $\phi=\1_{\{0\}}$ in Claim~\ref{claim:eigens} shows that
$S_{0},S_{1},\ldots,S_{m}$ are eigenspaces of the symmetric matrix
$J_{0}^{n,m,m}$ with eigenvalues $\Lambda_{0}^{n,m,m}(0),\Lambda_{0}^{n,m,m}(1),\ldots,\Lambda_{0}^{n,m,m}(m),$
respectively. But these $m+1$ numbers are pairwise distinct by \eqref{eq:m-n-half}
and~Lemma~\ref{lem:lambdas-distinct}. It now follows from Fact~\ref{fact:diff-eigens-orthog}
that $S_{0},S_{1},\ldots,S_{m}$ are pairwise orthogonal. 
\end{proof}
As we just established with Claim~\ref{claim:ortho}, the subspaces
$S_{0},S_{1},\ldots,S_{m}$ are pairwise orthogonal. Since they are
subspaces over the \emph{reals,} we infer that $\dim(\sum_{k=0}^{m}S_{k})=\sum_{k=0}^{m}\dim S_{k}.$
Using $\dim S_{k}=\binom{n}{k}_{q}-\binom{n}{k-1}_{q}$ from Claim~\ref{claim:dimensions},
we arrive at 
\[
\dim\left(\sum_{k=0}^{m}S_{k}\right)=\sum_{k=0}^{m}\left(\binom{n}{k}_{q}-\binom{n}{k-1}_{q}\right)=\binom{n}{m}_{q}-\binom{n}{-1}_{q}=\binom{n}{m}_{q}.
\]
In other words, a basis for the vector space in question can be obtained
by concatenating bases for $S_{0},S_{1},\ldots,S_{m}.$ Lastly, recall
from Claim~\ref{claim:eigens} that $S_{k}$ (for $k=0,1,\ldots,m$)
is an eigenspace of $J_{\phi}^{n,m,m}$ with eigenvalue $\Lambda_{\phi}^{n,m,m}(k).$
This settles~\ref{enu:eigen-m-small} and completes the proof of
the theorem.
\end{proof}
At last, we adapt the previous proof to the asymmetric case ($m\ne\ell$)
and determine the singular values of every subspace matrix $J_{\phi}^{n,m,\ell}$.
\begin{thm}[Singular values of $J_{\phi}^{n,m,\ell}$]
\label{thm:J-singular-values} Let $n,m,\ell$ be nonnegative integers
with $\max\{m,\ell\}\leq n$. 
\begin{enumerate}
\item \label{enu:singular-m-ell-small}If $m+\ell\leq n,$ then the singular
values of $J_{\phi}^{n,m,\ell}$ for a given function $\phi\colon\mathbb{Z}\to\Re$
are 
\begin{align*}
\sqrt{\Lambda_{\phi}^{n,m,\ell}(k)\,\Lambda_{\phi}^{n,\ell,m}(k)}, &  & k=0,1,\ldots,\min\{m,\ell\},
\end{align*}
with corresponding multiplicities $\binom{n}{k}_{q}-\binom{n}{k-1}_{q}$
for $k=0,1,\ldots,\min\{m,\ell\}.$
\item \label{enu:singular-m-ell-large}If $m+\ell\geq n,$ then the singular
values of $J_{\phi}^{n,m,\ell}$ for a given function $\phi\colon\mathbb{Z}\to\Re$
are 
\begin{align*}
\sqrt{\Lambda_{\psi}^{n,n-m,n-\ell}(k)\,\Lambda_{\psi}^{n,n-\ell,n-m}(k)}, &  & k=0,1,\ldots,\min\{n-m,n-\ell\},
\end{align*}
with corresponding multiplicities $\binom{n}{k}_{q}-\binom{n}{k-1}_{q}$
for $k=0,1,\ldots,\min\{n-m,n-\ell\},$ where $\psi\colon\mathbb{Z}\to\Re$
is given by $\psi(t)=\phi(t+m+\ell-n).$
\end{enumerate}
\end{thm}

\begin{proof}
We first show that~\ref{enu:singular-m-ell-small} implies~\ref{enu:singular-m-ell-large}.
Recall from Lemma~\ref{lem:J-orthog-complement} that the matrix
$J_{\phi}^{n,m,\ell}$ is the same, up to a reordering of the rows
and columns, as $J_{\psi}^{n,n-m,n-\ell}$ with $\psi\colon\mathbb{Z}\to\Re$
given by $\psi(t)=\phi(t+m+\ell-n).$ The singular values of $J_{\psi}^{n,n-m,n-\ell}$
are, by part~\ref{enu:eigen-m-small} of this theorem, $\sqrt{\Lambda_{\psi}^{n,n-m,n-\ell}(k)\,\Lambda_{\psi}^{n,n-\ell,n-m}(k)}$
for $k=0,1,\ldots,\min\{n-m,n-\ell\},$ with corresponding multiplicities
$\binom{n}{k}_{q}-\binom{n}{k-1}_{q}$ for $k=0,1,\ldots,\min\{n-m,n-\ell\}$.
It follows that these are also the singular values of $J_{\phi}^{n,m,\ell}$
because reordering the columns or rows does not change the singular
values or their multiplicities. This establishes~\ref{enu:singular-m-ell-large}.

It remains to settle~\ref{enu:singular-m-ell-small}, where by hypothesis
\begin{equation}
m+\ell\leq n.\label{eq:m+ell-leq-n}
\end{equation}
We may further assume that 
\begin{equation}
m\leq\ell,\label{eq:m-leq-ell}
\end{equation}
for otherwise we can work with the transposed matrix $(J_{\phi}^{n,m,\ell})^{\tr}=J_{\phi}^{n,\ell,m}$,
the singular values being invariant under matrix transposition. By~\eqref{eq:m+ell-leq-n},
\eqref{eq:m-leq-ell}, and Fact~\ref{fact:q-binomial-monotone},
\begin{equation}
\binom{n}{m}_{q}\leq\binom{n}{\ell}_{q}.\label{eq:choose-m-leq-choose-el}
\end{equation}
Another consequence of \eqref{eq:m+ell-leq-n} and~\eqref{eq:m-leq-ell}
is that $m\leq n/2,$ which makes it possible to define subspaces
$S_{0},S_{1},\ldots,S_{m}$ of the $\binom{n}{m}_{q}$-dimensional
real vector space as in the proof of part~\ref{enu:eigen-m-small}
of Theorem~\ref{thm:J-eigenvalues-sparse}. In particular, Claims~\ref{claim:dimensions}\textendash \ref{claim:ortho}
apply as before. 
\begin{claim}
\label{claim:singulars}Let $k\in\{0,1,\ldots,m\}$. Then every vector
of $S_{k}$ is an eigenvector of $J_{\phi}^{n,m,\ell}J_{\phi}^{n,\ell,m}$
with eigenvalue $\Lambda_{\phi}^{n,m,\ell}(k)\Lambda_{\phi}^{n,\ell,m}(k).$
\end{claim}

\begin{proof}
For $k=0$, a double application of Proposition~\ref{prop:J-normalization}
yields $J_{\phi}^{n,m,\ell}J_{\phi}^{n,\ell,m}\,\1=J_{\phi}^{n,m,\ell}\Lambda_{\phi}^{n,\ell,m}(0)\,\1=\Lambda_{\phi}^{n,m,\ell}(0)\Lambda_{\phi}^{n,\ell,m}(0)\,\1,$
as desired. Consider now $k\in\{1,2,\ldots,m\}.$ In this case, due
to~\eqref{eq:m-leq-ell}, the integers $n,m,\ell$ are positive and
satisfy $\min\{m,\ell\}=m.$ Then for any $x\in\ker J_{k-1}^{n,k-1,k},$
\begin{align*}
(J_{\phi}^{n,m,\ell}J_{\phi}^{n,\ell,m})J_{k}^{n,m,k}x & =J_{\phi}^{n,m,\ell}(J_{\phi}^{n,\ell,m}J_{k}^{n,m,k}x)\\
 & =J_{\phi}^{n,m,\ell}(\Lambda_{\phi}^{n,\ell,m}(k)J_{k}^{n,\ell,k}x)\\
 & =\Lambda_{\phi}^{n,\ell,m}(k)J_{\phi}^{n,m,\ell}J_{k}^{n,\ell,k}x\\
 & =\Lambda_{\phi}^{n,\ell,m}(k)\Lambda_{\phi}^{n,m,\ell}(k)J_{k}^{n,m,k}x,
\end{align*}
where the second and fourth steps apply Lemma~\ref{lem:eigenvectors}
with~\eqref{eq:m+ell-leq-n} (note that the roles of $m$ and $\ell$
are reversed in the first application). We have shown that for each
$x\in\ker J_{k-1}^{n,k-1,k},$ its image under $J_{k}^{n,m,k}$ is
an eigenvector of $J_{\phi}^{n,m,\ell}J_{\phi}^{n,\ell,m}$ with eigenvalue
$\Lambda_{\phi}^{n,\ell,m}(k)\Lambda_{\phi}^{n,m,\ell}(k).$ Since
$S_{k}$ is by definition the image of $\ker J_{k-1}^{n,k-1,k}$ under
$J_{k}^{n,m,k},$ the claim is proved.
\end{proof}
Recall from Claims~\ref{claim:dimensions} and~\ref{claim:ortho}
that the subspaces $S_{0},S_{1},\ldots,S_{m}$ are pairwise orthogonal,
with $\dim S_{k}=\binom{n}{k}_{q}-\binom{n}{k-1}_{q}$. As in the
proof of Theorem~\ref{thm:J-eigenvalues-sparse}, this implies that
the real vector space in question is a direct sum of $S_{0},S_{1},\ldots,S_{m}.$
In view of Claim~\ref{claim:singulars}, we conclude that the eigenvalues
of $J_{\phi}^{n,m,\ell}J_{\phi}^{n,\ell,m}$ are $\Lambda_{\phi}^{n,m,\ell}(k)\Lambda_{\phi}^{n,\ell,m}(k)$
for $k=0,1,\ldots,m,$ with corresponding multiplicities $\binom{n}{k}_{q}-\binom{n}{k-1}_{q}$
for $k=0,1,\ldots,m.$ This completes the proof since the singular
values of $J_{\phi}^{n,m,\ell}$ are, by~\eqref{eq:choose-m-leq-choose-el}
and Fact~\ref{fact:singular-values-eigenvalues}, the square roots
of the eigenvalues of $J_{\phi}^{n,m,\ell}(J_{\phi}^{n,m,\ell})^{\tr}=J_{\phi}^{n,m,\ell}J_{\phi}^{n,\ell,m},$
counting multiplicities.
\end{proof}

\subsection{Normalized subspace matrices\label{subsec:Normalized-subspace-matrices}}

To study the communication complexity of the subspace intersection
problem, we now define normalized versions of subspace matrices.
\begin{defn}
\label{def:J-bar}Let $n,m,\ell$ be nonnegative integers with $m+\ell\leq n,$
and let $\mathbb{F}$ be a finite field. Define
\begin{align}
\overline{J}_{r}^{\mathbb{F},n,m,\ell} & =\frac{1}{\|J_{r}^{\mathbb{F},n,m,\ell}\|_{1}}\cdot J_{r}^{\mathbb{F},n,m,\ell}, &  & r=0,1,2,\ldots,\min\{m,\ell\}.\label{eq:Jr-bar}
\end{align}
 For any function $\phi\colon\mathbb{Z}\to\Re,$ define
\begin{equation}
\overline{J}_{\phi}^{\mathbb{F},n,m,\ell}=\sum_{r=0}^{\min\{m,\ell\}}\phi(r)\overline{J}_{r}^{\mathbb{F},n,m,\ell}.\label{eq:J-phi-bar}
\end{equation}
\end{defn}

The requirement $m+\ell\leq n$ in Definition~\prettyref{def:J-bar}
serves to ensure that $J_{r}^{\mathbb{F},n,m,\ell}\ne0$ for each
$r=0,1,\ldots,\min\{m,\ell\},$ so that the normalization in~\prettyref{eq:Jr-bar}
is meaningful. As elsewhere in this manuscript, we will work with
the generic field $\mathbb{F}=\mathbb{F}_{q}$ and will henceforth
write $\overline{J}_{r}^{n,m,\ell}$ and $\overline{J}_{\phi}^{n,m,\ell}$
instead of $\overline{J}_{r}^{\mathbb{F},n,m,\ell}$ and $\overline{J}_{\phi}^{\mathbb{F},n,m,\ell}$. 

The following lemma relates the metric properties of a normalized
subspace matrix $\overline{J}_{\phi}^{n,m,\ell}$ to the corresponding
univariate function $\phi$.
\begin{lem}[Metric properties of $\overline{J}_{\phi}^{n,m,\ell}$]
\label{lem:J-bar-metric} Let $n,m,\ell$ be nonnegative integers
with $m+\ell\leq n.$ Let $\phi\colon\mathbb{Z}\to\Re$ be a given
function. Then:
\begin{align}
\sum_{\substack{S\in\Scal(\mathbb{F}_{q}^{n},m),\\
T\in\Scal(\mathbb{F}_{q}^{n},\ell):\\
\dim(S\cap T)=r
}
}(\overline{J}_{\phi}^{n,m,\ell})_{S,T} & =\phi(r), &  & r=0,1,\dots,\min\{m,\ell\}.\label{eq:J-bar-sum-dim-r}
\end{align}
Moreover, 
\begin{equation}
\|\overline{J}_{\phi}^{n,m,\ell}\|_{1}=\sum_{r=0}^{\min\{m,\ell\}}|\phi(r)|.\label{eq:J-bar-ell1}
\end{equation}
\end{lem}

\begin{proof}
We have
\begin{align}
\sum_{S,T:\dim(S\cap T)=r}(\overline{J}_{\phi}^{n,m,\ell})_{S,T} & =\sum_{S,T:\dim(S\cap T)=r}\;\sum_{i=0}^{\min\{m,\ell\}}\frac{\phi(i)}{\|J_{i}^{n,m,\ell}\|_{1}}\cdot(J_{i}^{n,m,\ell})_{S,T}\nonumber \\
 & =\sum_{S,T:\dim(S\cap T)=r}\frac{\phi(r)}{\|J_{r}^{n,m,\ell}\|_{1}}\cdot(J_{r}^{n,m,\ell})_{S,T}\nonumber \\
 & =\frac{\phi(r)}{\|J_{r}^{n,m,\ell}\|_{1}}\sum_{S,T:\dim(S\cap T)=r}(J_{r}^{n,m,\ell})_{S,T}\nonumber \\
 & =\phi(r),\label{eq:A+B-rank-sum-1}
\end{align}
where the second step uses $(J_{i}^{n,m,\ell})_{S,T}=0$ for $i\ne r,$
and the final step is valid because $(J_{r}^{n,m,\ell})_{S,T}$ equals
$1$ if $\dim(S\cap T)=r$ and $0$ otherwise. This proves~\prettyref{eq:J-bar-sum-dim-r}.
An analogous argument yields
\[
\sum_{S,T:\dim(S\cap T)=r}|(\overline{J}_{\phi}^{n,m,\ell})_{S,T}|=|\phi(r)|.
\]
Summing this equation over $r$ gives \prettyref{eq:J-bar-ell1}. 
\end{proof}
To describe the singular values of a normalized subspace matrix $\overline{J}_{\phi}^{n,m,\ell}$,
we introduce a normalized counterpart of the function~$\Lambda$
from Definition~\prettyref{def:Lambda}.
\begin{defn}
\label{def:Lambda-bar}Let $n,m,\ell,k$ be nonnegative integers with
$m+\ell\leq n$ and $k\leq\min\{m,\ell\}$. Define
\begin{align*}
\overline{\Lambda}_{r}^{\,n,m,\ell}(k) & =\frac{1}{\|J_{r}^{n,m,\ell}\|_{1}}\Lambda_{r}^{n,m,\ell}(k), &  & r=0,1,\ldots,\min\{m,\ell\}.
\end{align*}
More generally, for any $\phi\colon\mathbb{Z}\to\Re,$ define
\begin{align*}
\overline{\Lambda}_{\phi}^{\,n,m,\ell}(k) & =\sum_{r=0}^{\min\{m,\ell\}}\phi(r)\overline{\Lambda}_{r}^{\,n,m,\ell}(k).
\end{align*}

\noindent With this notation, we obtain the following counterpart
of \prettyref{thm:J-singular-values} for normalized subspace matrices.
\end{defn}

\begin{thm}
\label{thm:J-bar-singular-values}Let $n,m,\ell$ be nonnegative integers
with $m+\ell\leq n$. Let $\phi\colon\mathbb{Z}\to\Re$ be given.
Then the singular values of $\overline{J}_{\phi}^{n,m,\ell}$ are:
$\sqrt{\overline{\Lambda}_{\phi}^{\,n,m,\ell}(k)\,\overline{\Lambda}_{\phi}^{\,n,\ell,m}(k)}$
with multiplicity $\binom{n}{k}_{q}-\binom{n}{k-1}_{q},$ where $k=0,1,\ldots,\min\{m,\ell\}.$
\end{thm}

\begin{proof}
By definition, $\overline{J}_{\phi}^{n,m,\ell}=J_{\phi'}^{n,m,\ell}$
with $\phi'\colon\mathbb{Z}\to\Re$ given by 
\[
\phi'(r)=\begin{cases}
\phi(r)/\|J_{r}^{n,m,\ell}\|_{1} & \text{if }r\in\{0,1,\ldots,\min\{m,\ell\}\},\\
0 & \text{otherwise.}
\end{cases}
\]
Recall from \prettyref{thm:J-singular-values} that the singular values
of $J_{\phi'}^{n,m,\ell}$ are $\sqrt{\Lambda_{\phi'}^{n,m,\ell}(k)\,\Lambda_{\phi'}^{n,\ell,m}(k)}$
with multiplicity $\binom{n}{k}_{q}-\binom{n}{k-1}_{q},$ where $k=0,1,\ldots,\min\{m,\ell\}.$
Since $\Lambda_{\phi'}^{n,m,\ell}(k)=\overline{\Lambda}_{\phi}^{\,n,m,\ell}(k)$
and $\Lambda_{\phi'}^{n,\ell,m}(k)=\overline{\Lambda}_{\phi}^{\,n,\ell,m}(k)$
by Definition~\prettyref{def:Lambda-bar}, the proof is complete.
\end{proof}
With our next lemma, we establish key algebraic and analytic properties
of $\overline{\Lambda}_{r}^{\,n,m,\ell}(k)$. 
\begin{lem}
\label{lem:Lambda-r-bar}Let $n,m,\ell,k$ be nonnegative integers
with $m+\ell\leq n$ and $k\leq\min\{m,\ell\}.$ Then$:$
\begin{enumerate}
\item \label{enu:Lambda-bar-poly}for $n,m,\ell,k$ fixed, $\overline{\Lambda}_{r}^{\,n,m,\ell}(k)$
as a function of $r\in\{0,1,\ldots,\min\{m,\ell\}\}$ is a polynomial
in $q^{r}$ of degree at most $k;$
\item \label{enu:Lambda-bar-bound}$|\overline{\Lambda}_{r}^{\,n,m,\ell}(k)|\leq8\binom{n}{m}_{q}^{-1}q^{-k(m-r)/2}$
for $r=0,1,\ldots,\min\{m,\ell\}$.
\end{enumerate}
\end{lem}

\begin{proof}
Let $r\in\{0,1,\ldots,\min\{m,\ell\}\}$ be given. Then
\begin{align}
\overline{\Lambda}_{r}^{\,n,m,\ell}(k) & =\frac{1}{\|J_{r}^{n,m,\ell}\|_{1}}\Lambda_{r}^{n,m,\ell}(k)\nonumber \\
 & =\frac{1}{\|J_{r}^{n,m,\ell}\|_{1}}\sum_{i=0}^{k}(-1)^{i}\binom{k}{i}_{q}q^{\binom{i}{2}+(m-r)(\ell-r-i)}\binom{n-m-i}{\ell-r-i}_{q}\binom{m-k+i}{r-k+i}_{q}\nonumber \\
 & =\binom{n}{m}_{q}^{-1}\sum_{i=0}^{k}(-1)^{i}\binom{k}{i}_{q}\frac{q^{\binom{i}{2}+(m-r)(\ell-r-i)}}{q^{(m-r)(\ell-r)}}\cdot\frac{\binom{n-m-i}{\ell-r-i}_{q}}{\binom{n-m}{\ell-r}_{q}}\cdot\frac{\binom{m-k+i}{r-k+i}_{q}}{\binom{m}{r}_{q}},\label{eq:Lambda-bar-begin}
\end{align}
where the first step restates Definition~\prettyref{def:Lambda-bar},
the second step applies Definition~\prettyref{def:Lambda}, and the
final step uses~\prettyref{prop:J-normalization}. To simplify~\prettyref{eq:Lambda-bar-begin},
observe that
\begin{equation}
\frac{\binom{n-m-i}{\ell-r-i}_{q}}{\binom{n-m}{\ell-r}_{q}}=\frac{q^{\ell-r}-1}{q^{n-m}-1}\cdot\frac{q^{\ell-r}-q}{q^{n-m}-q}\cdot\cdots\cdot\frac{q^{\ell-r}-q^{i-1}}{q^{n-m}-q^{i-1}}.\label{eq:gaussian-ratio1}
\end{equation}
Indeed, if $\ell-r-i<0,$ then the left-hand side is zero by definition,
and the right-hand side also evaluates to zero. In the complementary
case $\ell-r-i\geq0,$ one obtains~\prettyref{eq:gaussian-ratio1}
directly from the definition of Gaussian binomial coefficients. One
analogously verifies that
\begin{equation}
\frac{\binom{m-k+i}{r-k+i}_{q}}{\binom{m}{r}_{q}}=\frac{q^{r}-1}{q^{m}-1}\cdot\frac{q^{r}-q}{q^{m}-q}\cdot\cdots\cdot\frac{q^{r}-q^{k-i-1}}{q^{m}-q^{k-i-1}},\label{eq:gaussian-ratio2}
\end{equation}
by considering the cases $r-k+i<0$ and $r-k+i\geq0.$ Substituting~\prettyref{eq:gaussian-ratio1}
and~\prettyref{eq:gaussian-ratio2} into~\prettyref{eq:Lambda-bar-begin}
gives
\begin{align}
\overline{\Lambda}_{r}^{\,n,m,\ell}(k) & =\binom{n}{m}_{q}^{-1}\sum_{i=0}^{k}(-1)^{i}\binom{k}{i}_{q}q^{\binom{i}{2}-(m-r)i}\cdot\frac{q^{\ell-r}-1}{q^{n-m}-1}\cdot\frac{q^{\ell-r}-q}{q^{n-m}-q}\cdot\cdots\cdot\frac{q^{\ell-r}-q^{i-1}}{q^{n-m}-q^{i-1}}\nonumber \\
 & \qquad\qquad\qquad\qquad\qquad\qquad\qquad\times\frac{q^{r}-1}{q^{m}-1}\cdot\frac{q^{r}-q}{q^{m}-q}\cdot\cdots\cdot\frac{q^{r}-q^{k-i-1}}{q^{m}-q^{k-i-1}}.\label{eq:Lambda-bar-simplified}
\end{align}

To verify the algebraic property~\prettyref{enu:Lambda-bar-poly},
rewrite~\prettyref{eq:Lambda-bar-simplified} to obtain
\begin{multline*}
\overline{\Lambda}_{r}^{\,n,m,\ell}(k)=\binom{n}{m}_{q}^{-1}\sum_{i=0}^{k}(-1)^{i}\binom{k}{i}_{q}q^{\binom{i}{2}-mi}\cdot\frac{q^{\ell}-q^{r}}{q^{n-m}-1}\cdot\frac{q^{\ell}-q^{r+1}}{q^{n-m}-q}\cdot\cdots\cdot\frac{q^{\ell}-q^{r+i-1}}{q^{n-m}-q^{i-1}}\\
\times\frac{q^{r}-1}{q^{m}-1}\cdot\frac{q^{r}-q}{q^{m}-q}\cdot\cdots\cdot\frac{q^{r}-q^{k-i-1}}{q^{m}-q^{k-i-1}}.\qquad
\end{multline*}
The $i$-th summand in this expression is, for fixed values of $n,m,\ell,k,$
clearly a polynomial in $q^{r}$ of degree at most $i+(k-i)=k.$ This
settles~\prettyref{enu:Lambda-bar-poly}.

We now turn to the analytic property,~\prettyref{enu:Lambda-bar-bound}.
Dropping the zero terms from the summation in~\prettyref{eq:Lambda-bar-simplified},
and applying the triangle inequality,
\begin{align*}
|\overline{\Lambda}_{r}^{\,n,m,\ell}(k)| & \leq\binom{n}{m}_{q}^{-1}\sum_{i=\max\{0,k-r\}}^{\min\{k,\ell-r\}}\binom{k}{i}_{q}q^{\binom{i}{2}-(m-r)i}\cdot\frac{q^{\ell-r}-1}{q^{n-m}-1}\cdot\frac{q^{\ell-r}-q}{q^{n-m}-q}\cdot\cdots\cdot\frac{q^{\ell-r}-q^{i-1}}{q^{n-m}-q^{i-1}}\\
 & \qquad\qquad\qquad\qquad\qquad\qquad\qquad\times\frac{q^{r}-1}{q^{m}-1}\cdot\frac{q^{r}-q}{q^{m}-q}\cdot\cdots\cdot\frac{q^{r}-q^{k-i-1}}{q^{m}-q^{k-i-1}}.
\end{align*}
The first $i$ fractions on the right-hand side are each bounded by
$q^{\ell-r}/q^{n-m},$ whereas the other $k-i$ fractions are each
bounded by $q^{r}/q^{m}$. Using these estimates leads to
\begin{align}
|\overline{\Lambda}_{r}^{\,n,m,\ell}(k)| & \leq\binom{n}{m}_{q}^{-1}\sum_{i=\max\{0,k-r\}}^{\min\{k,\ell-r\}}\binom{k}{i}_{q}q^{\binom{i}{2}-(m-r)i}\cdot q^{(\ell-r-n+m)i}\cdot q^{-(m-r)(k-i)}\nonumber \\
 & \leq\binom{n}{m}_{q}^{-1}\sum_{i=\max\{0,k-r\}}^{\min\{k,\ell-r\}}4q^{i(k-i)+\binom{i}{2}-(m-r)i+(\ell-r-n+m)i-(m-r)(k-i)}\nonumber \\
 & \leq\binom{n}{m}_{q}^{-1}\sum_{i=\max\{0,k-r\}}^{\min\{k,\ell-r\}}4q^{i(k-i)+\binom{i}{2}-(m-r)i-ri-(m-r)(k-i)}\nonumber \\
 & =\binom{n}{m}_{q}^{-1}\sum_{i=\max\{0,k-r\}}^{\min\{k,\ell-r\}}4q^{i(k-r-i)+\binom{i}{2}-k(m-r)},\label{eq:Lambda-bound-closer}
\end{align}
where the second step uses~\prettyref{cor:qbinbound}, and the third
step is valid since $n\geq m+\ell$ by hypothesis. Let $A(i)$ denote
the exponent of $q$ in~\prettyref{eq:Lambda-bound-closer}. Then
$A(i)$ is an integer-valued function of $i$ that strictly decreases
on $[k-r,\infty)$. As a result, \prettyref{eq:Lambda-bound-closer}
yields
\begin{align*}
|\overline{\Lambda}_{r}^{\,n,m,\ell}(k)| & \leq\binom{n}{m}_{q}^{-1}\sum_{t=0}^{\infty}4q^{A(\max\{0,k-r\})-t}\\
 & \leq8\binom{n}{m}_{q}^{-1}q^{A(\max\{0,k-r\})},
\end{align*}
where the second step uses a geometric series along with $q\geq2.$
Therefore, the proof of~\prettyref{enu:Lambda-bar-bound} will be
complete once we show that
\begin{equation}
A(\max\{0,k-r\})\leq-\frac{k(m-r)}{2}.\label{eq:A-max}
\end{equation}
There are two cases to consider. If $k\leq r,$ then $A(\max\{0,k-r\})=A(0)=-k(m-r)\leq-k(m-r)/2,$
where the last step uses the hypothesis that $r\leq m.$ If $k>r,$
then $A(\max\{0,k-r\})=A(k-r)=\binom{k-r}{2}-k(m-r)\leq k(k-r)/2-k(m-r)\leq-k(m-r)/2$,
where the last step uses the hypothesis that $k\leq m.$ This settles~\prettyref{eq:A-max}
and completes the proof of the lemma.
\end{proof}

\subsection{Approximate trace norm of the subspace problem\label{subsec:Approximate-trace-norm-subspace-problem}}

We have reached a pivotal point in our study of the subspace intersection
problem, where we analyze the approximate trace norm of its characteristic
matrix. As in our analysis of the rank problem (\prettyref{sec:rank-problem}),
we start by constructing a suitable univariate dual object.
\begin{lem}
\label{lem:phi-subspace}Let $\Delta,R,d_{1},d_{2}$ be nonnegative
integers with $0<R\leq\Delta-d_{1}-d_{2}$. Then there is a function
$\psi\colon\{0,1,\ldots,\Delta\}\to\Re$ such that:
\end{lem}

\begin{enumerate}
\item $\psi(0)=-1;$\label{enu:phi-at-0-subspace}
\item $\psi(R)>0;$\label{enu:phi-at-R-subspace}
\item $\psi(r)=0$ for $r\in\{0,1,\ldots,\Delta\}\setminus(\{R+d_{1}+1,R+d_{1}+2,\ldots,\Delta-d_{2}\}\cup\{0,R\});$\label{enu:phi-vanish-subspace}
\item \label{enu:phi-orthog-subspace}$\sum_{r=0}^{\Delta}\psi(r)\xi(q^{r})=0$
for every polynomial $\xi$ of degree at most $\Delta-R-d_{1}-d_{2};$
\item $\sum_{r\in\{0,\ldots,\Delta\}\setminus\{0,R\}}|\psi(r)|\leq32q^{-d_{1}-1}$.\label{enu:phi-small-tail-subspace}
\end{enumerate}
\begin{proof}
By hypothesis, $d_{1}+d_{2}\leq\Delta-R<\Delta.$ As a result, we
may invoke \prettyref{lem:phi} with parameters $n,k,\ell,m$ set
to $\Delta,\Delta-R,d_{2,}d_{1}$, respectively, to obtain a function
$\phi\colon\{0,1,\ldots,\Delta\}\to\Re$ such that:
\begin{enumerate}[label=(\roman*$'$)]
\item $\phi(\Delta)=1;$\label{enu:_psi-Delta}
\item $\phi(\Delta-R)<0;$\label{enu:_psi-Delta-R}
\item $\phi(r)=0$ for $r\in\{0,1,\ldots,\Delta\}\setminus(\{d_{2},d_{2}+1,\ldots,\Delta-R-d_{1}-1\}\cup\{\Delta-R,\Delta\});$
\label{enu:_psi-vanish}
\item $\sum_{r=0}^{\Delta}\phi(r)\xi(q^{-r})=0$ for every polynomial $\xi$
of degree at most $\Delta-R-d_{1}-d_{2};$\label{enu:_psi-orthog}
\item $\sum_{r\in\{0,\ldots,\Delta\}\setminus\{\Delta-R,\Delta\}}|\phi(r)|\leq32q^{-d_{1}-1}$.
\label{enu:_psi-tail}
\end{enumerate}
Define $\psi\colon\{0,1,\ldots,\Delta\}\to\Re$ by $\psi(r)=-\phi(\Delta-r).$
Then \prettyref{enu:phi-at-0-subspace}, \prettyref{enu:phi-at-R-subspace},
\prettyref{enu:phi-vanish-subspace}, and \prettyref{enu:phi-small-tail-subspace}
are immediate from \prettyref{enu:_psi-Delta}, \prettyref{enu:_psi-Delta-R},
\prettyref{enu:_psi-vanish}, and \prettyref{enu:_psi-tail}, respectively.
The remaining item~\prettyref{enu:phi-orthog-subspace} follows from~\prettyref{enu:_psi-orthog}
via
\[
\sum_{r=0}^{\Delta}\psi(r)\xi(q^{r})=-\sum_{r=0}^{\Delta}\phi(\Delta-r)\xi(q^{r-\Delta}\cdot q^{\Delta})=-\sum_{i=0}^{\Delta}\phi(i)\xi(q^{-i}\cdot q^{\Delta})=0.\qedhere
\]

\end{proof}
With the univariate dual object $\psi$ now constructed, we will use
the associated subspace matrix $\Phi=\overline{J}_{\psi}^{\,n,m,\ell}$
as a dual witness to prove our sought lower bound on the approximate
trace norm. The theorem below only treats a canonical case of the
subspace intersection problem. However, we will see shortly that this
result allows us to tackle all parameter settings.
\begin{thm}
\label{thm:approx-norm-subspace}Let $n,m,\ell,R$ be given integers
with $0<R\leq\min\{m,\ell\}$ and $m+\ell\leq n$. Let $F$ be the
characteristic matrix of $\INTERSECT_{0,R}^{\mathbb{F}_{q},n,m,\ell}$.
Then for all reals $\delta\geq0$ and all nonnegative integers $d_{1},d_{2}$
with $d_{1}+d_{2}\leq\min\{m,\ell\}-R,$
\begin{align}
\|F\|_{\Sigma,\delta} & \geq\frac{1}{8}\left(1-\delta-\frac{64}{q^{d_{1}+1}}\right)\binom{n}{m}_{q}^{1/2}\binom{n}{\ell}_{q}^{1/2}q^{(\min\{m,\ell\}-R-d_{1}-d_{2}+1)(m+\ell-2\min\{m,\ell\}+2d_{2})/4},\label{eq:F-m-l-trace-norm}\\
\|F\|_{\Sigma,\delta} & \geq\frac{1-\delta}{8}\binom{n}{m}_{q}^{1/2}\binom{n}{\ell}_{q}^{1/2}q^{(m+\ell-2R)/4}.\label{eq:F-m-l-trace-norm-alternate}
\end{align}
\end{thm}

\begin{proof}
Structurally, the proof is similar to that of~\prettyref{thm:approx-norm-rank-problem}.
Let $\Delta=\min\{m,\ell\}.$ Then $0<R\leq\Delta-d_{1}-d_{2}$ by
hypothesis. Let $\psi\colon\{0,1,\ldots,\Delta\}\to\Re$ be the function
constructed in \prettyref{lem:phi-subspace}, and extend $\psi$ to
all of $\mathbb{Z}$ by defining $\psi(r)=0$ for $r\notin\{0,1,\ldots,\Delta\}.$
Then
\begin{align}
\sum_{\dom F}F_{S,T} & (\overline{J}_{\psi}^{\,n,m,\ell})_{S,T}-\delta\|\overline{J}_{\psi}^{\,n,m,\ell}\|_{1}-\sum_{\overline{\dom F}}|(\overline{J}_{\psi}^{\,n,m,\ell})_{S,T}|\nonumber \\
 & =-\sum_{\dim(S\cap T)=0}(\overline{J}_{\psi}^{\,n,m,\ell})_{S,T}+\sum_{\dim(S\cap T)=R}(\overline{J}_{\psi}^{\,n,m,\ell})_{S,T}-\delta\|\overline{J}_{\psi}^{\,n,m,\ell}\|_{1}\nonumber \\
 & \qquad\qquad\qquad\qquad\qquad\qquad\qquad\qquad\qquad-\sum_{\dim(S\cap T)\notin\{0,R\}}|(\overline{J}_{\psi}^{\,n,m,\ell})_{S,T}|\nonumber \\
 & =-\psi(0)+\psi(R)-\delta\|\psi\|_{1}-\sum_{r\notin\{0,R\}}|\psi(r)|\nonumber \\
 & =|\psi(0)|+|\psi(R)|-\delta\|\psi\|_{1}-\sum_{r\notin\{0,R\}}|\psi(r)|\nonumber \\
 & =(1-\delta)\|\psi\|_{1}-2\sum_{r\notin\{0,R\}}|\psi(r)|\nonumber \\
 & \geq\left(1-\delta-2\sum_{r\notin\{0,R\}}|\psi(r)|\right)\|\psi\|_{1},\label{eq:Jbar-F-correlation}
\end{align}
where the second step uses \prettyref{lem:J-bar-metric}, the third
step is valid by \prettyref{lem:phi-subspace}\prettyref{enu:phi-at-0-subspace}\textendash \prettyref{enu:phi-at-R-subspace},
and the fifth step is justified by \prettyref{lem:phi-subspace}\prettyref{enu:phi-at-0-subspace}.

We now analyze the spectral norm of $\overline{J}_{\psi}^{\,n,m,\ell}$.
Recall from \prettyref{lem:Lambda-r-bar}\prettyref{enu:Lambda-bar-poly}
that for fixed $n,m,\ell$ and fixed $k\in\{0,1,\ldots,\Delta\},$
the quantity $\overline{\Lambda}_{r}^{\,n,m,\ell}(k)$ as a function
of $r\in\{0,1,\ldots,\Delta\}$ is a polynomial in $q^{r}$ of degree
at most $k.$ As a result,
\begin{align}
\max_{k\in\{0,1,\ldots,\Delta-R-d_{1}-d_{2}\}} & \sqrt{\overline{\Lambda}_{\psi}^{\,n,m,\ell}(k)\,\overline{\Lambda}_{\psi}^{\,n,\ell,m}(k)}\nonumber \\
 & =\max_{k\in\{0,1,\ldots,\Delta-R-d_{1}-d_{2}\}}\;\sqrt{\overline{\Lambda}_{\psi}^{\,n,\ell,m}(k)\cdot\sum_{r=0}^{\Delta}\psi(r)\overline{\Lambda}_{r}^{\,n,m,\ell}(k)}\nonumber \\
 & =\max_{k\in\{0,1,\ldots,\Delta-R-d_{1}-d_{2}\}}\;\sqrt{\overline{\Lambda}_{\psi}^{\,n,\ell,m}(k)\cdot0}\nonumber \\
 & =0,\label{eq:sing-value-k-small}
\end{align}
where the first step applies Definition~\prettyref{def:Lambda-bar},
and the second step uses~\prettyref{lem:phi-subspace}\prettyref{enu:phi-orthog-subspace}.
Next, for all $k\in\{0,1,\ldots,\Delta\}$,
\begin{align*}
|\overline{\Lambda}_{\psi}^{\,n,m,\ell}(k)\,\overline{\Lambda}_{\psi}^{\,n,\ell,m}(k)| & =\left|\sum_{r=0}^{\Delta}\psi(r)\overline{\Lambda}_{r}^{\,n,m,\ell}(k)\right|\cdot\left|\sum_{r=0}^{\Delta}\psi(r)\overline{\Lambda}_{r}^{\,n,\ell,m}(k)\right|\\
 & =\left|\sum_{r=0}^{\Delta-d_{2}}\psi(r)\overline{\Lambda}_{r}^{\,n,m,\ell}(k)\right|\cdot\left|\sum_{r=0}^{\Delta-d_{2}}\psi(r)\overline{\Lambda}_{r}^{\,n,\ell,m}(k)\right|\\
 & \leq\|\psi\|_{1}\max_{r=0,1,\ldots,\Delta-d_{2}}|\overline{\Lambda}_{r}^{\,n,m,\ell}(k)|\cdot\|\psi\|_{1}\max_{r=0,1,\ldots,\Delta-d_{2}}|\overline{\Lambda}_{r}^{\,n,\ell,m}(k)|\\
 & \leq\|\psi\|_{1}^{2}\cdot8\binom{n}{m}_{q}^{-1}q^{-k(m-\Delta+d_{2})/2}\cdot8\binom{n}{\ell}_{q}^{-1}q^{-k(\ell-\Delta+d_{2})/2},
\end{align*}
where the first step applies Definition~\prettyref{def:Lambda-bar},
the second step uses \prettyref{lem:phi-subspace}\prettyref{enu:phi-vanish-subspace},
and the last step invokes \prettyref{lem:Lambda-r-bar}\prettyref{enu:Lambda-bar-bound}
to bound $\overline{\Lambda}_{r}^{\,n,m,\ell}(k)$ and then again
(with the roles of $m$ and $\ell$ interchanged) to bound $\overline{\Lambda}_{r}^{\,n,\ell,m}(k)$.
It follows that
\begin{align}
 & \max_{k\in\{\Delta-R-d_{1}-d_{2}+1,\ldots,\Delta-1,\Delta\}}\sqrt{\overline{\Lambda}_{\psi}^{\,n,m,\ell}(k)\,\overline{\Lambda}_{\psi}^{\,n,\ell,m}(k)}\nonumber \\
 & \qquad\qquad\leq\max_{k\in\{\Delta-R-d_{1}-d_{2}+1,\ldots,\Delta-1,\Delta\}}\;8\|\psi\|_{1}\left(\binom{n}{m}_{q}\binom{n}{\ell}_{q}q^{k(m-\Delta+d_{2})/2}q^{k(\ell-\Delta+d_{2})/2}\right)^{-1/2}\nonumber \\
 & \qquad\qquad=8\|\psi\|_{1}\left(\binom{n}{m}_{q}\binom{n}{\ell}_{q}q^{(\Delta-R-d_{1}-d_{2}+1)(m+\ell-2\Delta+2d_{2})/2}\right)^{-1/2}.\label{eq:sing-value-k-large}
\end{align}
As a result,
\begin{align}
\|\overline{J}_{\psi}^{\,n,m,\ell}\| & =\max_{k\in\{0,1,\ldots,\Delta\}}\sqrt{\overline{\Lambda}_{\psi}^{\,n,m,\ell}(k)\,\overline{\Lambda}_{\psi}^{\,n,\ell,m}(k)}\nonumber \\
 & \leq8\|\psi\|_{1}\left(\binom{n}{m}_{q}\binom{n}{\ell}_{q}q^{(\Delta-R-d_{1}-d_{2}+1)(m+\ell-2\Delta+2d_{2})/2}\right)^{-1/2},\label{eq:Jbar-spectral-norm-bound}
\end{align}
where the first step appeals to~\prettyref{thm:J-bar-singular-values},
and the second step substitutes the upper bounds from~\prettyref{eq:sing-value-k-small}
and \prettyref{eq:sing-value-k-large}. 

We are now in a position to complete the proof of the theorem. \prettyref{prop:approxtracelower}
with $\Phi=\overline{J}_{\psi}^{\,n,m,\ell}$ implies, in view of~\prettyref{eq:Jbar-F-correlation}
and~\prettyref{eq:Jbar-spectral-norm-bound}, that
\[
\|F\|_{\Sigma,\delta}\geq\frac{1}{8}\left(1-\delta-2\sum_{r\notin\{0,R\}}|\psi(r)|\right)\binom{n}{m}_{q}^{1/2}\binom{n}{\ell}_{q}^{1/2}q^{(\Delta-R-d_{1}-d_{2}+1)(m+\ell-2\Delta+2d_{2})/4}.
\]
Since $\sum_{r\notin\{0,R\}}|\psi(r)|\leq32q^{-d_{1}-1}$ by \prettyref{lem:phi-subspace}\prettyref{enu:phi-small-tail-subspace},
this settles~\prettyref{eq:F-m-l-trace-norm}. In the special case
$d_{1}=0$ and $d_{2}=\Delta-R,$ we have $\sum_{r\notin\{0,R\}}|\psi(r)|=0$
from \prettyref{lem:phi-subspace}\prettyref{enu:phi-vanish-subspace},
whence~\prettyref{eq:F-m-l-trace-norm-alternate}.
\end{proof}

\subsection{\label{subsec:Communication-lower-bounds-subspace}Communication
lower bounds}

We will now prove an optimal lower bound on the communication complexity
of the subspace intersection problem. To simplify the exposition,
we will first consider the canonical case where Alice and Bob need
to determine whether the intersection of their subspaces has dimension
$0$ versus dimension $R$, corresponding to the approximate trace
norm result that we just obtained. The general lower bound for all
parameter settings will then follow using the reduction of \prettyref{prop:INTERSECT-reductions}.
\begin{lem}
\label{lem:INTERSECT-cc-lower-bound-canonical}Let $\mathbb{F}$ be
a finite field with $q=|\mathbb{F}|$ elements. Let $n,m,\ell,R$
be nonnegative integers with
\begin{align}
 & 0<R\leq\min\{m,\ell\},\label{eq:rmell}\\
 & R<\max\{m,\ell\},\label{eq:intersect-not-EQ}\\
 & m+\ell\leq n.\label{eq:m-ell-leq-n}
\end{align}
Then
\begin{equation}
Q_{(1-\gamma)/2}^{*}(\INTERSECT_{0,R}^{\mathbb{F},n,m,\ell})\geq c(\log_{q}\lceil q^{m-R}\gamma\rceil+1)(\log_{q}\lceil q^{\ell-R}\gamma\rceil+1)\log q\label{eq:INTERSECT-cc-lower-bound-canonical}
\end{equation}
for all $\gamma\in[\frac{1}{3}q^{-(m+\ell-2R)/5},1],$ where $c>0$
is an absolute constant independent of $\mathbb{F},n,m,\ell,R,\gamma.$
\end{lem}

\begin{proof}
Due to the symmetry between $m$ and $\ell$ in the statement of the
lemma, we may assume that
\begin{equation}
m\geq\ell,\label{eq:m-ge-ell}
\end{equation}
corresponding to the mnemonic ``$m$ for more, $\ell$ for less.''
The hypotheses~\prettyref{eq:m-ell-leq-n} ensures that there is
a pair of subspaces in $\Scal(\mathbb{F}^{n},m)\times\Scal(\mathbb{F}^{n},\ell)$
whose intersection has dimension $0$; analogously, \prettyref{eq:rmell}
and~\prettyref{eq:m-ell-leq-n} ensure that there is a pair of subspaces
in $\Scal(\mathbb{F}^{n},m)\times\Scal(\mathbb{F}^{n},\ell)$ whose
intersection has dimension $R.$ This makes $\INTERSECT_{0,R}^{\mathbb{F},n,m,\ell}$
a nonconstant function, with the trivial lower bound
\begin{equation}
Q_{(1-\gamma)/2}^{*}(\INTERSECT_{0,R}^{\mathbb{F},n,m,\ell})\geq1.\label{eq:INTERSECT-trivial}
\end{equation}
It suffices to prove that the characteristic matrix $F$ of this communication
problem satisfies
\begin{equation}
\|F\|_{\Sigma,1-\gamma}\geq c'\binom{n}{m}_{q}^{1/2}\binom{n}{\ell}_{q}^{1/2}q^{c'(\log_{q}\lceil q^{m-R}\gamma\rceil+1)(\log_{q}\lceil q^{\ell-R}\gamma\rceil+1)}\label{eq:INTERSECT-trace-norm}
\end{equation}
for some absolute constant $c'>0.$ Indeed, once this lower bound
is established, an appeal to \prettyref{thm:approx-trace-norm-method}
yields
\begin{align}
Q_{(1-\gamma)/2}^{*} & (\INTERSECT_{0,R}^{\mathbb{F},n,m,\ell})\nonumber \\
 & \qquad\geq\frac{1}{2}\log\frac{c'q^{c'(\log_{q}\lceil q^{m-R}\gamma\rceil+1)(\log_{q}\lceil q^{\ell-R}\gamma\rceil+1)}}{3}\nonumber \\
 & \qquad=\frac{c'}{2}(\log_{q}\lceil q^{m-R}\gamma\rceil+1)(\log_{q}\lceil q^{\ell-R}\gamma\rceil+1)\log q-\frac{1}{2}\log\frac{3}{c'}.\label{eq:Q-subspace-lower-bound}
\end{align}
Taking a weighted arithmetic average of~\prettyref{eq:INTERSECT-trivial}
and~\prettyref{eq:Q-subspace-lower-bound} settles~\prettyref{eq:INTERSECT-cc-lower-bound-canonical}.

In what follows, we prove~\prettyref{eq:INTERSECT-trace-norm}. We
first examine the case $\gamma\leq q^{-\ell+R+23}.$ Equation~\prettyref{eq:F-m-l-trace-norm-alternate}
of~\prettyref{thm:approx-norm-subspace} yields
\begin{align}
\|F\|_{\Sigma,1-\gamma} & \geq\frac{\gamma}{8}\binom{n}{m}_{q}^{1/2}\binom{n}{\ell}_{q}^{1/2}q^{(m+\ell-2R)/4}\nonumber \\
 & \geq\frac{1}{24}\binom{n}{m}_{q}^{1/2}\binom{n}{\ell}_{q}^{1/2}q^{(m+\ell-2R)/20},\label{eq:trace-large-case}
\end{align}
where the second step uses the lemma hypothesis that $\gamma\geq\frac{1}{3}q^{-(m+\ell-2R)/5}$.
Moreover,
\begin{align}
m+\ell-2R & \geq m-R\nonumber \\
 & \geq\frac{1}{2}(m-R+1)\nonumber \\
 & \geq\frac{1}{2}(m-R+1)\cdot\frac{1}{24}(\log_{q}\lceil q^{\ell-R}\gamma\rceil+1)\nonumber \\
 & \ge\frac{1}{48}(\log_{q}\lceil q^{m-R}\gamma\rceil+1)(\log_{q}\lceil q^{\ell-R}\gamma\rceil+1),\label{eq:mell-intermed}
\end{align}
where the first step uses~\prettyref{eq:rmell}, the second step
is valid by~\prettyref{eq:intersect-not-EQ} and~\prettyref{eq:m-ge-ell},
the third step is legitimate because $\gamma\leq q^{-\ell+R+23}$
in the case under consideration, and the last step uses the lemma
hypothesis that $\gamma\leq1.$ Equations~\prettyref{eq:trace-large-case}
and~\prettyref{eq:mell-intermed} imply \prettyref{eq:INTERSECT-trace-norm}
for $c'=1/960.$

We now examine the complementary case, $\gamma\geq q^{-\ell+R+23}.$
This assumption on $\gamma$, along with the lemma hypothesis that
$\gamma\leq1,$ implies that the integer $d_{1}=\lfloor\log_{q}(128/\gamma)\rfloor$
is an element of $\{0,1,2,\ldots,\ell-R\}.$ This in turn means that
the integer $d_{2}=\lceil(\ell-R-d_{1})/2\rceil$ is also an element
of $\{0,1,2,\ldots,\ell-R\}.$ We have $d_{1}+d_{2}=d_{1}+\lceil(\ell-R-d_{1})/2\rceil\leq d_{1}+(\ell-R-d_{1})=\ell-R=\min\{m,\ell\}-R,$
where the last step uses~\prettyref{eq:m-ge-ell}. As a result, \prettyref{thm:approx-norm-subspace}
is applicable with parameters $d_{1}$ and $d_{2}$, and equation~\prettyref{eq:F-m-l-trace-norm}
yields 
\begin{align}
\|F\|_{\Sigma,1-\gamma} & \geq\frac{1}{8}\left(\gamma-\frac{64}{q^{d_{1}+1}}\right)\binom{n}{m}_{q}^{1/2}\binom{n}{\ell}_{q}^{1/2}q^{(\min\{m,\ell\}-R-d_{1}-d_{2}+1)(m+\ell-2\min\{m,\ell\}+2d_{2})/4}\nonumber \\
 & \geq\frac{1}{8}\left(\gamma-\frac{64}{q^{d_{1}+1}}\right)\binom{n}{m}_{q}^{1/2}\binom{n}{\ell}_{q}^{1/2}q^{(\ell-R-d_{1}-d_{2}+1)(m-\ell+2d_{2})/4}\nonumber \\
 & \geq\frac{\gamma}{16}\binom{n}{m}_{q}^{1/2}\binom{n}{\ell}_{q}^{1/2}q^{(\ell-R-d_{1}-d_{2}+1)(m-\ell+2d_{2})/4}\nonumber \\
 & \geq\frac{\gamma}{16}\binom{n}{m}_{q}^{1/2}\binom{n}{\ell}_{q}^{1/2}q^{(\ell-R-d_{1})(m-R-d_{1})/8}\nonumber \\
 & =\frac{1}{16}\binom{n}{m}_{q}^{1/2}\binom{n}{\ell}_{q}^{1/2}q^{(\ell-R-\lfloor\log_{q}(128/\gamma)\rfloor)(m-R-\lfloor\log_{q}(128/\gamma)\rfloor)/8-\log_{q}(1/\gamma)},\label{eq:F-m-ell-R-intermed}
\end{align}
where the second step applies~\prettyref{eq:m-ge-ell}, the third
and fifth steps use the definition of $d_{1},$ and the fourth step
uses the definition of $d_{2}$. Recall that $\gamma\geq q^{-\ell+R+23}$
in the case under consideration, and $\gamma\in[\frac{1}{3}q^{-(m+\ell-2R)/5},1]$
by the lemma hypothesis. We may therefore use Claim~\ref{claim:ell-R-m-R},
stated and proved below, to simplify to the right-hand side of~\prettyref{eq:F-m-ell-R-intermed}
as follows:
\[
\|F\|_{\Sigma,1-\gamma}\geq\frac{1}{16}\binom{n}{m}_{q}^{1/2}\binom{n}{\ell}_{q}^{1/2}q^{(\log_{q}\lceil q^{m-R}\gamma\rceil+1)(\log_{q}\lceil q^{\ell-R}\gamma\rceil+1)/160}.
\]
This establishes~\prettyref{eq:INTERSECT-trace-norm} with $c'=1/160,$
completing the proof of the lemma.
\end{proof}
\begin{claim}
\label{claim:ell-R-m-R}For any $\gamma$ with $\max\{q^{-\ell+R+23},\frac{1}{3}q^{-(m+\ell-2R)/5}\}\leq\gamma\leq1,$
\begin{multline}
\frac{1}{8}\left(\ell-R-\left\lfloor \log_{q}\frac{128}{\gamma}\right\rfloor \right)\left(m-R-\left\lfloor \log_{q}\frac{128}{\gamma}\right\rfloor \right)-\log_{q}\frac{1}{\gamma}\\
\geq\frac{1}{160}(\log_{q}\lceil q^{m-R}\gamma\rceil+1)(\log_{q}\lceil q^{\ell-R}\gamma\rceil+1).\qquad\label{eq:ell-R-m-R}
\end{multline}
\end{claim}

\begin{proof}
The proof is somewhat tedious but straightforward. To begin with,
\begin{align}
\frac{1}{8}\left(\ell-R-\left\lfloor \log_{q}\frac{128}{\gamma}\right\rfloor \right) & \left(m-R-\left\lfloor \log_{q}\frac{128}{\gamma}\right\rfloor \right)-\log_{q}\frac{1}{\gamma}\nonumber \\
 & =\frac{1}{8}\left\lceil \log_{q}\frac{q^{\ell-R}\gamma}{128}\right\rceil \left(m-R-\left\lfloor \log_{q}\frac{128}{\gamma}\right\rfloor \right)-\log_{q}\frac{1}{\gamma}\nonumber \\
 & \geq\frac{1}{8}\left\lceil \log_{q}\frac{q^{\ell-R}\gamma}{128}\right\rceil \left(m-R-\left\lfloor \log_{q}\frac{128}{\gamma}\right\rfloor -\frac{1}{2}\log_{q}\frac{1}{\gamma}\right),\label{eq:m-R-ell-R-exponent}
\end{align}
where the last step uses the fact that $\frac{1}{8}\lceil\log_{q}(q^{\ell-R}\gamma/128)\rceil\geq2$
due to the hypothesis $\gamma\geq q^{-\ell+R+23}.$ We now bound from
below the factors in~\prettyref{eq:m-R-ell-R-exponent}. We have
\begin{equation}
\log_{q}\frac{q^{\ell-R}\gamma}{128}\geq\log_{q}\frac{q^{\ell-R}\gamma}{q^{7}}\geq\log_{q}\frac{\lceil q^{\ell-R}\gamma\rceil}{q^{8}}=\log_{q}\lceil q^{\ell-R}\gamma\rceil-8\geq\frac{1}{2}(\log_{q}\lceil q^{\ell-R}\gamma\rceil+1),\label{eq:m-R-ell-R-first-factor}
\end{equation}
where the second and fourth steps are valid because $q^{\ell-R}\gamma\geq q^{23}$
by hypothesis. The other factor in~\prettyref{eq:m-R-ell-R-exponent}
can be bounded as follows:
\begin{align}
m-R-\left\lfloor \log_{q}\frac{128}{\gamma}\right\rfloor -\frac{1}{2}\log_{q}\frac{1}{\gamma} & \geq m-R-\log_{q}\frac{q^{7}}{\gamma^{3/2}}\nonumber \\
 & \geq m-R-\log_{q}\frac{q^{7}}{(\max\{q^{-\ell+R+23},\frac{1}{3}q^{-(m+\ell-2R)/5}\})^{3/2}}\nonumber \\
 & \geq m-R-7-\frac{3}{2}\min\left\{ \ell-R-23,\frac{1}{5}(m+\ell-2R)+2\right\} \nonumber \\
 & \geq m-R-7-\frac{3}{2}\min\left\{ m-R-23,\frac{2}{5}(m-R)+2\right\} \nonumber \\
 & \geq m-R-7-\frac{3}{2}\left(\frac{1}{3}(m-R-23)+\frac{2}{3}\left(\frac{2}{5}(m-R)+2\right)\right)\nonumber \\
 & =\frac{1}{10}(m-R)+\frac{5}{2}\nonumber \\
 & \geq\frac{1}{10}(\log_{q}\lceil q^{m-R}\gamma\rceil+1),\label{eq:m-R-ell-R-second-factor}
\end{align}
where the first step applies the bound $128\leq q^{7}$ and drops
the floor operator, the second step uses the hypothesis for $\gamma,$
the fourth step is valid by~\prettyref{eq:m-ge-ell}, the fifth step
replaces the minimum by a weighted average, and the last step is legitimate
because $\gamma\leq1$ by hypothesis. Now~\prettyref{eq:ell-R-m-R}
follows from~\prettyref{eq:m-R-ell-R-exponent}\textendash \prettyref{eq:m-R-ell-R-second-factor}.
\end{proof}
We now extend the previous lemma to all possible parameter settings,
thus obtaining the desired communication lower bound for subspace
intersection.
\begin{thm}
\label{thm:INTERSECT-cc-lower-bound-general}Let $c>0$ be the absolute
constant from Lemma~\emph{\ref{lem:INTERSECT-cc-lower-bound-canonical}}.
Let $\mathbb{F}$ be a finite field with $q=|\mathbb{F}|$ elements,
and let $n,m,\ell,r,R$ be integers with $\max\{0,m+\ell-n\}\leq r<R\leq\min\{m,\ell\}.$
Then $\max\{m,\ell\}\leq n.$ Furthermore, for all $\gamma\in[\frac{1}{3}q^{-(m+\ell-2R)/5},1],$
\[
Q_{\frac{1-\gamma}{2}}^{*}(\INTERSECT_{r,R}^{\mathbb{F},n,m,\ell})\geq\begin{cases}
1 & \text{if }R=m=\ell,\\
c(\log_{q}\lceil q^{m-R}\gamma\rceil+1)(\log_{q}\lceil q^{\ell-R}\gamma\rceil+1)\log q & \text{otherwise.}
\end{cases}
\]
\end{thm}

\begin{proof}
The hypothesis $\max\{0,m+\ell-n\}\leq r<R\leq\min\{m,\ell\}$ implies
that $m+\ell-n\leq\min\{m,\ell\}$, which is equivalent to $\max\{m,\ell\}\leq n.$ 

Recall from \prettyref{prop:possible-intersections-and-sums} that
for each integer $d\in[\max\{0,m+\ell-n\},\min\{m,\ell\}],$ there
are subspaces $S\in\Scal(\mathbb{F}^{n},m)$ and $T\in\Scal(\mathbb{F}^{n},\ell)$
with $\dim(S\cap T)=d.$ This makes $\INTERSECT_{r,R}^{\mathbb{F},n,m,\ell}$
a nonconstant function, which means that its $\epsilon$-error quantum
communication complexity for each $\epsilon\in[0,1/2)$ is at least
$1$ bit. This settles the claimed communication lower bounds in the
case $R=m=\ell.$ 

In what follows, we focus on the complementary case when $R,m,\ell$
are not all equal. In view of $R\leq\min\{m,\ell\},$ we infer that
$R<\max\{m,\ell\}.$ This new inequality, and the theorem hypotheses
that $m+\ell-n\leq r<R\leq\min\{m,\ell\}$ and $\gamma\in[\frac{1}{3}q^{-(m+\ell-2R)/5},1]$,
can be equivalently stated as 
\begin{align}
 & R-r<\max\{m-r,\ell-r\},\label{eq:R-prime-max}\\
 & 0<R-r\leq\min\{m-r,\ell-r\},\\
 & (m-r)+(\ell-r)\leq n-r,\\
 & \gamma\in[{\textstyle \frac{1}{3}}q^{-((m-r)+(\ell-r)-2(R-r))/5},1].\label{eq:gamma-prime}
\end{align}
Now
\begin{align}
Q_{(1-\gamma)/2}^{*}(\INTERSECT_{r,R}^{\mathbb{F},n,m,\ell}) & \geq Q_{(1-\gamma)/2}^{*}(\INTERSECT_{0,R-r}^{\mathbb{F},n-r,m-r,\ell-r})\nonumber \\
 & \geq c(\log_{q}\lceil q^{m-R}\gamma\rceil+1)(\log_{q}\lceil q^{\ell-R}\gamma\rceil+1)\log q,\label{eq:intersect-intermed-lower}
\end{align}
where the first step uses~\prettyref{prop:INTERSECT-reductions},
and the second step is valid by~\prettyref{lem:INTERSECT-cc-lower-bound-canonical}
whose application is in turn justified by~\prettyref{eq:R-prime-max}\textendash \prettyref{eq:gamma-prime}.
\end{proof}
\prettyref{thm:INTERSECT-cc-lower-bound-general} settles the \emph{quantum}
communication lower bound of~\prettyref{thm:MAIN-subspace-intersection-general}
for the \emph{promise} subspace intersection problem, and hence also
the \emph{randomized} communication lower bound for the \emph{total}
subspace intersection problem.

\subsection{\label{subsec:Communication-upper-bounds-subspace-small-error}Communication
upper bounds for small error}

In this section and the next, we prove communication upper bounds
matching our lower bound for the subspace intersection problem. We
start with a technical lemma.
\begin{lem}
\label{lem:S-T-XS-XT}Let $n,m,\ell,r,\Delta$ be nonnegative integers
with $r\leq\min\{m,\ell\}$ and $\max\{m,\ell\}\leq n.$ Fix a finite
field $\mathbb{F},$ and let $S\in\Scal(\mathbb{F}^{n},m)$ and $T\in\Scal(\mathbb{F}^{n},\ell)$
be given subspaces. Let $X\in\mathbb{F}^{(m+\ell-r+\Delta)\times n}$
and $Y\in\mathbb{F}^{(m+\ell-2r+3\Delta)\times(m+\ell-r+\Delta)}$
be uniformly random matrices. Then with probability at least $1-16|\mathbb{F}|^{-\Delta-1},$
one has
\begin{align}
 & \dim(X(S))=\dim(S),\label{eq:preserve-S}\\
 & \dim(X(T))=\dim(T),\label{eq:preserve-T}\\
 & \dim(Y((X(S))^{\perp}))=\dim((X(S))^{\perp}),\label{eq:preserve-XS}\\
 & \dim(Y((X(T))^{\perp}))=\dim((X(T))^{\perp}).\label{eq:preserve-XT}
\end{align}
Assuming~\emph{\prettyref{eq:preserve-S}\textendash \prettyref{eq:preserve-XT}},
the subspaces $S'=Y((X(S))^{\perp})$ and $T'=Y((X(T))^{\perp})$
satisfy
\begin{align}
 & \dim(S')=\ell-r+\Delta,\label{eq:Sprime}\\
 & \dim(T')=m-r+\Delta,\label{eq:Tprime}\\
 & \dim(S'\cap T')=m+\ell-r+\Delta-\dim(X(S+T))\nonumber \\
 & \qquad\qquad\qquad\qquad+\dim((X(S))^{\perp}+(X(T))^{\perp})-\dim(Y((X(S))^{\perp}+(X(T))^{\perp})).\label{eq:Sprime-Tprime}
\end{align}
\end{lem}

\begin{proof}
Abbreviate $q=|\mathbb{F}|.$ Let $E_{1},E_{2},E_{3},E_{4}$ be the
events that correspond to~\prettyref{eq:preserve-S}\textendash \prettyref{eq:preserve-XT},
respectively. Applying \prettyref{lem:random-proj-subspaces} with
$t=m-1$ and $d=m+\ell-r+\Delta$ gives 
\begin{align}
\Prob[\neg E_{1}] & \leq4q^{-(\ell-r+\Delta+1)}\leq4q^{-\Delta-1}.\label{eq:E1-fails}
\end{align}
Analogously, applying \prettyref{lem:random-proj-subspaces} with
$t=\ell-1$ and $d=m+\ell-r+\Delta$ shows that
\begin{align}
\Prob[\neg E_{2}] & \leq4q^{-(m-r+\Delta+1)}\leq4q^{-\Delta-1}.\label{eq:E2-fails}
\end{align}
Conditioned on $E_{1}\wedge E_{2},$ we have 
\begin{align}
 & \dim((X(S))^{\perp})=m+\ell-r+\Delta-\dim(X(S))=\ell-r+\Delta,\label{eq:XS-perp-exact}\\
 & \dim((X(T))^{\perp})=m+\ell-r+\Delta-\dim(X(T))=m-r+\Delta.\label{eq:XT-perp-exact}
\end{align}
As a result, invoking \prettyref{lem:random-proj-subspaces} with
$t=\ell-r+\Delta-1$ and $d=m+\ell-2r+3\Delta$ shows that 
\begin{equation}
\Prob[\neg E_{3}\mid E_{1}\wedge E_{2}]\leq4q^{-(m+\ell-2r+3\Delta)+(\ell-r+\Delta-1)}\leq4q^{-2\Delta-1}.\label{eq:E3-fails}
\end{equation}
Analogously, invoking \prettyref{lem:random-proj-subspaces} with
$t=m-r+\Delta-1$ and $d=m+\ell-2r+3\Delta$ shows that 
\begin{equation}
\Prob[\neg E_{4}\mid E_{1}\wedge E_{2}]\leq4q^{-(m+\ell-2r+3\Delta)+(m-r+\Delta-1)}\leq4q^{-2\Delta-1}.\label{eq:E4-fails}
\end{equation}
Now
\begin{align*}
\Prob[E_{1}\wedge E_{2}\wedge E_{3}\wedge E_{4}] & =\Prob[E_{1}\wedge E_{2}]\Prob[E_{3}\wedge E_{4}\mid E_{1}\wedge E_{2}]\\
 & \geq\Prob[E_{1}\wedge E_{2}]-\Prob[\neg(E_{3}\wedge E_{4})\mid E_{1}\wedge E_{2}]\\
 & \geq1-\Prob[\neg E_{1}]-\Prob[\neg E_{2}]-\Prob[\neg E_{3}\mid E_{1}\wedge E_{2}]-\Prob[\neg E_{4}\mid E_{1}\wedge E_{2}]\\
 & \geq1-16q^{-\Delta-1},
\end{align*}
where the last step uses \prettyref{eq:E1-fails}\textendash \prettyref{eq:E4-fails}.
This settles the first part of the lemma.

In what follows, we assume~\prettyref{eq:preserve-S}\textendash \prettyref{eq:preserve-XT}.
Then \prettyref{eq:Sprime} follows from $\dim(S')=\dim((X(S))^{\perp})=\ell-r+\Delta,$
where the last step uses~\prettyref{eq:XS-perp-exact}. Analogously,
\prettyref{eq:Tprime} follows from $\dim(T')=\dim((X(T))^{\perp})=m-r+\Delta,$
where the last step uses~\prettyref{eq:XT-perp-exact}. Toward the
remaining equation~\prettyref{eq:Sprime-Tprime}, we have
\begin{align*}
\dim((X(S))^{\perp}+(X(T))^{\perp}) & =\dim(((X(S))\cap(X(T)))^{\perp})\\
 & =m+\ell-r+\Delta-\dim((X(S))\cap(X(T)))\\
 & =m+\ell-r+\Delta-(\dim(X(S))+\dim(X(T))-\dim(X(S)+X(T)))\\
 & =-r+\Delta+\dim(X(S+T)),
\end{align*}
where the first step uses \prettyref{fact:Aperp-intersect-Bperp},
and the last step uses~\prettyref{eq:preserve-S}, \prettyref{eq:preserve-T},
and the linearity of $X.$ With this substitution, \prettyref{eq:Sprime-Tprime}
is equivalent to 
\begin{equation}
\dim(S'\cap T')=m+\ell-2r+2\Delta-\dim(Y((X(S))^{\perp}+(X(T))^{\perp})).\label{eq:Sprime-Tprime-simplfied}
\end{equation}
Due to \prettyref{eq:Sprime}, \prettyref{eq:Tprime}, and $Y((X(S))^{\perp}+(X(T))^{\perp})=Y((X(S))^{\perp})+Y((X(T))^{\perp})=S'+T'$,
equation~\prettyref{eq:Sprime-Tprime-simplfied} is a restatement
of $\dim(S'\cap T')=\dim(S')+\dim(T')-\dim(S'+T'),$ which is a well-known
identity valid for any subspaces $S',T'.$
\end{proof}
We are now ready to prove our communication upper bound for subspace
intersection in the regime where the error probability is a small
constant or tends to $0$. In the next section, we will generalize
this result to the more challenging regime where the error tends to
$1/2$.
\begin{thm}[Small error]
\label{thm:subspace-upper-bound-small-error} Let $\mathbb{F}$ be
a finite field with $q=|\mathbb{F}|$ elements. Let $n,m,\ell,R$
be integers with $0<R\leq\min\{m,\ell\}$ and $\max\{m,\ell\}\leq n$.
Then for each $0<\epsilon\leq1/3,$
\begin{equation}
R_{\epsilon}(\INTERSECT_{R}^{\mathbb{F},n,m,\ell})=O\left(\left(m-R+\left\lceil \log_{q}\frac{1}{\epsilon}\right\rceil \right)\left(\ell-R+\left\lceil \log_{q}\frac{1}{\epsilon}\right\rceil \right)\log q\right).\label{eq:subspace-upper-bound-high-accuracy}
\end{equation}
If in addition $m=\ell=R,$ then for each $0<\epsilon\leq1/3,$
\begin{equation}
R_{\epsilon}(\INTERSECT_{R}^{\mathbb{F},n,m,\ell})=O\left(\log\frac{1}{\epsilon}\right).\label{eq:intersect-equality}
\end{equation}
\end{thm}

\begin{proof}
Define $r=R-1.$ For an integer $\Delta\geq0$ to be set later, consider
the following protocol $\Pi.$ On input a pair of subspaces $S\in\Scal(\mathbb{F}^{n},m)$
for Alice and $T\in\Scal(\mathbb{F}^{n},\ell)$ for Bob, the parties
use their shared randomness to pick independent and uniformly random
matrices $X\in\mathbb{F}^{(m+\ell-r+\Delta)\times n}$ and $Y\in\mathbb{F}^{(m+\ell-2r+3\Delta)\times(m+\ell-r+\Delta)}$.
Next, they verify the four conditions~\prettyref{eq:preserve-S}\textendash \prettyref{eq:preserve-XT}.
This can be done using only two bits of communication, with Alice
and Bob verifying the conditions pertaining to their respective inputs.
If any of these conditions fail, they output a uniformly random value
in $\{-1,1\}.$ In the complementary case, Alice and Bob compute
\begin{align*}
S' & =Y((X(S))^{\perp}),\\
T' & =Y((X(T))^{\perp}),
\end{align*}
respectively. The owner of the smaller of the subspaces $S'$ and
$T'$ sends it to the other party in the form of a basis, who then
computes $\dim(S'\cap T')$ and outputs $1$ if and only if $\dim(S'\cap T')\leq\Delta.$

We first analyze the communication cost of $\Pi.$ If any of the conditions~\prettyref{eq:preserve-S}\textendash \prettyref{eq:preserve-XT}
fail, the communication cost is $2$ bits. If all four conditions
hold, then $\dim(S')=\ell-r+\Delta$ and $\dim(T')=m-r+\Delta$ by
\prettyref{lem:S-T-XS-XT}. As a result, a basis for the smaller of
the subspaces $S'$ and $T'$ can be communicated using $(m+\ell-2r+3\Delta)(\min\{m,\ell\}-r+\Delta)\lceil\log q\rceil$
bits, where the first factor is the dimension of the ambient space.
Altogether, the communication cost is at most
\begin{multline}
2+(2\max\{m,\ell\}-2r+3\Delta)(\min\{m,\ell\}-r+\Delta)\lceil\log q\rceil+1\\
=O((m-r+\Delta)(\ell-r+\Delta)\log q).\qquad\label{eq:comm-cost}
\end{multline}

We now analyze the correctness probability. To this end, we prove
the following claim.
\begin{claim}
\label{claim:correct-when-XY-well-behaved}The output of the protocol
is correct whenever the matrices $X,Y$ satisfy \prettyref{eq:preserve-S}\textendash \prettyref{eq:preserve-XT}
as well as the additional conditions
\begin{align}
 & \dim(X(S+T))\geq\min\{\dim(S+T),m+\ell-r\},\label{eq:preserve-S+T}\\
 & \dim(Y((X(S))^{\perp}+(X(T))^{\perp}))=\dim((X(S))^{\perp}+(X(T))^{\perp}).\label{eq:preserve-XS+XT}
\end{align}
\end{claim}

\begin{proof}
Recall from \prettyref{lem:S-T-XS-XT} that~\prettyref{eq:preserve-S}\textendash \prettyref{eq:preserve-XT}
force \prettyref{eq:Sprime-Tprime}, which in view of~\prettyref{eq:preserve-XS+XT}
simplifies to
\begin{align}
 & \dim(S'\cap T')=m+\ell-r+\Delta-\dim(X(S+T)).\label{eq:Sprime-Tprime-simpler}
\end{align}
We first consider the case $\dim(S\cap T)\leq r.$ Here $\dim(S+T)\geq m+\ell-r,$
which along with \prettyref{eq:preserve-S+T} implies that $\dim(X(S+T))\geq m+\ell-r.$
Substituting this lower bound into \prettyref{eq:Sprime-Tprime-simpler}
gives $\dim(S'\cap T')\leq\Delta$. As a result, $\Pi$ outputs the
correct value in this case. 

In the complementary case $\dim(S\cap T)\geq r+1,$ we have $\dim(S+T)\leq m+\ell-r-1$
and therefore also $\dim(X(S+T))\leq m+\ell-r-1.$ Substituting this
upper bound into~\prettyref{eq:Sprime-Tprime-simpler} gives $\dim(S'\cap T')\geq\Delta+1,$
showing that the output of $\Pi$ is correct in this case as well.
\end{proof}
Condition~\prettyref{eq:preserve-S+T} fails with probability at
most $4q^{-\Delta-1}$, by \prettyref{lem:random-proj-subspaces}
with $d=m+\ell-r+\Delta$ and $t=\min\{\dim(S+T),m+\ell-r\}-1$. Moreover,
conditioned on \prettyref{eq:preserve-S} and~\prettyref{eq:preserve-T},
one has 
\begin{align*}
\dim((X(S))^{\perp}+(X(T))^{\perp}) & \leq\dim((X(S))^{\perp})+\dim((X(T))^{\perp})\\
 & =2(m+\ell-r+\Delta)-\dim(X(S))-\dim(X(T))\\
 & \leq m+\ell-2r+2\Delta
\end{align*}
and hence \prettyref{eq:preserve-XS+XT} fails with probability at
most $4q^{-(m+\ell-2r+3\Delta)+(m+\ell-2r+2\Delta-1)}\leq4q^{-\Delta-1}$,
by \prettyref{lem:random-proj-subspaces} with $d=m+\ell-2r+3\Delta$
and $t=\dim((X(S))^{\perp}+(X(T))^{\perp})-1$. Since \prettyref{eq:preserve-S}\textendash \prettyref{eq:preserve-XT}
are simultaneously true with probability at least $1-16q^{-\Delta-1}$
(by \prettyref{lem:S-T-XS-XT}), we conclude that the six conditions
\prettyref{eq:preserve-S}\textendash \prettyref{eq:preserve-XT},
\prettyref{eq:preserve-S+T}, \prettyref{eq:preserve-XS+XT} hold
simultaneously with probability at least $1-16q^{-\Delta-1}-4q^{-\Delta-1}-4q^{-\Delta-1}=1-24q^{-\Delta-1}.$
Now Claim~\prettyref{claim:correct-when-XY-well-behaved} implies
that the described protocol $\Pi$ has error probability at most $24q^{-\Delta-1}.$
Since we calculated $\Pi$'s cost to be~\prettyref{eq:comm-cost},
we conclude that
\[
R_{24/q^{\Delta+1}}(\INTERSECT_{R}^{\mathbb{F},n,m,\ell})=O((m-r+\Delta)(\ell-r+\Delta)\log q).\qquad
\]
Taking $\Delta=\lfloor\log_{q}(24/\epsilon)\rfloor$ now settles~\prettyref{eq:subspace-upper-bound-high-accuracy}.
For the additional upper bound~\prettyref{eq:intersect-equality},
observe that $\INTERSECT_{R}^{\mathbb{F},n,m,\ell}$ for $m=\ell=R$
is the equality problem with domain $\Scal(\mathbb{F}^{n},m)\times\Scal(\mathbb{F}^{n},m).$
The claimed upper bound now follows because it is well-known~\cite[Chapter~3.3]{ccbook}
that the equality problem over any domain has $\epsilon$-error randomized
communication complexity $O(\log(1/\epsilon))$.
\end{proof}

\subsection{Communication upper bounds for large error\label{subsec:Communication-upper-bounds-large-error}}

To study the large-error regime, we recall a basic fact on vector
spaces.

\begin{prop}
\label{prop:relative-dimensions}Let $A,A'$ be subspaces such that
$A'\subseteq A.$ Then for any subspace $B,$
\begin{equation}
\dim(A\cap B)-\dim(A'\cap B)\leq\dim(A)-\dim(A').\label{eq:relative-dimensions}
\end{equation}
\end{prop}

\begin{proof}
Since $A\cap B+A'$ is a subspace of $A,$ we have $\dim(A\cap B+A')\leq\dim(A).$
Expanding the left-hand side yields $\dim(A\cap B)+\dim(A')-\dim(A\cap B\cap A')\leq\dim(A),$
which is clearly equivalent to~\prettyref{eq:relative-dimensions}.
\end{proof}
We now revisit the subspaces $S'$ and $T'$ in \prettyref{lem:S-T-XS-XT}
and study the distribution of $\dim(S'\cap T')$.
\begin{lem}
\label{lem:low-bias-expectation}Let $n,m,\ell,r,\Delta$ be nonnegative
integers with $r<\min\{m,\ell\}$ and $\max\{m,\ell\}\leq n.$ Fix
a finite field $\mathbb{F}$ with $q=|\mathbb{F}|$ elements, and
let $S\in\Scal(\mathbb{F}^{n},m)$ and $T\in\Scal(\mathbb{F}^{n},\ell)$
be given subspaces. Let $X\in\mathbb{F}^{(m+\ell-r+\Delta)\times n}$
and $Y\in\mathbb{F}^{(m+\ell-2r+3\Delta)\times(m+\ell-r+\Delta)}$
be uniformly random matrices. Let $Z$ be the indicator random variable
for the event that~\emph{\prettyref{eq:preserve-S}\textendash \prettyref{eq:preserve-XT}}
hold. Define $S'=Y((X(S))^{\perp})$ and $T'=Y((X(T))^{\perp})$.
Then:
\begin{enumerate}
\item \label{enu:intersection-small}$\Exp[Zq^{\dim(S'\cap T')}]\leq q^{\Delta}(1+8q^{-\Delta})^{2}$
whenever $\dim(S\cap T)\leq r;$
\item \label{enu:intersection-large}$\Exp[Zq^{\min\{\dim(S'\cap T'),\Delta+1\}}]\geq q^{\Delta+1}(1-16q^{-\Delta-1})$
whenever $\dim(S\cap T)\geq r+1.$
\end{enumerate}
\end{lem}

\begin{proof}
\prettyref{enu:intersection-small} Consider the random variables
\begin{align*}
A & =m+\ell-r-\min\{\dim(X(S+T)),m+\ell-r\},\\
B & =\dim((X(S))^{\perp}+(X(T))^{\perp})-\dim(Y((X(S))^{\perp}+(X(T))^{\perp})).
\end{align*}
Then the inequality
\begin{equation}
Zq^{\dim(S'\cap T')}\leq Zq^{A+B+\Delta}\label{eq:ZABDelta}
\end{equation}
is trivially true for $Z=0$ and follows from equation \prettyref{eq:Sprime-Tprime}
of~\prettyref{lem:S-T-XS-XT} for $Z=1.$ The hypothesis $\dim(S\cap T)\leq r$
implies that $\dim(S+T)=\dim(S)+\dim(T)-\dim(S\cap T)\geq m+\ell-r.$
As a result, applying \prettyref{lem:random-proj-subspaces} with
$d=m+\ell-r+\Delta$ and $T=m+\ell-r$ gives
\begin{equation}
\Exp_{X}[q^{A}]\leq1+8q^{-\Delta}.\label{eq:qA}
\end{equation}
Now, let $Z'$ be the indicator random variable for the event that~\eqref{eq:preserve-S}
and~\eqref{eq:preserve-T} hold. Then $Z'=1$ implies that 
\begin{align*}
\dim((X(S))^{\perp}+(X(T))^{\perp}) & \leq\dim((X(S))^{\perp})+\dim((X(T))^{\perp})\\
 & =2(m+\ell-r+\Delta)-\dim(X(S))-\dim(X(T))\\
 & \leq m+\ell-2r+2\Delta.
\end{align*}
As a result, \prettyref{lem:random-proj-subspaces} is applicable
with $d=m+\ell-2r+3\Delta$ and $T=\dim((X(S))^{\perp}+(X(T))^{\perp})$
and gives 
\begin{equation}
Z'\Exp_{Y}[q^{B}\mid X]\leq Z'(1+8q^{-\Delta}).\label{eq:qB}
\end{equation}
It remains to put these ingredients together:
\begin{align*}
\Exp[Zq^{\dim(S'\cap T')}] & \leq\Exp[Zq^{A+B+\Delta}]\\
 & \leq\Exp[Z'q^{A+B+\Delta}]\\
 & =q^{\Delta}\Exp_{X}q^{A}Z'\Exp_{Y}[q^{B}\mid X]\\
 & \leq q^{\Delta}\Exp_{X}q^{A}Z'(1+8q^{-\Delta})\\
 & \leq q^{\Delta}\Exp_{X}q^{A}(1+8q^{-\Delta})\\
 & \leq q^{\Delta}(1+8q^{-\Delta})^{2},
\end{align*}
where the first step uses~\prettyref{eq:ZABDelta}, the second step
is justified by~$Z\leq Z',$ the fourth step applies~\prettyref{eq:qB},
and the last step uses~\prettyref{eq:qA}.

\prettyref{enu:intersection-large} Assume now that $\dim(S\cap T)\geq r+1$.
For $Z=1,$ equation \prettyref{eq:Sprime-Tprime} of~\prettyref{lem:S-T-XS-XT}
gives
\begin{align*}
\dim(S'\cap T') & \geq m+\ell-r+\Delta-\dim(X(S+T))\\
 & \geq m+\ell-r+\Delta-\dim(S+T)\\
 & =m+\ell-r+\Delta-\dim(S)-\dim(T)+\dim(S\cap T)\\
 & =-r+\Delta+\dim(S\cap T)\\
 & \geq\Delta+1.
\end{align*}
Now 
\begin{align*}
\Exp[Zq^{\min\{\dim(S'\cap T'),\Delta+1\}}] & \geq q^{\Delta+1}\Exp[Z]\\
 & \geq q^{\Delta+1}(1-16q^{-\Delta-1}),
\end{align*}
where the second step uses~\prettyref{lem:S-T-XS-XT}.
\end{proof}
At last, we are in a position to prove our claimed communication upper
bound for the subspace intersection problem.
\begin{thm}[Large error]
\label{thm:subspace-upper-bound-large-error} Let $\mathbb{F}$ be
a finite field with $q=|\mathbb{F}|$ elements, and let $n,m,\ell,R$
be integers with $\max\{0,m+\ell-n\}<R\leq\min\{m,\ell\}.$ Then $\max\{m,\ell\}\leq n$
and 
\begin{equation}
R_{\frac{1}{2}-\frac{1}{16q^{m+\ell-2R+16}}}(\INTERSECT_{R}^{\mathbb{F},n,m,\ell})\leq2.\label{eq:large-error-subspace-upper-4bit}
\end{equation}
Furthermore, for each $\gamma\in[\frac{1}{3}q^{-(m+\ell-2R)/3},\frac{1}{3}],$
\begin{equation}
R_{(1-\gamma)/2}(\INTERSECT_{R}^{\mathbb{F},n,m,\ell})=O((\log_{q}\lceil\gamma q^{m-R}\rceil+1)(\log_{q}\lceil\gamma q^{\ell-R}\rceil+1)\log q).\label{eq:large-error-subspace-upper-gamma}
\end{equation}
If in addition $m=\ell=R,$ then
\begin{equation}
R_{1/3}(\INTERSECT_{R}^{\mathbb{F},n,m,\ell})=O(1).\label{eq:intersect-equality-1}
\end{equation}
\end{thm}

\begin{proof}
The hypothesis $\max\{0,m+\ell-n\}<R\leq\min\{m,\ell\}$ implies that
$m+\ell-n\leq\min\{m,\ell\}$, which is equivalent to $\max\{m,\ell\}\leq n.$
The bound~\prettyref{eq:intersect-equality-1} is immediate from~\prettyref{thm:subspace-upper-bound-small-error}. 

In the rest of the proof, define $r=R-1$. We will first settle~\prettyref{eq:large-error-subspace-upper-4bit}.
Let $\Delta$ be a nonnegative integer to be chosen later. Consider
the following protocol $\Pi'$. On input a pair of subspaces $S\in\Scal(\mathbb{F}^{n},m)$
for Alice and $T\in\Scal(\mathbb{F}^{n},\ell)$ for Bob, the parties
use their shared randomness to pick independent and uniformly random
matrices $X\in\mathbb{F}^{(m+\ell-r+\Delta)\times n}$ and $Y\in\mathbb{F}^{(m+\ell-2r+3\Delta)\times(m+\ell-r+\Delta)}$.
Alice and Bob compute $S'=Y((X(S))^{\perp})$ and $T'=Y((X(T))^{\perp}),$
respectively. Note that $S'$ and $T'$ are subspaces in an ambient
vector space $V$ of dimension $m+\ell-2r+3\Delta.$ Let $Z$ be the
indicator random variable for the event that the four conditions~\prettyref{eq:preserve-S}\textendash \prettyref{eq:preserve-XT}
hold. Alice and Bob use shared randomness to pick a uniformly random
vector $v\in V$. They output $1$ in the event that $Z=1$ and $v\in S'\cap T',$
and output a uniformly random $\pm1$ value otherwise. The communication
cost of this protocol is $2$ bits since Alice can privately verify
the conditions \prettyref{eq:preserve-S}, \prettyref{eq:preserve-XS},
and $v\in S',$ and likewise Bob can privately verify the conditions
\prettyref{eq:preserve-T}, \prettyref{eq:preserve-XT}, and $v\in T'$.
Observe further that
\[
\Exp[\Pi'(S,T)\mid X,Y]=Zq^{\dim(S'\cap T')-(m+\ell-2r+3\Delta)}.
\]
Passing to expectations over $X$ and $Y,$ we arrive at
\[
\Exp\Pi'(S,T)=q^{-(m+\ell-2r+3\Delta)}\Exp[Zq^{\dim(S'\cap T')}].
\]
Applying \prettyref{lem:low-bias-expectation}, we find that $\Pi'(S,T)$
has expectation at most $\alpha'=q^{-m-\ell+2r-2\Delta}(1+8q^{-\Delta})^{2}$
if $\dim(S\cap T)\leq r,$ and at least $\beta'=q^{-m-\ell+2r-2\Delta}q(1-16q^{-\Delta-1})$
if $\dim(S\cap T)\geq r+1.$ Taking $\Delta=7,$ one calculates that
$\beta'-\alpha'\geq q^{-m-\ell+2r-14}/2$. Now~\prettyref{prop:shift-probab}
implies that
\[
R_{\frac{1}{2}-\frac{1}{16q^{m+\ell-2r+14}}}(\neg\INTERSECT_{R}^{\mathbb{F},n,m,\ell})\leq2,
\]
which is equivalent to \prettyref{eq:large-error-subspace-upper-4bit}.

In what follows, we prove the remaining upper bound~\prettyref{eq:large-error-subspace-upper-gamma}.
Due to the symmetry between $m$ and $\ell$, we may assume without
loss of generality that
\begin{equation}
m\geq\ell.\label{eq:m-geq-ell}
\end{equation}
Let $k$ and $\Delta$ be nonnegative integers to be set later, where
\begin{equation}
1\leq k\leq\ell-r.\label{eq:k-restriction}
\end{equation}
We will adapt $\Pi'$ to obtain a new protocol $\Pi''$ that satisfies
the following inequalities for all subspaces $S,T\subseteq\mathbb{F}^{n}$
of dimension $m$ and $\ell,$ respectively:
\begin{align}
 & \Exp[\Pi''(S,T)\mid X,Y]\geq q^{-k-\Delta}Zq^{\min\{\dim(S'\cap T'),\Delta+1\}},\label{eq:ZPi-lower}\\
 & \Exp[\Pi''(S,T)\mid X,Y]\leq q^{-k-\Delta}Zq^{\dim(S'\cap T')}.\label{eq:ZPi-upper}
\end{align}
On input $S$ and $T,$ Alice and Bob in $\Pi''$ choose uniformly
random matrices $X$ and $Y$ as before. They then compute the indicator
random variable $Z,$ which is a function of $X,Y,S,T.$ If $Z=0,$
they output a uniformly random $\pm1$ value. Clearly,~\prettyref{eq:ZPi-lower}
and~\prettyref{eq:ZPi-upper} hold in this case.

In the complementary case $Z=1,$ Alice and Bob compute $S'=Y((X(S))^{\perp})$
and $T'=Y((X(T))^{\perp}),$ respectively. \prettyref{lem:S-T-XS-XT}
with $Z=1$ implies that $S'$ and $T'$ are subspaces of dimension
$\ell-r+\Delta$ and $m-r+\Delta$, respectively, in an ambient vector
space of dimension $m+\ell-2r+3\Delta.$ This makes it possible for
Bob to find a subspace $U$ of dimension $m-r+\Delta+k$ such that
$T'\subseteq U$, and send $U$ to Alice. An application of \prettyref{prop:relative-dimensions}
yields
\begin{equation}
\dim(S'\cap U)-\dim(S'\cap T')\leq\dim(U)-\dim(T')=k.\label{eq:mask-U}
\end{equation}
What Alice does next depends on the dimension of $S'\cap U.$ 
\begin{enumerate}
\item If $\dim(S'\cap U)\geq k+\Delta+1,$ then~\prettyref{eq:mask-U}
implies that $\dim(S'\cap T')\geq\Delta+1.$ Therefore, \prettyref{eq:ZPi-lower}
and~\prettyref{eq:ZPi-upper} amount to the requirement that the
protocol's output have expectation at least $q^{-k+1}$ and at most
$q^{-k-\Delta+\dim(S'\cap T')}\in[q^{-k+1},\infty).$ To meet this
requirement, Alice simply outputs a random $\pm1$ value with expectation
$q^{-k+1}$. 
\item If $\dim(S'\cap U)\leq k+\Delta,$ Alice identifies a $(k+\Delta)$-dimensional
subspace $S''$ with the property that $S'\cap U\subseteq S''\subseteq S'$,
which exists because $k+\Delta\leq\ell-r+\Delta$ due to~\prettyref{eq:k-restriction}.
She then picks a uniformly random vector $v\in S''$ and sends it
to Bob, who outputs $1$ if $v\in T'$ and a uniformly random $\pm1$
value otherwise. In this case, Alice and Bob's expected output is
\[
\frac{q^{\dim(S''\cap T')}}{q^{\dim(S'')}}=\frac{q^{\dim(S''\cap T')}}{q^{k+\Delta}}=\frac{q^{\dim(S''\cap U\cap T')}}{q^{k+\Delta}}=\frac{q^{\dim(S'\cap U\cap T')}}{q^{k+\Delta}}=\frac{q^{\dim(S'\cap T')}}{q^{k+\Delta}},
\]
where the second step uses $T'\subseteq U,$ the third step uses the
defining property $S'\cap U\subseteq S''\subseteq S'$ of the set
$S'',$ and the last step is valid due to $T'\subseteq U.$ This agrees
with~\prettyref{eq:ZPi-lower} and~\prettyref{eq:ZPi-upper}, which
require that the protocol's output have expectation between $q^{-k-\Delta}q^{\min\{\dim(S'\cap T'),\Delta+1\}}$
and $q^{-k-\Delta}q^{\dim(S'\cap T')}.$
\end{enumerate}
The proof of~\prettyref{eq:ZPi-lower} and~\prettyref{eq:ZPi-upper}
is now complete. 

Since $U$ has co-dimension $\ell-r+2\Delta-k$, it can be communicated
in the form of a basis for $U^{\perp}$ using $(\ell-r+2\Delta-k)(m+\ell-2r+3\Delta)\lceil\log q\rceil$
bits. The vector $v$ takes $(m+\ell-2r+3\Delta)\lceil\log q\rceil$
bits to send. In view of~\prettyref{eq:m-geq-ell}, we conclude that
\begin{equation}
\cost(\Pi'')=O((\ell-r+2\Delta-k+1)(m-r+\Delta)\lceil\log q\rceil+1).\label{eq:Pi-prime-prime-cost}
\end{equation}
Lastly, we will show that $\Pi''$ is a distinguisher for the subspace
intersection problem. For this, pass to expectations with respect
to $X$ and $Y$ in ~\prettyref{eq:ZPi-lower} and~\prettyref{eq:ZPi-upper}
to obtain
\begin{align}
 & \Exp\Pi''(S,T)\geq q^{-k-\Delta}\Exp[Zq^{\min\{\dim(S'\cap T'),\Delta+1\}}],\label{eq:ZPi-lower-1}\\
 & \Exp\Pi''(S,T)\leq q^{-k-\Delta}\Exp[Zq^{\dim(S'\cap T')}].\label{eq:ZPi-upper-1}
\end{align}
Now \prettyref{lem:low-bias-expectation} implies that $\Pi''(S,T)$
has expectation at most $\alpha''=q^{-k}(1+8q^{-\Delta})^{2}$ if
$\dim(S\cap T)\leq r,$ and at least $\beta''=q^{-k+1}(1-16q^{-\Delta-1})$
if $\dim(S\cap T)\geq r+1.$ Taking $\Delta=7,$ one calculates that
$\beta''-\alpha''\geq q^{-k}/2$. Now~\prettyref{eq:Pi-prime-prime-cost}
and \prettyref{prop:shift-probab} imply that
\begin{align}
R_{\frac{1}{2}-\frac{1}{16q^{k}}}(\neg\INTERSECT_{R}^{\mathbb{F},n,m,\ell}) & =O((\ell-r-k+1)(m-r+1)\log q)\label{eq:small-bias-bound-subspace-k}
\end{align}
for every positive integer $k\leq\ell-r.$

Let $\gamma\in[\frac{1}{3}q^{-(m+\ell-2R)/3},\frac{1}{3}]$ be given.
For $\gamma\in[q^{-4},\frac{1}{3}],$ one obtains \prettyref{eq:large-error-subspace-upper-gamma}
from the bound $R_{1/3}(\INTERSECT_{R}^{\mathbb{F},n,m,\ell})=O((\ell-R+1)(m-R+1)\log q)$
of \prettyref{thm:subspace-upper-bound-small-error}. For $\gamma\in[\frac{1}{3}q^{-(m+\ell-2R)/3},q^{-4}],$
setting $k=\min\{\lfloor\log_{q}(1/\gamma)\rfloor-3,\ell-r\}$ in~\prettyref{eq:small-bias-bound-subspace-k}
gives
\begin{align*}
R_{(1-\gamma)/2}(\neg\INTERSECT_{R}^{\mathbb{F},n,m,\ell}) & =O((\ell-r-\min\{\lfloor\log_{q}(1/\gamma)\rfloor-3,\ell-r\}+1)(m-r+1)\log q)\\
 & =O((\max\{\lceil\log_{q}(\gamma q^{\ell-r})\rceil,0\}+4)(m-r+1)\log q)\\
 & =O((\log_{q}\lceil\gamma q^{\ell-r}\rceil+1)(m-r+1)\log q)\\
 & =O((\log_{q}\lceil\gamma q^{\ell-R}\rceil+1)(m-R+1)\log q)\\
 & =O((\log_{q}\lceil\gamma q^{\ell-R}\rceil+1)(\log_{q}\lceil\gamma q^{m-R}\rceil+1)\log q),
\end{align*}
where the last step uses \prettyref{eq:m-geq-ell} and $\gamma\geq\frac{1}{3}q^{-(m+\ell-2R)/3}.$
This completes the proof of \prettyref{eq:large-error-subspace-upper-gamma}.
\end{proof}
\prettyref{thm:subspace-upper-bound-large-error} settles the \emph{randomized}
communication upper bounds of~\prettyref{thm:MAIN-subspace-intersection-general}
for the \emph{total} subspace intersection problem, and hence also
the \emph{quantum} communication upper bound for the \emph{promise}
subspace intersection problem.

\section*{Acknowledgments}

The authors are thankful to Alan Joel, Xiaoming Sun, Chengu Wang,
and David Woodruff for their feedback on an earlier version of this
manuscript. We are also grateful to Alan for useful discussions during
this work, and to David for suggesting the application of our communication
lower bounds to bilinear query complexity.

\bibliographystyle{siamplain}
\bibliography{refs,new-refs}

\appendix

\section{\label{sec:nonsingularity}Fourier spectrum of nonsingularity}

Throughout this section, the underlying field is $\mathbb{F}_{q}$
for an arbitrary prime power $q.$ The root of unity $\omega$ and
the notation $\omega^{x}$ for $x\in\mathbb{F}_{q}$ are as defined
in \prettyref{sec:Fourier}. The objective of this appendix is to
prove \prettyref{lem:sunwang}. Our proof is shorter and simpler than
the approach of Sun and Wang~\cite{sunwang12communication-linear},
who proved \prettyref{lem:sunwang} for fields of prime order. What
our proofs have in common is the following proposition~\cite{sunwang12communication-linear}.
\begin{prop}[Sun and Wang]
For any integers $n\geq1$ and $r\in\{0,1,\ldots,n\},$ 
\[
\Gamma_{n}(n,r)=\Exp_{X\in\Mcal_{n}}\omega^{X_{1,1}+\cdots+X_{r,r}}.
\]
\label{prop:gammatrace-n}
\end{prop}

\begin{proof}[Proof \emph{(due to Sun and Wang)}]
 Let $X,Y,Z$ be independent and uniformly random nonsingular matrices
of order $n$. Then by \prettyref{prop:random-matrices}\prettyref{enu:XAY},
the product $XI_{r}Y$ is a uniformly random matrix of rank $r$.
Therefore,
\begin{align*}
\Gamma_{n}(n,r) & =\Exp\omega^{\langle Z,XI_{r}Y\rangle}\\
 & =\Exp\omega^{\langle X\tr ZY\tr,I_{r}\rangle}\\
 & =\Exp\omega^{\langle X,I_{r}\rangle}\\
 & =\Exp\omega^{X_{1,1}+X_{2,2}+\cdots+X_{r,r}},
\end{align*}
where the second step uses~\prettyref{fact:inner-product-vs-trace}\prettyref{enu:inner-product-transfer},
and the third step is legitimate because $X\tr ZY\tr$ is a uniformly
random nonsingular matrix by \prettyref{prop:random-matrices}\prettyref{enu:XA-AX}.
\end{proof}
We are now ready to establish our result of interest.
\begin{lem*}[restatement of \prettyref{lem:sunwang}]
For any integers $n\geq1$ and $r\in\{0,1,\ldots,n\},$ 
\begin{align*}
\Gamma_{n}(n,r) & =\frac{(-1)^{r}q^{\binom{r}{2}}}{(q^{n}-1)(q^{n}-q)\cdots(q^{n}-q^{r-1})}.
\end{align*}
\end{lem*}
\begin{proof}
Consider independent random matrices $X$ and $L$ of order $n$,
where $X$ is a uniformly random nonsingular matrix and $L$ is a
uniformly random nonsingular \emph{lower-diagonal} matrix. By~\prettyref{prop:random-matrices}\prettyref{enu:XA-AX},
the product $XL$ is a uniformly random nonsingular matrix. Therefore,
\prettyref{prop:gammatrace-n} implies that
\begin{align*}
\Gamma_{n}(n,r) & =\Exp\omega^{\sum_{i=1}^{r}(XL)_{i,i}}.
\end{align*}

We will say that $X$ is \emph{nice} if $X_{i,j}=0$ for all $(i,j)$
pairs such that $i\in\{1,2,\ldots,r\}$ and $j>i.$ 
\begin{claim}
\label{claim:Gamma-nice-not-nice}One has
\[
\Exp_{L}\left[\omega^{\sum_{i=1}^{r}(XL)_{i,i}}\mid X\right]=\begin{cases}
(-1)^{r}(q-1)^{-r} & \text{if \ensuremath{X} is nice,}\\
0 & \text{otherwise.}
\end{cases}
\]
\end{claim}

This claim, to be proved shortly, implies that
\begin{equation}
\Gamma_{n}(n,r)=\frac{(-1)^{r}}{(q-1)^{r}}\Prob[X\text{ is nice}].\label{eq:Gamma-nice}
\end{equation}
The probability of the nonsingular matrix $X$ being nice is straightforward
to calculate: there are $q-1$ choices for the first row, $q(q-1)$
choices for the second row, $q^{2}(q-1)$ choices for the third row,
and so on up to row $r$, whence 
\begin{align*}
\Prob_{X}[X\text{ is nice}] & =\frac{\prod_{i=1}^{r}q^{i-1}(q-1)}{(q^{n}-1)(q^{n}-q)\cdots(q^{n}-q^{r-1})}.
\end{align*}
Making this substitution in~\prettyref{eq:Gamma-nice} completes
the proof.
\end{proof}
\begin{proof}[Proof of Claim~\emph{\prettyref{claim:Gamma-nice-not-nice}}.]
Conditioned on $X,$ the columns of $XL$ are independent random
variables. Therefore,
\begin{align}
\Exp_{L}\left[\omega^{\sum_{i=1}^{r}(XL)_{i,i}}\mid X\right] & =\prod_{i=1}^{r}\Exp_{L}\left[\omega^{(XL)_{i,i}}\mid X\right]\nonumber \\
 & =\prod_{i=1}^{r}\Exp_{L}\left[\omega^{\sum_{j=i}^{n}X_{i,j}L_{j,i}}\mid X\right].\label{eq:claim-nice-not-nice-proof-product}
\end{align}
The entries of $L$ are independent random variables, with the diagonal
entries distributed uniformly on $\mathbb{F}_{q}\setminus\{0\}$ and
the subdiagonal entries distributed uniformly on $\mathbb{F}_{q}.$
If $X$ is not nice, then $X_{i,k}\ne0$ for some $i\in\{1,2,\ldots,r\}$
and $k>i,$ which means that the corresponding summation $\sum_{j=i}^{n}X_{i,j}L_{j,i}$
is a uniformly random field element. This forces~\prettyref{eq:claim-nice-not-nice-proof-product}
to vanish, due to~\prettyref{eq:omegasum}. When $X$ is nice, on
the other hand, \prettyref{eq:claim-nice-not-nice-proof-product}
simplifies as follows:
\begin{align*}
\prod_{i=1}^{r}\Exp_{L}\left[\omega^{\sum_{j=i}^{n}X_{i,j}L_{j,i}}\mid X\right] & =\prod_{i=1}^{r}\Exp_{L}\left[\omega^{X_{i,i}L_{i,i}}\mid X\right]\\
 & =\prod_{i=1}^{r}\Exp_{a_{i}\in\mathbb{F}_{q}\setminus\{0\}}\omega^{a_{i}}\\
 & =\left(\frac{\sum_{a\in\mathbb{F}_{q}}\omega^{a}-1}{q-1}\right)^{r}\\
 & =\frac{(-1)^{r}}{(q-1)^{r}},
\end{align*}
where the second step is legitimate because $X_{i,i}$ is nonzero
and $L_{i,i}$ is a uniformly random nonzero field element, and the
last step uses~\prettyref{eq:omegasum}.
\end{proof}

\section{\label{sec:2party-to-multiparty}Multiparty lower bounds via symmetrization}

The purpose of this appendix is to prove \prettyref{prop:2party-multiparty},
which gives a generic method for transforming two-party communication
lower bounds for a class of problems into corresponding multiparty
lower bounds. Recall that we adopt the number-in-hand blackboard model
of multiparty communication, reviewed in the introduction. The notation
$R_{\epsilon}(F)$ stands for the $\epsilon$-error randomized communication
complexity of the two-party or multiparty problem $F.$ The \emph{cost}
of a protocol $\Pi,$ denoted $\cost(\Pi)$, is the total number of
bits written to the blackboard in the worst-case execution of $\Pi.$
\begin{prop*}[restatement of \prettyref{prop:2party-multiparty}]
 Let $(X,+)$ be a finite Abelian group, and let $f\colon X\to\{-1,1,*\}$
be a given function. For $t\geq2,$ let $F_{t}\colon X^{t}\to\{-1,1,*\}$
be the $t$-party communication problem given by $F_{t}(x_{1},x_{2},\ldots,x_{t})=f(x_{1}+x_{2}+\cdots+x_{t}).$
Then for all $t\geq2,$
\[
R_{1/6}(F_{t})\geq\frac{1}{12}tR_{1/3}(F_{2}).
\]
\end{prop*}
\begin{proof}
The proof uses the symmetrization technique of Phillips, Verbin, and
Zhang~\cite{NIH12symmetrization}. Let $\Pi$ be a randomized protocol
for $F_{t}$ with error probability $1/6$. We will use $\Pi$ to
construct a protocol for the two-party problem $F_{2}$ with error
probability $1/3$ and communication cost at most $\cost(\Pi)\cdot12/t$. 

The protocol for $F_{2}$ is as follows. On input $(a,b)\in X\times X,$
Alice and Bob use their shared randomness to pick uniformly random
elements $r_{1},r_{2},\ldots,r_{t-1}\in X$ and uniformly random integers
$i,j$ with $1\leq i<j\leq n.$ Let
\[
\mathbf{x}=(r_{1},r_{2},\ldots,r_{t-1},-r_{1}-r_{2}-\cdots-r_{t-1})+(0,\ldots,0,a,0,\ldots,0,b,0,\ldots,0),
\]
where the rightmost tuple has $a$ in the $i$-th component, $b$
in the $j$-th component, and zeroes everywhere else. Since the components
of \textbf{$\mathbf{x}$} sum to $a+b,$ we have
\begin{equation}
F_{t}(\mathbf{x})=F_{2}(a,b).\label{eq:F2-Ft}
\end{equation}
Moreover, $\mathbf{x}$ is a \emph{uniformly random} tuple whose components
sum to $a+b$ because the first $t-1$ components are distributed
independently and uniformly at random on $X$, whereas the sum of
the components is $a+b$. In particular, $\mathbf{x}$ and $(i,j)$
are independent random variables.

By construction, Alice knows all the components of $\mathbf{x}$ except
for the $j$-th, and Bob knows all the components except for the $i$-th.
This makes it possible for them to run $\Pi$ on $\mathbf{x},$ with
Alice simulating all the parties other than the $j$-th, and Bob simulating
all the parties other than the $i$-th. When $\Pi$ requires the $i$-th
party to speak, Alice sends his message to Bob, and analogously for
the $j$-th party. For $k=1,2,\ldots,t,$ consider the random variable
$C(\mathbf{x},k)$ defined as the total number of bits sent in $\Pi$
by the $k$-th party on input $\mathbf{x}$. Then the number of bits
exchanged by Alice and Bob is $C(\mathbf{x},i)+C(\mathbf{x},j).$
Using the independence of $\mathbf{x}$ and $(i,j),$ we can now bound
Alice and Bob's expected communication cost on input $(a,b)$ as follows:
\begin{equation}
\Exp[C(\mathbf{x},i)+C(\mathbf{x},j)]=\frac{2}{t}\Exp[C(\mathbf{x},1)+\cdots+C(\mathbf{x},t)]\leq\frac{2}{t}\cost(\Pi).\label{eq:Ft-cost}
\end{equation}

By~\prettyref{eq:F2-Ft}, the described two-party protocol computes
$F_{2}$ with the same error probability that $\Pi$ computes $F_{t}$,
namely, $1/6.$ Furthermore, by~\prettyref{eq:Ft-cost}, the \emph{expected
}communication cost of the two-party protocol on any given input is
at most $\cost(\Pi)\cdot2/t.$ By Markov's inequality, the probability
of Alice and Bob exchanging at least $\cost(\Pi)\cdot12/t$ bits is
at most $1/6.$ Therefore, one can obtain a protocol for $F_{2}$
with error $1/6+1/6=1/3$ by terminating the described protocol as
soon as $\lfloor\cost(\Pi)\cdot12/t\rfloor$ bits have been communicated.
\end{proof}

\end{document}